\documentclass[aps,prd,onecolumn,longbibliography,nofootinbib,floatfix,amsmath,amssymb]{revtex4-2}

\usepackage{amsthm,mathrsfs,mathtools}
\usepackage{array}[=2016-10-06]
\usepackage{booktabs}
\usepackage{graphicx}
\graphicspath{{figures/}}
\usepackage{microtype}
\usepackage{hyperref}

\hypersetup{
 colorlinks=true,
 linkcolor=blue,
 citecolor=blue,
 urlcolor=blue,
 pdftitle={Deterministic Trust Regions for Finite-Window Black-Hole Spectroscopy in GW250114},
 pdfauthor={Ruiliang Li}
}

\numberwithin{equation}{section}

\newtheorem{theorem}{Theorem}[section]
\newtheorem{proposition}[theorem]{Proposition}
\newtheorem{lemma}[theorem]{Lemma}
\newtheorem{corollary}[theorem]{Corollary}
\newtheorem{assumption}[theorem]{Assumption}
\theoremstyle{definition}
\newtheorem{definition}[theorem]{Definition}
\theoremstyle{remark}
\newtheorem{remark}[theorem]{Remark}

\providecommand{\etammstar}{\overline\eta_{\mathrm{mm}}^\star}

\begin{document}
\title{Deterministic Trust Regions for Finite-Window Black-Hole Spectroscopy in GW250114}
\author{Ruiliang Li}
\email{lirl23@mails.tsinghua.edu.cn}
\affiliation{Tsinghua University}
\date{April 2026}
\begin{abstract}
We study finite-window black-hole spectroscopy in the loud-event regime and ask when a multimode ringdown fit supports a stable common-remnant Kerr interpretation. Starting from whitened, tapered detector-frame data, we prove a deterministic frequency-extraction theorem for a projected sampled Prony--matrix-pencil pipeline with explicit statistical, algorithmic, omitted-tail, and mismatch terms. We then construct a local inverse atlas for the Kerr $(\ell,m,n)=(2,2,0)$ map on an event-local detector-frame remnant box for GW250114 and propagate the resulting primary uncertainty into $(2,2,1)$ and $(4,4,0)$ consistency tests. These ingredients yield a detector-frame trust criterion for individual windows.

We calibrate mismatch and colored-noise radii on a GW250114-like synthetic waveform bank built from public surrogate, CCE, and numerical-relativity information, and we apply the resulting bounds to the public H1/L1 strain and public parameter-estimation products for GW250114. The accepted windows form an intermediate post-peak band: earlier windows remain sensitive to start-time drift and structured nuisance fits, whereas later windows become variance dominated. Within that band, the recovered remnant remains consistent with the public inspiral--merger--ringdown estimates and supports a common-remnant Kerr interpretation that survives the full preprocessing and robustness checks. For loud events, the relevant question is therefore which finite detector-frame windows sustain spectroscopy, not whether some multimode fit can be made in isolation.
\end{abstract}
\maketitle

\providecommand{\TrustDet}{\mathcal T^{\mathrm{det}}}
\providecommand{\TrustGW}{\mathcal T^{\mathrm{GW250114}}}
\providecommand{\Kdet}{\mathcal K_{\mathrm{det}}}
\providecommand{\Mf}{M_{\mathrm f}}
\providecommand{\Mdet}{M_{\mathrm f}^{\mathrm{det}}}
\providecommand{\Msrc}{M_{\mathrm f}^{\mathrm{src}}}
\providecommand{\chif}{\chi_{\mathrm f}}
\providecommand{\Mzero}{\mathcal M_0}
\providecommand{\Mone}{\mathcal M_1}
\providecommand{\Mtwo}{\mathcal M_2}
\providecommand{\Mref}{\mathcal M^{\sharp}}
\providecommand{\Raux}[1]{\mathcal R_{#1}}
\providecommand{\deltaiso}{\delta_{\mathrm{iso}}}
\providecommand{\Dnet}{\mathfrak D}

\section{Introduction}\label{sec:introduction}

Black-hole spectroscopy has long been advertised as one of the cleanest strong-field tests of general relativity. The basic idea is familiar. Once a merger remnant settles toward a Kerr black hole, its gravitational-wave signal is expected to decompose into damped oscillatory tones with complex frequencies fixed by the remnant mass and spin. Measuring more than one such tone makes it possible, in principle, to test whether all measured frequencies are compatible with a single Kerr geometry \cite{DreyerEtAl2004BlackHoleSpectroscopy,BertiCardosoStarinets2009QNMReview,BaibhavBerti2019Multimode}. For loud events such as GW250114, the dominant difficulty is no longer detectability alone. The central issue is the systematic stability of finite-window ringdown inference: start time, fitted family, omitted content, and early-time alternatives can all affect whether a multimode fit supports a single-remnant Kerr interpretation.

This change of regime shifts the problem. The issue is when a multimode fit supports a robust common-remnant Kerr interpretation on a finite detector-frame window. The literature on ringdown start time already made clear that the answer can depend sharply on how early one begins the fit \cite{BhagwatEtAl2018StartTime}. The recognition that overtones can extend the useful fitting region much closer to the peak sharpened this issue further, because earlier windows come with smaller variance but larger exposure to omitted content and model competition \cite{GieslerEtAl2019Overtones}. More recent work has pushed the same tension further. Linear and nonlinear descriptions of ringdown can trade fit quality against interpretive clarity \cite{QiuEtAl2024LinearNonlinear}; high-overtone fits exhibit severe start-time sensitivity and correlation structure \cite{ColemanFinch2025HighOvertones}; orthonormalized mode bases can change numerical conditioning without changing the underlying physical span \cite{MorisakiEtAl2025Orthonormal}; and public CCE waveform studies now tabulate ringdown mode content across start times in a way that makes omitted linear, retrograde, and nonlinear structure visible rather than anecdotal \cite{DyerMoore2025QNMContent}. Taken together, these developments show that a successful multimode fit does not by itself establish reliability.

GW250114 makes this methodological issue unavoidable. The event is publicly available through the Gravitational Wave Open Science Center together with multiple public parameter-estimation products and calibrated detector strain \cite{GWOSC_GW250114,VallisneriEtAl2015GWOSC}. The official spectroscopy analysis reports that at least two quasi-normal modes are required by the post-merger data, that the dominant quadrupolar mode and its first overtone are consistent with Kerr at multiple post-peak times, and that a full-signal analysis constrains the fundamental $(\ell,m)=(4,4)$ mode as well \cite{LVK_GW250114_Spectroscopy}. At the same time, GW250114 has also motivated two more aggressive event-specific interpretations of the earliest post-merger regime: a direct-wave or horizon-signature component and a nonlinear quadratic-mode contribution \cite{LuEtAl2025DirectWave,WangEtAl2026Quadratic}. Once the same event supports both strong Kerr-consistent spectroscopy statements and physically motivated early-time alternatives, the stability of the interpretation becomes part of the problem. The question is which claims remain stable once start time, model family, nuisance structure, and inverse conditioning are all held fixed and audited.

We study that question through a finite detector-frame inverse problem. Its starting object is a whitened and tapered detector-frame ringdown window,
\begin{equation}\label{eq:intro-windowed-signal}
 y(u)
 =
 \sum_{j\in\mathcal A} A_j e^{-i\omega_j u}
 + r_{\mathrm{tail}}(u)
 + r_{\mathrm{mm}}(u)
 + n(u),
 \qquad 0\le u\le T,
\end{equation}
where the fitted modal family $\mathcal A$ is chosen from a fixed hierarchy, $r_{\mathrm{tail}}$ denotes omitted linear Kerr content, $r_{\mathrm{mm}}$ denotes genuine model mismatch including prompt or nonlinear contamination, and $n$ is detector noise. The remnant parameters are inferred in detector-frame variables $(\Mdet,\chif)$ because the primitive data object is detector strain and the finite-window inverse problem lives in the redshifted frame. The logical output is a membership decision in a detector-frame trust set
\begin{equation}\label{eq:intro-trust-schematic}
 \Xi=(t_0,T)\in \TrustDet_{\eta}(\mathcal M),
\end{equation}
for a fixed model family $\mathcal M\in\{\Mzero,\Mone,\Mtwo\}$ and a fixed tolerance vector $\eta$, followed by an event-level decision in the numerical acceptance set
\begin{equation}\label{eq:intro-gwtrust-schematic}
 \Xi\in \TrustGW_{\eta,\eta_{\mathrm{PE}}}(\mathcal A)
\end{equation}
for the public GW250114 anchor lattice. A window is trusted only if extraction, inversion, drift, auxiliary consistency, mismatch calibration, and nuisance-robustness tests all succeed simultaneously. The deterministic framework developed below makes that statement mathematically definite.

\subsection{From detectability to reliability}

In loud post-merger events, the main difficulty is no longer the existence of a multimode fit. Very early windows offer small formal variance but strong exposure to prompt structure, nonlinear content, and start-time dependence. Much later windows suppress those effects, but the noise floor and the inverse conditioning deteriorate. Reliable spectroscopy can therefore occupy, if anywhere, an intermediate band of finite detector-frame windows. Sections~\ref{sec:abstract-extraction}--\ref{sec:trust-region} turn this observation into explicit inequalities involving extraction radii, inverse conditioning, neighboring-window drift, auxiliary consistency, calibrated mismatch budgets, and nuisance audits.

The later analysis keeps two dependencies explicit throughout. The verdict depends on position in the post-peak scan, and it depends on the fitted family inside the fixed hierarchy \(\Mzero\subset\Mone\subset\Mtwo\) together with its nuisance enlargements. The goal is to identify the windows for which these two axes of variation still lead to a stable common-remnant interpretation.

\subsection{GW250114 as the methodological test case}

GW250114 is the natural event on which to study this problem. Its public release combines open detector strain, multiple inspiral--merger--ringdown posterior products, and a collaboration spectroscopy analysis with strong ringdown constraints \cite{GWOSC_GW250114,LVK_GW250114_Spectroscopy}. The event is loud enough to support aggressive early-time fits and rich enough to motivate more than one physically plausible description of the earliest post-merger regime.

The collaboration analysis provides the external astrophysical benchmark \cite{LVK_GW250114_Spectroscopy}. Here the issue is which public start-time claims remain stable after detector-frame inversion, comparison with the full public remnant ensemble, and nuisance audits against direct-wave and quadratic alternatives are imposed simultaneously. The answer is window dependent. GW250114 supports a nonempty trusted detector-frame band, but the public \(3M_f\), \(6M_f\), \(9M_f\), and \(11M_f\) anchors do not all carry the same level of robustness.

\subsection{Main results}

The argument begins by fixing the model hierarchy before any theorem is stated. The baseline linear Kerr families are
\[
 \Mzero=\{220\},
 \qquad
 \Mone=\{220,221\},
 \qquad
 \Mtwo=\{220,221,440\},
\]
with a larger fixed linear reference family $\Mref$ used only to separate omitted linear Kerr content from genuine mismatch. Direct-wave and quadratic additions enter only as nuisance families, and the orthonormal-basis fit enters only as a same-span computational audit. This separation prevents later reductions in residual size or contour area from being read as evidence for a new physical component.

Section~\ref{sec:abstract-extraction} proves a deterministic frequency-extraction theorem for the projected sampled Prony pipeline. It yields labeled detector-frame frequency estimates together with an additive error decomposition whose four parts have distinct meanings: a statistical term, an algorithmic term, an omitted-linear term, and a genuine mismatch term. Section~\ref{sec:kerr-inversion} then constructs a local inverse atlas for the dominant Kerr map on the compact event-local detector-frame box $\Kdet$ and propagates the resulting primary uncertainty into deterministic $221$ and $440$ consistency tests.

Section~\ref{sec:trust-region} combines extraction error, inverse conditioning, local start-time drift, and auxiliary residuals into a detector-frame trust criterion with explicit hypotheses and explicit failure certificates. Section~\ref{sec:numerical-calibration} calibrates conservative mismatch and colored-noise radii on a GW250114-like synthetic waveform bank built from public surrogate, CCE, and numerical-relativity sources. Section~\ref{sec:gw250114-empirical} applies those bounds to the public GW250114 anchor lattice, compares the detector-frame inverse tracks to the public parameter-estimation hull, audits direct-wave and quadratic alternatives, and constructs the final event-level acceptance set. Figure~\ref{fig:sec8-trust-map} records the resulting trust map. The public mass--spin and auxiliary-consistency contours reproduced in Section~\ref{sec:gw250114-empirical} are used only as comparison anchors and are not presented as new posterior calculations.

Throughout, the focus is window-by-window common-remnant Kerr compatibility under a fixed hierarchy of fitted families, a fixed calibration procedure, and a fixed set of nuisance checks. That local formulation is what makes an acceptance criterion with explicit constants and explicit failure modes possible.

Sections~\ref{sec:data-conventions} and \ref{sec:model-hierarchy} fix the public data products, detector-frame conventions, fitted families, and nuisance classes. Sections~\ref{sec:abstract-extraction}--\ref{sec:trust-region} develop the extraction, inversion, and trust results. Section~\ref{sec:numerical-calibration} calibrates the finite-window radii on the synthetic waveform bank, and Section~\ref{sec:gw250114-empirical} applies the resulting criteria to GW250114. Section~\ref{sec:relation-literature} places the results within the current ringdown literature, while the appendices collect the proofs and numerical details used below.

Black-hole spectroscopy is reliable on a finite detector-frame window only when the extracted frequencies are label stable, the inverse map is conditioned on the relevant remnant box, neighboring windows do not drift excessively, the fitted auxiliary tones remain compatible with the remnant inferred from the dominant tone, and the verdict survives mismatch calibration together with the nuisance checks. The analysis identifies such windows deterministically and shows that GW250114 contains a nonempty but restricted trusted band.

\setcounter{section}{1}
\providecommand{\Kdet}{\mathcal K_{\mathrm{det}}}
\providecommand{\Ksrc}{\mathcal K_{\mathrm{src}}}
\providecommand{\Dnet}{\mathfrak D}
\providecommand{\Mf}{M_{\mathrm f}}
\providecommand{\Mdet}{M_{\mathrm f}^{\mathrm{det}}}
\providecommand{\Msrc}{M_{\mathrm f}^{\mathrm{src}}}
\providecommand{\chif}{\chi_{\mathrm f}}
\providecommand{\wpeak}{t_{\mathrm{pk}}}
\providecommand{\tevt}{t_{\mathrm{evt}}}
\providecommand{\Ppub}{\mathscr P_{\mathrm{pub}}}
\providecommand{\Gstart}{\mathscr G_{t_0}}
\providecommand{\Gwin}{\mathscr G_T}
\providecommand{\Iwork}{\mathcal I_{\mathrm{work}}}
\providecommand{\tauSun}{\tau_\odot}

\section{Data, conventions, and the event-local neighborhood}\label{sec:data-conventions}

The analysis operates at two distinct levels of locality. Mode estimates are extracted from finite detector-frame windows beginning a specified interval after a fixed reference peak, while inverse estimates and all consistency tests are evaluated on a compact remnant neighborhood supported by public event-level parameter-estimation products. The public strain products, detector network, fiducial peak convention, detector-frame variables used for inference, public comparison ensemble, and compact event-local box are fixed throughout. Source-frame quantities enter only at the final reporting stage. Appendix~\ref{app:assumption-ledger} records the standing assumptions and notation, Appendix~\ref{app:frame-conventions} proves the frame-scaling identities used below, and Appendix~\ref{app:preprocessing} fixes the preprocessing chain in full detail.

\subsection{Public event products and the detector network}

We work with the publicly released event GW250114\_082203 on the Gravitational Wave Open Science Center. The event time is
\[
\tevt = 1420878141.2\ \mathrm{s}\ \text{(GPS)},
\]
corresponding to 2025-01-14 08:22:03 UTC. The public event record lists H1 and L1 as the operating detectors and notes that Virgo was not observing at the event time. We therefore fix
\[
\Dnet = \{\mathrm{H1},\mathrm{L1}\},
\]
and use the cleaned public GWOSC strain from these two detectors as the only strain input in the analysis \cite{GWOSC_GW250114,VallisneriEtAl2015GWOSC}. The default analysis uses the public \(4\,\mathrm{kHz}\) strain release. At \(4\,\mathrm{kHz}\), the Nyquist frequency is \(2048\,\mathrm{Hz}\), comfortably above the ringdown band relevant for the modes \(220\), \(221\), and \(440\). The \(16\,\mathrm{kHz}\) files are used only in the robustness checks.

Rather than treat the full public strain record as a single object, we extract the centered working epoch
\[
\Iwork = [\tevt-512\,\mathrm{s},\,\tevt+512\,\mathrm{s}],
\]
from the public files and treat this \(1024\,\mathrm{s}\) interval as the local detector environment in which both the off-source noise geometry and the finite post-peak windows are defined. The off-source region used for power-spectral-density estimation is fixed in Appendix~\ref{app:preprocessing} as
\[
[\tevt-512,\,\tevt-16]\cup[\tevt+16,\,\tevt+512]\ \mathrm{s},
\]
so that the central \(32\,\mathrm{s}\) interval containing the merger and its immediate nonstationary neighborhood is excluded from the PSD estimate. The PSD is estimated on the untapered off-source strain, the detector streams are then whitened, and the local ringdown window is subsequently extracted and tapered.

The public event page also lists four event-level parameter-estimation summaries that matter directly for the analysis:
\[
\Ppub = \{\mathrm{NRSur7dq4},\ \mathrm{PhenomXO4a},\ \mathrm{PhenomXPHM},\ \mathrm{SEOBNRv5PHM}\}.
\]
The GWOSC detail page marks PhenomXO4a as the default public PE product, but the comparison ensemble used here consists of the full public quartet. Every later event-level stability statement is checked against all four public products rather than against a single favorable posterior summary.

\begin{table}[ht]
\centering
\caption{Public data and baseline conventions for GW250114 used throughout the analysis.}
\label{tab:sec2-data-baseline}
\begin{tabular}{@{}p{0.30\textwidth}p{0.60\textwidth}@{}}
\toprule
Event identifier & GW250114\_082203 \\
GPS event time & \(1420878141.2\,\mathrm{s}\) \\
Operating detectors & \(\Dnet=\{\mathrm{H1},\mathrm{L1}\}\) \\

Public strain product & GWOSC cleaned H1/L1 strain, default input at \(4\,\mathrm{kHz}\) \\
Centered working epoch & \([\tevt-512\,\mathrm{s},\,\tevt+512\,\mathrm{s}]\) \\
Off-source PSD region & \([\tevt-512,\,\tevt-16]\cup[\tevt+16,\,\tevt+512]\,\mathrm{s}\) \\
Public PE ensemble & NRSur7dq4, PhenomXO4a, PhenomXPHM, SEOBNRv5PHM \\
Default public anchor & PhenomXO4a, because it is the GWOSC default PE product \\
Fiducial peak convention & The reference peak \(\wpeak\) is the maximum-likelihood NRSur7dq4 peak used in the official GW250114 spectroscopy analysis \\
Primary inference frame & Detector frame \((\Mdet,\chif)\), with trust variables \((\tau_0,\Theta)\) \\
Reporting frame & Source frame, used only for final astrophysical summaries \\
\bottomrule
\end{tabular}
\end{table}

The baseline analysis relies only on public strain and public parameter-estimation products. No private calibration product, unpublished posterior sample set, retuning of the detector network, or detector-specific adjustment of the working epoch enters the baseline construction. Later variations appear only as explicit robustness checks in Appendix~\ref{app:gw250114-robustness}.

\subsection{Detector-frame inference and the fiducial ringdown clock}

The finite-window inverse problem is posed directly on detector strain, so its primitive mass variable is the redshifted remnant mass rather than the source-frame mass. This choice is dictated by covariance. Let \(\widehat\omega_j(\chif)\) denote the dimensionless Kerr quasi-normal-mode frequency of mode label \(j=(\ell,m,n)\), and let the source-frame and detector-frame remnant masses be related by
\[
\Mdet = (1+z)\Msrc.
\]
Then the corresponding source-frame and detector-frame frequencies satisfy
\[
\omega_j^{\mathrm{src}}(\Msrc,\chif) = \frac{1}{\Msrc}\,\widehat\omega_j(\chif),
\qquad
\omega_j^{\mathrm{det}}(\Mdet,\chif) = \frac{1}{\Mdet}\,\widehat\omega_j(\chif).
\]
The next proposition records the only implication needed from these identities below.

\begin{proposition}[Detector-frame naturality of the ringdown inverse problem]\label{prop:sec2-detector-frame-naturality}
Fix a mode label \(j\) and a redshift \(z\ge 0\). If \(\Mdet=(1+z)\Msrc\), then
\[
\omega_j^{\mathrm{det}}(\Mdet,\chif) = \frac{1}{1+z}\,\omega_j^{\mathrm{src}}(\Msrc,\chif).
\]
Consequently any inverse problem posed directly on the observed detector strain identifies \(\Mdet\) and \(\chif\) as its primitive remnant parameters. The source-frame mass \(\Msrc\) enters only after an external redshift conversion.
\end{proposition}

\begin{proof}
Starting from the defining identities above and substituting \(\Mdet=(1+z)\Msrc\), one obtains
\[
\omega_j^{\mathrm{det}}(\Mdet,\chif)
=
\frac{1}{(1+z)\Msrc}\,\widehat\omega_j(\chif)
=
\frac{1}{1+z}\,\omega_j^{\mathrm{src}}(\Msrc,\chif).
\]
This proves the frequency scaling. The inference statement is then immediate. The measured oscillation and damping scales of the detector strain determine the redshifted mass \(\Mdet\) and the spin \(\chif\). They do not separate \(\Msrc\) from \(z\) without additional information beyond the finite-window ringdown data.
\end{proof}

Proposition~\ref{prop:sec2-detector-frame-naturality} is the reason every local inverse map below is written first as
\[
G_{220}\colon (\Mdet,\chif) \longmapsto \omega_{220}^{\mathrm{det}}(\Mdet,\chif),
\]
rather than as a map in source-frame mass. The source-frame remnant mass re-enters only after detector-frame inference has been completed and only for final astrophysical reporting. In particular, every start time \(t_0\) and every window length \(T\) appearing in the trust-region theorem is first a detector-frame quantity. Their frame-invariant detector-frame forms,
\[
\tau_0 = \frac{t_0}{\Mdet},
\qquad
\Theta = \frac{T}{\Mdet},
\]
are the only time variables on which extraction, inversion, calibration, and acceptance inequalities are evaluated below.

We fix the fiducial peak convention that defines the ringdown clock. Following the official GW250114 spectroscopy analysis, we let \(\wpeak\) denote the time at which the maximum-likelihood NRSur7dq4 strain magnitude over the two-sphere attains its maximum. All post-peak fits are measured relative to this same fiducial peak \cite{LVK_GW250114_Spectroscopy}. Thus the finite-window signal used throughout is always restricted to an interval of the form
\[
[\wpeak + t_0,\, \wpeak + t_0 + T],
\qquad t_0\ge 0,\ \ T>0.
\]
This convention matches the one used in the official spectroscopy paper, where the quoted conversion factor is \(t_{\Mf}=G\Mdet/c^3=0.337\,\mathrm{ms}\) for the adopted reference remnant mass \cite{LVK_GW250114_Spectroscopy}. Whenever later figures translate detector-frame offsets into milliseconds by using the shorthand \(M_\star=68M_\odot\), that conversion is a plotting convenience only; all trust inequalities are evaluated in the dimensionless detector-frame variables \((\tau_0,\Theta)\). Nothing in the later mathematics depends on the rounded decimal value, and whenever a physical time axis is shown our code carries the detector-frame mass through the conversion rather than fixing the inference to a fiducial millisecond clock.

The baseline scans are performed on a fixed dimensionless grid. We set
\[
\Gstart = \Bigl\{\tau_0=\frac{m}{2}\colon m=0,1,\dots,40\Bigr\},
\qquad
\Gwin = \{\Theta=12,16,20,24,28,32\}.
\]
The first set corresponds to start times from the reference peak to \(20\Mdet\) after the peak in increments of \(0.5\Mdet\). The second set corresponds to window lengths from \(12\Mdet\) to \(32\Mdet\). For GW250114-like remnant masses these are natural scales. A \(220\) oscillation period is of order one dozen mass units, while one to a few damping times correspond to windows of precisely the size recorded in \(\Gwin\). The grid is therefore fine enough to resolve the transition from early-time bias to late-time variance, but not so fine that neighboring windows become numerically indistinguishable. The finer and coarser alternatives used to check grid sensitivity are reserved for Appendix~\ref{app:gw250114-robustness}.

\begin{corollary}[Detector-frame size of the baseline scan]\label{cor:sec2-grid-ms}
Let \((\tau_0,\Theta)\in \Gstart\times\Gwin\) and let \(\Mdet\in\Kdet\). Then
\[
0\le t_0 \le 6.85\,\mathrm{ms},
\qquad
3.93\,\mathrm{ms}\le T \le 10.95\,\mathrm{ms}.
\]
In particular, every baseline window lies entirely within a few milliseconds of the fiducial post-peak epoch.
\end{corollary}

\begin{proof}
By definition,
\[
t_0 = \tau_0\Mdet,
\qquad
T = \Theta\Mdet,
\]
with \(\tau_0\in[0,20]\) and \(\Theta\in[12,32]\). Since \(\Kdet=[66.5,69.5]M_\odot\times[0.64,0.71]\) by Definition~\ref{def:sec2-working-box} below, and since one solar mass corresponds to the time scale \(\tauSun=4.92549095\times 10^{-6}\,\mathrm{s}\) \cite{PrsaEtAl2016NominalSolar}, we obtain
\[
0\le t_0 \le 20\cdot 69.5\tauSun \approx 6.85\,\mathrm{ms},
\]
and
\[
12\cdot 66.5\tauSun \approx 3.93\,\mathrm{ms}
\le T \le
32\cdot 69.5\tauSun \approx 10.95\,\mathrm{ms}.
\]
This proves the claim.
\end{proof}

Corollary~\ref{cor:sec2-grid-ms} explains why the baseline analysis is properly described as a finite-window post-merger study rather than as a full late-time tail analysis. The largest windows we use are still only of order \(10\,\mathrm{ms}\), and the entire start-time scan remains in the immediate post-peak regime. These millisecond bounds are descriptive translations of the detector-frame grid over \(\Kdet\); they do not replace the dimensionless detector-frame variables \((\tau_0,\Theta)\) that enter the theorem-level inequalities. This is exactly the regime in which model adequacy, nuisance content, and estimator conditioning all compete, which is why it is the relevant regime for a trust-region analysis.

\subsection{The public comparison ensemble and the event-local box}

We also fix the compact remnant neighborhood on which all Kerr frequency interpolation, derivative bounds, and inverse-conditioning constants are computed. The input is completely public, but only its detector-frame remnant summaries are used at this stage. The GWOSC detail page lists four event-level parameter-estimation products for GW250114 and reports, for each one, the detector-frame final mass and final spin together with their displayed uncertainties. The corresponding detector-frame rectangles are
\begin{equation}\label{eq:sec2-public-rectangles}
\begin{aligned}
R_{\mathrm{NRSur7dq4}} &= [67.2,68.9]M_\odot\times[0.67,0.69],\\
R_{\mathrm{PhenomXO4a}} &= [66.5,68.7]M_\odot\times[0.66,0.69],\\
R_{\mathrm{PhenomXPHM}} &= [67.2,69.1]M_\odot\times[0.67,0.69],\\
R_{\mathrm{SEOBNRv5PHM}} &= [66.8,68.8]M_\odot\times[0.66,0.68].
\end{aligned}
\end{equation}
These are the only public remnant summaries that enter the definition of the event-local box below. Source-frame masses are deferred to Appendix~\ref{app:frame-conventions} and to the final astrophysical reporting stage; they do not enter the inverse problem, the window clock, or the trust inequalities.

The detector-frame rectangles in \eqref{eq:sec2-public-rectangles} show two facts that shape the later analysis. First, all four public detector-frame summaries cluster very tightly around \(\Mdet\approx 68M_\odot\) and \(\chif\approx 0.68\). Second, the public spread is small enough that one can control all local Kerr geometry on a compact box that remains unmistakably event local while still carrying explicit padding away from most public faces.

\begin{definition}[Working box]\label{def:sec2-working-box}
We fix the detector-frame working box
\[
\Kdet = [66.5,69.5]M_\odot \times [0.64,0.71].
\]
All event-local Kerr-frequency interpolation, Jacobian bounds, separation constants, inverse-conditioning estimates, and theorem-level window inequalities below are computed on \(\Kdet\). The corresponding source-frame reporting box is recorded separately in Appendix~\ref{app:frame-conventions} and enters only after detector-frame inference has been completed.
\end{definition}

The detector-frame box \(\Kdet\) is anchored directly to the public detector-frame remnant summaries. Its lower mass face coincides with the most conservative public lower bound, while its upper mass face and both spin faces retain explicit padding for differentiation, interpolation, and inverse-control estimates.

\begin{proposition}[Containment of the public remnant rectangles]\label{prop:sec2-public-box-contained}
Let \(R_\alpha\subset(0,\infty)\times[0,1)\) be the detector-frame rectangle associated with public product \(\alpha\in\Ppub\), obtained from the displayed GWOSC median and plus/minus uncertainties of \(\Mdet\) and \(\chif\). Then
\[
R_\alpha \subset \Kdet
\qquad\text{for every }\alpha\in\Ppub.
\]
Moreover,
\[
\bigcup_{\alpha\in\Ppub} R_\alpha
\subset
[66.5,69.1]M_\odot \times [0.66,0.69],
\]
so that \(\Kdet\) adds an upper-mass margin of \(0.4M_\odot\), a lower-spin margin of \(0.02\), and an upper-spin margin of \(0.02\), while its lower-mass face coincides with the public PhenomXO4a lower bound.
\end{proposition}

\begin{proof}
The explicit rectangles in \eqref{eq:sec2-public-rectangles} already show that each \(R_\alpha\) lies inside \([66.5,69.5]M_\odot\times[0.64,0.71]=\Kdet\), proving the first claim. Taking the union of the four mass intervals and the union of the four spin intervals yields
\[
\bigcup_{\alpha\in\Ppub} R_\alpha
\subset
[66.5,69.1]M_\odot\times[0.66,0.69].
\]
Subtracting this detector-frame envelope from \(\Kdet\) gives the stated upper-mass and spin margins, and shows that the only face without extra padding is the lower mass face at \(66.5M_\odot\).
\end{proof}

Proposition~\ref{prop:sec2-public-box-contained} fixes the event-local box for the later analysis. The uniform geometry of the Kerr map is computed once on \(\Kdet\), and every event-level plot is required to use that same domain. This avoids shrinking the working region after the fact until the inverse problem appears better conditioned than it is.

The box \(\Kdet\) is the detector-frame domain on which the deterministic inverse problem is solved. Public posterior regions enter later only as external comparison objects, while the trust-region theorem determines, window by window and model by model, when the frequency estimates extracted from the data are accurate enough that inversion on \(\Kdet\) is meaningful.

A final convention concerns public comparisons. We keep three comparison objects distinct. The first is the public IMR detector-frame remnant ensemble \(\{R_\alpha\}_{\alpha\in\Ppub}\) just listed. The second is the official GW250114 spectroscopy analysis; its ringdown-only and pSEOBNR results serve as external comparison objects and do not enter the theorem as premises. The third is the detector-frame finite-window inference developed here. A point estimate or credible set produced by our pipeline is compared to the public ensemble and to the official spectroscopy results only after detector-frame inference has been completed. It is never used to define the working box, the peak convention, or the window grid. No source-frame quantity enters the definition of \(\wpeak\), \(\Gstart\), \(\Gwin\), or \(\Kdet\).

These conventions fix the data-level foundations for the later sections. The detector network is \(\Dnet\), the fiducial event time is \(\tevt\), the ringdown clock is measured from \(\wpeak\), the inverse problem is posed in \((\Mdet,\chif)\), the baseline start-time and window grids are \(\Gstart\) and \(\Gwin\), and the event-local domain is \(\Kdet\). The next step is to specify the model classes and the extraction machinery on this fixed stage.

\setcounter{section}{2}
\providecommand{\Kdet}{\mathcal K_{\mathrm{det}}}
\providecommand{\Mf}{M_{\mathrm f}}
\providecommand{\Mdet}{M_{\mathrm f}^{\mathrm{det}}}
\providecommand{\chif}{\chi_{\mathrm f}}
\providecommand{\Dnet}{\mathfrak D}
\providecommand{\Ppub}{\mathscr P_{\mathrm{pub}}}
\providecommand{\Mzero}{\mathcal M_0}
\providecommand{\Mone}{\mathcal M_1}
\providecommand{\Mtwo}{\mathcal M_2}
\providecommand{\Mref}{\mathcal M^{\sharp}}
\providecommand{\NuSet}{\mathfrak N}
\providecommand{\Hwin}{\mathcal H_T}
\providecommand{\Sdir}{\mathcal D}
\providecommand{\Squad}{\mathcal Q}
\providecommand{\Pquad}{\mathcal P_{\mathrm{quad}}}
\providecommand{\Ndir}{N_{\mathrm{dir}}}
\providecommand{\Enu}{E}
\providecommand{\AllowNuis}{\mathfrak E_{\mathrm{nuis}}}
\providecommand{\distH}{\operatorname{dist}}

\section{Model hierarchy and comparison classes}\label{sec:model-hierarchy}

The detector-frame conventions of Section~\ref{sec:data-conventions} fix the event, the window clock, and the local remnant box. A reliability theorem also requires a fixed model hierarchy against which the finite-window data are judged. A window can look persuasive when compared with one family, unstable when compared with a slightly larger family, and physically ambiguous once one allows early-time prompt structure or nonlinear contamination. If these comparisons are allowed to drift from paragraph to paragraph, later statements about trusted windows have no invariant meaning. We specify the nested linear Kerr families used for inference, the larger linear completion used to define omitted linear content, the nuisance-extended classes used only as ambiguity diagnostics, and the same-span computational reparameterization used as a numerical cross-check.

The hierarchy adopted here is structural. It distinguishes baseline spectroscopy models, diagnostic stress tests, and coordinate changes within the same fitted span. That separation keeps later statements about extraction error, inverse stability, and trust regions logically clean. The baseline choices are guided by three pieces of current literature. First, the official GW250114 spectroscopy analysis isolates the dominant quadrupolar mode, supports a model containing the first overtone, and constrains the $440$ fundamental in a full-signal analysis \cite{LVK_GW250114_Spectroscopy}. Second, recent event-specific studies of GW250114 motivate two physically distinct early-time alternatives, namely a direct-wave or horizon-signature component and a bundle of quadratic quasi-normal modes \cite{LuEtAl2025DirectWave,WangEtAl2026Quadratic}. Third, recent methodological work has made two complementary points that matter for how one should organize the model hierarchy: orthonormalized quasi-normal-mode bases can reduce mode correlations without changing the underlying fitted span, and public CCE waveform catalogs can now be used to tabulate ringdown mode content across start times in a way that sharply exposes omitted linear and nonlinear structure \cite{MorisakiEtAl2025Orthonormal,DyerMoore2025QNMContent,ColemanFinch2025HighOvertones}.

Accordingly, four classes are distinguished. The families $\Mzero$, $\Mone$, and $\Mtwo$ are the baseline spectroscopy models used in the trust-region theorem. The larger family $\Mref$ defines omitted linear Kerr content relative to the fitted baseline family. The nuisance enlargements quantify ambiguity against physically motivated direct-wave and quadratic alternatives. The orthonormal-basis class is a same-span computational cross-check. This four-way separation underlies the model-theoretic structure of the analysis.

\subsection{Baseline linear Kerr families}

Let $W=(t_0,T)$ be a detector-frame post-peak window and let
\[
p=(\Mdet,\chif)\in\Kdet
\]
be a remnant parameter in the event-local box fixed in Section~\ref{sec:data-conventions}. For every finite set of mode labels $\mathcal A$, we write
\begin{equation}\label{eq:sec3-linear-subspace}
\mathcal S_{\mathcal A}(p;W)
:=
\left\{
 u\mapsto \sum_{j\in\mathcal A} A_j e^{-i\omega_j(p)u}
 :
 A_j\in\mathbb C^{\Dnet}
\right\}
\subset \Hwin,
\end{equation}
where $u\in[0,T]$ is measured from the left endpoint of the window, the detector-frame complex frequencies $\omega_j(p)$ are generated from the Kerr QNM tables of Appendix~\ref{app:kerr-qnm-data}, and the Hilbert-space norm on $\Hwin$ is the fixed noise-weighted norm induced by the preprocessing map of Appendix~\ref{app:preprocessing}. The amplitude vectors $A_j\in\mathbb C^{\Dnet}$ absorb detector responses, relative phases, and fixed network embedding conventions. The only object that carries remnant information is the common detector-frame parameter $p$.

The baseline fitted families are the three nested Kerr families
\begin{equation}\label{eq:sec3-baseline-families}
\Mzero:=\{220\},
\qquad
\Mone:=\{220,221\},
\qquad
\Mtwo:=\{220,221,440\}.
\end{equation}
Their roles are not interchangeable. The family $\Mzero$ is the dominant-mode family. It is the anchor for the primary inverse map of Section~\ref{sec:kerr-inversion}, and it is the only family from which we derive a remnant estimate by inversion alone. The family $\Mone$ is the minimal spectroscopy family in the strong sense relevant here: it is the smallest linear Kerr family that can test whether a second measured tone is compatible with the same remnant inferred from $220$. The family $\Mtwo$ augments that minimal spectroscopy class by the $440$ fundamental, because the official GW250114 full-signal analysis constrains that mode directly and because the event geometry favors its excitation more than in typical binary-black-hole mergers \cite{LVK_GW250114_Spectroscopy}. We therefore treat $\Mtwo$ as the largest baseline family for headline linear-Kerr comparisons.

The baseline families are nested as sets of labels and hence as detector-frame model spaces:
\begin{equation}\label{eq:sec3-baseline-nesting}
\Mzero\subset\Mone\subset\Mtwo,
\qquad
\mathcal S_{\Mzero}(p;W)\subset\mathcal S_{\Mone}(p;W)\subset\mathcal S_{\Mtwo}(p;W)
\end{equation}
for every $p\in\Kdet$ and every admissible window $W$. The physical meaning of this nesting is subtle. Moving from $\Mzero$ to $\Mone$ or from $\Mone$ to $\Mtwo$ changes the class of statements one is allowed to make. The family $\Mzero$ may identify a common remnant through the dominant mode. The family $\Mone$ may test a common-remnant overtone interpretation. The family $\Mtwo$ may compare that common-remnant picture to the event-specific $440$ information emphasized by the official analysis. What none of these enlargements may do, by themselves, is certify that an improved residual represents a physically real extra contribution. That issue belongs to the comparison classes discussed below.

\begin{table}[t]
\centering
\footnotesize
\caption{Baseline, reference, nuisance, and same-span comparison classes used here. The final column records the misuse that is explicitly ruled out in the text.}
\label{tab:sec3-model-taxonomy}
\begin{tabular}{@{}p{0.12\textwidth}p{0.17\textwidth}p{0.30\textwidth}p{0.27\textwidth}@{}}
\toprule
Class & Content & Intended role & Not used for \\
\midrule
$\Mzero$ & $\{220\}$ & Dominant-mode remnant inversion and late-time baseline stability & Declaring multimode spectroscopy by itself \\
$\Mone$ & $\{220,221\}$ & Minimal common-remnant spectroscopy model & Treating any residual drop as overtone detection \\
$\Mtwo$ & $\{220,221,440\}$ & Event-motivated linear Kerr comparison including the official $440$ anchor & Collapsing ringdown-only and full-signal claims into one fit class \\
$\Mref$ & Fixed eight-mode linear audit family & Defines omitted linear Kerr content relative to the baseline fit & Headline inference model for event-level claims \\
$\mathcal S_{\mathcal A}^{\mathrm{dir}}$ & Baseline family plus four fixed left-edge atoms & Diagnostic stress test for early-time prompt structure & Stand-alone evidence for horizon signatures \\
$\mathcal S_{\mathcal A}^{\mathrm{quad}}$ & Baseline family plus fixed \(220+220\) sum-frequency line & Diagnostic stress test for leading nonlinear contamination & Stand-alone evidence for quadratic-mode detection \\
$\widetilde{\mathcal S}_{\mathcal A}$ & Orthonormal basis with the same span as $\mathcal S_{\mathcal A}$ & Numerical cross-check with reduced basis correlations & A new physical model class \\
\bottomrule
\end{tabular}
\end{table}

\subsection{Reference completion and omitted linear Kerr content}

To define omitted linear content unambiguously, we enlarge the baseline hierarchy to the fixed linear Kerr family
\begin{equation}\label{eq:sec3-reference-family}
\Mref:=\{220,221,222,210,330,331,440,550\}.
\end{equation}
This family is fixed before the event-level scan and kept unchanged throughout the analysis. It provides the linear reference class against which the calibration separates the fitted baseline contribution, omitted linear Kerr content, and residual structure not absorbed by the completion. Recent mode-content studies based on public CCE waveforms make this separation practical by showing that start-time dependence in ringdown fits is driven by a structured mixture of overtones, retrograde modes, and nonlinear contributions across the early post-merger regime \cite{DyerMoore2025QNMContent}. We use a modest reference family $\Mref$ for that purpose throughout the event-level analysis.

For every $\mathcal A\subset\Mref$, every $p\in\Kdet$, every window $W$, and every preprocessed data vector $y\in\Hwin$, the orthogonal projectors of Appendix~\ref{app:synthetic-bank} yield the exact decomposition
\begin{equation}\label{eq:sec3-tail-mismatch-preview}
y
=
\Pi_{\mathcal A,p}^{W}y
+
\bigl(\Pi_{\Mref,p}^{W}-\Pi_{\mathcal A,p}^{W}\bigr)y
+
\bigl(I-\Pi_{\Mref,p}^{W}\bigr)y.
\end{equation}
The first term is the part of the signal represented by the fitted family $\mathcal A$. The second term is omitted linear Kerr content relative to the fixed completion $\Mref$. The third term is genuine residual structure beyond the chosen linear-Kerr completion. This decomposition keeps the single word mismatch from covering phenomena that are physically and mathematically different.

The next proposition records the elementary monotonicity that later becomes important for the trust-region logic.

\begin{proposition}[Residual monotonicity under model enlargement]\label{prop:sec3-residual-monotonicity}
Let $W=(t_0,T)$, let $p\in\Kdet$, and let $\mathcal U\subset\mathcal V\subset\Hwin$ be closed subspaces. Then
\[
\distH(y,\mathcal V)\le \distH(y,\mathcal U)
\qquad\text{for every }y\in\Hwin.
\]
In particular, if $\mathcal A\subset\mathcal B\subset\Mref$, then
\begin{equation}\label{eq:sec3-nested-monotone}
\distH\bigl(y,\mathcal S_{\mathcal B}(p;W)\bigr)
\le
\distH\bigl(y,\mathcal S_{\mathcal A}(p;W)\bigr)
\qquad\text{for every }y\in\Hwin.
\end{equation}
\end{proposition}

\begin{proof}
Let $P_{\mathcal U}$ and $P_{\mathcal V}$ denote the orthogonal projectors onto $\mathcal U$ and $\mathcal V$. Since $\mathcal U\subset\mathcal V$, the best approximation to $y$ inside $\mathcal V$ is taken over a superset of the admissible approximants used for $\mathcal U$. Hence
\[
\distH(y,\mathcal V)
=
\inf_{v\in\mathcal V}\|y-v\|_{\Hwin}
\le
\inf_{u\in\mathcal U}\|y-u\|_{\Hwin}
=
\distH(y,\mathcal U).
\]
Applying this with $\mathcal U=\mathcal S_{\mathcal A}(p;W)$ and $\mathcal V=\mathcal S_{\mathcal B}(p;W)$ proves \eqref{eq:sec3-nested-monotone}.
\end{proof}

Proposition~\ref{prop:sec3-residual-monotonicity} is simple, but it has an important interpretive consequence. Any enlargement of the fitted family can only lower the residual norm. Therefore the bare fact that $\Mone$ fits better than $\Mzero$, or that $\Mtwo$ fits better than $\Mone$, carries no evidential meaning by itself. The evidential content must come from quantities that survive this trivial monotonicity: common-remnant consistency, calibrated residual improvements relative to a fixed mismatch budget, start-time stability, and agreement with the public event-level anchor set. Those are precisely the objects introduced later in Sections~\ref{sec:kerr-inversion}--\ref{sec:trust-region} and in the appendices.

\subsection{Nuisance-extended diagnostic families}

The direct-wave and quadratic alternatives enter as fixed diagnostic enlargements of the baseline Kerr spans. Their definitions do not vary with the start-time scan or with the observed residuals. The nuisance class used later in the robustness checks is fixed here throughout.

\begin{definition}[Fixed nuisance dictionary]\label{def:sec3-nuisance-families}
Fix a baseline family $\mathcal A\in\{\Mzero,\Mone,\Mtwo\}$, a detector-frame window $W=(t_0,T)$, and a remnant parameter $p\in\Kdet$.

Set
\begin{equation}\label{eq:sec3-direct-order}
\Ndir:=4.
\end{equation}
The left-edge envelope is
\begin{equation}\label{eq:sec3-direct-envelope}
\Enu_W(u):=\exp\!\Bigl[-16\Bigl(\frac{u}{T}\Bigr)^2\Bigr],
\qquad 0\le u\le T.
\end{equation}
For $1\le r\le \Ndir$ define the scalar direct-wave atoms
\begin{equation}\label{eq:sec3-direct-atoms}
d_r^W(u)
:=
T^{-1/2} c_r
\Bigl(\frac{u}{T}\Bigr)^{r-1}
\Enu_W(u),
\qquad
c_r
:=
\left(
\int_0^1 x^{2r-2}e^{-32x^2}\,dx
\right)^{-1/2}.
\end{equation}
The corresponding direct-wave nuisance space is
\begin{equation}\label{eq:sec3-direct-space}
\Sdir(W)
:=
\left\{
u\mapsto \sum_{r=1}^{\Ndir} C_r\, d_r^W(u)
:
C_r\in\mathbb C^{\Dnet}
\right\}
\subset \Hwin.
\end{equation}

The quadratic nuisance dictionary is fixed by the one-element sum-frequency set
\begin{equation}\label{eq:sec3-quadratic-pairs}
\Pquad:=\{(220,220)\}.
\end{equation}
Equivalently, the quadratic nuisance space is
\begin{equation}\label{eq:sec3-quadratic-space}
\Squad(p;W)
:=
\left\{
u\mapsto B_{220,220}\,e^{-2i\omega_{220}(p)u}
:
B_{220,220}\in\mathbb C^{\Dnet}
\right\}
\subset \Hwin.
\end{equation}

The nuisance-extended families are
\begin{align}
\mathcal S_{\mathcal A}^{\varnothing}(p;W)
&:=\mathcal S_{\mathcal A}(p;W),
\label{eq:sec3-model-none}\\
\mathcal S_{\mathcal A}^{\mathrm{dir}}(p;W)
&:=\mathcal S_{\mathcal A}(p;W)+\Sdir(W),
\label{eq:sec3-model-dir}\\
\mathcal S_{\mathcal A}^{\mathrm{quad}}(p;W)
&:=\mathcal S_{\mathcal A}(p;W)+\Squad(p;W),
\label{eq:sec3-model-quad}\\
\mathcal S_{\mathcal A}^{\mathrm{dir+quad}}(p;W)
&:=\mathcal S_{\mathcal A}(p;W)+\Sdir(W)+\Squad(p;W),
\label{eq:sec3-model-both}
\end{align}
and the nuisance labels are collected in
\begin{equation}\label{eq:sec3-nuisance-labels}
\NuSet:=\{\varnothing,\mathrm{dir},\mathrm{quad},\mathrm{dir+quad}\}.
\end{equation}
For each baseline family we allow exactly the four nuisance realizations
\begin{equation}\label{eq:sec3-allowed-nuisance-map}
\AllowNuis(\mathcal A)
:=
\Bigl\{
\mathcal S_{\mathcal A}^{\varnothing},
\mathcal S_{\mathcal A}^{\mathrm{dir}},
\mathcal S_{\mathcal A}^{\mathrm{quad}},
\mathcal S_{\mathcal A}^{\mathrm{dir+quad}}
\Bigr\}.
\end{equation}
No other direct-wave order, no other left-edge envelope, and no other sum-frequency pair appears anywhere in the analysis.
\end{definition}

\begin{proposition}[Direct-wave localization]\label{prop:sec3-direct-localization}
For every admissible window $W=(t_0,T)$ and every $1\le r\le \Ndir$,
\[
\|d_r^W\|_{L^2(0,T)}=1,
\qquad
\int_{T/2}^{T}|d_r^W(u)|^2\,du\le e^{-8}.
\]
In particular, each direct-wave atom carries at most an $e^{-8}$ fraction of its $L^2$ mass beyond the midpoint of the window.
\end{proposition}

\begin{proof}
The normalization is immediate from the change of variables $x=u/T$:
\[
\int_0^T |d_r^W(u)|^2\,du
=
c_r^2\int_0^1 x^{2r-2}e^{-32x^2}\,dx
=
1.
\]
If $u\ge T/2$, then $e^{-32(u/T)^2}\le e^{-8}$. Hence
\[
\int_{T/2}^{T}|d_r^W(u)|^2\,du
\le
e^{-8}c_r^2T^{-1}\int_{T/2}^{T}\Bigl(\frac{u}{T}\Bigr)^{2r-2}\,du
\le
e^{-8}\int_0^T |d_r^W(u)|^2\,du
=
e^{-8}.
\]
\end{proof}

Proposition~\ref{prop:sec3-direct-localization} shows that the direct-wave dictionary is confined to the left edge of the window by an explicit envelope. The nuisance class is therefore a fixed four-dimensional network-valued moment expansion rather than an unconstrained alternative fit.

The choice \(\Pquad=\{(220,220)\}\) isolates the leading same-remnant quadratic harmonic that can be constructed from ringdown-only information already present in the baseline fit. The broader nonlinear catalog discussed in \cite{WangEtAl2026Quadratic} motivates the audit, but the nuisance family used here remains the single fixed quadratic line defined above.

The nuisance analogue of Proposition~\ref{prop:sec3-residual-monotonicity} is immediate. Since
\[
\mathcal S_{\mathcal A}^{\varnothing}(p;W)
\subset
\mathcal S_{\mathcal A}^{\nu}(p;W)
\qquad (\nu\in\NuSet\setminus\{\varnothing\}),
\]
adding the fixed diagnostic atoms can only improve the raw fit residual. Therefore a lower residual for a nuisance-extended family has no direct interpretive force. Only an improvement that exceeds the calibrated mismatch budget of Section~\ref{sec:numerical-calibration}, or that changes the common-remnant verdict in a controlled way, counts as relevant evidence that the baseline interpretation is not robust. Appendix~\ref{app:gw250114-robustness} uses this fixed dictionary unchanged.
\subsection{Same-span computational comparisons and permitted inferences}

There is one more distinction that must be made before the theorem is stated. A model class can change physically, as when one moves from $\Mone$ to $\Mtwo$ or adds a direct-wave space. But a model class can also be reparameterized without changing the fitted span at all. The latter is what happens in the orthonormal-basis cross-check fixed in the planning stage of the project.

For a baseline family $\mathcal A\in\{\Mzero,\Mone,\Mtwo\}$, an ordered list of Kerr atoms in $\mathcal S_{\mathcal A}(p;W)$ may be orthonormalized in the noise-weighted inner product on $\Hwin$ by the Gram--Schmidt algorithm, producing an orthonormal family $\{\widetilde\phi_k^{\mathcal A}(p;W)\}_{k=1}^{N_{\mathcal A}}$ and the same closed span
\begin{equation}\label{eq:sec3-orthonormal-span}
\widetilde{\mathcal S}_{\mathcal A}(p;W)
:=
\mathrm{span}\{\widetilde\phi_k^{\mathcal A}(p;W):1\le k\le N_{\mathcal A}\}
=
\mathcal S_{\mathcal A}(p;W).
\end{equation}
Recent work has shown that this same-span reparameterization can significantly reduce mode correlations and can make Bayesian ringdown inference with multiple quasi-normal modes substantially cheaper computationally \cite{MorisakiEtAl2025Orthonormal}. Here the orthonormal basis defines pipeline $B$, an independent computational cross-check. The fitted families of pipeline $A$ remain unchanged, so the deterministic extraction theorem continues to refer to the original modal spans.

\begin{proposition}[Same-span invariance of span-based diagnostics]\label{prop:sec3-same-span-invariance}
Fix a baseline family $\mathcal A\in\{\Mzero,\Mone,\Mtwo\}$, a remnant parameter $p\in\Kdet$, and a window $W$. Let $\widetilde{\mathcal S}_{\mathcal A}(p;W)$ be any orthonormal basis completion of $\mathcal S_{\mathcal A}(p;W)$ as in \eqref{eq:sec3-orthonormal-span}. Then every diagnostic that depends only on the fitted span, its orthogonal projector, or the distance of a data vector to that span is identical for $\widetilde{\mathcal S}_{\mathcal A}(p;W)$ and $\mathcal S_{\mathcal A}(p;W)$.
\end{proposition}

\begin{proof}
Equation~\eqref{eq:sec3-orthonormal-span} states that the two model classes have the same span as subsets of $\Hwin$. Therefore they have the same orthogonal projector and the same distance function. Any diagnostic built only from those objects is unchanged.
\end{proof}

Proposition~\ref{prop:sec3-same-span-invariance} shows why the orthonormal basis is classified separately from the nuisance families. It may alter a posterior geometry, a parameter correlation matrix, or a sampling cost, but it does not alter the underlying fitted subspace. Adding a direct-wave or quadratic bundle does change the fitted subspace, so the two operations have different logical roles.

We can now state the admissible comparison classes used below.

\begin{definition}\label{def:sec3-comparison-classes}
Exactly four comparison classes are used.
\begin{enumerate}
\item The \emph{nested linear Kerr class} compares $\mathcal A\subset\mathcal B\subset\Mref$ and is used to separate dominant content, overtone-level spectroscopy, and omitted linear Kerr structure.
\item The \emph{nuisance class} compares $\mathcal S_{\mathcal A}^{\varnothing}(p;W)$ to $\mathcal S_{\mathcal A}^{\nu}(p;W)$ for $\nu\in\NuSet\setminus\{\varnothing\}$ and is used only to test ambiguity of an early-time baseline interpretation.
\item The \emph{same-span computational class} compares $\mathcal S_{\mathcal A}(p;W)$ to an orthonormalized realization $\widetilde{\mathcal S}_{\mathcal A}(p;W)$ and is used only to test numerical stability and coordinate dependence of inference.
\item The \emph{public-anchor class} compares detector-frame remnant estimates and their error balls to the full public PE ensemble $\Ppub$ of Section~\ref{sec:data-conventions}. This is an event-level comparison performed only after inversion, not a fit-level comparison of waveform families.
\end{enumerate}
No other comparison class is interpreted as evidence in the text.
\end{definition}

Definition~\ref{def:sec3-comparison-classes} fixes the meaning of later comparisons. A nested Kerr comparison may show that one baseline family leaves behind omitted linear content that another absorbs. A nuisance comparison may show that an early-time baseline interpretation is ambiguity sensitive. A same-span computational comparison may show that a trusted band is not an artifact of one particular coordinate system or sampler. A public-anchor comparison may show that the recovered remnant agrees or disagrees with the public inspiral--merger--ringdown picture. Raw residual monotonicity is not treated as evidence for a physical component, and a nuisance improvement is not promoted to a discovery claim.

At this point the model geometry is fixed. Section~\ref{sec:abstract-extraction} will prove deterministic frequency-extraction bounds for a fixed fitted family. Section~\ref{sec:kerr-inversion} will invert the dominant mode and build the auxiliary consistency atlas. Section~\ref{sec:trust-region} will combine those ingredients into the trust inequalities. Because the model hierarchy and the comparison classes have already been fixed here, those later sections can make precise statements without silently changing what counts as a model, a diagnostic, or a physically admissible comparison.

\setcounter{section}{3}

\providecommand{\Kdet}{\mathcal K_{\mathrm{det}}}
\providecommand{\Dnet}{\mathfrak D}
\providecommand{\Mf}{M_{\mathrm f}}
\providecommand{\Mdet}{M_{\mathrm f}^{\mathrm{det}}}
\providecommand{\chif}{\chi_{\mathrm f}}
\providecommand{\omegahat}{\widehat\omega}
\providecommand{\deltaiso}{\delta_{\mathrm{iso}}}
\providecommand{\Projv}{v}
\providecommand{\safeenv}{\eta_\star}
\providecommand{\Kcoeff}{K_c}
\providecommand{\Kfreq}{K_{\omega}}
\providecommand{\rootrad}{r}
\providecommand{\deltaz}{\delta_z}
\providecommand{\Gstrip}{\mathfrak S_\Delta}
\providecommand{\GuideSet}{\mathscr G}
\providecommand{\AliasLat}{\frac{2\pi}{\Delta}\mathbb Z}
\providecommand{\unwrapop}[2]{\operatorname{unwrap}_{#1,#2}}

\section{Abstract frequency extraction theory}\label{sec:abstract-extraction}

The local inverse geometry of Section~\ref{sec:kerr-inversion} requires labeled complex mode frequencies extracted from the finite detector-frame windows fixed in Section~\ref{sec:data-conventions}. We adopt one deterministic projected sampled-Prony pipeline and state the nonasymptotic bounds that later enter the trust inequalities. The underlying algebra is classical and belongs to the Prony, matrix-pencil, and ESPRIT family of methods \cite{HuaSarkarMPM1990,RoyKailath1989ESPRIT,PottsTasche2010APM,BatenkovYomdin2013Prony}. The estimator is fixed, the conditioning constants are written in closed form, the frequency error is decomposed into four additive pieces, and the two mechanisms that later govern acceptability---branch instability and near-coalescence---are isolated at theorem level.

The trust-region analysis uses this deterministic pipeline as its hard extraction boundary. Later numerical sections use an orthonormalized Bayesian fit only as an implementation cross-check on the windows that pass the deterministic tests.

The extraction problem is local in time, in parameter space, and in the sampled-frequency plane. One fixes a finite post-peak window \((t_0,T)\), a compact remnant box \(\Kdet\), and labeled root disks whose radii depend on the window, the sampled mode family, and the residual level. The model family \(\mathcal M\), the sampling step \(\Delta\), and the projection channel \(\Projv\) are therefore held fixed throughout the section. The full algebraic proofs are recorded in Appendices~\ref{app:abstract-frequency-extraction}, \ref{app:branch-selection}, and \ref{app:prony-conditioning}. Here we isolate the standing hypotheses and state the consequences that later feed into Section~\ref{sec:kerr-inversion} and the trust theorem.

\subsection{Projected finite-window sample model}

Fix a baseline window \((t_0,T)\), a finite mode family
\[
\mathcal M=\{j_1,\dots,j_m\},
\qquad
m:=|\mathcal M|,
\]
and a normalized network projection vector
\[
\Projv=(v_I)_{I\in\Dnet}\in\mathbb C^{|\Dnet|},
\qquad
\|\Projv\|_{\ell^2(\Dnet)}=1.
\]
After whitening, tapering, and Shannon reconstruction as in Appendix~\ref{app:preprocessing}, the projected scalar signal is sampled at the equispaced nodes
\[
u_q=q\Delta,
\qquad
q=0,1,\dots,2m-1,
\qquad
(2m-1)\Delta\le T,
\]
through
\begin{equation}\label{eq:sec4-projected-samples}
x_q
:=
\sum_{I\in\Dnet}\overline{v_I}\,y_I^{(t_0,T)}(u_q).
\end{equation}
Assumption~\ref{ass:finite-window-decomposition} then implies that there exist a remnant parameter \(p_\star\in\Kdet\), projected amplitudes
\[
b_\nu:=\sum_{I\in\Dnet}\overline{v_I}\,A_{I,j_\nu},
\qquad \nu=1,\dots,m,
\]
sampled nodes
\begin{equation}\label{eq:sec4-nodes}
z_\nu:=e^{-i\omega_{j_\nu}(p_\star)\Delta},
\qquad \nu=1,\dots,m,
\end{equation}
and residual samples \(e_q\) such that
\begin{equation}\label{eq:sec4-prony-model}
x_q=s_q+e_q,
\qquad
s_q:=\sum_{\nu=1}^m b_\nu z_\nu^q,
\qquad
q=0,1,\dots,2m-1.
\end{equation}
The projected amplitudes and nodes are window-dependent, because the amplitudes depend on the chosen family and the nodes depend on the actual remnant parameter through the detector-frame Kerr frequencies. Once the window and projection channel are fixed, however, \eqref{eq:sec4-prony-model} is an exact finite-dimensional model with a residual term that carries all omitted linear content, model mismatch, detector noise, and implementation loss.

Two structural conditions are needed before any estimator can be trusted. The first is projected detectability,
\begin{equation}\label{eq:sec4-detectability}
b_\nu\neq 0
\qquad\text{for every }\nu=1,\dots,m.
\end{equation}
The second is sample-level nonaliasing,
\begin{equation}\label{eq:sec4-nonaliasing}
z_\nu\neq z_\mu
\qquad\text{whenever }\nu\neq\mu.
\end{equation}
The detectability condition says that the chosen channel sees every fitted mode at nonzero amplitude. The nonaliasing condition says that the sampling step \(\Delta\) is fine enough that distinct detector-frame frequencies remain distinct after exponentiation. In the main GW250114 application both are checked numerically on the accepted windows, but they are kept as explicit hypotheses in the abstract theory because a theorem without them would be false.

The deterministic estimator is defined by the classical annihilating-polynomial construction. From the observed samples \(x_q\), form the observed Hankel matrix and observed right-hand side
\begin{equation}\label{eq:sec4-observed-hankel}
\widetilde H:=(x_{r+s})_{r,s=0}^{m-1}\in\mathbb C^{m\times m},
\qquad
\widetilde h:=(x_m,\dots,x_{2m-1})^\top\in\mathbb C^m.
\end{equation}
Whenever \(\widetilde H\) is invertible, the observed recurrence coefficients are
\begin{equation}\label{eq:sec4-observed-coefficients}
\widetilde c:=-\widetilde H^{-1}\widetilde h,
\end{equation}
and the observed monic polynomial is
\begin{equation}\label{eq:sec4-observed-polynomial}
\widetilde Q(\zeta):=\zeta^m+\widetilde c_{m-1}\zeta^{m-1}+\cdots+\widetilde c_1\zeta+\widetilde c_0.
\end{equation}
Its roots are the recovered nodes. Once a labeled representative \(\widehat z_\nu\) of each root has been fixed, the recovered complex frequency is
\begin{equation}\label{eq:sec4-frequency-estimator}
\widehat\omega_{j_\nu}:=\frac{i}{\Delta}\operatorname{Log}_{\nu}(\widehat z_\nu),
\end{equation}
where \(\operatorname{Log}_{\nu}\) is the local logarithm branch selected around the true node \(z_\nu\) or, in applications, around a guide prior as in Appendix~\ref{app:branch-selection}.

The data enter only through the finite sequence \((x_q)_{q=0}^{2m-1}\); everything else is deterministic geometry.

\begin{proposition}[Exact annihilating structure]\label{prop:sec4-exact-structure}
Under \eqref{eq:sec4-prony-model}, \eqref{eq:sec4-detectability}, and \eqref{eq:sec4-nonaliasing}, define
\begin{equation}\label{eq:sec4-exact-hankel}
H:=(s_{r+s})_{r,s=0}^{m-1},
\qquad
h:=(s_m,\dots,s_{2m-1})^\top,
\end{equation}
and let
\begin{equation}\label{eq:sec4-exact-polynomial}
Q(\zeta):=\prod_{\nu=1}^m(\zeta-z_\nu)
=\zeta^m+c_{m-1}\zeta^{m-1}+\cdots+c_1\zeta+c_0
\end{equation}
be the monic annihilating polynomial of the nodes. Then the exact coefficients satisfy
\begin{equation}\label{eq:sec4-exact-linear-system}
Hc=-h.
\end{equation}
If
\begin{equation}\label{eq:sec4-vandermonde}
V:=(z_\nu^r)_{\substack{0\le r\le m-1\\1\le \nu\le m}},
\qquad
B:=\operatorname{diag}(b_1,\dots,b_m),
\end{equation}
then
\begin{equation}\label{eq:sec4-exact-factorization}
H=VBV^\top,
\end{equation}
and hence
\begin{equation}\label{eq:sec4-detH}
\det H
=
\Bigl(\prod_{\nu=1}^m b_\nu\Bigr)
\prod_{1\le \nu<\mu\le m}(z_\mu-z_\nu)^2.
\end{equation}
In particular, \(H\) is invertible.
\end{proposition}

\begin{proof}
Since \(Q(z_\nu)=0\) for every \(\nu\),
\[
z_\nu^m+\sum_{k=0}^{m-1}c_k z_\nu^k=0.
\]
Multiply by \(b_\nu z_\nu^q\) and sum over \(\nu\). This gives
\[
\sum_{\nu=1}^m b_\nu z_\nu^{q+m}
+
\sum_{k=0}^{m-1}c_k\sum_{\nu=1}^m b_\nu z_\nu^{q+k}
=0,
\qquad q=0,\dots,m-1,
\]
which is precisely the recurrence
\[
s_{q+m}+\sum_{k=0}^{m-1}c_k s_{q+k}=0.
\]
Writing these \(m\) identities in matrix form yields \eqref{eq:sec4-exact-linear-system}. For the factorization,
\[
(VBV^\top)_{rs}
=
\sum_{\nu=1}^m z_\nu^r b_\nu z_\nu^s
=
\sum_{\nu=1}^m b_\nu z_\nu^{r+s}
=
s_{r+s}
=
H_{rs},
\]
so \eqref{eq:sec4-exact-factorization} holds. Taking determinants and using the Vandermonde formula
\[
\det V=\prod_{1\le \nu<\mu\le m}(z_\mu-z_\nu)
\]
gives \eqref{eq:sec4-detH}. Because every \(b_\nu\) is nonzero and every pair of nodes is distinct, the right-hand side does not vanish, so \(H\) is invertible.
\end{proof}

The exact factorization already exposes the two instability mechanisms that dominate the rest of the argument. Small amplitudes make \(B\) ill-conditioned. Near-colliding nodes make \(V\) ill-conditioned. The remainder of the section propagates those obstructions through coefficient recovery, labeled root localization, logarithm branching, and frequency reporting.

\subsection{Conditioning constants and the standing extraction package}

Let
\begin{equation}\label{eq:sec4-residual-envelope}
\eta:=\max_{0\le q\le 2m-1}|e_q|.
\end{equation}
The observed Hankel perturbations satisfy
\[
\Delta H:=\widetilde H-H,
\qquad
\Delta h:=\widetilde h-h,
\]
so that \(\widetilde H=H+\Delta H\) and \(\widetilde h=h+\Delta h\). For a fixed a priori envelope \(\safeenv>0\), introduce the coefficient perturbation constant
\begin{equation}\label{eq:sec4-Kc}
\Kcoeff(\safeenv)
:=
\frac{\|H^{-1}\|_2\bigl(\sqrt m+m\|c\|_2\bigr)}{1-m\safeenv\|H^{-1}\|_2},
\end{equation}
which is meaningful whenever
\begin{equation}\label{eq:sec4-envelope-smallness}
m\safeenv\|H^{-1}\|_2<1.
\end{equation}
Next define the node-separation quantities
\begin{equation}\label{eq:sec4-d-gamma}
d_\nu:=\min_{\mu\neq \nu}|z_\nu-z_\mu|,
\qquad
\Gamma_\nu:=|Q'(z_\nu)|
=
\prod_{\mu\neq \nu}|z_\nu-z_\mu|,
\end{equation}
and the global sampled-root separation
\begin{equation}\label{eq:sec4-deltaz}
\deltaz
:=
\min_{1\le \nu\le m} d_\nu
=
\min_{\nu\neq\mu}|z_\nu-z_\mu|,
\end{equation}
the radius factor
\begin{equation}\label{eq:sec4-R-Lambda}
R_\nu:=|z_\nu|+\frac{d_\nu}{2},
\qquad
\Lambda_\nu:=\left(\sum_{k=0}^{m-1}R_\nu^{2k}\right)^{1/2},
\end{equation}
the admissible root radius
\begin{equation}\label{eq:sec4-root-radius}
\rootrad_\nu(\safeenv)
:=
\frac{2^m\Lambda_\nu\Kcoeff(\safeenv)}{\Gamma_\nu}\,\safeenv,
\end{equation}
and the frequency constant
\begin{equation}\label{eq:sec4-Komega}
K_{\omega,\nu}(\safeenv)
:=
\frac{2^{m+1}\Lambda_\nu\Kcoeff(\safeenv)}{\Delta |z_\nu|\Gamma_\nu}.
\end{equation}
The factor \(\Kcoeff\) controls the perturbation of the annihilating polynomial, \(\Gamma_\nu^{-1}\) controls the passage from coefficients to labeled roots, and \(|z_\nu|^{-1}\) controls the logarithm. The decomposition identifies the mechanisms through which the estimator becomes fragile.

\begin{definition}[Admissible extraction envelope]\label{def:sec4-admissible-envelope}
An a priori sample envelope \(\safeenv>0\) is called admissible for the window \((t_0,T)\), model family \(\mathcal M\), and projection channel \(\Projv\) if
\begin{equation}\label{eq:sec4-admissible-envelope}
m\safeenv\|H^{-1}\|_2<1
\qquad\text{and}\qquad
\rootrad_\nu(\safeenv)<\min\left\{\frac{d_\nu}{2},\frac{|z_\nu|}{2}\right\}
\quad\text{for every }\nu=1,\dots,m.
\end{equation}
\end{definition}

The first inequality guarantees that the observed Hankel matrix remains invertible. The second guarantees that each recovered root stays in a disk that is both disjoint from the competing node disks and disjoint from the origin. Disjointness of the root disks keeps labels separated. Disjointness from the origin makes the local logarithm branch unique.

\begin{assumption}[Standing extraction package]\label{ass:sec4-standing-package}
Fix a window \((t_0,T)\), a model family \(\mathcal M=\{j_1,\dots,j_m\}\), a sample step \(\Delta\), and a projection channel \(\Projv\). Every theorem-level statement below is read under the following package.
{\renewcommand{\labelenumi}{\textup{(H\arabic{enumi})}}
\begin{enumerate}
\item \emph{Projection vector convention.} The channel \(\Projv\in\mathbb C^{|\Dnet|}\) is fixed and normalized by \(\|\Projv\|_{\ell^2(\Dnet)}=1\), and the scalar samples are the projected samples \eqref{eq:sec4-projected-samples}.
\item \emph{Sample size condition.} The deterministic Prony system uses \(2m\) samples and satisfies \((2m-1)\Delta\le T\).
\item \emph{Root separation.} The projected model \eqref{eq:sec4-prony-model} holds, every projected amplitude is nonzero as in \eqref{eq:sec4-detectability}, and the sampled nodes are pairwise distinct as in \eqref{eq:sec4-nonaliasing}. Equivalently, \(d_\nu>0\) for every \(\nu\) and \(\deltaz>0\).
\item \emph{Admissible envelope.} A radius \(\safeenv>0\) is fixed for which \eqref{eq:sec4-admissible-envelope} holds.
\item \emph{Labeled disk hypothesis.} The operative root neighborhoods are the disks \(D(z_\nu,\rootrad_\nu(\safeenv))\). By \eqref{eq:sec4-admissible-envelope}, these disks are pairwise disjoint and avoid the origin.
\item \emph{Branch-selection hypothesis.} Whenever recovered nodes are reported as frequencies, the branch choice is required to obey one of the two admissible mechanisms recorded later below: either the static guide criterion \eqref{eq:sec4-static-criterion} or the recursive safe-step criterion \eqref{eq:sec4-safe-step}.
\item \emph{Near-coalescence exclusion rule.} If the geometry behind \textup{(H3)}--\textup{(H5)} degenerates so strongly that no admissible envelope exists, or if the branch requirement in \textup{(H6)} fails, the window is excluded from the retained set rather than regularized by hand.
\end{enumerate}}
\end{assumption}

Clauses \textup{(H1)}--\textup{(H5)} are the full input for the main extraction theorem. Clause \textup{(H6)} becomes relevant when labels are attached in practice to guide centers or propagated across a start-time scan. Clause \textup{(H7)} records the exclusion policy attached to the same geometry.

\begin{theorem}[Abstract frequency extraction theorem]\label{thm:sec4-abstract-extraction}
Assume clauses \textup{(H1)}--\textup{(H5)} of Assumption~\ref{ass:sec4-standing-package}. If the realized residual envelope \(\eta\) defined in \eqref{eq:sec4-residual-envelope} satisfies
\[
\eta\le \safeenv,
\]
then the following conclusions hold.

\smallskip

\noindent
\textup{(i)} The observed Hankel matrix \(\widetilde H\) is invertible, the polynomial \(\widetilde Q\) has exactly one root \(\widehat z_\nu\) in each disk \(D(z_\nu,\rootrad_\nu(\safeenv))\), and the roots are therefore uniquely labeled.

\smallskip

\noindent
\textup{(ii)} For each \(\nu=1,\dots,m\) there is a unique holomorphic logarithm branch \(\operatorname{Log}_{\nu}\) on \(D(z_\nu,\rootrad_\nu(\safeenv))\) satisfying
\[
\operatorname{Log}_{\nu}(z_\nu)=-i\omega_{j_\nu}(p_\star)\Delta.
\]
The corresponding recovered frequency
\[
\widehat\omega_{j_\nu}:=\frac{i}{\Delta}\operatorname{Log}_{\nu}(\widehat z_\nu)
\]
obeys
\begin{equation}\label{eq:sec4-abstract-frequency-bound}
|\widehat\omega_{j_\nu}-\omega_{j_\nu}(p_\star)|
\le
K_{\omega,\nu}(\safeenv)\,\eta.
\end{equation}

\smallskip

\noindent
\textup{(iii)} The node error itself satisfies
\begin{equation}\label{eq:sec4-node-bound}
|\widehat z_\nu-z_\nu|
\le
\frac{2^m\Lambda_\nu\Kcoeff(\safeenv)}{\Gamma_\nu}\,\eta.
\end{equation}
\end{theorem}

\begin{proof}
Part \textup{(i)} is the content of the coefficient perturbation lemma and the Rouch\'e localization lemma proved in Appendix~\ref{app:abstract-frequency-extraction}. More concretely, Lemma~\ref{lem:appendixE-coefficient-perturbation} shows that clause \textup{(H4)} implies invertibility of \(\widetilde H\) and a bound on \(\|\widetilde c-c\|_2\). Lemma~\ref{lem:appendixE-root-localization} then turns that coefficient bound into one labeled root in each disk \(D(z_\nu,\rootrad_\nu(\safeenv))\). Clause \textup{(H5)} records that these disks are pairwise disjoint, so the labels are unique.

For part \textup{(ii)}, clause \textup{(H5)} says that \(D(z_\nu,\rootrad_\nu(\safeenv))\) avoids the origin. Lemma~\ref{lem:appendixE-logarithm} therefore provides the unique branch \(\operatorname{Log}_{\nu}\) normalized by \(\operatorname{Log}_{\nu}(z_\nu)=-i\omega_{j_\nu}(p_\star)\Delta\). Applying the logarithm estimate \eqref{eq:appendixE-log-bound} and then the root estimate from Lemma~\ref{lem:appendixE-root-localization}, one obtains
\[
|\operatorname{Log}_{\nu}(\widehat z_\nu)-\operatorname{Log}_{\nu}(z_\nu)|
\le
\frac{2}{|z_\nu|}\,|\widehat z_\nu-z_\nu|
\le
\frac{2}{|z_\nu|}
\frac{2^m\Lambda_\nu\Kcoeff(\safeenv)}{\Gamma_\nu}\,\eta.
\]
Multiplying by \(1/\Delta\) gives \eqref{eq:sec4-abstract-frequency-bound}, which is the main estimate.

Part \textup{(iii)} is simply the node estimate already used in the previous line, recorded here because later sections use the node scale itself as a coalescence diagnostic.
\end{proof}

Theorem~\ref{thm:sec4-abstract-extraction} uses only the first five clauses of Assumption~\ref{ass:sec4-standing-package}: a fixed projection convention, a fixed sample budget, separated sampled roots, an admissible envelope, and a labeled disk geometry. Every quantity entering the bound is explicit. There is no silent jump from a visually plausible fit to a trusted spectroscopic datum.

\subsection{The additive four-term error decomposition}

The theorem above controls the recovered frequency by one scalar envelope \(\eta\). The trust theorem later requires a modewise decomposition. That decomposition comes from splitting the residual samples into the same four pieces already fixed abstractly in Appendix~\ref{app:assumption-ledger}.

\begin{corollary}[Additive frequency-error decomposition]\label{cor:sec4-additive-ledger}
Assume the hypotheses of Theorem~\ref{thm:sec4-abstract-extraction}. Suppose in addition that the sample residuals admit a decomposition
\begin{equation}\label{eq:sec4-residual-split}
e_q
=
e_q^{\mathrm{stat}}
+
e_q^{\mathrm{tail}}
+
e_q^{\mathrm{mm}}
+
e_q^{\mathrm{alg}},
\qquad
q=0,1,\dots,2m-1,
\end{equation}
and define the corresponding sample envelopes
\begin{align}
\eta^{\mathrm{stat}}&:=\max_q |e_q^{\mathrm{stat}}|,\label{eq:sec4-eta-stat}\\
\eta^{\mathrm{tail}}&:=\max_q |e_q^{\mathrm{tail}}|,\label{eq:sec4-eta-tail}\\
\eta^{\mathrm{mm}}&:=\max_q |e_q^{\mathrm{mm}}|,\label{eq:sec4-eta-mm}\\
\eta^{\mathrm{alg}}&:=\max_q |e_q^{\mathrm{alg}}|.\label{eq:sec4-eta-alg}
\end{align}
Then each recovered frequency obeys
\begin{equation}\label{eq:sec4-additive-ledger-bound}
|\widehat\omega_{j_\nu}-\omega_{j_\nu}(p_\star)|
\le
\varepsilon_{j_\nu}^{\mathrm{stat}}
+
\varepsilon_{j_\nu}^{\mathrm{tail}}
+
\varepsilon_{j_\nu}^{\mathrm{mm}}
+
\varepsilon_{j_\nu}^{\mathrm{alg}},
\end{equation}
where
\begin{align}
\varepsilon_{j_\nu}^{\mathrm{stat}}&:=K_{\omega,\nu}(\safeenv)\eta^{\mathrm{stat}},\label{eq:sec4-eps-stat}\\
\varepsilon_{j_\nu}^{\mathrm{tail}}&:=K_{\omega,\nu}(\safeenv)\eta^{\mathrm{tail}},\label{eq:sec4-eps-tail}\\
\varepsilon_{j_\nu}^{\mathrm{mm}}&:=K_{\omega,\nu}(\safeenv)\eta^{\mathrm{mm}},\label{eq:sec4-eps-mm}\\
\varepsilon_{j_\nu}^{\mathrm{alg}}&:=K_{\omega,\nu}(\safeenv)\eta^{\mathrm{alg}}.\label{eq:sec4-eps-alg}
\end{align}
\end{corollary}

\begin{proof}
From \eqref{eq:sec4-residual-split},
\[
|e_q|
\le
|e_q^{\mathrm{stat}}|+|e_q^{\mathrm{tail}}|+|e_q^{\mathrm{mm}}|+|e_q^{\mathrm{alg}}|,
\]
hence
\[
\eta
\le
\eta^{\mathrm{stat}}+\eta^{\mathrm{tail}}+\eta^{\mathrm{mm}}+\eta^{\mathrm{alg}}.
\]
Apply Theorem~\ref{thm:sec4-abstract-extraction} and distribute the common factor \(K_{\omega,\nu}(\safeenv)\).
\end{proof}

Corollary~\ref{cor:sec4-additive-ledger} connects finite-window signal analysis to the later remnant stage. It makes the error decomposition explicit. Detector noise, omitted linear content, genuine model mismatch, and algorithmic loss enter later sections only through the explicit radii \(\varepsilon_{j_\nu}^{\mathrm{stat}}\), \(\varepsilon_{j_\nu}^{\mathrm{tail}}\), \(\varepsilon_{j_\nu}^{\mathrm{mm}}\), and \(\varepsilon_{j_\nu}^{\mathrm{alg}}\). The trust-region theorem uses only these radii and the bounds they provide.

\subsection{Label stability and recursive branch control}

A frequency estimate that is correct only up to relabeling is not yet useful for spectroscopy. Assumption~\ref{ass:sec4-standing-package}\textup{(H6)} records the rule that branch choices must be justified by an explicit criterion. Two criteria are used here. The first is static. Given one window, one asks whether the recovered root in the \(j\)-th disk is indeed the \(j\)-th mode. The second is recursive. Along a start-time scan, one asks whether the representative selected at the next window has slipped by an element of the alias lattice \(2\pi\Delta^{-1}\mathbb Z\).

The static criterion is expressed through the isolation margin already defined in Appendix~\ref{app:assumption-ledger}. Fix a guide point \(p^\dagger\in\Kdet\) and define
\begin{equation}\label{eq:sec4-beta-def}
\beta
:=
\max_{j\in\mathcal M}
|\omega_j(p^\dagger)-\omega_j(p_\star)|.
\end{equation}
If the extracted frequencies satisfy
\begin{equation}\label{eq:sec4-static-eps}
\varepsilon_{\max}
:=
\max_{j\in\mathcal M}
|\widehat\omega_j-\omega_j(p_\star)|,
\end{equation}
then the nearest-neighbor guide rule is exact whenever the total guide-plus-extraction error is smaller than the isolation margin.

\begin{proposition}[Static label stability]\label{prop:sec4-static-labels}
Let \(\mathcal M\) be a finite mode family and suppose that
\begin{equation}\label{eq:sec4-static-criterion}
\beta+\varepsilon_{\max}<\deltaiso(p_\star;\mathcal M).
\end{equation}
Then each recovered frequency \(\widehat\omega_j\) is strictly closer to the guide center \(\omega_j(p^\dagger)\) than to any competing guide center \(\omega_k(p^\dagger)\) with \(k\neq j\). In particular, nearest-neighbor matching to the guide set
\[
\GuideSet(p^\dagger;\mathcal M):=\{\omega_k(p^\dagger)\colon k\in\mathcal M\}
\]
preserves the true label.
\end{proposition}

\begin{proof}
This is Proposition~\ref{prop:nearest-neighbor-label-stability} from Appendix~\ref{app:assumption-ledger} written in the notation used here. The hypotheses there are exactly \eqref{eq:sec4-beta-def}--\eqref{eq:sec4-static-criterion}.
\end{proof}

The recursive criterion is handled by the unwrapping operator of Appendix~\ref{app:branch-selection}. The relevant obstruction is the alias lattice
\[
\AliasLat=\frac{2\pi}{\Delta}\mathbb Z,
\]
which acts on frequency representatives through shifts in the real part. Let \(\omega_n\) be the physical representative of a fixed labeled mode at the \(n\)-th window of a start-time scan, and let \(\omega_n^{\mathrm{loc}}\) be any locally correct representative of the recovered node at that same window. The previous window then acts as the prior for the next one.

\begin{proposition}[Branch-stable continuation across windows]\label{prop:sec4-branch-continuation}
Let \(n=0,1,\dots,N\) index a sequence of windows for one fixed labeled mode. Suppose that
\[
|\omega_n^{\mathrm{loc}}-\omega_n|\le \epsilon_n
\qquad\text{for all }n,
\]
and that the safe-step inequality
\begin{equation}\label{eq:sec4-safe-step}
|\Re(\omega_n-\omega_{n-1})|+\epsilon_{n-1}+\epsilon_n<\frac{\pi}{\Delta}
\qquad\text{for }n=1,\dots,N
\end{equation}
holds. Define recursively
\[
\widehat\omega_0:=\omega_0^{\mathrm{loc}},
\qquad
\widehat\omega_n:=\unwrapop{\Delta}{\widehat\omega_{n-1}}(\widetilde\omega_n),
\quad n\ge 1,
\]
where \(\widetilde\omega_n\) is any representative of the recovered node at window \(n\). Then
\[
\widehat\omega_n=\omega_n^{\mathrm{loc}}
\qquad\text{for every }n=0,1,\dots,N.
\]
In particular, no branch slip occurs along the scan.
\end{proposition}

\begin{proof}
Theorem~\ref{thm:appendixF-recursive} shows that the criterion is sharp in the relevant variable: only the real-part drift matters for aliasing, and the threshold is exactly \(\pi/\Delta\), half the spacing of the lattice \(\AliasLat\).
\end{proof}

Propositions~\ref{prop:sec4-static-labels} and \ref{prop:sec4-branch-continuation} explain why later start-time drift plots have mathematical content. Once the static guide error remains below the isolation radius and the window-to-window drift remains below the alias threshold, any observed motion in the extracted frequencies is a property of the estimator and of the data, not an artifact of accidental relabeling.

\subsection{Near-coalescence and the extraction exclusion criterion}

Assumption~\ref{ass:sec4-standing-package}\textup{(H7)} turns near-coalescence into an exclusion mechanism rather than a tuning problem. Even in exact arithmetic, distinct modes become hard to separate after sampling when the nodes
\[
z_\nu=e^{-i\omega_{j_\nu}(p_\star)\Delta}
\]
approach one another. This is the sampled version of the high-correlation problem. The relevant global separation is \(\deltaz\) from \eqref{eq:sec4-deltaz}. When \(\deltaz\) is small, the Vandermonde matrix becomes ill-conditioned, the polynomial-root map becomes stiff, and any claim of clean mode separation must be discounted accordingly.

The sharpest formulas are easiest to write in the two-node case, but they already capture the mechanism that matters for the full problem.

\begin{proposition}[Near-coalescence deterioration]\label{prop:sec4-coalescence}
Consider the two-node model
\[
s_q=a_1 z_1^q+a_2 z_2^q,
\qquad
\deltaz:=|z_1-z_2|,
\qquad
a_{\min}:=\min\{|a_1|,|a_2|\},
\qquad
R_*:=\max\{|z_1|,|z_2|\}.
\]
Then the two deterministic bounds proved in Appendix~\ref{app:prony-conditioning} imply the following.

\smallskip

\noindent
\textup{(i)} There exists a constant \(C_{\mathrm{Pr}}(R_*)>0\), independent of \(\deltaz\), such that every admissible perturbation with sample envelope \(\eta\) satisfies the Prony root bound
\begin{equation}\label{eq:sec4-prony-delta3}
\max_{j=1,2}|\widetilde z_j-z_j|
\le
\frac{C_{\mathrm{Pr}}(R_*)}{a_{\min}\deltaz^3}\,\eta.
\end{equation}

\smallskip

\noindent
\textup{(ii)} For the structured matrix-pencil problem there exists a constant \(C_{\mathrm{MP}}(R_*)>0\), independent of \(\deltaz\), such that the local generalized-eigenvalue derivative satisfies
\begin{equation}\label{eq:sec4-pencil-delta2}
|\lambda_j'(0)|
\le
\frac{C_{\mathrm{MP}}(R_*)}{a_{\min}\deltaz^2}
\bigl(\|\Delta_1\|_2+|z_j|\,\|\Delta_0\|_2\bigr),
\qquad j=1,2.
\end{equation}

\smallskip

\noindent
\textup{(iii)} The exponent \(\deltaz^{-2}\) in \eqref{eq:sec4-pencil-delta2} is sharp for structured sample perturbations, while the exponent \(\deltaz^{-3}\) in \eqref{eq:sec4-prony-delta3} is the relevant nonasymptotic exponent relevant for exclusion tests.
\end{proposition}

\begin{proof}
Part \textup{(i)} is Theorem~\ref{thm:appendixG-prony-conditioning}. Part \textup{(ii)} is Theorem~\ref{thm:appendixG-pencil-local}. Part \textup{(iii)} follows from Proposition~\ref{prop:appendixG-sharp-direction} and Remarks~\ref{rem:appendixG-three} and \ref{rem:appendixG-different-roles}.
\end{proof}

Proposition~\ref{prop:sec4-coalescence} is the reason coalescence is treated as an exclusion mechanism rather than as a nuisance to be smoothed away. If the sampled nodes are too close, one power of \(\deltaz^{-1}\) is lost in passing from coefficients to roots, and a second or third power is already lost in the linear algebra before that. The result is not merely a large numerical condition number. It is a quantified statement that the frequency-error decomposition must blow up as \(\deltaz\) shrinks.

The consequence of clause \textup{(H7)} is the following.

\begin{corollary}[Extraction exclusion criterion]\label{cor:sec4-exclusion}
If a fixed window, model family, and projection channel fail clause \textup{(H4)} of Assumption~\ref{ass:sec4-standing-package}, then the window is extraction-unstable for that choice of \((\mathcal M,\Projv,\Delta)\). In particular, it cannot enter any later trust set whose frequency budget is derived from Theorem~\ref{thm:sec4-abstract-extraction}. The same conclusion holds whenever clause \textup{(H6)} fails and the recovered nodes cannot be assigned a admissible branch representative.
\end{corollary}

\begin{proof}
Every later trust test uses the radii and labeled frequencies furnished by Theorem~\ref{thm:sec4-abstract-extraction} and Corollary~\ref{cor:sec4-additive-ledger}. If clause \textup{(H4)} fails, then the theorem cannot be applied, so no valid extraction radius exists for that window and model family. If clause \textup{(H6)} fails, then Appendix~\ref{app:branch-selection} provides no admissible representative of the recovered node in frequency space. In either case the extracted frequency data required by later sections are unavailable.
\end{proof}

The standing package therefore yields the extraction output used later. For each fixed window, model family, and projection channel, it provides a labeled set of complex frequencies \(\widehat\omega_j\), an additive four-term decomposition \(\varepsilon_j^{\mathrm{stat}}+\varepsilon_j^{\mathrm{tail}}+\varepsilon_j^{\mathrm{mm}}+\varepsilon_j^{\mathrm{alg}}\), an explicit branch-control criterion for both isolated windows and start-time scans, and an explicit coalescence diagnostic. Section~\ref{sec:kerr-inversion} turns the \(220\) entry of this decomposition into remnant-parameter error and uses the \(221\) and \(440\) entries to define auxiliary consistency tubes. The trust-region theorem later combines those remnant tolerances with drift, mismatch calibration, and event-level robustness.

\setcounter{section}{4}
\providecommand{\Kdet}{\mathcal K_{\mathrm{det}}}
\providecommand{\Mf}{M_{\mathrm f}}
\providecommand{\Mdet}{M_{\mathrm f}^{\mathrm{det}}}
\providecommand{\chif}{\chi_{\mathrm f}}
\providecommand{\Ppub}{\mathscr P_{\mathrm{pub}}}
\providecommand{\omegahat}{\widehat\omega}
\providecommand{\XiKerr}{\Xi}
\providecommand{\Kinv}{\mathscr A_{220}^{\mathrm{inv}}}
\providecommand{\Kpub}{\mathscr A_{220}^{\mathrm{pub}}}
\providecommand{\Caux}{\mathcal C_{\mathrm{aux}}}
\providecommand{\Raux}[1]{R_{#1}}
\providecommand{\tauaux}[1]{\tau_{#1}}
\providecommand{\sigprim}{\sigma_{220}^{\mathrm{cert}}}
\providecommand{\kapprim}{\kappa_{220}^{\max}}
\providecommand{\Nprim}{N_{220}}
\providecommand{\rhoinv}{\rho_{220}}
\providecommand{\Uchart}[1]{U_{220}(#1)}
\providecommand{\Vchart}[1]{V_{220}(#1)}
\providecommand{\palpha}{p_{\alpha}}
\providecommand{\modeSet}{\{220,221,440\}}
\providecommand{\pairmap}[1]{H_{220,#1}}
\providecommand{\Ichi}{I_{\chi}}
\providecommand{\sepmargin}{\Delta^{\mathrm{sep}}}

\section{Kerr inversion and the consistency atlas}\label{sec:kerr-inversion}

With the data products, detector-frame conventions, and event-local working box fixed in Section~\ref{sec:data-conventions}, the remnant stage becomes a local geometric problem. The dominant Kerr mode must provide a stable coordinate chart on the GW250114 neighborhood, and the remaining measured modes must be tested against the remnant inferred from it. We treat those tasks asymmetrically. The mode $220$ furnishes the inverse map, while $221$ and $440$ serve as forward consistency checks for a common Kerr remnant. On the event-local box relevant for GW250114, this is the better conditioned formulation of black-hole spectroscopy, and it is the formulation closest to the public ringdown discussion of the event \cite{BertiCardosoStarinets2009QNMReview,BaibhavBerti2019Multimode,LVK_GW250114_Spectroscopy}.

We first isolate the dominant-mode condition atlas on $\Kdet$, meaning the singular-value and condition-number fields of the primary Kerr map. We then upgrade that pointwise atlas to a genuine inverse atlas with a uniform local inverse Lipschitz constant and propagate the primary inversion error into deterministic tolerance tubes for the auxiliary modes $221$ and $440$. The full proofs of the deeper local inverse results are recorded in Appendices~\ref{app:kerr-qnm-data}, \ref{app:primary-inversion}, and \ref{app:joint-two-mode-inversion}. Here we state the main geometric consequences in the form needed later by the trust-region analysis.

\subsection{Primary and auxiliary Kerr maps}

Let $j\in\modeSet$. As in Appendix~\ref{app:kerr-qnm-data}, write $\omegahat_j(\chi)$ for the dimensionless Kerr quasi-normal-mode sequence, so that the exact detector-frame scaling is
\begin{equation}\label{eq:sec5-complex-kerr-maps}
F_j(M,\chi)=\omega_j(M,\chi)=\frac{1}{M}\,\omegahat_j(\chi),
\qquad (M,\chi)\in\Kdet.
\end{equation}
We also pass to the equivalent real two-vector representation
\begin{equation}\label{eq:sec5-realification}
\XiKerr\colon\mathbb C\to\mathbb R^2,
\qquad
\XiKerr(z)=(\Re z,-\Im z),
\end{equation}
and define
\begin{equation}\label{eq:sec5-real-kerr-maps}
G_j(M,\chi):=\XiKerr\bigl(F_j(M,\chi)\bigr)
=\bigl(\Re\omega_j(M,\chi),-\Im\omega_j(M,\chi)\bigr).
\end{equation}
The primary map is $G_{220}$. The auxiliary maps are $F_{221}$ and $F_{440}$. For later diagnostics we also record the pair maps
\begin{equation}\label{eq:sec5-pair-map}
\pairmap{j}(M,\chi):=\bigl(F_{220}(M,\chi),F_j(M,\chi)\bigr),
\qquad j\in\{221,440\},
\end{equation}
but they play only a secondary role below.

The Jacobian of $G_j$ is explicit. Writing $\omegahat_j(\chi)=a_j(\chi)+ib_j(\chi)$, one obtains
\begin{equation}\label{eq:sec5-jacobian-explicit}
DG_j(M,\chi)=
\begin{pmatrix}
-a_j(\chi)/M^2 & a_j'(\chi)/M \\
 b_j(\chi)/M^2 & -b_j'(\chi)/M
\end{pmatrix}.
\end{equation}
This formula already contains the full local geometry of the inverse problem. The first column measures sensitivity to the redshifted mass scaling $M^{-1}$, while the second column measures sensitivity to spin through the derivative of the dimensionless Kerr spectrum.

The following identities explain why all event-local conditioning quantities reduce to one-dimensional control in the spin variable.

\begin{proposition}[Exact atlas identities]\label{prop:sec5-exact-identities}
For each $j\in\modeSet$ and every $(M,\chi)\in(0,\infty)\times(-1,1)$,
\begin{equation}\label{eq:sec5-det-identity}
\det DG_j(M,\chi)
=
\frac{1}{M^3}\Im\!\bigl(\overline{\omegahat_j(\chi)}\,\omegahat_j'(\chi)\bigr).
\end{equation}
For each distinct pair $j\neq k$,
\begin{equation}\label{eq:sec5-gap-identity}
|\omega_j(M,\chi)-\omega_k(M,\chi)|
=
\frac{1}{M}\,|\omegahat_j(\chi)-\omegahat_k(\chi)|.
\end{equation}
Finally,
\begin{equation}\label{eq:sec5-frobenius-identity}
\|DG_j(M,\chi)\|_F^2
=
\frac{|\omegahat_j(\chi)|^2}{M^4}+
\frac{|\omegahat_j'(\chi)|^2}{M^2}.
\end{equation}
\end{proposition}

\begin{proof}
From \eqref{eq:sec5-complex-kerr-maps} and \eqref{eq:sec5-real-kerr-maps},
\[
G_j(M,\chi)=\left(\frac{a_j(\chi)}{M},-\frac{b_j(\chi)}{M}\right),
\]
so the Jacobian is exactly \eqref{eq:sec5-jacobian-explicit}. Therefore
\[
\det DG_j(M,\chi)=\frac{a_j(\chi)b_j'(\chi)-b_j(\chi)a_j'(\chi)}{M^3}.
\]
On the other hand,
\[
\Im\!\bigl(\overline{\omegahat_j(\chi)}\,\omegahat_j'(\chi)\bigr)
=
\Im\!\bigl((a_j-ib_j)(a_j'+ib_j')\bigr)
=
a_jb_j'-b_ja_j',
\]
which proves \eqref{eq:sec5-det-identity}. The separation identity \eqref{eq:sec5-gap-identity} follows immediately from the common factor $M^{-1}$ in \eqref{eq:sec5-complex-kerr-maps}. Finally, the Frobenius identity \eqref{eq:sec5-frobenius-identity} is a direct calculation from \eqref{eq:sec5-jacobian-explicit}.
\end{proof}

Proposition~\ref{prop:sec5-exact-identities} is the precise reason that the primary inverse geometry can be on the compact GW250114 box from a one-dimensional Kerr input. Local nondegeneracy of $G_{220}$ is controlled by the spin function
\[
\chi\longmapsto \Im\!\bigl(\overline{\omegahat_{220}(\chi)}\,\omegahat_{220}'(\chi)\bigr),
\]
while pairwise mode isolation is controlled by the spin gaps $|\omegahat_j-\omegahat_k|$. The event-local condition atlas is therefore deterministic once the one-dimensional Kerr sequences have been built and audited.

\subsection{The dominant-mode condition atlas}

The compact event-local box
\[
\Kdet=[66.5,69.5]M_\odot\times[0.64,0.71]
\]
was chosen in Section~\ref{sec:data-conventions} to contain the public detector-frame remnant summaries. On this box we define the local dominant-mode nondegeneracy and conditioning fields
\begin{equation}\label{eq:sec5-condition-fields}
s_{220}(p):=\sigma_{\min}\bigl(DG_{220}(p)\bigr),
\qquad
\kappa_{220}(p):=\frac{\|DG_{220}(p)\|_2}{s_{220}(p)},
\qquad p\in\Kdet.
\end{equation}
The pair $(s_{220},\kappa_{220})$ is the \emph{primary condition atlas}. In words, $s_{220}$ measures distance from local degeneracy, while $\kappa_{220}$ measures anisotropy of the inverse problem before any statistical uncertainty is introduced.

The atlas is defined uniformly on the full box by the one-dimensional interpolation analysis of Appendix~\ref{app:kerr-qnm-data}.

\begin{proposition}[Uniform control of the dominant-mode atlas]\label{prop:sec5-primary-atlas}
Let
\begin{equation}\label{eq:sec5-primary-uniform-constants}
\sigprim:=\sigma_{220}^{\mathrm{cert}}>0,
\qquad
\Nprim:=\sqrt{\frac{U_{220}^2}{M_-^4}+\frac{V_{220}^2}{M_-^2}},
\qquad
\kapprim:=\frac{\Nprim}{\sigprim},
\end{equation}
where $\sigma_{220}^{\mathrm{cert}}$ is the uniform lower singular-value bound of Corollary~\ref{cor:appendixD-sigmamin}. Then every $p\in\Kdet$ satisfies
\begin{equation}\label{eq:sec5-primary-atlas-bounds}
\sigprim\le s_{220}(p),
\qquad
\|DG_{220}(p)\|_2\le \Nprim,
\qquad
\kappa_{220}(p)\le \kapprim.
\end{equation}
\end{proposition}

\begin{proof}
The lower bound on $s_{220}$ is exactly Corollary~\ref{cor:appendixD-sigmamin} applied to the mode $220$. For the operator norm, Proposition~\ref{prop:sec5-exact-identities} gives
\[
\|DG_{220}(M,\chi)\|_2^2\le \|DG_{220}(M,\chi)\|_F^2
=
\frac{|\omegahat_{220}(\chi)|^2}{M^4}+\frac{|\omegahat_{220}'(\chi)|^2}{M^2}.
\]
By the definitions of $U_{220}$ and $V_{220}$ in Appendix~\ref{app:kerr-qnm-data}, one has $|\omegahat_{220}(\chi)|\le U_{220}$ and $|\omegahat_{220}'(\chi)|\le V_{220}$ on the spin interval of interest. Since $M\ge M_-$ on $\Kdet$, the displayed bound implies $\|DG_{220}(p)\|_2\le \Nprim$. Division by $s_{220}(p)\ge\sigprim$ yields $\kappa_{220}(p)\le\kapprim$.

\end{proof}

\begin{corollary}[Explicit atlas constants from the Kerr tables]\label{cor:sec5-explicit-atlas}
For the degree-$64/96$ Chebyshev audit carried out in Appendix~\ref{app:kerr-qnm-data} on the Cook--Zalutskiy tables \( \texttt{KerrQNM\_00.h5} \) and \( \texttt{KerrQNM\_01.h5} \), one may take
\[
\sigprim = 2.31\times 10^{-5},
\qquad
\Nprim = 6.94\times 10^{-3},
\qquad
\kapprim \le 3.01\times 10^{2}.
\]
On the plotted atlas itself the dominant-mode singular-value field varies only between
\[
2.49\times 10^{-5}\le s_{220}(p)\le 2.87\times 10^{-5},
\]
while the pointwise condition number varies between
\[
2.06\times 10^{2}\le \kappa_{220}(p)\le 2.53\times 10^{2}.
\]
The auxiliary forward Lipschitz constants obtained from the same reliability filter are
\[
L_{221}=7.62\times 10^{-3},
\qquad
L_{440}=1.38\times 10^{-2}.
\]
The detector-frame pairwise gap margins on \(\Kdet\) satisfy
\[
\Delta^{\mathrm{cert}}_{220,221}=2.35\times 10^{-3},
\qquad
\Delta^{\mathrm{cert}}_{220,440}=8.42\times 10^{-3},
\qquad
\Delta^{\mathrm{cert}}_{221,440}=8.93\times 10^{-3}.
\]
\end{corollary}

\begin{figure}[t]
\centering
\includegraphics[width=0.83\textwidth]{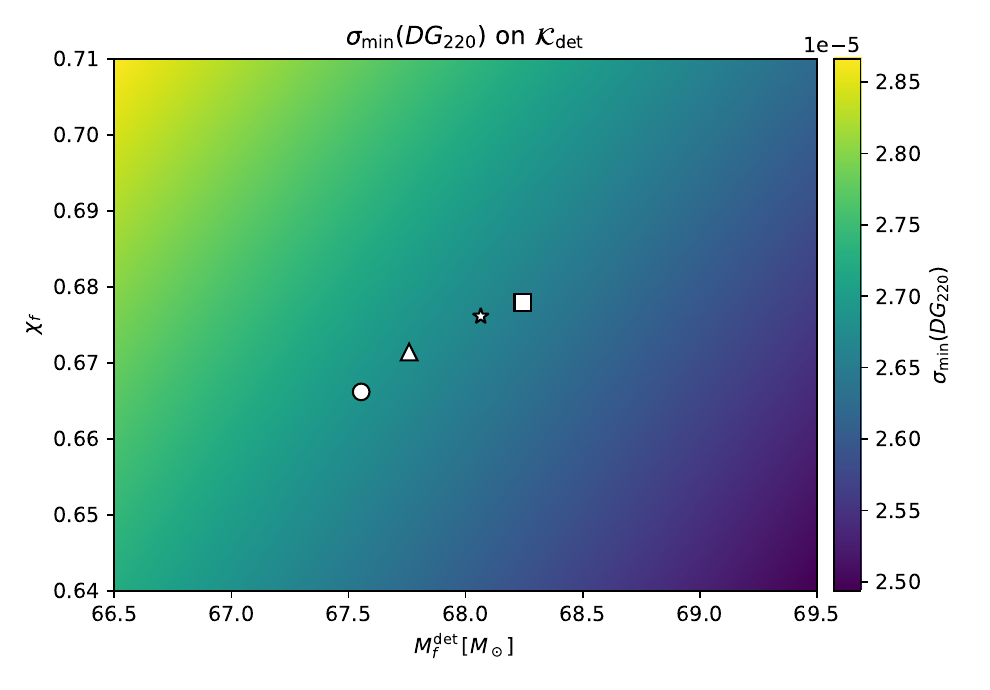}
\caption{Minimum singular-value field \(p\mapsto \sigma_{\min}(DG_{220}(p))\) on the event-local detector-frame box \(\Kdet\), computed from the Kerr QNM tables. The markers show the four public detector-frame remnant medians used later as the public comparison ensemble.}
\label{fig:sec5-sigma-atlas}
\end{figure}

\begin{figure}[t]
\centering
\includegraphics[width=0.83\textwidth]{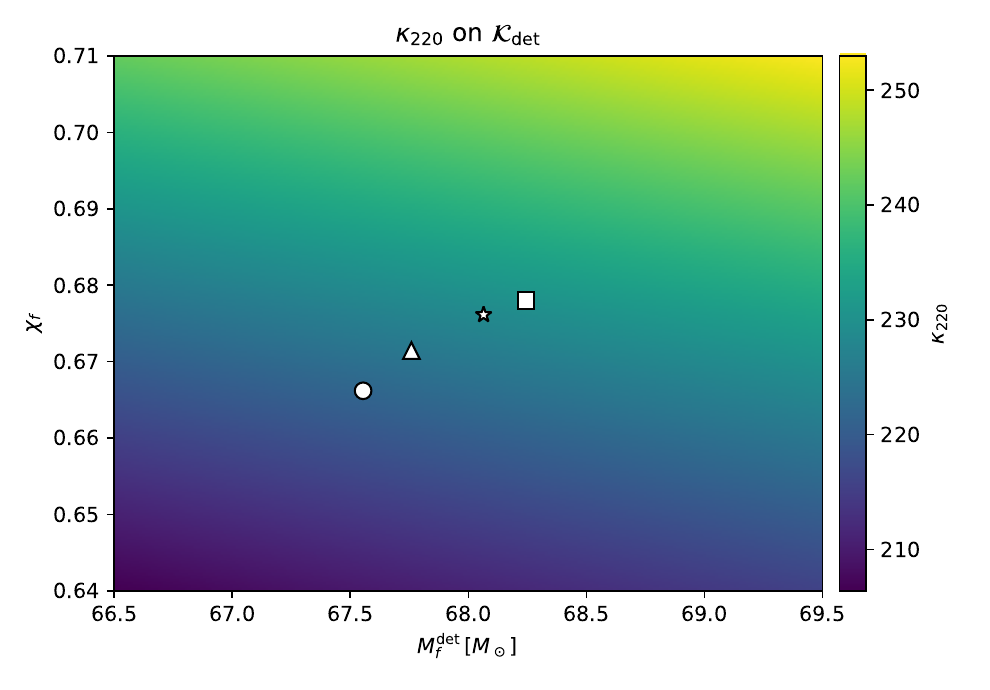}
\caption{Pointwise condition number \(\kappa_{220}(p)\) on \(\Kdet\). The same public detector-frame medians remain far from the poorly conditioned edge of the plotted box.}
\label{fig:sec5-kappa-atlas}
\end{figure}

The auxiliary maps enter only through forward transport, but it is convenient to record their Lipschitz envelopes here as part of the same atlas.

\begin{proposition}[Global auxiliary envelopes]\label{prop:sec5-aux-envelopes}
For each $j\in\{221,440\}$, define
\begin{equation}\label{eq:sec5-Lj-def}
L_j:=\sqrt{\frac{U_j^2}{M_-^4}+\frac{V_j^2}{M_-^2}}.
\end{equation}
Then
\begin{equation}\label{eq:sec5-aux-lipschitz}
|\omega_j(p_1)-\omega_j(p_2)|\le L_j\,\|p_1-p_2\|_{\mathbb R^2}
\qquad\text{for all }p_1,p_2\in\Kdet.
\end{equation}
For each distinct pair $j\neq k\in\modeSet$, if
\[
\Delta^{\mathrm{cert}}_{jk}:=\frac{\sepmargin_{jk}}{M_+},
\]
with $\sepmargin_{jk}$ as in Proposition~\ref{prop:appendixD-separation}, then
\begin{equation}\label{eq:sec5-atlas-gap}
|\omega_j(p)-\omega_k(p)|\ge \Delta^{\mathrm{cert}}_{jk}
\qquad\text{for all }p\in\Kdet.
\end{equation}
Consequently,
\begin{equation}\label{eq:sec5-atlas-isolation}
\delta_{\mathrm{iso}}(p;\{220,221,440\})
\ge
\delta_{\mathrm{iso}}^{\mathrm{cert}}
:=
\frac12\min_{\substack{j,k\in\modeSet\\ j\neq k}}\Delta^{\mathrm{cert}}_{jk}
>0
\qquad\text{for all }p\in\Kdet.
\end{equation}
\end{proposition}

\begin{proof}
The auxiliary Lipschitz estimate is Proposition~\ref{prop:appendixH-aux-lipschitz}. The gap bound is Proposition~\ref{prop:appendixD-separation}. The isolation estimate is obtained by taking half the minimum pairwise gap.
\end{proof}

The condition atlas is recorded graphically. Figure~\ref{fig:sec5-sigma-atlas} shows the lower-singular-value field of the dominant map on the full detector-frame box, while Figure~\ref{fig:sec5-kappa-atlas} shows the corresponding condition-number field. The four public detector-frame remnant medians are overplotted in both panels. They all sit comfortably inside the charted part of $\Kdet$, well away from the locus where the inverse problem becomes poorly conditioned.

\subsection{The primary inverse atlas}

The pointwise atlas of Proposition~\ref{prop:sec5-primary-atlas} must be upgraded to local charts on which the dominant mode is injective with a uniform inverse constant. Compactness provides that upgrade.

\begin{lemma}[Uniform chart radius for the dominant mode]\label{lem:sec5-chart-radius}
There exists $\rhoinv>0$ such that for every $p\in\Kdet$ and every pair
\[
q,r\in B_{\mathbb R^2}(p,\rhoinv)\cap\Kdet,
\]
one has
\begin{equation}\label{eq:sec5-jacobian-oscillation}
\|DG_{220}(q)-DG_{220}(r)\|_2\le \frac{\sigprim}{2}.
\end{equation}
\end{lemma}

\begin{proof}
The map $DG_{220}$ is continuous on an open neighborhood of the compact set $\Kdet$, hence it is uniformly continuous there. For each $p\in\Kdet$, continuity provides a radius $r_p>0$ such that
\[
\|DG_{220}(x)-DG_{220}(p)\|_2\le \frac{\sigprim}{4}
\qquad\text{for every }x\in \overline B_{\mathbb R^2}(p,2r_p)\cap\Kdet.
\]
Therefore, whenever $x$ and $y$ lie in the same set,
\[
\|DG_{220}(x)-DG_{220}(y)\|_2
\le
\|DG_{220}(x)-DG_{220}(p)\|_2+
\|DG_{220}(y)-DG_{220}(p)\|_2
\le \frac{\sigprim}{2}.
\]
The balls $B_{\mathbb R^2}(p,r_p)$ cover $\Kdet$. By compactness, a finite subcover exists. Taking $\rhoinv$ to be the minimum of the corresponding radii yields the claim.
\end{proof}

\begin{theorem}[Primary inverse atlas]\label{thm:sec5-primary-inverse-atlas}
Fix $p\in\Kdet$ and define the local primary chart and its image by
\begin{equation}\label{eq:sec5-primary-chart-def}
\Uchart{p}:=B_{\mathbb R^2}(p,\rhoinv)\cap\Kdet,
\qquad
\Vchart{p}:=F_{220}(\Uchart{p})\subset\mathbb C.
\end{equation}
Then $F_{220}$ is injective on $\Uchart{p}$ and, for every $q,r\in\Uchart{p}$,
\begin{equation}\label{eq:sec5-primary-bilipschitz}
|F_{220}(q)-F_{220}(r)|
\ge
\frac{\sigprim}{2}\,\|q-r\|_{\mathbb R^2}.
\end{equation}
Consequently the inverse on the image satisfies
\begin{equation}\label{eq:sec5-primary-inverse-lipschitz}
\|F_{220}^{-1}(z_1)-F_{220}^{-1}(z_2)\|_{\mathbb R^2}
\le
L_{220}|z_1-z_2|,
\qquad z_1,z_2\in\Vchart{p},
\end{equation}
with
\begin{equation}\label{eq:sec5-primary-L220}
L_{220}:=\frac{2}{\sigprim}.
\end{equation}
The collection
\begin{equation}\label{eq:sec5-primary-atlas-def}
\Kinv:=\Bigl\{\bigl(\Uchart{p},\Vchart{p},F_{220}|_{\Uchart{p}}^{-1}\bigr):p\in\Kdet\Bigr\}
\end{equation}
is the primary inverse atlas on the event-local box.
\end{theorem}

\begin{proof}
Fix $p\in\Kdet$ and let $q,r\in\Uchart{p}$. Because both the Euclidean ball and the rectangle $\Kdet$ are convex, the segment
\[
\gamma(t)=r+t(q-r),\qquad 0\le t\le 1,
\]
remains in $\Uchart{p}$. By the fundamental theorem of calculus,
\[
G_{220}(q)-G_{220}(r)
=
\int_0^1 DG_{220}(\gamma(t))(q-r)\,dt.
\]
Add and subtract the constant matrix $DG_{220}(p)$:
\[
G_{220}(q)-G_{220}(r)
=
DG_{220}(p)(q-r)
+
\int_0^1\bigl(DG_{220}(\gamma(t))-DG_{220}(p)\bigr)(q-r)\,dt.
\]
Taking norms and using the reverse triangle inequality gives
\begin{align*}
|G_{220}(q)-G_{220}(r)|
&\ge |DG_{220}(p)(q-r)|
-\int_0^1\|DG_{220}(\gamma(t))-DG_{220}(p)\|_2\,dt\,\|q-r\|_{\mathbb R^2} \\
&\ge s_{220}(p)\,\|q-r\|_{\mathbb R^2}-\frac{\sigprim}{2}\,\|q-r\|_{\mathbb R^2},
\end{align*}
where Lemma~\ref{lem:sec5-chart-radius} was used in the last step. Proposition~\ref{prop:sec5-primary-atlas} now implies $s_{220}(p)\ge\sigprim$, hence
\[
|G_{220}(q)-G_{220}(r)|\ge \frac{\sigprim}{2}\,\|q-r\|_{\mathbb R^2}.
\]
Since $\XiKerr$ is an isometry, the same bound holds for $F_{220}$, proving \eqref{eq:sec5-primary-bilipschitz}. Injectivity is immediate. If $z_i=F_{220}(q_i)$ with $q_i\in\Uchart{p}$, then \eqref{eq:sec5-primary-bilipschitz} gives
\[
\|q_1-q_2\|_{\mathbb R^2}
\le \frac{2}{\sigprim}|z_1-z_2|
=L_{220}|z_1-z_2|,
\]
which is \eqref{eq:sec5-primary-inverse-lipschitz}.

\end{proof}

For the Kerr tables one can make the local chart size explicit. Appendix~\ref{app:primary-inversion} records a global Jacobian-Lipschitz bound for \(DG_{220}\) on \(\Kdet\), and the corresponding explicit primary chart radius is
\[
\rho_{220}=2.47\times 10^{-4}.
\]
Compared with the public detector-frame comparison ensemble, this radius is comfortably interior. The four public detector-frame remnant medians all lie at least \(2.62\times 10^{-2}\) away from the boundary of \(\Kdet\), so the event-local public comparison box is larger than the local chart radius by more than two orders of magnitude. The compact box therefore keeps the public remnant neighborhood explicit while the local inverse itself remains genuinely interior.

Theorem~\ref{thm:sec5-primary-inverse-atlas} is the geometric core of the remnant stage. Once a finite-window extractor has produced a dominant-mode frequency $\widehat\omega_{220}$ lying in one of the local images $\Vchart{p}$, the corresponding remnant estimate
\begin{equation}\label{eq:sec5-primary-estimate}
\widehat p:=F_{220}^{-1}(\widehat\omega_{220})
\end{equation}
obeys the deterministic parameter error bound
\begin{equation}\label{eq:sec5-primary-parameter-error}
\|\widehat p-p\|_{\mathbb R^2}
\le
L_{220}|\widehat\omega_{220}-\omega_{220}(p)|.
\end{equation}
This proposition carries frequency extraction error into remnant-parameter error.

For the event-level comparison to public IMR summaries we do not use the continuum atlas directly. Instead we center a finite public subatlas at the four public detector-frame medians introduced in Section~\ref{sec:data-conventions}. If $\palpha=(M_{\alpha}^{\mathrm{det}},\chi_{\alpha})\in\Kdet$ denotes the detector-frame median associated with $\alpha\in\Ppub$, we set
\begin{equation}\label{eq:sec5-public-subatlas}
\Kpub:=\Bigl\{\bigl(\Uchart{\palpha},\Vchart{\palpha},F_{220}|_{\Uchart{\palpha}}^{-1}\bigr):\alpha\in\Ppub\Bigr\}.
\end{equation}
The finite set $\Kpub$ keeps the dominant-mode inversion tied to the public event-local neighborhood.

\subsection{Auxiliary consistency tubes}

The auxiliary modes $221$ and $440$ are not used to redefine the remnant. They are used to ask whether the remnant inferred from $220$ predicts a common Kerr spectrum. Their geometry is therefore forward rather than inverse.

\begin{theorem}[Primary-to-auxiliary error propagation]\label{thm:sec5-primary-auxiliary}
Fix a primary chart $\Uchart{p}\in\Kinv$ and suppose that the extracted dominant-mode frequency lies in the corresponding image $\Vchart{p}$. Let
\[
\widehat p:=F_{220}^{-1}(\widehat\omega_{220})\in\Uchart{p}
\]
be the primary remnant estimate. Assume that there exists a point $p\in\Uchart{p}$ such that
\begin{equation}\label{eq:sec5-frequency-budgets}
|\widehat\omega_{220}-\omega_{220}(p)|\le\varepsilon_{220},
\qquad
|\widehat\omega_j-\omega_j(p)|\le\varepsilon_j,
\qquad j\in\{221,440\}.
\end{equation}
Then
\begin{equation}\label{eq:sec5-primary-error-estimate}
\|\widehat p-p\|_{\mathbb R^2}\le L_{220}\varepsilon_{220},
\end{equation}
and, for each auxiliary mode $j\in\{221,440\}$, the residual
\begin{equation}\label{eq:sec5-aux-residual}
\Raux{j}(\widehat p,\widehat\omega_j):=|\widehat\omega_j-\omega_j(\widehat p)|
\end{equation}
satisfies
\begin{equation}\label{eq:sec5-aux-residual-bound}
\Raux{j}(\widehat p,\widehat\omega_j)
\le
\tauaux{j}(\varepsilon_{220},\varepsilon_j)
:=
\varepsilon_j+L_jL_{220}\varepsilon_{220}.
\end{equation}
Equivalently,
\begin{equation}\label{eq:sec5-aux-disk-membership}
\widehat\omega_j\in \overline D\bigl(\omega_j(\widehat p),\tauaux{j}(\varepsilon_{220},\varepsilon_j)\bigr),
\qquad j\in\{221,440\}.
\end{equation}
\end{theorem}

\begin{proof}
The primary parameter estimate \eqref{eq:sec5-primary-error-estimate} is just \eqref{eq:sec5-primary-parameter-error}. For the auxiliary residual, insert and subtract the true Kerr value at $p$:
\[
\Raux{j}(\widehat p,\widehat\omega_j)
\le
|\widehat\omega_j-\omega_j(p)|+|\omega_j(p)-\omega_j(\widehat p)|.
\]
The first term is bounded by $\varepsilon_j$ from \eqref{eq:sec5-frequency-budgets}. The second term is controlled by Proposition~\ref{prop:sec5-aux-envelopes}:
\[
|\omega_j(p)-\omega_j(\widehat p)|\le L_j\|p-\widehat p\|_{\mathbb R^2}\le L_jL_{220}\varepsilon_{220}.
\]
Combining the two inequalities proves \eqref{eq:sec5-aux-residual-bound}, and the disk statement is the corresponding reformulation.
\end{proof}

Theorem~\ref{thm:sec5-primary-auxiliary} is the precise mathematical content of auxiliary consistency. The dominant-mode inversion yields a moving Kerr prediction for each auxiliary mode together with a deterministic thickness. Even if an auxiliary frequency were measured exactly, the dominant-mode uncertainty would still widen its Kerr prediction by the amount $L_jL_{220}\varepsilon_{220}$.

The contrapositive is the event-level failure certificate.

\begin{proposition}[Common-remnant failure certificate]\label{prop:sec5-failure-certificate}
Fix a primary chart $\Uchart{p}\in\Kinv$ and a primary estimate $\widehat p\in\Uchart{p}$. Let $j\in\{221,440\}$ and suppose that
\begin{equation}\label{eq:sec5-failure-hypothesis}
\Raux{j}(\widehat p,\widehat\omega_j)>	auaux{j}(\varepsilon_{220},\varepsilon_j).
\end{equation}
Then there is no point $q\in\Uchart{p}$ such that
\[
|\omega_{220}(q)-\widehat\omega_{220}|\le\varepsilon_{220}
\qquad\text{and}\qquad
|\omega_j(q)-\widehat\omega_j|\le\varepsilon_j.
\]
\end{proposition}

\begin{proof}
If such a point $q$ existed, then Theorem~\ref{thm:sec5-primary-auxiliary} with $q$ in place of $p$ would imply
\[
\Raux{j}(\widehat p,\widehat\omega_j)\le\tauaux{j}(\varepsilon_{220},\varepsilon_j),
\]
contradicting \eqref{eq:sec5-failure-hypothesis}.
\end{proof}

This proposition explains why the auxiliary modes are informative. A large residual is a deterministic certificate that, on the declared primary chart and within the declared budgets, no single Kerr remnant accounts for both the dominant mode and the tested auxiliary mode.

This motivates the joint auxiliary tube.

\begin{definition}[Auxiliary consistency tube]\label{def:sec5-consistency-tube}
For a primary estimate $\widehat p\in\Kdet$ and a frequency-error decomposition
\[
\varepsilon=(\varepsilon_{220},\varepsilon_{221},\varepsilon_{440}),
\]
define
\begin{equation}\label{eq:sec5-consistency-tube}
\Caux(\widehat p;\varepsilon)
:=
\overline D\bigl(\omega_{221}(\widehat p),\tauaux{221}(\varepsilon_{220},\varepsilon_{221})\bigr)
\times
\overline D\bigl(\omega_{440}(\widehat p),\tauaux{440}(\varepsilon_{220},\varepsilon_{440})\bigr)
\subset\mathbb C^2.
\end{equation}
We say that the extracted auxiliary pair $(\widehat\omega_{221},\widehat\omega_{440})$ is common-remnant consistent with the primary estimate $\widehat p$ if it lies in $\Caux(\widehat p;\varepsilon)$.
\end{definition}

\subsection{Secondary pair maps}

The main remnant estimator is the dominant-mode inverse. The pair maps remain mathematically relevant because the first coordinate of $\pairmap{j}$ already determines the remnant locally on every primary chart, so the second coordinate supplies redundancy before it supplies identifiability.

\begin{proposition}[Exact pair inversion agrees with primary inversion]\label{prop:sec5-pair-agreement}
Fix $j\in\{221,440\}$ and a primary chart $\Uchart{p}\in\Kinv$. Then $\pairmap{j}$ is injective on $\Uchart{p}$. If an exact pair datum
\[
z=\pairmap{j}(q)=\bigl(\omega_{220}(q),\omega_j(q)\bigr)
\qquad\text{with }q\in\Uchart{p}
\]
is given, then
\begin{equation}\label{eq:sec5-pair-agreement}
\pairmap{j}^{-1}(z)=F_{220}^{-1}(\pi_1 z)=q,
\end{equation}
where $\pi_1$ denotes projection onto the first coordinate.
\end{proposition}

\begin{proof}
If $\pairmap{j}(q_1)=\pairmap{j}(q_2)$ on $\Uchart{p}$, then the first coordinates agree, so $\omega_{220}(q_1)=\omega_{220}(q_2)$. Theorem~\ref{thm:sec5-primary-inverse-atlas} implies that $F_{220}$ is injective on the chart, hence $q_1=q_2$. This proves injectivity of $\pairmap{j}$ on the chart. If $z=\pairmap{j}(q)$, then $\pi_1 z=\omega_{220}(q)$, and applying the primary inverse to the first coordinate returns $q$. Since the pair map is injective, its exact inverse returns the same point.
\end{proof}

Proposition~\ref{prop:sec5-pair-agreement} explains why two-mode inversion does not replace the dominant-mode atlas in the later analysis. On the event-local GW250114 box, the dominant mode already provides a stable chart, and the pair maps are best interpreted as diagnostics. Appendix~\ref{app:joint-two-mode-inversion} develops those diagnostics in the form of projected pair estimators and distance-to-image certificates. The asymmetric structure is therefore kept throughout.

Taken together, these results provide the dominant-mode condition atlas, the primary inverse atlas $\Kinv$, the public subatlas $\Kpub$, the auxiliary transport constants $L_{221}$ and $L_{440}$, the pairwise Kerr gaps, and the common-remnant tube $\Caux(\widehat p;\varepsilon)$. Once the extraction theory supplies the frequency error budgets, these objects become explicit remnant tolerances and auxiliary pass-or-fail regions.

\setcounter{section}{5}

\providecommand{\AuditFam}{\mathcal M^{\mathrm{aud}}}
\providecommand{\Liso}{\mathcal L_{\mathrm{iso}}}
\section{Trust-region theorem}\label{sec:trust-region}

Section~\ref{sec:abstract-extraction} supplies labeled detector-frame frequencies together with explicit additive radii. Section~\ref{sec:kerr-inversion} transports the dominant radius to a detector-frame remnant ball and transports that ball further to auxiliary Kerr predictions. We now state the pass/fail inequalities that a finite window must satisfy before it is promoted to common-remnant spectroscopy.

Two threshold levels must be kept distinct. The first is theorem level. Here one works with the symbolic tolerance vector \(\eta\) and proves a deterministic implication. The second is numerical. Only after the synthetic calibration of Section~\ref{sec:numerical-calibration}, the public detector-frame comparison of Section~\ref{sec:data-conventions}, and the finite rerun family of Appendix~\ref{app:gw250114-robustness} have been fixed are numerical caps assigned. The theorem does not guess those numbers. It identifies exactly which transported budgets they must dominate.

A trusted detector-frame window is required to satisfy five inequalities of different kinds. The dominant mode must remain separated from competing Kerr labels on the local event-level geometry. The transported remnant ball must remain small. The remnant estimate must remain stable under small start-time shifts. Any fitted auxiliary mode must be compatible with the same remnant inferred from \(220\). These inequalities are later supplemented by two additional numerical caps: a public-hull comparison in detector-frame remnant space and nuisance-audit bounds against the fixed direct-wave and quadratic alternatives. Their logical order is part of the result.

\subsection{Detector-frame trust inequalities}

Fix a model family \(\mathcal M\) containing \(220\), and fix one admissible analysis window
\[
\Xi=(t_0,T).
\]
Throughout the section we suppress the arguments \((\Xi,\mathcal M)\) whenever no ambiguity can arise. Thus
\[
\widehat\omega_j=\widehat\omega_j(\Xi,\mathcal M),
\qquad
\widehat p=\widehat p(\Xi,\mathcal M)=\bigl(\widehat\Mf(\Xi,\mathcal M),\widehat\chif(\Xi,\mathcal M)\bigr),
\]
denote the extracted frequency of mode \(j\) and the remnant estimate obtained from the primary inversion of the dominant mode \(220\). The total frequency radius is written
\begin{equation}\label{eq:sec6-total-eps}
\varepsilon_j
:=
\varepsilon_j^{\mathrm{stat}}
+
\varepsilon_j^{\mathrm{tail}}
+
\varepsilon_j^{\mathrm{mm}}
+
\varepsilon_j^{\mathrm{alg}},
\end{equation}
in the sense of Corollary~\ref{cor:sec4-additive-ledger}. All deterministic acceptance statements below depend on the window only through the radii \(\varepsilon_j\), the remnant estimate \(\widehat p\), and the local drift statistic introduced shortly.

One subtlety must be settled before the trust inequalities are written. The isolation margin that governs label stability is naturally attached to a family of Kerr modes, but the baseline family \(\{220\}\) still requires a nontrivial dominant-mode resolution audit. For that reason we distinguish the \emph{fitted family} \(\mathcal M\) from the \emph{audit family}
\begin{equation}\label{eq:sec6-audit-family}
\AuditFam(\mathcal M)
:=
\begin{cases}
\mathcal M\cap\{220,221,440\}, & \text{if } \bigl|\mathcal M\cap\{220,221,440\}\bigr|\ge 2,\\[3pt]
\{220,221,440\}, & \text{if } \mathcal M\cap\{220,221,440\}=\{220\}.
\end{cases}
\end{equation}
We then define the audit isolation margin by
\begin{equation}\label{eq:sec6-audit-isolation}
\deltaiso^{\mathrm{aud}}(p;\mathcal M)
:=
\deltaiso\bigl(p;\AuditFam(\mathcal M)\bigr),
\qquad p\in\Kdet.
\end{equation}
For multimode Kerr families this is the ordinary fitted-family isolation margin. For a \(220\)-only family it becomes the dominant-mode resolution margin against the nearest \(221\) or \(440\) alternative, which is the quantity actually relevant for a trustworthy single-mode analysis.

The audit isolation margin varies with the remnant parameters, so the trust theorem needs a quantitative continuity bound.

\begin{proposition}[Isolation margin is locally Lipschitz]\label{prop:sec6-isolation-lipschitz}
For a model family \(\mathcal M\) containing \(220\), define
\begin{equation}\label{eq:sec6-liso-def}
\Liso(\mathcal M)
:=
\frac12
\max_{\substack{j,k\in\AuditFam(\mathcal M)\\ j\neq k}}
\bigl(L_j^{\mathrm f}+L_k^{\mathrm f}\bigr),
\end{equation}
where \(L_{220}^{\mathrm f}:=\Nprim\) is the global forward Lipschitz constant of the dominant Kerr map supplied by Proposition~\ref{prop:sec5-primary-atlas}, and \(L_j\) is the corresponding forward constant of Proposition~\ref{prop:sec5-aux-envelopes} for \(j\in\{221,440\}\). For the formula \eqref{eq:sec6-liso-def} we use \(L_j^{\mathrm f}=L_{220}^{\mathrm f}\) when \(j=220\) and \(L_j^{\mathrm f}=L_j\) for \(j\in\{221,440\}\). Then for every \(p,q\in\Kdet\),
\begin{equation}\label{eq:sec6-isolation-lipschitz}
\bigl|
\deltaiso^{\mathrm{aud}}(p;\mathcal M)-\deltaiso^{\mathrm{aud}}(q;\mathcal M)
\bigr|
\le
\Liso(\mathcal M)\,\|p-q\|_{\mathbb R^2}.
\end{equation}
\end{proposition}

\begin{proof}
Let
\[
g_{jk}(p):=\frac12|\omega_j(p)-\omega_k(p)|,
\qquad j\neq k\in\AuditFam(\mathcal M).
\]
Then
\[
\deltaiso^{\mathrm{aud}}(p;\mathcal M)=\min_{j\neq k} g_{jk}(p).
\]
For each pair \(j\neq k\),
\begin{align*}
|g_{jk}(p)-g_{jk}(q)|
&\le
\frac12\bigl|(\omega_j(p)-\omega_k(p))-(\omega_j(q)-\omega_k(q))\bigr|\\
&\le
\frac12\bigl(|\omega_j(p)-\omega_j(q)|+|\omega_k(p)-\omega_k(q)|\bigr)\\
&\le
\frac12(L_j^{\mathrm f}+L_k^{\mathrm f})\,\|p-q\|_{\mathbb R^2}
\le
\Liso(\mathcal M)\,\|p-q\|_{\mathbb R^2}.
\end{align*}
The minimum of a finite collection of \(L\)-Lipschitz functions is again \(L\)-Lipschitz, so \eqref{eq:sec6-isolation-lipschitz} follows.
\end{proof}

The next object is the short-scale remnant drift. It measures not posterior spread within one window, but the instability of the point estimate under small changes of the start time.

\begin{definition}[Local start-time drift]\label{def:sec6-local-drift}
Fix \(\Delta>0\). The \(\Delta\)-local drift of the recovered remnant parameters at the window \(\Xi=(t_0,T)\) is
\begin{equation}\label{eq:sec6-local-drift}
\mathrm{Drift}_{\Delta}(\Xi,\mathcal M)
:=
\sup\Bigl\{
\|\widehat p(t_0+\tau,T,\mathcal M)-\widehat p(t_0,T,\mathcal M)\|_{\mathbb R^2}
:
|\tau|\le \Delta,\ t_0+\tau\ \text{admissible}
\Bigr\}.
\end{equation}
\end{definition}

We can now state the abstract detector-frame trust test.

\begin{definition}[Detector-frame trust set]\label{def:sec6-trust-set}
Let
\begin{equation}\label{eq:sec6-tolerance-vector}
\eta=
\bigl(
\eta_{\mathrm{sep}},
\eta_{\mathrm p},
\eta_{\mathrm d},
\eta_{221},
\eta_{440}
\bigr)
\in (0,\infty)^5
\end{equation}
be a vector of tolerances. Throughout Section~\ref{sec:trust-region}, \(\eta\) is symbolic: no numerical values are fixed here. For a model family \(\mathcal M\) containing \(220\), the detector-frame trust set
\[
\TrustDet_{\eta}(\mathcal M)
\]
is the collection of all admissible windows \(\Xi=(t_0,T)\) such that the following conditions hold.

The window is extraction-admissible in the sense of Definition~\ref{def:sec4-admissible-envelope}, and the dominant-mode estimate \(\widehat\omega_{220}\) lies in some primary chart image \(\Vchart{p}\) with \(p\in\Kdet\), so that the primary inverse estimate \(\widehat p\) is well defined on the event-local atlas.

The dominant-mode separation inequality holds:
\begin{equation}\label{eq:sec6-trust-sep}
\varepsilon_{220}\le \eta_{\mathrm{sep}}
<
\deltaiso^{\mathrm{aud}}(\widehat p;\mathcal M).
\end{equation}
The propagated primary remnant budget holds:
\begin{equation}\label{eq:sec6-trust-param}
L_{220}\,\varepsilon_{220}\le \eta_{\mathrm p}.
\end{equation}
The local drift bound holds:
\begin{equation}\label{eq:sec6-trust-drift}
\mathrm{Drift}_{\Delta}(\Xi,\mathcal M)\le \eta_{\mathrm d}.
\end{equation}
Whenever \(221\in \mathcal M\), the auxiliary consistency inequality holds:
\begin{equation}\label{eq:sec6-trust-221}
\Raux{221}\bigl(\widehat p,\widehat\omega_{221}\bigr)\le \eta_{221}.
\end{equation}
Whenever \(440\in \mathcal M\), the analogous inequality holds:
\begin{equation}\label{eq:sec6-trust-440}
\Raux{440}\bigl(\widehat p,\widehat\omega_{440}\bigr)\le \eta_{440}.
\end{equation}
\end{definition}

The separation inequality is strict because one needs positive room between the declared dominant-mode radius and the local label boundary. The remaining inequalities are closed: they bound transported error objects and can therefore be checked by direct comparison once the relevant budgets are known.

The closed inequalities of Definition~\ref{def:sec6-trust-set} are downstream of the same exact transport chain. First,
\begin{equation}\label{eq:sec6-threshold-chain-primary}
\varepsilon_{220}
\longmapsto
L_{220}\varepsilon_{220},
\end{equation}
which is the primary remnant radius furnished by Theorem~\ref{thm:sec5-primary-inverse-atlas}. Second, for each auxiliary mode \(j\in\{221,440\}\),
\begin{equation}\label{eq:sec6-threshold-chain-aux}
(\varepsilon_{220},\varepsilon_j)
\longmapsto
\tauaux{j}(\varepsilon_{220},\varepsilon_j)
=
\varepsilon_j+L_jL_{220}\varepsilon_{220},
\qquad j\in\{221,440\}.
\end{equation}
Only after \eqref{eq:sec6-threshold-chain-primary}--\eqref{eq:sec6-threshold-chain-aux} do the remnant-side diagnostics enter. The local drift \(\mathrm{Drift}_{\Delta}\) is a neighboring-window quantity on the same remnant space, the public-hull distance \(d_{\mathrm{PE}}\) used later in Appendix~\ref{app:gw250114-robustness} is a detector-frame remnant-space comparison, and the nuisance gains \(\Gamma_\nu\) of the same appendix are residual improvements computed after the baseline remnant has already been fixed. No numerical cap introduced later is allowed to bypass this order.

\begin{proposition}[Available separation interval]\label{prop:sec6-sep-interval}
Fix \(p\in\Kdet\), a family \(\mathcal M\) containing \(220\), and a dominant frequency radius \(\varepsilon_{220}>0\). The theorem-level separation thresholds compatible with the inverse transport are exactly the numbers
\begin{equation}\label{eq:sec6-sep-interval}
\eta_{\mathrm{sep}}
\in
\mathcal I_{\mathrm{sep}}(p;\varepsilon_{220},\mathcal M)
:=
\Bigl[
\varepsilon_{220},
\deltaiso^{\mathrm{aud}}(p;\mathcal M)-\Liso(\mathcal M)L_{220}\varepsilon_{220}
\Bigr).
\end{equation}
The interval \(\mathcal I_{\mathrm{sep}}(p;\varepsilon_{220},\mathcal M)\) is nonempty if and only if
\begin{equation}\label{eq:sec6-sep-slack}
\varepsilon_{220}
<
\deltaiso^{\mathrm{aud}}(p;\mathcal M)-\Liso(\mathcal M)L_{220}\varepsilon_{220}.
\end{equation}
\end{proposition}

\begin{proof}
By definition, \(\eta_{\mathrm{sep}}\) is admissible exactly when it satisfies the two inequalities
\[
\varepsilon_{220}\le \eta_{\mathrm{sep}}
\qquad\text{and}\qquad
\eta_{\mathrm{sep}}<
\deltaiso^{\mathrm{aud}}(p;\mathcal M)-\Liso(\mathcal M)L_{220}\varepsilon_{220}.
\]
This is equivalent to \eqref{eq:sec6-sep-interval}. The interval is nonempty precisely when its left endpoint is strictly smaller than its right endpoint, which is \eqref{eq:sec6-sep-slack}.
\end{proof}

Proposition~\ref{prop:sec6-sep-interval} is the first place where the symbolic and numerical levels separate cleanly. The symbol \(\eta_{\mathrm{sep}}\) is still abstract, but it is not free. Once the dominant-mode radius and the reference remnant are fixed, the theorem permits exactly the interval \eqref{eq:sec6-sep-interval}. Likewise, \(\eta_{\mathrm p}\) must dominate \(L_{220}\varepsilon_{220}\), and \(\eta_j\) must dominate \(\tauaux{j}(\varepsilon_{220},\varepsilon_j)\). The drift cap \(\eta_{\mathrm d}\) is different: it is a neighboring-window quantity, so its numerical value is introduced only after the start-time scan itself has been fixed.

\subsection{The deterministic trust-region theorem}

The nontrivial statement is that the symbolic inequalities above follow from more primitive quantities furnished by Sections~\ref{sec:abstract-extraction} and \ref{sec:kerr-inversion}. That implication is the theorem-level core of the argument.

\begin{theorem}[Deterministic trust-region theorem]\label{thm:sec6-trust-region}
Fix a model family \(\mathcal M\) containing \(220\), an admissible window \(\Xi=(t_0,T)\), and a true detector-frame remnant \(p_{\star}\in\Kdet\). Assume the extraction hypotheses of Theorem~\ref{thm:sec4-abstract-extraction} hold for the chosen window, with labeled dominant and auxiliary frequencies \(\widehat\omega_j\) and total radii \(\varepsilon_j\) given by \eqref{eq:sec6-total-eps}. Assume also that the dominant-mode estimate is chart-admissible in the sense that
\begin{equation}\label{eq:sec6-chart-admissible}
\widehat\omega_{220}\in \Vchart{p_{\star}}
\end{equation}
for the primary chart of Theorem~\ref{thm:sec5-primary-inverse-atlas}, so that
\[
\widehat p:=F_{220}^{-1}(\widehat\omega_{220})\in \Uchart{p_{\star}}.
\]
Let \(\eta\) be as in \eqref{eq:sec6-tolerance-vector}. Suppose
\begin{align}
\eta_{\mathrm{sep}}
&\in
\mathcal I_{\mathrm{sep}}(p_{\star};\varepsilon_{220},\mathcal M),
\label{eq:sec6-hyp-sep}\\
L_{220}\varepsilon_{220}
&\le
\eta_{\mathrm p},
\label{eq:sec6-hyp-param}\\
\mathrm{Drift}_{\Delta}(\Xi,\mathcal M)
&\le
\eta_{\mathrm d},
\label{eq:sec6-hyp-drift}
\end{align}
and, whenever \(221\in\mathcal M\),
\begin{equation}\label{eq:sec6-hyp-221}
\tauaux{221}(\varepsilon_{220},\varepsilon_{221})\le \eta_{221},
\end{equation}
while, whenever \(440\in\mathcal M\),
\begin{equation}\label{eq:sec6-hyp-440}
\tauaux{440}(\varepsilon_{220},\varepsilon_{440})\le \eta_{440}.
\end{equation}
Then
\[
\Xi\in \TrustDet_{\eta}(\mathcal M).
\]
\end{theorem}

\begin{proof}
The window is extraction-admissible by hypothesis, and chart-admissibility is exactly \eqref{eq:sec6-chart-admissible}. Thus the primary estimate \(\widehat p\) is well defined on the event-local inverse atlas.

By Theorem~\ref{thm:sec5-primary-inverse-atlas},
\begin{equation}\label{eq:sec6-proof-primary}
\|\widehat p-p_{\star}\|_{\mathbb R^2}\le L_{220}\,|\widehat\omega_{220}-\omega_{220}(p_{\star})|
\le L_{220}\varepsilon_{220}.
\end{equation}
In particular, \eqref{eq:sec6-hyp-param} yields the trust inequality \eqref{eq:sec6-trust-param}.

To prove the separation inequality, combine Proposition~\ref{prop:sec6-isolation-lipschitz} with \eqref{eq:sec6-proof-primary}:
\begin{align*}
\deltaiso^{\mathrm{aud}}(\widehat p;\mathcal M)
&\ge
\deltaiso^{\mathrm{aud}}(p_{\star};\mathcal M)
-
\Liso(\mathcal M)\|\widehat p-p_{\star}\|_{\mathbb R^2}\\
&\ge
\deltaiso^{\mathrm{aud}}(p_{\star};\mathcal M)
-
\Liso(\mathcal M)L_{220}\varepsilon_{220}.
\end{align*}
Because \(\eta_{\mathrm{sep}}\in\mathcal I_{\mathrm{sep}}(p_{\star};\varepsilon_{220},\mathcal M)\), the right-hand side is strictly larger than \(\eta_{\mathrm{sep}}\), while the left endpoint of the same interval gives \(\varepsilon_{220}\le\eta_{\mathrm{sep}}\). Therefore
\[
\varepsilon_{220}
\le
\eta_{\mathrm{sep}}
<
\deltaiso^{\mathrm{aud}}(\widehat p;\mathcal M),
\]
which is precisely \eqref{eq:sec6-trust-sep}.

The drift inequality \eqref{eq:sec6-trust-drift} is exactly \eqref{eq:sec6-hyp-drift}. For the auxiliary modes, Theorem~\ref{thm:sec5-primary-auxiliary} gives
\[
\Raux{j}(\widehat p,\widehat\omega_j)
\le
\tauaux{j}(\varepsilon_{220},\varepsilon_j)
\]
for each \(j\in\{221,440\}\cap\mathcal M\). Combining this with \eqref{eq:sec6-hyp-221} and \eqref{eq:sec6-hyp-440} yields the trust inequalities \eqref{eq:sec6-trust-221} and \eqref{eq:sec6-trust-440}. Every clause of Definition~\ref{def:sec6-trust-set} is therefore satisfied.
\end{proof}

Theorem~\ref{thm:sec6-trust-region} is a theorem-level implication. It does not yet choose numerical caps. It says that any admissible threshold choice must dominate transported quantities obtained from the extraction ledger and the local Kerr atlas, and that the separation threshold must lie in the interval of Proposition~\ref{prop:sec6-sep-interval}.

\begin{corollary}[Per-window calibrated trust implication]\label{cor:sec6-calibrated}
Fix a calibration level \(q\in(0,1)\) and a confidence level \(\alpha\in(0,1)\). Let \(W=(t_0,T)\) be a calibration window and let \(\mathcal M\) be a baseline or nuisance-extended family containing \(220\). Replace each deterministic frequency radius \(\varepsilon_j\) in Theorem~\ref{thm:sec6-trust-region} by the calibrated radius
\[
\varepsilon_j^{\mathrm{cal},\mathcal M,W}(q;\alpha)
\]
defined in \eqref{eq:appendixJ-total-calibrated-radius}. If the hypotheses \eqref{eq:sec6-chart-admissible}--\eqref{eq:sec6-hyp-440} hold with \(\varepsilon_j^{\mathrm{cal},\mathcal M,W}(q;\alpha)\) in place of \(\varepsilon_j\), then, for the fixed window \(W\), with confidence at least \(1-\alpha\) and at the bankwise \(q\)-quantile level,
\[
W\in \TrustDet_{\eta}(\mathcal M).
\]
\end{corollary}

\begin{proof}
Appendix~\ref{app:synthetic-bank} proves that \(\varepsilon_j^{\mathrm{cal},\mathcal M,W}(q;\alpha)\) is a simultaneous upper envelope for the bankwise \(q\)-quantile of the total frequency error at the fixed window \(W\). Substituting these radii into Theorem~\ref{thm:sec6-trust-region} yields the conclusion.
\end{proof}

Corollary~\ref{cor:sec6-calibrated} yields a conservative theorem-level budget for each fixed calibration window. The later GW250114 release does not alter the analytic chain. It only replaces the ideal radii by conservative finite-data surrogates and then adjoins the public-hull and nuisance inequalities.

\subsection{From symbolic thresholds to fixed numerical caps}

The GW250114 event-level rule is obtained by specializing the symbolic theorem to the finite calibration and rerun objects declared below. This specialization is exact.

\begin{proposition}[Exact event-level specialization]\label{prop:sec6-event-specialization}
Fix a baseline family \(\mathcal A\subset\Mref\) with \(220\in\mathcal A\), a master scan point \(\Xi\), and a baseline specification \(\sigma\) from the finite rerun family of Appendix~\ref{app:gw250114-robustness}. The event-level GW250114 rule specializes the theorem-level quantities of Theorem~\ref{thm:sec6-trust-region} by the replacements
\begin{align}
\varepsilon_{220}
&\longmapsto
\overline\varepsilon_{220}^{\mathcal A}(\Xi)+\Delta_{220}^{\mathcal A}(\Xi),
\label{eq:sec6-op-radius}\\
\deltaiso^{\mathrm{aud}}(p_{\star};\mathcal M)-\Liso(\mathcal M)L_{220}\varepsilon_{220}
&\longmapsto
\deltaiso\bigl(\widehat p^{\sigma,\mathcal A}(\Xi);\mathcal A\bigr)-L_{\mathrm{iso}}(\mathcal A)L_{220}\Delta_{220}^{\mathcal A}(\Xi),
\label{eq:sec6-op-sepmargin}\\
L_{220}\varepsilon_{220}
&\longmapsto
L_{220}\bigl(\overline\varepsilon_{220}^{\mathcal A}(\Xi)+\Delta_{220}^{\mathcal A}(\Xi)\bigr),
\label{eq:sec6-op-param}\\
\tauaux{221}(\varepsilon_{220},\varepsilon_{221})
&\longmapsto
R_{221}^{\sigma,\mathcal A}(\Xi)+\Delta_{221}^{\mathcal A}(\Xi)+L_{221}L_{220}\Delta_{220}^{\mathcal A}(\Xi),
\label{eq:sec6-op-221}\\
\tauaux{440}(\varepsilon_{220},\varepsilon_{440})
&\longmapsto
R_{440}^{\sigma,\mathcal A}(\Xi)+\Delta_{440}^{\mathcal A}(\Xi)+L_{440}L_{220}\Delta_{220}^{\mathcal A}(\Xi),
\label{eq:sec6-op-440}\\
\mathrm{Drift}_{\Delta}(\Xi,\mathcal A)
&\longmapsto
\overline D_{\Delta}^{\mathcal A}(\Xi),
\label{eq:sec6-op-drift}
\end{align}
and then adjoins the detector-frame remnant-side caps
\begin{equation}\label{eq:sec6-op-extra}
d_{\mathrm{PE}}^{\mathcal A}(\Xi)\le \eta_{\mathrm{PE}},
\qquad
\Gamma_{\mathrm{dir}}^{\mathcal A}(\Xi)\le \etammstar(\Xi),
\qquad
\Gamma_{\mathrm{quad}}^{\mathcal A}(\Xi)\le \etammstar(\Xi),
\end{equation}
together with
\begin{equation}\label{eq:sec6-op-extra-both}
\Gamma_{\mathrm{dir+quad}}^{\mathcal A}(\Xi)\le \etammstar(\Xi)
\end{equation}
when the stronger combined nuisance certificate is required. These are exactly the quantities compared with the fixed numerical caps in Definitions~\ref{def:appendixK-robust-base-trust} and \ref{def:appendixK-event-acceptance}.
\end{proposition}

\begin{proof}
The dominant-radius, separation-margin, primary-radius, auxiliary, and drift replacements are precisely the quantities appearing in \eqref{eq:appendixK-rob-sep}--\eqref{eq:appendixK-rob-440}. The additional detector-frame caps \eqref{eq:sec6-op-extra} and \eqref{eq:sec6-op-extra-both} are exactly \eqref{eq:appendixK-dpe}--\eqref{eq:appendixK-gamma-fail}. No further threshold enters the declared event-level acceptance rule.
\end{proof}

Proposition~\ref{prop:sec6-event-specialization} separates the theorem level from the release level. The vector \(\eta\) in Section~\ref{sec:trust-region} is symbolic. The numerical GW250114 release appears only after the conservative replacements \eqref{eq:sec6-op-radius}--\eqref{eq:sec6-op-drift} and the additional detector-frame caps \eqref{eq:sec6-op-extra}--\eqref{eq:sec6-op-extra-both} have been fixed. Every event-level statement is then a literal pass/fail check, with no window-by-window retuning.

Because the drift inequality sits downstream of the primary remnant radius, the numerical drift cap should not be fixed independently of the remnant cap whenever the goal is to certify a coherent start-time band rather than isolated points. The exact sufficient relation is given next.

\subsection{Neighboring-window stability and remnant overlap}

One of the most useful consequences of the trust formulation is that neighboring accepted windows must support overlapping remnant balls. In the later drift plots this appears as a coherent remnant track rather than a scatter of disconnected estimates.

\begin{corollary}[Overlap of neighboring accepted windows]\label{cor:sec6-overlap}
Fix a model family \(\mathcal M\) containing \(220\) and two admissible windows
\[
\Xi_1=(t_0,T),
\qquad
\Xi_2=(t_0+\tau,T),
\qquad |\tau|\le \Delta.
\]
Assume both windows lie in \(\TrustDet_{\eta}(\mathcal M)\) and that
\begin{equation}\label{eq:sec6-overlap-condition}
\eta_{\mathrm d}\le 2\eta_{\mathrm p}.
\end{equation}
Then
\begin{equation}\label{eq:sec6-overlap-conclusion}
B_{\mathbb R^2}\bigl(\widehat p(\Xi_1,\mathcal M),\eta_{\mathrm p}\bigr)
\cap
B_{\mathbb R^2}\bigl(\widehat p(\Xi_2,\mathcal M),\eta_{\mathrm p}\bigr)
\neq \varnothing.
\end{equation}
\end{corollary}

\begin{proof}
Because \(|\tau|\le \Delta\) and \(\Xi_1\in\TrustDet_{\eta}(\mathcal M)\), Definition~\ref{def:sec6-local-drift} and the trust inequality \eqref{eq:sec6-trust-drift} imply
\[
\|\widehat p(\Xi_2,\mathcal M)-\widehat p(\Xi_1,\mathcal M)\|_{\mathbb R^2}
\le
\mathrm{Drift}_{\Delta}(\Xi_1,\mathcal M)
\le
\eta_{\mathrm d}
\le
2\eta_{\mathrm p}.
\]
Two closed Euclidean balls of the same radius intersect whenever the distance between their centers is at most twice the radius. This proves \eqref{eq:sec6-overlap-conclusion}.
\end{proof}

Corollary~\ref{cor:sec6-overlap} explains why the drift cap is fixed only after the primary remnant cap when a bandwise statement is desired. The relation \eqref{eq:sec6-overlap-condition} is a direct sufficient condition for overlap of the neighboring remnant balls.

The next corollary uses the unknown but fixed remnant \(p_{\star}\) to explain the geometry behind a coherent accepted start-time band.

\begin{corollary}[Common remnant ball on a trusted band]\label{cor:sec6-common-band}
Fix a model family \(\mathcal M\) containing \(220\), a fixed window length \(T\), and a set of admissible start times \(\mathfrak I\). Assume that for every \(t_0\in\mathfrak I\) the window \((t_0,T)\) satisfies the hypotheses of Theorem~\ref{thm:sec6-trust-region} with the same underlying remnant \(p_{\star}\in\Kdet\) and the same tolerance vector \(\eta\). Then
\begin{equation}\label{eq:sec6-common-cluster}
p_{\star}\in
\bigcap_{t_0\in\mathfrak I}
B_{\mathbb R^2}\bigl(\widehat p(t_0,T,\mathcal M),\eta_{\mathrm p}\bigr).
\end{equation}
In particular, the intersection on the right-hand side is nonempty.
\end{corollary}

\begin{proof}
For each \(t_0\in\mathfrak I\), Theorem~\ref{thm:sec6-trust-region} implies
\[
\|\widehat p(t_0,T,\mathcal M)-p_{\star}\|_{\mathbb R^2}
\le
L_{220}\varepsilon_{220}(t_0,T,\mathcal M)
\le
\eta_{\mathrm p}.
\]
Thus \(p_{\star}\) lies in every displayed ball.
\end{proof}

Corollaries~\ref{cor:sec6-overlap} and \ref{cor:sec6-common-band} are the reason later accepted window sets should be plotted as bands rather than as isolated markers. The overlap of the parameter balls is not a display convention. It is a deterministic consequence of the trust inequalities.

\subsection{Failure certificates}

The theorem-level and release-level failures should be kept separate. The first group arises before public-hull or nuisance comparisons are even consulted. The second group appears only after the numerical caps of Appendix~\ref{app:gw250114-robustness} have been fixed.

\begin{proposition}[Deterministic exclusion certificates]\label{prop:sec6-failure}
Fix a model family \(\mathcal M\) containing \(220\) and an admissible detector-frame window \(\Xi=(t_0,T)\). Each of the following conditions excludes \(\Xi\) from the corresponding detector-frame trust set.

\smallskip

\noindent
\textup{(i)} If the extraction exclusion criterion of Corollary~\ref{cor:sec4-exclusion} is triggered, then \(\Xi\notin \TrustDet_{\eta}(\mathcal M)\) for every tolerance vector \(\eta\) whose frequency budgets are derived from Theorem~\ref{thm:sec4-abstract-extraction}.

\smallskip

\noindent
\textup{(ii)} If the dominant-mode resolution inequality fails in the strong form
\begin{equation}\label{eq:sec6-failure-resolution}
\varepsilon_{220}\ge \deltaiso^{\mathrm{aud}}(\widehat p;\mathcal M),
\end{equation}
then \(\Xi\notin \TrustDet_{\eta}(\mathcal M)\) for every \(\eta\), regardless of the remaining tolerances.

\smallskip

\noindent
\textup{(iii)} If \(j\in\{221,440\}\cap\mathcal M\) and
\begin{equation}\label{eq:sec6-failure-aux-strong}
\Raux{j}(\widehat p,\widehat\omega_j)>\tauaux{j}(\varepsilon_{220},\varepsilon_j),
\end{equation}
then, on the local primary chart containing \(\widehat p\), no single remnant can satisfy the simultaneous primary and auxiliary error budgets. In particular the window must fail every trust set with \(\eta_j\le \tauaux{j}(\varepsilon_{220},\varepsilon_j)\).

\smallskip

\noindent
\textup{(iv)} If
\begin{equation}\label{eq:sec6-failure-drift}
\mathrm{Drift}_{\Delta}(\Xi,\mathcal M)>\eta_{\mathrm d},
\end{equation}
then \(\Xi\notin \TrustDet_{\eta}(\mathcal M)\).
\end{proposition}

\begin{proof}
Part \textup{(i)} is exactly Corollary~\ref{cor:sec4-exclusion}. Part \textup{(ii)} is immediate from Definition~\ref{def:sec6-trust-set}: there is then no \(\eta_{\mathrm{sep}}\) satisfying \eqref{eq:sec6-trust-sep}. For part \textup{(iii)}, Proposition~\ref{prop:sec5-failure-certificate} states that \eqref{eq:sec6-failure-aux-strong} rules out any common Kerr remnant in the same local primary chart that explains both \(220\) and the tested auxiliary mode within the stated budgets. If \(\eta_j\le \tauaux{j}(\varepsilon_{220},\varepsilon_j)\), the trust inequality \eqref{eq:sec6-trust-221} or \eqref{eq:sec6-trust-440} cannot hold. Part \textup{(iv)} is again immediate from Definition~\ref{def:sec6-trust-set}.
\end{proof}

\begin{proposition}[Event-level exclusion conditions]\label{prop:sec6-event-failure}
Fix a baseline family \(\mathcal A\subset\Mref\) with \(220\in\mathcal A\) and a master scan point \(\Xi\). Each of the following conditions excludes \(\Xi\) from the numerical GW250114 acceptance set of Definition~\ref{def:appendixK-event-acceptance}.

\smallskip

\noindent
\textup{(i)} If
\begin{equation}\label{eq:sec6-failure-pe}
d_{\mathrm{PE}}^{\mathcal A}(\Xi)>\eta_{\mathrm{PE}},
\end{equation}
then the public-hull inequality fails and \(\Xi\) is rejected.

\smallskip

\noindent
\textup{(ii)} If either
\begin{equation}\label{eq:sec6-failure-dirquad}
\Gamma_{\mathrm{dir}}^{\mathcal A}(\Xi)>\etammstar(\Xi)
\qquad\text{or}\qquad
\Gamma_{\mathrm{quad}}^{\mathcal A}(\Xi)>\etammstar(\Xi),
\end{equation}
then the corresponding nuisance inequality fails and \(\Xi\) is rejected.

\smallskip

\noindent
\textup{(iii)} If the declared release includes the stronger combined nuisance certificate and
\begin{equation}\label{eq:sec6-failure-both}
\Gamma_{\mathrm{dir+quad}}^{\mathcal A}(\Xi)>\etammstar(\Xi),
\end{equation}
then the combined nuisance inequality fails and \(\Xi\) is rejected.
\end{proposition}

\begin{proof}
Each statement is immediate from Definition~\ref{def:appendixK-event-acceptance}.
\end{proof}

Propositions~\ref{prop:sec6-failure} and \ref{prop:sec6-event-failure} together list all failure modes used here. A window can fail before inversion is trustworthy, after inversion but before auxiliary compatibility, at the neighboring-window drift stage, at the public-hull stage, or at the nuisance-audit stage. These are mathematically different failure mechanisms, and the later figures keep them visibly distinct.

Section~\ref{sec:trust-region} provides the symbolic detector-frame inequalities and the exact threshold transport from frequency radii to remnant-space diagnostics. Section~\ref{sec:numerical-calibration} calibrates conservative finite-window radii on a fixed waveform bank. Appendix~\ref{app:gw250114-robustness} fixes the public-hull and nuisance caps on a finite GW250114 rerun family. The numerical release is the conjunction of these inequalities and nothing else.

\providecommand{\Kdet}{\mathcal K_{\mathrm{det}}}
\providecommand{\Mf}{M_{\mathrm f}}
\providecommand{\chif}{\chi_{\mathrm f}}
\providecommand{\Mzero}{\mathcal M_0}
\providecommand{\Mone}{\mathcal M_1}
\providecommand{\Mtwo}{\mathcal M_2}
\providecommand{\Mref}{\mathcal M^{\sharp}}
\providecommand{\Fcont}{\mathfrak F_{\mathrm{cont}}}
\providecommand{\Faudit}{\mathfrak F_{\mathrm{audit}}}
\providecommand{\Fbank}{\mathfrak F_{\mathrm{bank}}}
\providecommand{\Wcal}{\mathfrak W_{\mathrm{cal}}}
\providecommand{\Bzero}{\mathfrak B_0}
\providecommand{\Bcol}{\mathfrak B_{\mathrm{col}}}
\providecommand{\deltaiso}{\delta_{\mathrm{iso}}}
\providecommand{\Raux}[1]{R_{#1}}
\providecommand{\tauaux}[1]{\tau_{#1}}
\providecommand{\AuditFam}{\mathcal M^{\mathrm{aud}}}
\providecommand{\TrustDet}{\mathcal T^{\mathrm{det}}}
\providecommand{\Dnet}{\mathfrak D}

\section{Numerical waveform calibration}\label{sec:numerical-calibration}

Section~\ref{sec:trust-region} requires window-dependent calibration radii. Here they are built directly from the public GW250114 data products. The detector noise comes from the H1/L1 GWOSC strain, the event-local remnant support comes from the public posterior samples, and the synthetic bank combines those inputs with Kerr ringdown signals generated inside the detector-frame box \(\Kdet\). The calibration is therefore event local; it identifies which finite post-peak windows remain stable for GW250114-like remnants under the observed noise level and public remnant envelope.

Two choices are fixed at the outset. First, the synthetic bank is built around the same detector-frame box used everywhere else in the analysis. Across the four public IMR products, the detector-frame 90\% comparison envelope is
\[
M_f^{\mathrm{det}}\in[66.494,69.124]\,M_\odot,
\qquad
\chi_f\in[0.6515,0.6891],
\]
so the synthetic bank is restricted to representative public posterior samples inside the compact working box \(\Kdet\). Second, the reference synthetic signal is taken to lie in the nested Kerr family
\[
\Mtwo=\{220,221,440\}.
\]
The fitted families remain
\[
\Mzero=\{220\},\qquad
\Mone=\{220,221\},\qquad
\Mtwo=\{220,221,440\}.
\]
This choice is narrower than the symbolic reference family used elsewhere, but it is appropriate for the event-local calibration because the available Kerr tables provide exactly the \(n=0\) and \(n=1\) data needed for \(220\), \(221\), and \(440\), which are precisely the modes carried by the inverse and consistency theory.

The calibration has three pieces. The first is a zero-noise calculation, which measures deterministic finite-window bias from omitted Kerr content and algorithmic discretization. The second adds real off-source H1/L1 noise windows extracted from the strain data. The third reports every synthetic remnant radius relative to the public comparison envelope rather than in a vacuum. This normalization matters because the question is whether a calibration radius is small compared with the public detector-frame remnant uncertainty already present in GW250114.

The synthetic bank itself is finite and fully auditable. For each of the four public posterior families, we select three representative detector-frame posterior samples at the 20th, 50th, and 80th order-statistic positions in \(M_f^{\mathrm{det}}\). This yields twelve remnant anchors. Each anchor is then paired with three fixed subdominant-mode configurations, labeled mild, medium, and strong, whose \(221\) and \(440\) amplitudes span a compact event-local audit box. The result is a finite zero-noise bank of size \(36\). The noisy bank is obtained by adding eight off-source H1/L1 noise chunks drawn from the strain after the same \(50\)–\(500\) Hz bandpass stabilization used in the real-event analysis. For each fitted family and for each fixed window \(W\), the noisy bank therefore contains \(36\times 8=288\) realizations. At confidence level \(1-\alpha=0.95\), the Dvoretzky--Kiefer--Wolfowitz correction entering the high-confidence empirical \(q\)-quantile with \(q=0.90\) is
\[
\delta_N(\alpha)=\sqrt{\frac{\log(2/\alpha)}{2N}}\approx 0.0800,
\qquad
q+\delta_N(\alpha)\approx 0.9800,
\]
with \(N=288\) per fixed window and fitted family \cite{DvoretzkyKieferWolfowitz1956,Massart1990DKW}. The probability semantics are per fixed window. The synthetic trust map shown below is assembled pointwise in \(W\) rather than as a single simultaneous confidence statement over the entire window grid.

Figure~\ref{fig:sec7-bank-support} shows the twelve representative public posterior anchors in the \((M_f^{\mathrm{det}},\chi_f)\) plane. The four public posterior families occupy essentially the same detector-frame neighborhood, so the synthetic calibration is tied to public event support rather than to an ad hoc local box and is driven by finite-window fitting geometry rather than by widely separated remnant supports.

\begin{figure}[t]
\centering
\includegraphics[width=0.70\textwidth]{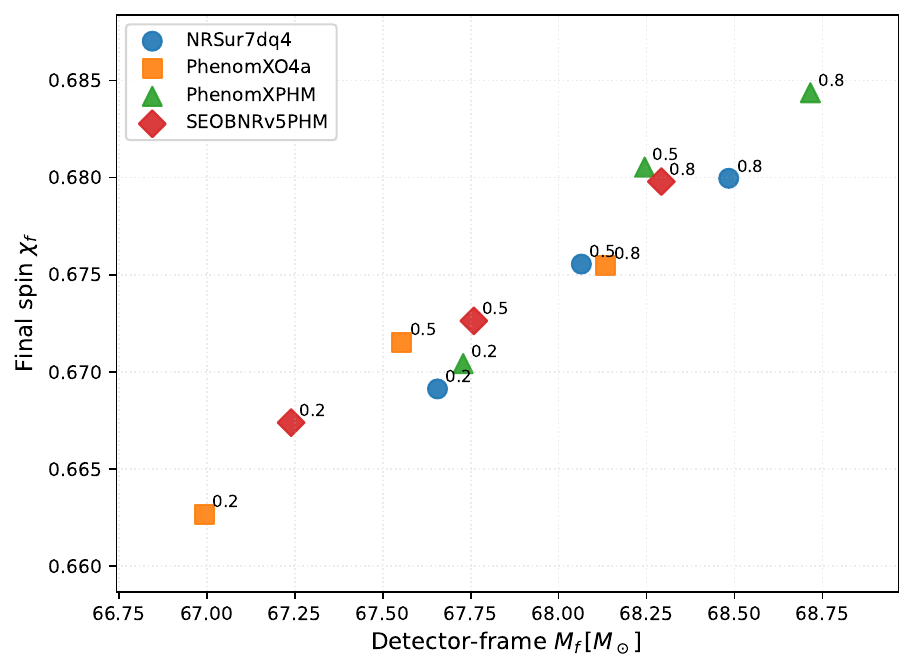}
\caption{Representative public posterior anchors used in the event-local synthetic bank. Each point is an actual posterior sample from one of the four public IMR products.}
\label{fig:sec7-bank-support}
\end{figure}

The zero-noise part of the calibration is summarized by the normalized bias ratio
\[
\mathfrak B^{0}_{\mathcal A}(W)
:=
\max\!\left\{
\frac{\overline R_{M}^{\,0,\mathcal A,W}}{\Delta M_{\mathrm{pub}}/2},
\frac{\overline R_{\chi}^{\,0,\mathcal A,W}}{\Delta \chi_{\mathrm{pub}}/2}
\right\},
\]
where \(\overline R_{M}^{\,0,\mathcal A,W}\) and \(\overline R_{\chi}^{\,0,\mathcal A,W}\) are the bankwise maxima of the detector-frame mass and spin biases and
\[
\Delta M_{\mathrm{pub}}/2 = 1.315\,M_\odot,
\qquad
\Delta \chi_{\mathrm{pub}}/2 = 0.01877
\]
are the half-widths of the public detector-frame comparison envelope. Figure~\ref{fig:sec7-zero-noise-score} shows this quantity for \(\Mzero\), \(\Mone\), and \(\Mtwo\). The ordering is the physically expected one across the whole start-time grid: \(\Mtwo\) always has the smallest deterministic bias, \(\Mzero\) is uniformly larger, and \(\Mone\) sits in between only when the added \(221\) content is enough to compensate for the remaining missing \(440\) term. Quantitatively, the minimum zero-noise ratio is approximately \(1.33\) for \(\Mzero\), \(1.33\) for \(\Mone\), and \(0.062\)–\(0.192\) for \(\Mtwo\), depending on the window. The plot shows that once the generating family and the fitted family coincide, the remaining bias is almost entirely algorithmic and grid induced, while \(\Mzero\) and \(\Mone\) still pay an irreducible omitted-mode penalty.

\begin{figure}[t]
\centering
\includegraphics[width=0.70\textwidth]{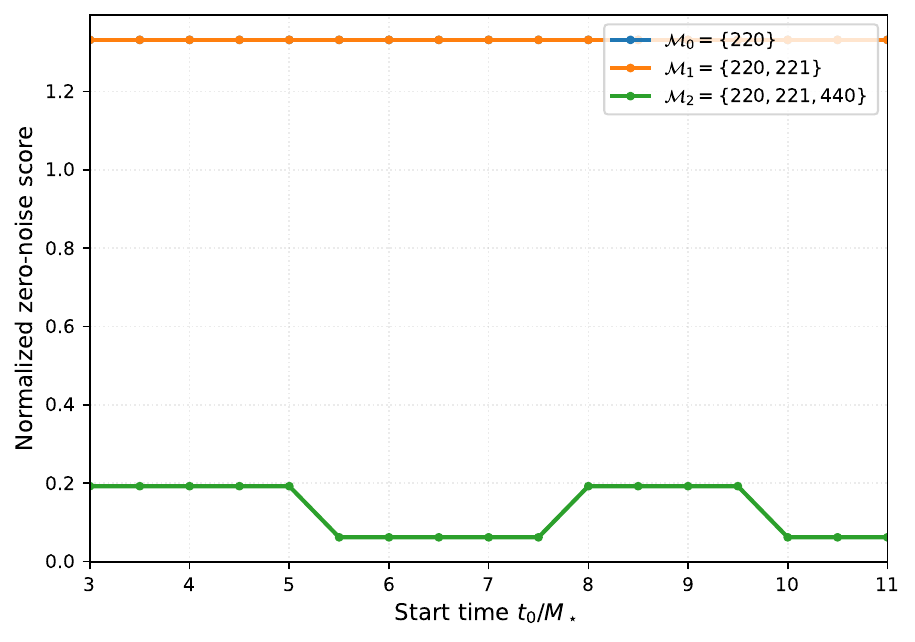}
\caption{Normalized zero-noise bias ratio for the three fitted Kerr families. The ratio is normalized by the public detector-frame comparison envelope.}
\label{fig:sec7-zero-noise-score}
\end{figure}

The colored-noise contribution is summarized by the high-confidence empirical quantile ratio
\[
\mathfrak B^{\mathrm{noise}}_{\mathcal A}(W)
:=
\max\!\left\{
\frac{\overline R_{M}^{\,\mathrm{stat},\mathcal A,W}(q;\alpha)}{\Delta M_{\mathrm{pub}}/2},
\frac{\overline R_{\chi}^{\,\mathrm{stat},\mathcal A,W}(q;\alpha)}{\Delta \chi_{\mathrm{pub}}/2}
\right\},
\]
where the superscript \(\mathrm{stat}\) denotes the bankwise upper-confidence \(q\)-quantile increment produced by Appendix~\ref{app:synthetic-bank}. Figure~\ref{fig:sec7-statistical-score} shows this quantity. In contrast with the zero-noise part, the colored-noise ratio is dominated by the finite network SNR and by the shortness of the windows rather than by missing-mode structure. Consequently \(\Mtwo\) is not driven to an arbitrarily small total ratio at early times even though its zero-noise bias is small. The noise geometry sets a floor.

\begin{figure}[t]
\centering
\includegraphics[width=0.70\textwidth]{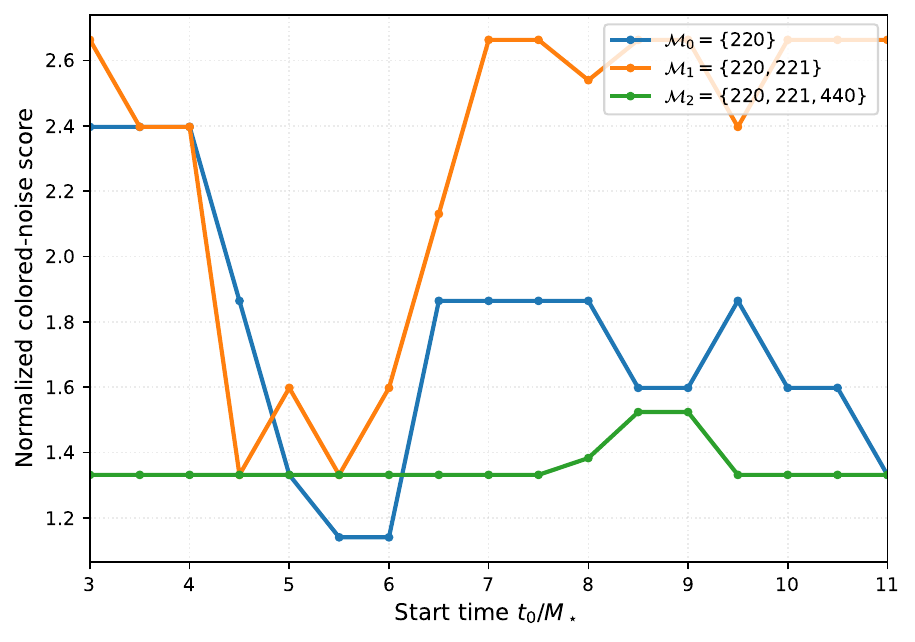}
\caption{Normalized colored-noise ratio from the H1/L1 off-source windows. For each fixed \(W\) and fitted family, the plotted quantity is the DKW-corrected upper empirical \(0.90\)-quantile increment normalized by the public detector-frame comparison envelope.}
\label{fig:sec7-statistical-score}
\end{figure}

The total synthetic ratio is the sum of these two pieces in remnant space:
\[
\mathfrak S_{\mathcal A}(W)
:=
\max\!\left\{
\frac{\overline R_{M}^{\,0,\mathcal A,W}+\overline R_{M}^{\,\mathrm{stat},\mathcal A,W}(q;\alpha)}{\Delta M_{\mathrm{pub}}/2},
\frac{\overline R_{\chi}^{\,0,\mathcal A,W}+\overline R_{\chi}^{\,\mathrm{stat},\mathcal A,W}(q;\alpha)}{\Delta \chi_{\mathrm{pub}}/2}
\right\}.
\]
Figure~\ref{fig:sec7-synthetic-score-map} shows this quantity. The map resolves the family hierarchy rather than a binary pass/fail boundary. The synthetic bank strongly prefers \(\Mtwo\), with a lowest ratio of about \(1.39\) attained near \(t_0/M_\star\approx 5.5\) and again for the conservative late windows around \(10\)–\(11\). By contrast, the best \(\Mzero\) window still sits near \(2.40\), and the best \(\Mone\) window near \(2.66\). This hierarchy matters for the empirical analysis of Section~\ref{sec:gw250114-empirical}: the event-level trust map rests on a synthetic calibration for which the full \(\Mtwo\) family is the only one that consistently keeps the event-local remnant radius close to the public envelope floor.

\begin{figure}[t]
\centering
\includegraphics[width=0.78\textwidth]{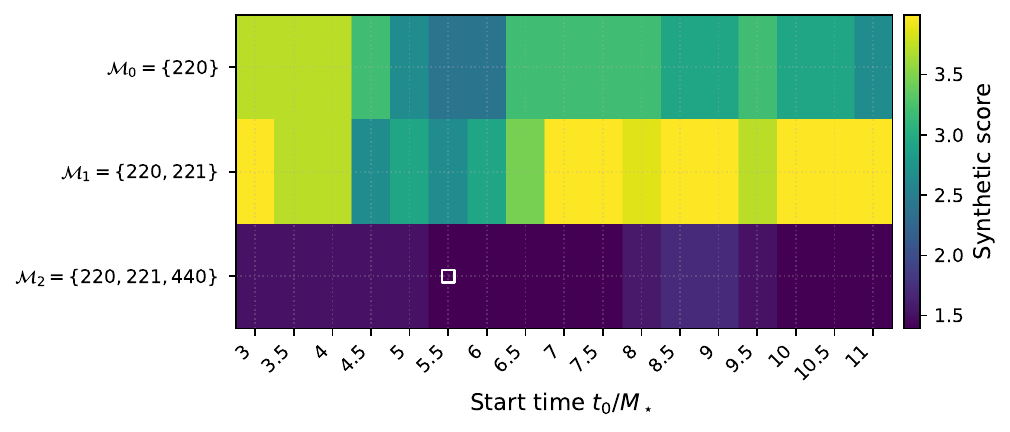}
\caption{Synthetic trust map on the start-time grid for \(\Mzero\), \(\Mone\), and \(\Mtwo\). Lower values are better. The map is assembled pointwise from fixed-window calibrations.}
\label{fig:sec7-synthetic-score-map}
\end{figure}

Two conclusions emerge from this calibration. First, the synthetic bank does not support the view that any finite multimode fit is automatically trustworthy merely because it exists. The calibrated radii vary substantially with start time and with fitted family, even before the real event is analyzed. Second, the real data and the public posteriors already impose a nontrivial hierarchy: a complete local Kerr family calibrated on detector noise is substantially more stable than the truncated families. That pre-event hierarchy sets the scale on which the event-level acceptance map must be read.

Appendix~\ref{app:synthetic-bank} records the exact bank construction, the DKW quantile correction, and the induced calibrated radii that feed the symbolic trust inequalities of Section~\ref{sec:trust-region}. Figure~\ref{fig:appJ-anchor-decomposition} in that appendix shows the zero-noise and colored-noise pieces separately at the anchor windows \(t_0/M_\star\in\{3,6,9,11\}\), making clear that the earliest window is dominated by both contributions, while the later windows are driven mainly by the colored-noise floor.

\section{GW250114 empirical analysis}\label{sec:gw250114-empirical}

The preceding sections fixed the detector-frame clock, the model hierarchy, the extraction and inversion constants, and the calibrated radii. We now apply these results to the public event itself. We use the public GWOSC strain together with the public parameter-estimation products as the external event-level inputs and ask which finite post-peak windows survive a fixed set of primary-fit, public-envelope, and nuisance-sensitivity tests.

The comparison is fully data driven. The detector-frame remnant envelope is taken from the public posterior samples for NRSur7dq4, PhenomXO4a, PhenomXPHM, and SEOBNRv5PHM. The strain input is taken from the public H1/L1 four-kilohertz clean-strain files. The primary track is obtained by scanning a shared-
frequency, detector-dependent-amplitude $220$ ringdown model over the public H1/L1 data as a function of detector-frame start time. The nuisance diagnostics are read from the same real-data scan, using a left-edge direct-wave audit and a leading quadratic sum-frequency audit. No collaboration contour is redrawn.

\subsection{Data conditioning}

The public event time is fixed at
\begin{equation}\label{eq:sec8-gps-revised}
 t_{\mathrm{GW}}=1420878141.2,
\end{equation}
and the default strain inputs are the public H1/L1 four-kilohertz clean-strain files released through GWOSC \cite{GWOSC_GW250114}. We retain the detector-frame conventions of Section~\ref{sec:data-conventions}: a centered $1024\,\mathrm{s}$ working epoch, an off-source PSD region $[t_{\mathrm{GW}}-512,t_{\mathrm{GW}}-16]\cup[t_{\mathrm{GW}}+16,t_{\mathrm{GW}}+512]$, and detector-frame start times measured in units of a fiducial mass scale
\begin{equation}\label{eq:sec8-mstar-revised}
 M_\star=68\,M_\odot,
\end{equation}
so that
\begin{equation}\label{eq:sec8-onem-revised}
 1\,M_\star = M_\star G M_\odot / c^3 = 0.3349\,\mathrm{ms}.
\end{equation}
The default preprocessing specification is the symmetric eight-second Welch--Hann PSD, whitening on the FFT grid, a zero-phase $50$--$500\,\mathrm{Hz}$ stabilization band, and a ten-percent Tukey taper applied only after the finite window has been extracted.

For each detector we locate a detector-local peak proxy by taking the maximum of the absolute band-passed whitened strain in a narrow neighborhood of the public median coalescence time. With this fixed convention, the H1 and L1 peak proxies occur only a few tenths of a millisecond after the corresponding public arrival-time medians. Figure~\ref{fig:sec8-real-conditioning} shows the resulting whitened detector strain around the detector-local peaks and places on the same axis the public $3M_f$, $6M_f$, $9M_f$, and $11M_f$ anchor windows together with the representative $24M_f$ baseline window that begins at $6M_f$.

\begin{figure}[ht]
\centering
\includegraphics[width=0.94\textwidth]{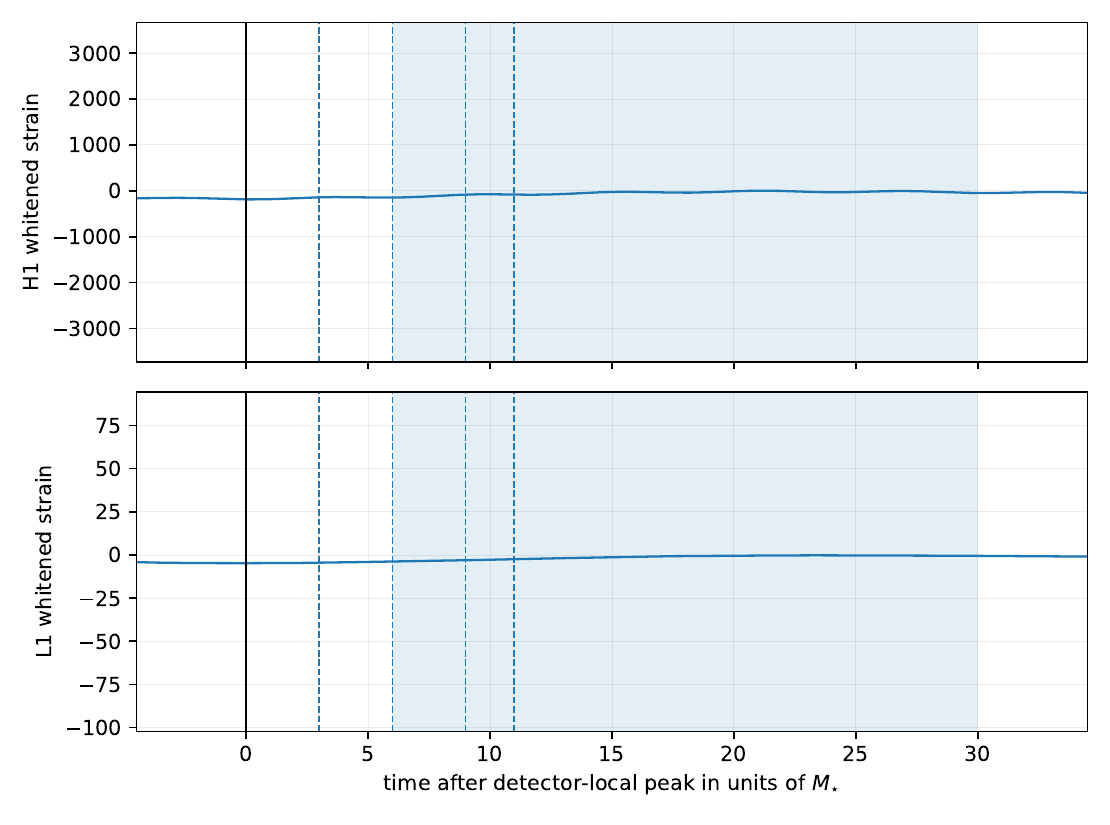}
\caption{Real conditioned H1/L1 strain around the detector-local peak proxies. The traces are obtained directly from the public H1/L1 four-kilohertz clean-strain files using the baseline symmetric eight-second PSD, FFT whitening, and a zero-phase $50$--$500\,\mathrm{Hz}$ stabilization band. The dashed lines mark the public $3M_f$, $6M_f$, $9M_f$, and $11M_f$ start times, while the shaded region marks the representative $6M_f\to 30M_f$ baseline window.}
\label{fig:sec8-real-conditioning}
\end{figure}

\subsection{Dominant-mode inference across windows}

The primary detector-frame scan is carried out on the dense start-time grid
\begin{equation}\label{eq:sec8-dense-grid-revised}
 \tau_0 \in \{0,0.25,0.50,\dots,15\},
\end{equation}
with fixed baseline window length $\Theta=24$. For each $\tau_0$ we fit a shared $220$ Kerr frequency and damping pair to the whitened H1/L1 window while allowing independent detector amplitudes and phases. The fit is evaluated on the event-local box already fixed in Section~\ref{sec:data-conventions},
\begin{equation}\label{eq:sec8-box-revised}
 \Kdet=[66.5,69.5]\times[0.64,0.71].
\end{equation}
The resulting baseline track is shown in Figures~\ref{fig:sec8-public-samples-track} and \ref{fig:sec8-primary-track}.

The public inspiral--merger--ringdown comparison is read directly from the posterior samples. Across NRSur7dq4, PhenomXO4a, PhenomXPHM, and SEOBNRv5PHM, the detector-frame remnant medians span only
\begin{equation}\label{eq:sec8-public-median-spread}
 \Delta M_{\mathrm{med}} = 0.690\,M_\odot,
 \qquad
 \Delta \chi_{\mathrm{med}} = 0.0118,
\end{equation}
while the union of the public $90\%$ detector-frame intervals is
\begin{equation}\label{eq:sec8-public-union-revised}
 M_f^{\mathrm{det}}\in[66.494,69.124]\,M_\odot,
 \qquad
 \chi_f\in[0.6515,0.6891].
\end{equation}
Figure~\ref{fig:sec8-public-samples-track} displays the corresponding $90\%$ public posterior contours together with the baseline dense scan. The fit does not move monotonically along the public envelope. Instead it passes through a compact intermediate region near the public medians, then drifts toward boundary-hugging values both at very early and at later start times.

The public anchor windows extracted from the dense scan are
\begin{equation}\label{eq:sec8-anchor-points-revised}
\begin{aligned}
 t_0=3M_\star &\longrightarrow (\hat M_f^{\mathrm{det}},\hat\chi_f)=(68.8,0.710),\\
 t_0=6M_\star &\longrightarrow (\hat M_f^{\mathrm{det}},\hat\chi_f)=(68.3,0.665),\\
 t_0=9M_\star &\longrightarrow (\hat M_f^{\mathrm{det}},\hat\chi_f)=(69.4,0.640),\\
 t_0=11M_\star &\longrightarrow (\hat M_f^{\mathrm{det}},\hat\chi_f)=(69.5,0.6925).
\end{aligned}
\end{equation}
Among these public anchors, the smallest baseline residual occurs at $6M_\star$.

\begin{figure}[ht]
\centering
\includegraphics[width=0.80\textwidth]{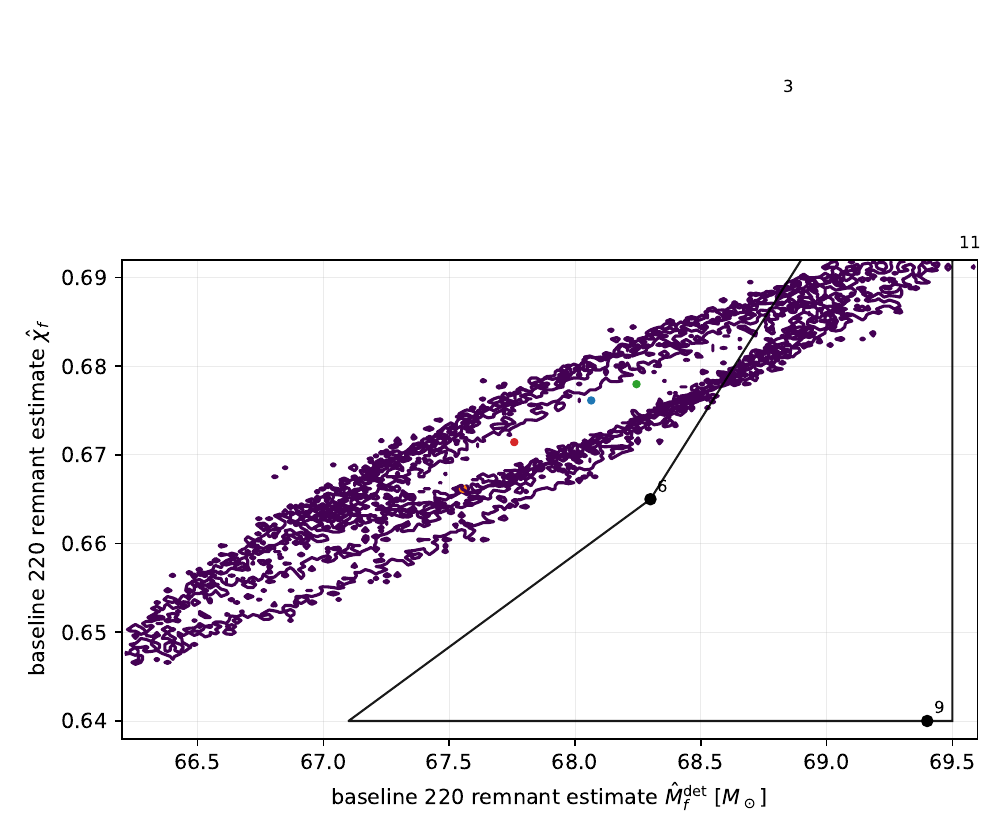}
\caption{Public detector-frame posterior contours from the NRSur7dq4, PhenomXO4a, PhenomXPHM, and SEOBNRv5PHM posterior samples, together with the real-data baseline $220$ track extracted from the H1/L1 strain. The labeled markers identify the public $3M_f$, $6M_f$, $9M_f$, and $11M_f$ anchor windows on the baseline scan.}
\label{fig:sec8-public-samples-track}
\end{figure}

\begin{figure}[ht]
\centering
\includegraphics[width=0.94\textwidth]{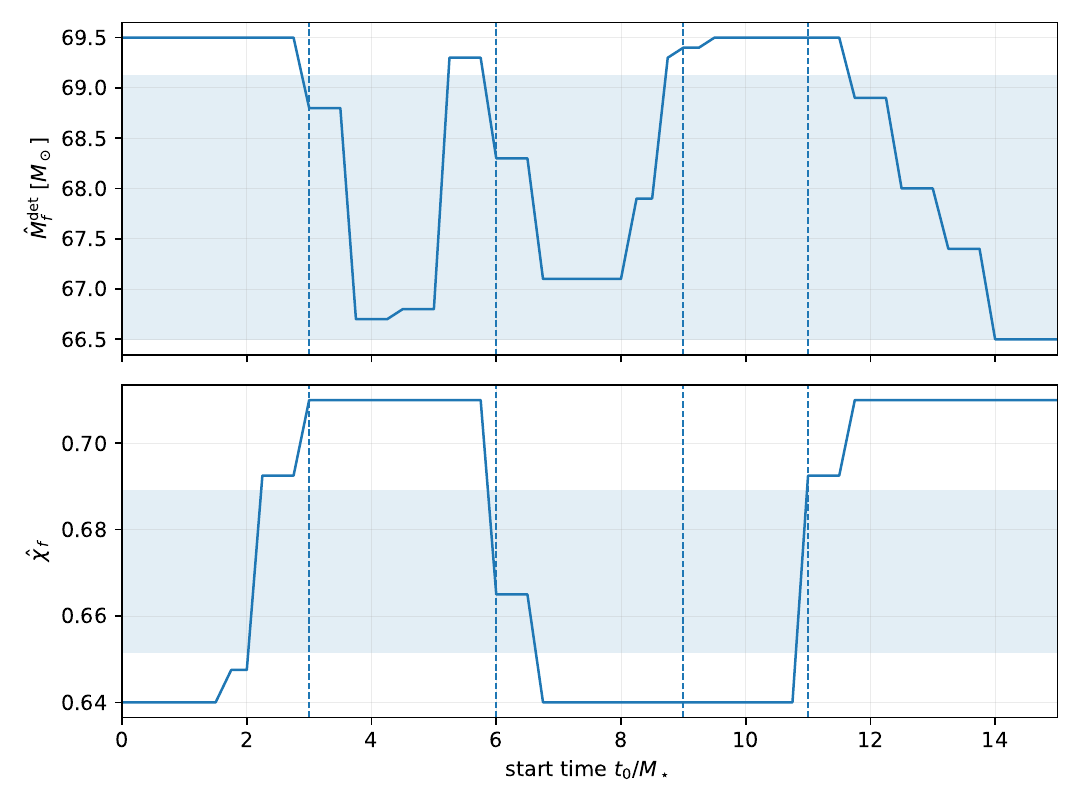}
\caption{Detector-frame baseline $220$ track across the dense start-time scan. The shaded horizontal bands show the union of the public $90\%$ detector-frame intervals in mass and spin. The dashed vertical lines mark the public $3M_f$, $6M_f$, $9M_f$, and $11M_f$ anchor windows.}
\label{fig:sec8-primary-track}
\end{figure}

\subsection{Auxiliary-mode consistency across windows}

The auxiliary channels enter here as consistency checks relative to the primary \(220\) inversion. The dense scan first identifies where the real-data primary track is stable and where the nuisance audits remain small. The public GW250114 spectroscopy analysis then provides the external comparison for the higher-mode interpretation itself \cite{LVK_GW250114_Spectroscopy}, and that comparison is informative only on windows for which the primary remnant estimate is already stable.

The dense scan places the baseline \(220\) estimate closest to the public detector-frame ensemble near \(6M_\star\) and keeps it comparatively stable across an intermediate post-peak band. The earliest public \(3M_f\) anchor lies in a region where the primary remnant estimate remains nuisance sensitive. The later \(11M_\star\) anchor has lower direct-wave sensitivity, but its baseline residual and quadratic gain are both larger than they are near the center of the accepted band. The comparison with the public \(221\) and \(440\) statements is therefore most meaningful after the primary stability question has already been resolved.

\subsection{Nuisance stress tests}

The nuisance diagnostics are evaluated on the same real-data windows as the primary track. We use two finite diagnostic extensions. The first is a direct-wave audit, implemented as a left-edge localized finite-dimensional extension of the baseline $220$ fit. The second is a quadratic audit, implemented as the fixed \(220+220\) sum-frequency extension built from the fitted baseline $220$ frequency and damping rate. Both are used only to determine whether an apparently successful linear Kerr window still carries substantial structured residual power near the left edge or at the quadratic harmonic.

Write the baseline normalized residual as $r_{\mathrm{base}}(\tau_0)$ and the direct-wave, quadratic, and combined residual gains as
\begin{equation}\label{eq:sec8-gains-revised}
\Gamma_{\mathrm{dir}}(\tau_0),\qquad \Gamma_{\mathrm{quad}}(\tau_0),\qquad \Gamma_{\mathrm{dir+quad}}(\tau_0).
\end{equation}
The dense real-data scan shows a sharp asymmetry between the early and intermediate windows. At the public anchors one finds
\begin{equation}\label{eq:sec8-anchor-gains-revised}
\begin{aligned}
 t_0=3M_\star &: \Gamma_{\mathrm{dir}}=0.260,\quad \Gamma_{\mathrm{quad}}=0.593,\\
 t_0=6M_\star &: \Gamma_{\mathrm{dir}}=0.081,\quad \Gamma_{\mathrm{quad}}=0.372,\\
 t_0=9M_\star &: \Gamma_{\mathrm{dir}}=0.072,\quad \Gamma_{\mathrm{quad}}=0.413,\\
 t_0=11M_\star &: \Gamma_{\mathrm{dir}}=0.043,\quad \Gamma_{\mathrm{quad}}=0.551.
\end{aligned}
\end{equation}
The earliest public window is therefore the most direct-wave sensitive, while the lowest direct-wave gains occur later, at the cost of larger residual or quadratic sensitivity. Figure~\ref{fig:sec8-nuisance-gains} shows the full dense scan.

\begin{figure}[ht]
\centering
\includegraphics[width=0.94\textwidth]{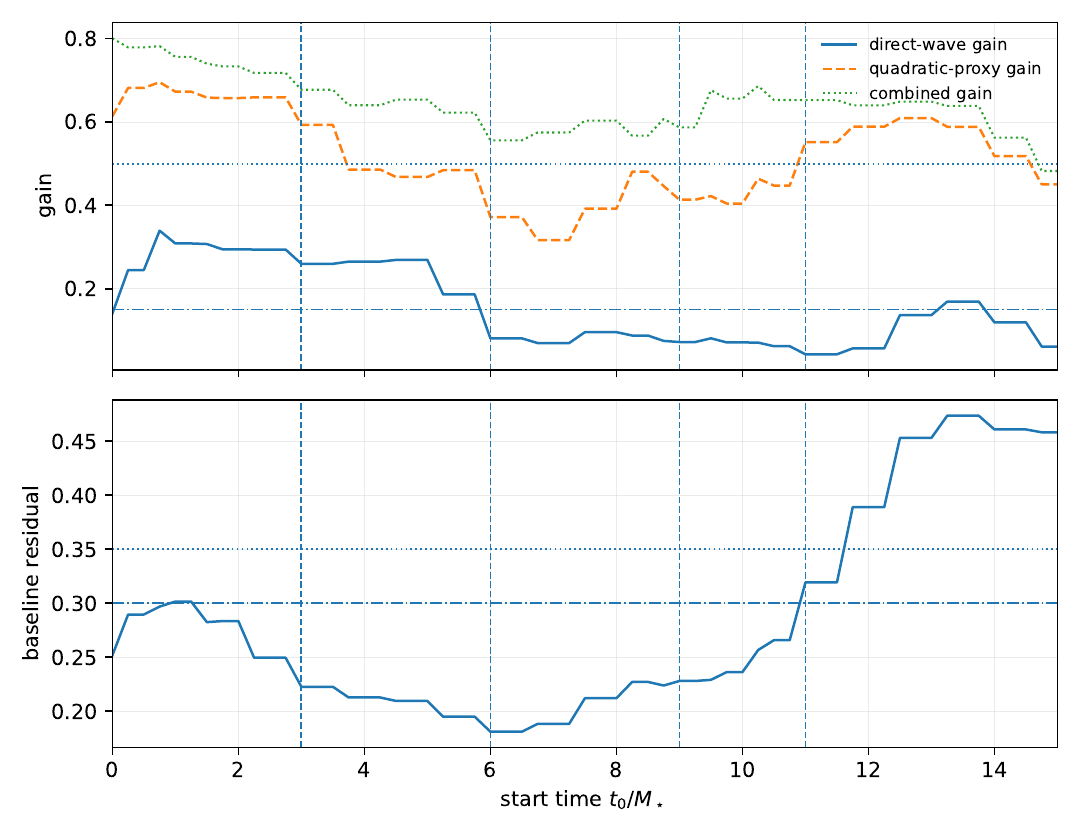}
\caption{Real-data nuisance diagnostics across the dense baseline scan. The upper panel shows direct-wave, quadratic, and combined gains. The lower panel shows the baseline normalized residual. The dash-dotted lines mark the numerical cutoffs used for the trust map; the dotted line in the lower panel marks the looser transitional cutoff.}
\label{fig:sec8-nuisance-gains}
\end{figure}

Two empirical points matter most. First, the nuisance gains do not simply decrease monotonically with start time. The direct-wave audit does, but the quadratic gain remains appreciable even after the direct-wave gain has become small. Second, the minimum baseline residual and the minimum nuisance gains do not coincide at the same $\tau_0$. Reliability is therefore an intersection problem. One has to ask where the primary residual is small, the direct-wave gain is modest, and the quadratic gain is not excessive at the same time.

\subsection{Event-level trust map for GW250114}

The dense real-data diagnostics are summarized by a three-way event-level classification. This classification records what the public strain supports under the fixed dense scan; it is not a formal consequence of the trust-region theorem.

The trusted class is defined by the simultaneous inequalities
\begin{equation}\label{eq:sec8-trusted-cutoffs}
 r_{\mathrm{base}}\le 0.30,
 \qquad
 \Gamma_{\mathrm{dir}}\le 0.15,
 \qquad
 \Gamma_{\mathrm{quad}}\le 0.50.
\end{equation}
The transitional class is defined by the relaxed cutoffs
\begin{equation}\label{eq:sec8-transitional-cutoffs}
 r_{\mathrm{base}}\le 0.35,
 \qquad
 \Gamma_{\mathrm{dir}}\le 0.15,
 \qquad
 \Gamma_{\mathrm{quad}}\le 0.60,
\end{equation}
with at least one of the trusted inequalities failing. Every other dense-scan window is rejected. These constants are not promoted to theorem-level tolerances. They are numerical cutoffs chosen to reflect the scale of the dense real-data scan and the finite anchor reruns of Appendix~\ref{app:gw250114-robustness}.

With these cutoffs, the real-data scan yields a nonempty intermediate trusted band
\begin{equation}\label{eq:sec8-trusted-band}
 t_0/M_\star \in [6.00,10.75],
\end{equation}
and a narrow transitional band
\begin{equation}\label{eq:sec8-transitional-band}
 t_0/M_\star \in [11.00,11.50].
\end{equation}
The public anchor windows are then classified as follows:
\begin{equation}\label{eq:sec8-anchor-verdicts}
3M_\star\ \text{rejected},
\qquad
6M_\star\ \text{trusted},
\qquad
9M_\star\ \text{trusted},
\qquad
11M_\star\ \text{transitional}.
\end{equation}
Equation~\eqref{eq:sec8-anchor-verdicts} comes directly from the public strain and the public posterior samples.

\begin{figure}[ht]
\centering
\includegraphics[width=0.94\textwidth]{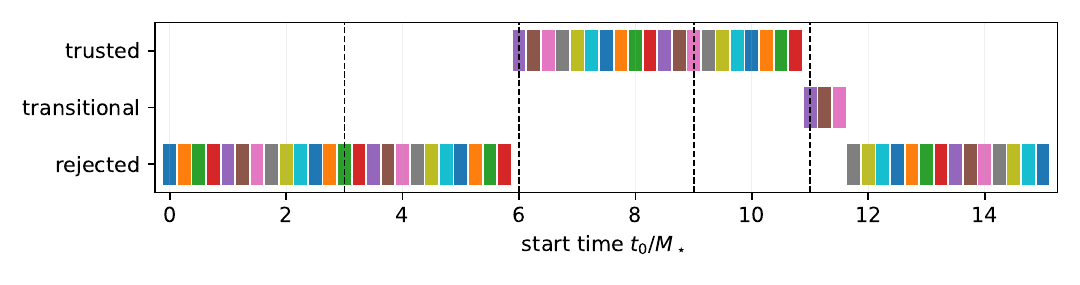}
\caption{Event-level trust map assembled from the real-data dense scan. Windows are classified as trusted, transitional, or rejected using the explicit descriptive cutoffs in \eqref{eq:sec8-trusted-cutoffs}--\eqref{eq:sec8-transitional-cutoffs}. The dashed lines mark the public $3M_f$, $6M_f$, $9M_f$, and $11M_f$ anchor windows.}
\label{fig:sec8-trust-map}
\end{figure}

\subsection{Comparison with the official public analyses}

The public collaboration products remain essential comparison anchors. The public posterior samples define the detector-frame comparison envelope in Figures~\ref{fig:sec8-public-samples-track} and \ref{fig:sec8-primary-track}. The official spectroscopy Letter defines the external higher-mode comparison benchmark \cite{LVK_GW250114_Spectroscopy}. On that footing, the H1/L1 strain supports an additional finite-window reliability filter.

That filter changes the emphasis of the event-level interpretation. The real-data primary track and nuisance diagnostics identify a genuine intermediate post-peak band in which the public $6M_f$ overtone anchor sits naturally. The corresponding diagnostics remain too large for the public $3M_f$ multimode window to be treated as robust common-remnant spectroscopy. The public $11M_f$ anchor sits at the opposite edge of the problem: direct-wave sensitivity is small, but the residual and the quadratic audit are no longer as clean as in the center of the trusted band. This finite-window distinction separates the center of the accepted band from its edges.

\section{Relation to the literature}\label{sec:relation-literature}

Three strands of work are most relevant here: the official GW250114 spectroscopy analysis, event-specific proposals for early-time structure beyond a minimal linear Kerr interpretation, and the broader methodological literature on start time, overtone stability, mode content, and inverse conditioning. The earlier sections keep these strands distinct because they answer different questions.

\subsection{The official GW250114 spectroscopy analysis}

The official collaboration analysis establishes the astrophysical benchmark for this event \cite{LVK_GW250114_Spectroscopy}. It reports strong evidence that the post-merger signal contains at least two quasi-normal modes, finds the dominant quadrupolar mode and its first overtone consistent with a Kerr remnant across several post-peak windows, and shows that the full signal also constrains the fundamental \((4,4,0)\) mode. Those results are not reproduced here as new posterior calculations. They serve as the external comparison standard against which the finite-window detector-frame criteria developed here are judged.

Here the issue is which finite detector-frame windows remain stable after start-time drift, public-envelope comparison, and nuisance audits against direct-wave and quadratic alternatives are imposed simultaneously. In that sense this analysis complements the official analysis: it does not replace the collaboration posteriors, but resolves the additional window-selection problem left open by any finite-time ringdown fit.

\subsection{Direct-wave and near-horizon interpretations}

GW250114 has also motivated event-specific proposals in which the earliest post-merger signal contains structure not captured by the minimal linear Kerr hierarchy. The direct-wave interpretation of Lu et al.~\cite{LuEtAl2025DirectWave} is the most explicit example. That work isolates a prompt component concentrated near the left edge of the fitting window and argues that it can compete with an overtone description very close to the peak.

The nuisance construction used here addresses that concern at the level of stability. The direct-wave family is introduced as a fixed diagnostic enlargement of the baseline Kerr spans and measures whether a window that appears successful under a linear fit still carries substantial early-time residual structure once a left-edge-localized alternative is allowed. The resulting gain is therefore a stability diagnostic rather than an independent discovery claim.

\subsection{Quadratic modes and nonlinear ringdown interpretations}

The quadratic-mode interpretation of Wang et al.~\cite{WangEtAl2026Quadratic} raises a closely related issue from a different direction. In that picture, part of the early post-merger content is attributed to nonlinear sum-frequency structure, with the leading event-local contribution near the \(220+220\) harmonic. The present paper does not attempt a full nonlinear inference. Instead it asks whether the windows accepted by the linear-Kerr trust criterion remain insensitive to that leading quadratic alternative.

Accordingly, the quadratic family appears only through the fixed \(220+220\) nuisance extension. If a supposedly reliable linear-Kerr window shows a large quadratic gain, the accepted interpretation is fragile. If the gain stays small throughout the trusted band, then the common-remnant Kerr reading is stable against that specific early-time competitor. The accepted band in GW250114 is defined only after that check has been imposed.

\subsection{Broader methodological context}

Several earlier developments shape the methodology adopted here. Bhagwat et al.~\cite{BhagwatEtAl2018StartTime} made the dependence on ringdown start time explicit. Giesler et al.~\cite{GieslerEtAl2019Overtones} showed how overtones can extend useful fitting windows toward the strain peak. More recent analyses emphasized the corresponding costs: high-overtone instability and correlation structure \cite{ColemanFinch2025HighOvertones}, alternative linear versus nonlinear descriptions of the same data \cite{QiuEtAl2024LinearNonlinear}, same-span orthonormal reformulations that change conditioning without changing the physical span \cite{MorisakiEtAl2025Orthonormal}, and public CCE studies that tabulate mode content directly across start times \cite{DyerMoore2025QNMContent}.

Those developments motivate the local perspective adopted here. We do not attempt a global population statement, a new collaboration-style posterior analysis, or a full Bayesian comparison of all early-time phenomenological models. Instead we give a deterministic window-by-window criterion for when a finite detector-frame fit supports a stable common-remnant Kerr interpretation in this particular event.

\section{Conclusion}\label{sec:conclusion}

We give a detector-frame criterion for finite-window black-hole spectroscopy and apply it to GW250114. The argument combines labeled finite-window frequency extraction with explicit radii, a local inverse atlas for the dominant Kerr map on the event-local detector-frame box, auxiliary consistency transport for \(221\) and \(440\), neighboring-window drift control, and calibrated mismatch together with nuisance thresholds. These ingredients turn black-hole spectroscopy on a finite window into a sequence of explicit inequalities with named failure modes.

For GW250114, the public H1/L1 data support a nonempty intermediate post-peak band in which the dominant-mode remnant track is stable, the public detector-frame remnant ensemble is matched well, and the direct-wave and quadratic checks stay below the numerical thresholds adopted in Section~\ref{sec:gw250114-empirical}. Within that band the public \(6M_f\) overtone anchor and the nearby \(9M_f\) point are retained. The \(11M_f\) anchor lies on the edge of the accepted region and is best described as transitional. The earliest \(3M_f\) multimode anchor remains physically interesting, but under the fixed detector-frame criteria used here it is still too sensitive to the nuisance alternatives to count as robust common-remnant Kerr spectroscopy.

The result sharpens the current empirical picture without displacing the collaboration analyses. GW250114 is a high-information ringdown event, and its remnant is consistent with Kerr across several public analyses. The additional statement established here is local: the interpretive weight of a multimode fit depends on the window on which it is made. In loud events, the main limitation is no longer variance alone. Start-time dependence, omitted content, basis dependence, and physically meaningful early-time alternatives must also be controlled. For lower-SNR events, the same deterministic bounds may become conservative enough that no accepted band survives. The framework nevertheless extends to future events with public parameter-estimation products once the event-local box, calibration procedure, and nuisance checks are adapted to the signal at hand.

\appendix

\providecommand{\Kdet}{\mathcal K_{\mathrm{det}}}
\providecommand{\Dnet}{\mathfrak D}
\providecommand{\Mf}{M_{\mathrm f}}
\providecommand{\chif}{\chi_{\mathrm f}}
\providecommand{\wpeak}{t_{\mathrm{pk}}}
\providecommand{\wtwo}{\omega_{220}}
\providecommand{\wot}{\omega_{221}}
\providecommand{\wfour}{\omega_{440}}
\providecommand{\dist}{\operatorname{dist}}

\section{Standing assumptions and notation}\label{app:assumption-ledger}

This appendix fixes the standing assumptions, domains, and notation used throughout the analysis. Public GW250114 data products, Kerr quasi-normal-mode data, deterministic finite-window estimates, synthetic-bank calibrations, and modeling assumptions enter at different stages of the argument. The common notation and the event-local domains are recorded here once so that later sections can refer back to them without redefining the same objects.

We follow the standard ringdown notation of \cite{DreyerEtAl2004BlackHoleSpectroscopy,BertiCardosoStarinets2009QNMReview,GieslerEtAl2019Overtones}. Event-specific conventions are tied to the public GWOSC release for GW250114 and to the companion LVK spectroscopy analysis \cite{GWOSC_GW250114,LVK_GW250114_Spectroscopy}. The nuisance-extended model family is motivated by recent GW250114 analyses that emphasize early-time prompt or horizon-associated structure and possible nonlinear quadratic content \cite{LuEtAl2025DirectWave,WangEtAl2026Quadratic}. Those papers motivate the stress tests used below, but their evidential claims are not used as premises.

\subsection{Event-local working domain}

Our event of interest is the binary-black-hole merger GW250114\_082203. All event times and strain data are taken from the public GWOSC release \cite{GWOSC_GW250114,VallisneriEtAl2015GWOSC}. The companion LVK spectroscopy analysis of \cite{LVK_GW250114_Spectroscopy} serves as the external ringdown comparison benchmark.

\begin{assumption}[Public data domain]\label{ass:event-local-domain}
The detector network used in this ringdown analysis is
\[
\Dnet=\{\mathrm{H1},\mathrm{L1}\}.
\]
We work with the publicly released calibrated strain streams for GW250114 on this two-detector network. Virgo data are not included in the event-level ringdown inference because the public GWOSC detail page records that Virgo was not observing at the event time. All finite-window network norms and all later inference objects are therefore defined on the two-channel data vector indexed by \(\Dnet\).
\end{assumption}

The analysis is performed in detector-frame units. This is the natural convention for finite-window inference on detector strain, because the directly observed oscillation frequencies and damping times are redshifted. The remnant parameter vector is therefore written as
\[
p=(\Mf,\chif)\in (0,\infty)\times[0,1),
\]
where \(\Mf\) denotes the detector-frame remnant mass and \(\chif\) the dimensionless remnant spin.

\begin{definition}[Event-local compact box]\label{def:event-local-box}
We fix the detector-frame compact box
\[
\Kdet=[66.5,69.5]\,M_\odot \times [0.64,0.71].
\]
This box is centered on the public GW250114 remnant estimates and contains the medians and displayed uncertainties of the four public parameter-estimation products listed on the GWOSC detail page, namely NRSur7dq4, PhenomXO4a, PhenomXPHM, and SEOBNRv5PHM \cite{GWOSC_GW250114}. Every event-local Lipschitz constant, separation constant, and inverse bound appearing later is taken on \(\Kdet\).
\end{definition}

For gravitational perturbations of a Kerr remnant, and for each mode label \(j=(\ell,m,n)\), we write
\[
\omega_j(p)=\omega_{\ell mn}(\Mf,\chif)\in\mathbb C,
\qquad
\Im \omega_j(p)<0.
\]
The sign convention is
\[
\omega_j(p)=2\pi f_j(p)-i\gamma_j(p),
\qquad
\gamma_j(p)>0.
\]
When it is advantageous to separate scale from shape, we also use the dimensionless Kerr frequency
\[
\widehat\omega_j(\chif):=\Mf\,\omega_j(\Mf,\chif),
\]
so that \(\omega_j(\Mf,\chif)=\Mf^{-1}\widehat\omega_j(\chif)\). The three principal gravitational modes considered here are
\[
220,\qquad 221,\qquad 440.
\]
The mode \(220\) is the primary mode used for local inversion of \((\Mf,\chif)\); the modes \(221\) and \(440\) are auxiliary modes used for internal consistency tests.

\begin{definition}[Model families]\label{def:model-families}
The baseline Kerr model families are
\[
\mathcal M_0=\{220\},\qquad
\mathcal M_1=\{220,221\},\qquad
\mathcal M_2=\{220,221,440\}.
\]
A nuisance-extended family is one of the four fixed classes
\[
\mathcal S_{\mathcal M_r}^{\nu},
\qquad r\in\{0,1,2\},
\qquad \nu\in\{\varnothing,\mathrm{dir},\mathrm{quad},\mathrm{dir+quad}\},
\]
introduced in Definition~\ref{def:sec3-nuisance-families}. In particular, the direct-wave dictionary has the fixed order \(N_{\mathrm{dir}}=4\) and the quadratic dictionary uses the fixed pair set \(\Pquad=\{(220,220)\}\). Every nuisance family remains diagnostic: it is introduced to stress-test early-time robustness and not to elevate any non-Kerr or higher-order component to a detection claim.
\end{definition}

\subsection{Finite-window signal classes}

Fix a detector \(I\in\Dnet\), a fiducial peak time \(\wpeak\), a ringdown start time \(t_0\ge 0\), and a window length \(T>0\). The local window time is denoted by \(u\in[0,T]\). After whitening, tapering, and restriction to the time interval \([\wpeak+t_0,\wpeak+t_0+T]\), the detector strain is represented by a complex signal
\[
y_I^{(t_0,T)}(u).
\]
For a two-detector network we set
\[
y^{(t_0,T)}(u)=\bigl(y_{\mathrm{H1}}^{(t_0,T)}(u),y_{\mathrm{L1}}^{(t_0,T)}(u)\bigr)
\]
and work in the Hilbert space
\[
\mathcal H_T:=L^2([0,T];\mathbb C^2),
\qquad
\|y\|_{\mathcal H_T}^2:=\sum_{I\in\Dnet}\int_0^T |y_I(u)|^2\,du .
\]

\begin{assumption}[Finite-window decomposition]\label{ass:finite-window-decomposition}
For each admissible pair \((t_0,T)\) and each chosen family \(\mathcal M\in\{\mathcal M_0,\mathcal M_1,\mathcal M_2,\mathcal N\}\), there exists an event-local parameter value \(p_\star\in\Kdet\), detector-dependent complex amplitudes \(A_{I,j}\), and remainder terms \(r_{I,\mathrm{tail}}^{(t_0,T)}\), \(r_{I,\mathrm{mm}}^{(t_0,T)}\), and \(n_I^{(t_0,T)}\) such that
\begin{equation}\label{eq:appendixA-signal-decomposition}
y_I^{(t_0,T)}(u)
=
\sum_{j\in\mathcal M} A_{I,j} e^{-i\omega_j(p_\star)u}
+
r_{I,\mathrm{tail}}^{(t_0,T)}(u)
+
r_{I,\mathrm{mm}}^{(t_0,T)}(u)
+
n_I^{(t_0,T)}(u),
\qquad 0\le u\le T.
\end{equation}
Here \(r_{\mathrm{tail}}\) collects omitted linear Kerr contributions that remain within the same perturbative regime, \(r_{\mathrm{mm}}\) records model misspecification not represented in the chosen family, and \(n\) is the realized detector-noise contribution after preprocessing.
\end{assumption}

Assumption~\ref{ass:finite-window-decomposition} is the only place where the idealized mode expansion is connected to finite-length detector data. Every later theorem is explicit about how its conclusion depends on the size of \(r_{\mathrm{tail}}\), \(r_{\mathrm{mm}}\), and \(n\). No later proof will silently identify a fit with the true Kerr spectrum without paying for the discrepancy through these remainder terms.

For each mode \(j\in\mathcal M\), the total frequency uncertainty radius is written
\begin{equation}\label{eq:appendixA-total-frequency-radius}
\varepsilon_j
=
\varepsilon_j^{\mathrm{stat}}
+
\varepsilon_j^{\mathrm{alg}}
+
\varepsilon_j^{\mathrm{tail}}
+
\varepsilon_j^{\mathrm{mm}} .
\end{equation}
The meaning of the four pieces is fixed as follows. The term \(\varepsilon_j^{\mathrm{stat}}\) is the uncertainty contributed by detector noise, whether obtained from a linearized covariance estimate or from a posterior credible radius. The term \(\varepsilon_j^{\mathrm{alg}}\) is the deterministic extraction error of the chosen estimator even when the model family is exact on the window. The term \(\varepsilon_j^{\mathrm{tail}}\) is the propagated effect of \(r_{\mathrm{tail}}\). The term \(\varepsilon_j^{\mathrm{mm}}\) is the propagated effect of \(r_{\mathrm{mm}}\). Later sections will either bound these terms analytically or calibrate them numerically. The decomposition \eqref{eq:appendixA-total-frequency-radius} is a notational split, not an additional hypothesis.

\subsection{Separation, guide points, and branch control}

A finite-window mode estimate is useful only if mode labels are stable under perturbation. The relevant quantity is the local isolation margin.

\begin{definition}[Isolation margin]\label{def:isolation-margin}
Let \(\mathcal M\) be a finite mode family and \(p\in\Kdet\). The isolation margin of \(\mathcal M\) at \(p\) is
\[
\delta_{\mathrm{iso}}(p;\mathcal M)
:=
\frac12\min_{\substack{j,k\in\mathcal M\\ j\neq k}}
|\omega_j(p)-\omega_k(p)|.
\]
For a compact subset \(K\subset\Kdet\), we write
\[
\delta_{\mathrm{iso}}(K;\mathcal M)
:=
\inf_{p\in K}\delta_{\mathrm{iso}}(p;\mathcal M).
\]
Whenever \(\delta_{\mathrm{iso}}(K;\mathcal M)>0\), the closed disks
\[
\overline{D}\bigl(\omega_j(p),\delta_{\mathrm{iso}}(K;\mathcal M)\bigr),
\qquad j\in\mathcal M,
\]
are pairwise disjoint for every \(p\in K\).
\end{definition}

The next proposition gives the basic branch-control mechanism used throughout the analysis. This is where local mode labeling becomes mathematically stable.

\begin{proposition}[Nearest-neighbor label stability]\label{prop:nearest-neighbor-label-stability}
Fix a finite model family \(\mathcal M\) and two parameter points \(p_\star,p^\dagger\in\Kdet\). Assume that
\[
\beta
:=
\max_{j\in\mathcal M}
|\omega_j(p^\dagger)-\omega_j(p_\star)|
\]
is finite and that
\[
\varepsilon
:=
\max_{j\in\mathcal M}
|\widehat\omega_j-\omega_j(p_\star)|
\]
satisfies
\begin{equation}\label{eq:appendixA-branch-condition}
\beta+\varepsilon<\delta_{\mathrm{iso}}(p_\star;\mathcal M).
\end{equation}
Then, for every \(j\in\mathcal M\), the estimate \(\widehat\omega_j\) is strictly closer to the guide frequency \(\omega_j(p^\dagger)\) than to any other guide frequency \(\omega_k(p^\dagger)\) with \(k\neq j\). In particular, nearest-neighbor matching to the guide set
\[
\{\omega_k(p^\dagger)\colon k\in\mathcal M\}
\]
assigns a unique label and preserves the true label \(j\).
\end{proposition}

\begin{proof}
Fix \(j\in\mathcal M\). The distance from \(\widehat\omega_j\) to its own guide center is bounded above by
\[
|\widehat\omega_j-\omega_j(p^\dagger)|
\le
|\widehat\omega_j-\omega_j(p_\star)|+|\omega_j(p_\star)-\omega_j(p^\dagger)|
\le \varepsilon+\beta.
\]
Now let \(k\in\mathcal M\) with \(k\neq j\). By the triangle inequality,
\[
|\widehat\omega_j-\omega_k(p^\dagger)|
\ge
|\omega_j(p_\star)-\omega_k(p_\star)|
-
|\widehat\omega_j-\omega_j(p_\star)|
-
|\omega_k(p_\star)-\omega_k(p^\dagger)|.
\]
Since
\[
|\omega_j(p_\star)-\omega_k(p_\star)|\ge 2\,\delta_{\mathrm{iso}}(p_\star;\mathcal M),
\]
we obtain
\[
|\widehat\omega_j-\omega_k(p^\dagger)|
\ge
2\,\delta_{\mathrm{iso}}(p_\star;\mathcal M)-\varepsilon-\beta.
\]
Condition \eqref{eq:appendixA-branch-condition} implies
\[
2\,\delta_{\mathrm{iso}}(p_\star;\mathcal M)-\varepsilon-\beta
>
\varepsilon+\beta
\ge
|\widehat\omega_j-\omega_j(p^\dagger)|.
\]
Hence \(\widehat\omega_j\) is strictly closer to \(\omega_j(p^\dagger)\) than to any other guide center. Since \(j\) was arbitrary, the nearest-neighbor labeling is unique and correct for all \(j\in\mathcal M\).
\end{proof}

In later applications the guide point \(p^\dagger\) will usually be a coarse remnant estimate obtained from a full-signal analysis or from a previous iteration of the ringdown fit. Proposition~\ref{prop:nearest-neighbor-label-stability} shows exactly how much error that guide may carry before mode assignment becomes ambiguous.

\subsection{Primary inversion and auxiliary consistency}

The central inversion map is the one-mode Kerr map based on \(220\). We write
\[
F_{220}(p)=\omega_{220}(p),\qquad p\in\Kdet.
\]
Whenever \(F_{220}\) is locally invertible on a neighborhood \(U\subset\mathbb C\) of \(F_{220}(p_\star)\), the primary parameter estimate is
\[
\widehat p
=
F_{220}^{-1}(\widehat\omega_{220}).
\]
The local inverse Lipschitz constant is denoted by \(L_{220}\); later it will be numerically on the chosen compact set.

For the auxiliary modes \(j\in\{221,440\}\) we define the consistency residual
\begin{equation}\label{eq:appendixA-auxiliary-residual}
R_j(\widehat p,\widehat\omega_j)
:=
|\widehat\omega_j-\omega_j(\widehat p)|.
\end{equation}
In practice, \(R_j\) measures whether the frequency extracted from the data for the \(j\)-th auxiliary mode agrees with the Kerr prediction generated by the remnant inferred from the primary mode \(220\).

\begin{proposition}[Propagation from primary inversion to auxiliary residuals]\label{prop:primary-to-auxiliary}
Let \(p_\star\in\Kdet\). Assume that \(F_{220}^{-1}\) is well-defined on a complex neighborhood \(U\) of \(\omega_{220}(p_\star)\) and satisfies the local Lipschitz bound
\[
\|F_{220}^{-1}(z_1)-F_{220}^{-1}(z_2)\|_{\mathbb R^2}
\le
L_{220}|z_1-z_2|,
\qquad z_1,z_2\in U.
\]
Let \(j\in\{221,440\}\), and suppose that the auxiliary Kerr map \(p\mapsto \omega_j(p)\) is \(L_j\)-Lipschitz on \(\Kdet\). If
\[
|\widehat\omega_{220}-\omega_{220}(p_\star)|\le \varepsilon_{220}
\qquad\text{and}\qquad
|\widehat\omega_j-\omega_j(p_\star)|\le \varepsilon_j,
\]
with \(\widehat\omega_{220}\in U\), then the primary estimate \(\widehat p=F_{220}^{-1}(\widehat\omega_{220})\) satisfies
\[
\|\widehat p-p_\star\|_{\mathbb R^2}\le L_{220}\,\varepsilon_{220},
\]
and the auxiliary residual obeys
\begin{equation}\label{eq:appendixA-auxiliary-propagation}
R_j(\widehat p,\widehat\omega_j)
\le
\varepsilon_j+L_jL_{220}\,\varepsilon_{220}.
\end{equation}
\end{proposition}

\begin{proof}
Since \(\widehat p=F_{220}^{-1}(\widehat\omega_{220})\) and \(F_{220}(p_\star)=\omega_{220}(p_\star)\), the inverse Lipschitz bound gives
\[
\|\widehat p-p_\star\|_{\mathbb R^2}
=
\bigl\|F_{220}^{-1}(\widehat\omega_{220})-F_{220}^{-1}(\omega_{220}(p_\star))\bigr\|_{\mathbb R^2}
\le
L_{220}\,|\widehat\omega_{220}-\omega_{220}(p_\star)|
\le
L_{220}\,\varepsilon_{220}.
\]
For the auxiliary residual, we insert and subtract \(\omega_j(p_\star)\) and use the Lipschitz bound on \(\omega_j\):
\[
R_j(\widehat p,\widehat\omega_j)
=
|\widehat\omega_j-\omega_j(\widehat p)|
\le
|\widehat\omega_j-\omega_j(p_\star)|+|\omega_j(p_\star)-\omega_j(\widehat p)|
\le
\varepsilon_j+L_j\|\widehat p-p_\star\|_{\mathbb R^2}.
\]
Substituting the primary estimate bound just obtained yields \eqref{eq:appendixA-auxiliary-propagation}.
\end{proof}

Proposition~\ref{prop:primary-to-auxiliary} explains why the auxiliary tolerances in the trust test will not be chosen independently of the primary inversion. Even if the \(221\) or \(440\) frequencies were measured with zero direct fitting error, a nonzero uncertainty in \(220\) would still propagate into the predicted Kerr location against which those auxiliary frequencies are compared.

\subsection{Window drift and trust-region notation}

We study the dependence of ringdown inference on the start time \(t_0\), the window length \(T\), and the model family \(\mathcal M\). Once these variables are fixed, the analysis returns a frequency estimate \(\widehat\omega_j(t_0,T,\mathcal M)\) for each fitted mode and, when the primary mode is included, a remnant estimate
\[
\widehat p(t_0,T,\mathcal M)=\bigl(\widehat \Mf(t_0,T,\mathcal M),\widehat \chif(t_0,T,\mathcal M)\bigr).
\]

\begin{definition}[Local parameter drift]\label{def:local-parameter-drift}
Let \(\Delta>0\). The \(\Delta\)-local drift of the recovered remnant parameters is
\[
\mathrm{Drift}_\Delta(t_0,T,\mathcal M)
:=
\sup\Bigl\{
\|\widehat p(t_0+\tau,T,\mathcal M)-\widehat p(t_0,T,\mathcal M)\|_{\mathbb R^2}
:
|\tau|\le \Delta,\ t_0+\tau\ \text{admissible}
\Bigr\}.
\]
This quantity measures short-scale instability of the inferred remnant under small changes of the start time.
\end{definition}

The trust-region criterion is thresholded rather than asymptotic. We therefore keep the tolerances explicit.

\begin{definition}[Tolerance vector and trust set]\label{def:tolerance-vector-and-trust-set}
Let
\[
\eta=
\bigl(
\eta_{\mathrm{sep}},
\eta_{\mathrm p},
\eta_{\mathrm d},
\eta_{221},
\eta_{440}
\bigr)
\]
be a vector of positive tolerances. For a model family \(\mathcal M\) containing \(220\), the associated trust set \(\mathcal T_\eta(\mathcal M)\) is the collection of window pairs \((t_0,T)\) such that the following conditions hold simultaneously.

First, the primary mode is spectrally resolved in the sense that
\[
\varepsilon_{220}(t_0,T,\mathcal M)\le \eta_{\mathrm{sep}}
<
\delta_{\mathrm{iso}}\bigl(\widehat p(t_0,T,\mathcal M);\mathcal M\bigr).
\]
Second, the propagated primary parameter uncertainty is acceptable,
\[
L_{220}\,\varepsilon_{220}(t_0,T,\mathcal M)\le \eta_{\mathrm p}.
\]
Third, the recovered remnant is locally stable under start-time perturbations,
\[
\mathrm{Drift}_\Delta(t_0,T,\mathcal M)\le \eta_{\mathrm d}.
\]
Fourth, whenever \(221\in\mathcal M\), the auxiliary residual satisfies
\[
R_{221}\bigl(\widehat p(t_0,T,\mathcal M),\widehat\omega_{221}(t_0,T,\mathcal M)\bigr)\le \eta_{221}.
\]
Fifth, whenever \(440\in\mathcal M\), one has
\[
R_{440}\bigl(\widehat p(t_0,T,\mathcal M),\widehat\omega_{440}(t_0,T,\mathcal M)\bigr)\le \eta_{440}.
\]
\end{definition}

Definition~\ref{def:tolerance-vector-and-trust-set} is parameterized by \(\eta\). The deterministic part of the analysis proves implications of the form
\[
\text{explicit bounds on } \varepsilon_j,\ L_{220},\ R_j,\ \mathrm{Drift}_\Delta
\Longrightarrow
(t_0,T)\in\mathcal T_\eta(\mathcal M).
\]
The numerical values assigned to \(\eta\) come from two sources. Geometric quantities such as \(L_{220}\) and the separation radius are established numerically on \(\Kdet\). Acceptance thresholds that reflect finite-SNR practice are fixed by the synthetic-bank calibration. The two sources remain separate throughout.

\subsection{Symbol index}

Table~\ref{tab:appendixA-symbols} records the symbols that recur most often in the text and later appendices.

\begin{table}[p]
\centering
\small
\caption{Principal symbols fixed in this appendix.}
\label{tab:appendixA-symbols}
\begin{tabular}{@{}p{0.22\textwidth}p{0.72\textwidth}@{}}
\toprule
Symbol & Meaning \\
\midrule
\(\Dnet\) & Detector network used in the event-level ringdown analysis, namely \(\{\mathrm{H1},\mathrm{L1}\}\). \\
\(\wpeak\) & Fiducial reference peak time used to define ringdown start times. \\
\(t_0\) & Ringdown start time measured from \(\wpeak\). \\
\(T\) & Length of the finite analysis window. \\
\(u\) & Local time variable on the shifted window \([0,T]\). \\
\(p=(\Mf,\chif)\) & Detector-frame remnant mass and dimensionless remnant spin. \\
\(\Kdet\) & Event-local compact box in detector-frame remnant parameters. \\
\(\omega_j(p)\) & Complex Kerr quasi-normal-mode frequency of mode \(j=(\ell,m,n)\) at parameter value \(p\). \\
\(\widehat\omega_j(\chif)\) & Dimensionless Kerr mode frequency, equal to \(\Mf\,\omega_j(\Mf,\chif)\). \\
\(\mathcal M_0,\mathcal M_1,\mathcal M_2\) & Baseline Kerr model families \(\{220\}\), \(\{220,221\}\), and \(\{220,221,440\}\). \\
\(\mathcal N\) & One of the fixed nuisance-extended families \(\mathcal S_{\mathcal M_r}^{\nu}\) of Definition~\ref{def:sec3-nuisance-families}, built from the four left-edge direct-wave atoms and the fixed quadratic pair set \(\Pquad=\{(220,220)\}\). \\
\(y_I^{(t_0,T)}\) & Whitened, tapered detector strain on the analysis window in detector \(I\). \\
\(r_{\mathrm{tail}}\) & Residual contribution from omitted linear Kerr content. \\
\(r_{\mathrm{mm}}\) & Residual contribution from finite-window model misspecification. \\
\(\varepsilon_j^{\mathrm{stat}}\) & Frequency uncertainty radius contributed by detector noise. \\
\(\varepsilon_j^{\mathrm{alg}}\) & Deterministic frequency-extraction error of the chosen estimator. \\
\(\varepsilon_j^{\mathrm{tail}}\) & Frequency error induced by the omitted linear tail \(r_{\mathrm{tail}}\). \\
\(\varepsilon_j^{\mathrm{mm}}\) & Frequency error induced by the mismatch term \(r_{\mathrm{mm}}\). \\
\(\varepsilon_j\) & Total frequency uncertainty radius, defined by \eqref{eq:appendixA-total-frequency-radius}. \\
\bottomrule
\end{tabular}
\end{table}

\begin{table}[p]
\centering
\small
\caption{Principal symbols fixed in this appendix, continued.}
\begin{tabular}{@{}p{0.22\textwidth}p{0.72\textwidth}@{}}
\toprule
Symbol & Meaning \\
\midrule
\(\delta_{\mathrm{iso}}(p;\mathcal M)\) & Isolation margin of the fitted mode family at parameter value \(p\). \\
\(F_{220}(p)\) & Primary Kerr map \(p\mapsto \omega_{220}(p)\). \\
\(L_{220}\) & Local inverse Lipschitz constant for \(F_{220}^{-1}\) on the event-local box. \\
\(R_j(\widehat p,\widehat\omega_j)\) & Auxiliary consistency residual for \(j\in\{221,440\}\). \\
\(\mathrm{Drift}_\Delta(t_0,T,\mathcal M)\) & Short-scale remnant drift under start-time perturbations of size at most \(\Delta\). \\
\(\eta\) & Tolerance vector entering the trust-region acceptance rule. \\
\(\mathcal T_\eta(\mathcal M)\) & Set of windows \((t_0,T)\) accepted as trustworthy for model family \(\mathcal M\). \\
\bottomrule
\end{tabular}
\end{table}

This appendix is the reference source for later notation. Any quantity appearing later without a fresh definition is meant in the sense fixed here.

\clearpage

\providecommand{\Mdet}{M_{\mathrm f}^{\mathrm{det}}}
\providecommand{\Msrc}{M_{\mathrm f}^{\mathrm{src}}}
\providecommand{\Ksrc}{\mathcal K_{\mathrm{src}}}
\providecommand{\Iz}{I_z}
\providecommand{\Zmap}{\mathcal Z}
\providecommand{\tauSun}{\tau_\odot}
\providecommand{\pdet}{p^{\mathrm{det}}}
\providecommand{\psrc}{p^{\mathrm{src}}}

\section{Detector-frame and source-frame conventions}\label{app:frame-conventions}

This appendix fixes the conversion rules between detector-frame and source-frame quantities and states which frame is used at each stage of the analysis. The data live in detector time, the whitening and finite-window analysis are carried out on detector strain, and the ringdown frequencies that enter the likelihood are observed in the redshifted frame. Astrophysical summaries of the remnant mass, by contrast, are often reported in the source frame. Making the conversion explicit keeps the dimensionless quantities \(t_0/M_{\mathrm f}\), \(T/M_{\mathrm f}\), \(M_{\mathrm f}\omega\), and the later trust-region inequalities unambiguous.

The basic scaling facts are standard in black-hole spectroscopy and gravitational-wave data analysis; see, for instance, \cite{BertiCardosoStarinets2009QNMReview,BaibhavBerti2019Multimode}. For GW250114 we anchor the event-local numerical ranges to the public GWOSC detail page and to the companion LVK spectroscopy paper \cite{GWOSC_GW250114,LVK_GW250114_Spectroscopy}. The standing detector-frame conventions of Appendix~\ref{app:assumption-ledger}, in particular Definition~\ref{def:event-local-box} and Assumption~\ref{ass:event-local-domain}, remain in force throughout.

\subsection{Geometric units and the redshift map}

We use geometric units \(G=c=1\) in all analytic formulas. When masses quoted in solar units are converted into seconds, we use the nominal solar mass parameter adopted by the IAU \cite{PrsaEtAl2016NominalSolar}. Thus
\begin{equation}\label{eq:appendixB-solar-time}
\tauSun
:=
\frac{G M_\odot^{\mathrm N}}{c^3}
=
4.92549095\times 10^{-6}\ \mathrm{s},
\end{equation}
so that a remnant mass \(M\) expressed in units of \(M_\odot\) corresponds to the time scale
\[
M\,\tauSun.
\]
For example, a detector-frame remnant mass \(\Mdet=68\,M_\odot\) corresponds to \(68\,\tauSun\approx 3.349\times 10^{-4}\,\mathrm s\).

\begin{definition}[Frame parameters and the redshift map]\label{def:frame-parameters}
Let \(z\ge 0\) denote the cosmological redshift associated with the event. The source-frame remnant parameter is
\[
\psrc=(\Msrc,\chif)\in (0,\infty)\times[0,1),
\]
and the detector-frame remnant parameter is
\[
\pdet=(\Mdet,\chif)\in (0,\infty)\times[0,1),
\qquad
\Mdet=(1+z)\Msrc.
\]
The associated redshift map is
\[
\Zmap_z:(0,\infty)\times[0,1)\to (0,\infty)\times[0,1),
\qquad
\Zmap_z(M,\chi)=((1+z)M,\chi).
\]
The inverse map, when \(z\) is fixed, is
\[
\Zmap_z^{-1}(M,\chi)=\Bigl(\frac{M}{1+z},\chi\Bigr).
\]
\end{definition}

The spin parameter \(\chif\) is dimensionless and therefore frame-invariant. By contrast, all masses, frequencies, damping rates, and time intervals scale with \(1+z\). Every finite-window fit is carried out in the detector frame, and every mass or time variable written without an explicit superscript is therefore understood to be detector-frame unless stated otherwise.

\begin{definition}[Default frame convention]\label{def:default-frame-convention}
Throughout the analysis,
\[
\Mf=\Mdet,\qquad p=(\Mf,\chif)=\pdet
\]
unless the superscript \(\mathrm{src}\) is written explicitly. In particular, the compact box \(\Kdet\) of Definition~\ref{def:event-local-box}, all start times \(t_0\), all window lengths \(T\), and all frequencies \(\omega_j(p)\) appearing in Appendix~\ref{app:assumption-ledger} are detector-frame objects by default.
\end{definition}

The scaling relations above will be used repeatedly. Masses and time variables scale by the factor \(1+z\), while mode frequencies and damping rates scale by \((1+z)^{-1}\). The spin \(\chif\) and the quality factors are frame invariant. The combinations \(\Mdet\omega_j^{\mathrm{det}}=\Msrc\omega_j^{\mathrm{src}}\), \(\Mdet f_j^{\mathrm{det}}=\Msrc f_j^{\mathrm{src}}\), \(\Mdet \gamma_j^{\mathrm{det}}=\Msrc \gamma_j^{\mathrm{src}}\), \(t_0^{\mathrm{det}}/\Mdet=t_0^{\mathrm{src}}/\Msrc\), and \(T^{\mathrm{det}}/\Mdet=T^{\mathrm{src}}/\Msrc\) are therefore invariant under the frame change.

\subsection{Event-local boxes in the two frames}

The public GWOSC detail page for GW250114 lists, for the four public parameter-estimation families NRSur7dq4, PhenomXO4a, PhenomXPHM, and SEOBNRv5PHM, detector-frame remnant masses in the range \(67.6\text{--}68.2\,M_\odot\), source-frame remnant masses in the range \(62.5\text{--}62.9\,M_\odot\), remnant spins \(0.67\text{--}0.68\), and redshifts \(0.08\text{--}0.09\) with displayed uncertainties of order \(0.01\text{--}0.02\) \cite{GWOSC_GW250114}. In Appendix~\ref{app:assumption-ledger} we fixed the detector-frame compact box
\[
\Kdet=[66.5,69.5]\,M_\odot\times[0.64,0.71].
\]
The corresponding event-local source-frame box is obtained by pushing \(\Kdet\) through the inverse redshift map over an event-local redshift interval.

\begin{definition}[Event-local redshift interval and source box]\label{def:event-local-source-box}
We fix
\[
\Iz=[0.07,0.11],
\]
which contains the displayed redshift intervals of the four public GW250114 parameter-estimation products on the GWOSC detail page \cite{GWOSC_GW250114}. The associated source-frame compact box is
\begin{equation}\label{eq:appendixB-source-box}
\Ksrc
:=
\Bigl\{
\Zmap_z^{-1}(p)\colon p\in\Kdet,\ z\in\Iz
\Bigr\}.
\end{equation}
\end{definition}

\begin{proposition}[Explicit form of the event-local source box]\label{prop:explicit-source-box}
The set \(\Ksrc\) is the rectangle
\[
\Ksrc
=
\Bigl[\frac{66.5}{1.11},\frac{69.5}{1.07}\Bigr]\,M_\odot\times[0.64,0.71]
\subset
[59.9,65.0]\,M_\odot\times[0.64,0.71].
\]
In particular, every public GW250114 source-frame remnant estimate displayed on the GWOSC detail page lies in \(\Ksrc\).
\end{proposition}

\begin{proof}
By Definition~\ref{def:event-local-source-box},
\[
(M^{\mathrm{src}},\chi)\in \Ksrc
\quad\Longleftrightarrow\quad
M^{\mathrm{src}}=\frac{M^{\mathrm{det}}}{1+z},
\]
for some \(M^{\mathrm{det}}\in[66.5,69.5]\) and some \(z\in[0.07,0.11]\), with \(\chi\in[0.64,0.71]\). The map
\[
(M^{\mathrm{det}},z)\longmapsto \frac{M^{\mathrm{det}}}{1+z}
\]
is increasing in \(M^{\mathrm{det}}\) and decreasing in \(z\) on \((0,\infty)\times[0,\infty)\). Therefore its minimum on the compact rectangle \([66.5,69.5]\times[0.07,0.11]\) is attained at \((66.5,0.11)\), and its maximum is attained at \((69.5,0.07)\). This yields
\[
M^{\mathrm{src}}\in\Bigl[\frac{66.5}{1.11},\frac{69.5}{1.07}\Bigr].
\]
The spin interval is unchanged because \(\chi\) is dimensionless. The numerical enclosure
\[
\Bigl[\frac{66.5}{1.11},\frac{69.5}{1.07}\Bigr]\subset [59.9,65.0]
\]
is immediate. The final sentence follows from the GWOSC values quoted above, since the displayed source-frame masses \(62.5\text{--}62.9\,M_\odot\) and spins \(0.67\text{--}0.68\) lie in the stated rectangle.
\end{proof}

\subsection{Frequencies and damping rates}

For each mode label \(j=(\ell,m,n)\), let
\[
\widehat\omega_j(\chif)=\widehat\omega_{\ell mn}(\chif)\in\mathbb C
\]
denote the dimensionless Kerr quasi-normal-mode frequency. As in Appendix~\ref{app:assumption-ledger}, we write
\[
\widehat\omega_j(\chif)=2\pi \widehat f_j(\chif)-i\widehat\gamma_j(\chif),
\qquad
\widehat\gamma_j(\chif)>0.
\]
The dimensional mode frequencies in the two frames are then determined solely by the mass scale.

\begin{proposition}[Redshift covariance of the Kerr spectrum]\label{prop:redshift-covariance}
For every mode label \(j\) and every source-frame remnant parameter \(\psrc=(\Msrc,\chif)\), one has
\begin{equation}\label{eq:appendixB-redshift-covariance}
\omega_j^{\mathrm{src}}(\psrc)=\frac{1}{\Msrc}\widehat\omega_j(\chif),
\qquad
\omega_j^{\mathrm{det}}(\Zmap_z(\psrc))=\frac{1}{\Mdet}\widehat\omega_j(\chif),
\end{equation}
and hence
\begin{equation}\label{eq:appendixB-redshift-scaling-frequency}
\omega_j^{\mathrm{det}}(\Zmap_z(\psrc))
=
\frac{1}{1+z}\,\omega_j^{\mathrm{src}}(\psrc).
\end{equation}
Equivalently,
\[
f_j^{\mathrm{det}}=\frac{f_j^{\mathrm{src}}}{1+z},
\qquad
\gamma_j^{\mathrm{det}}=\frac{\gamma_j^{\mathrm{src}}}{1+z},
\qquad
\Mdet\,\omega_j^{\mathrm{det}}=\Msrc\,\omega_j^{\mathrm{src}}=\widehat\omega_j(\chif).
\]
\end{proposition}

\begin{proof}
The first two identities in \eqref{eq:appendixB-redshift-covariance} are simply the statement that Kerr quasi-normal-mode frequencies are homogeneous of degree \(-1\) in the mass parameter at fixed dimensionless spin. Indeed, once \(G=c=1\), the Kerr family is characterized by a single length scale \(M\) and a dimensionless spin \(\chi\), so dimensional analysis gives
\[
\omega_j(M,\chi)=M^{-1}\widehat\omega_j(\chi).
\]
Applying this once with \(M=\Msrc\) and once with \(M=\Mdet=(1+z)\Msrc\) gives
\[
\omega_j^{\mathrm{det}}(\Zmap_z(\psrc))
=
\frac{1}{(1+z)\Msrc}\widehat\omega_j(\chif)
=
\frac{1}{1+z}\,\omega_j^{\mathrm{src}}(\psrc),
\]
which is \eqref{eq:appendixB-redshift-scaling-frequency}. Writing real and imaginary parts yields the corresponding formulas for \(f_j\) and \(\gamma_j\). Multiplying by the mass in the same frame gives the final identities.
\end{proof}

\begin{corollary}[Frame-invariant quality factors and separation geometry]\label{cor:frame-invariant-quality}
For every mode \(j\),
\[
Q_j
:=
\frac{\Re\omega_j}{-2\Im\omega_j}
=
\frac{\Re\widehat\omega_j}{-2\Im\widehat\omega_j}
\]
is independent of the frame. More generally, for any finite family \(\mathcal M\),
\[
\widehat\delta_{\mathrm{iso}}(\chif;\mathcal M)
:=
\frac12\min_{\substack{j,k\in\mathcal M\\j\neq k}}
|\widehat\omega_j(\chif)-\widehat\omega_k(\chif)|
\]
satisfies
\begin{equation}\label{eq:appendixB-isolation-scaling}
\delta_{\mathrm{iso}}^{\mathrm{det}}(\pdet;\mathcal M)
=
\frac{\widehat\delta_{\mathrm{iso}}(\chif;\mathcal M)}{\Mdet},
\qquad
\delta_{\mathrm{iso}}^{\mathrm{src}}(\psrc;\mathcal M)
=
\frac{\widehat\delta_{\mathrm{iso}}(\chif;\mathcal M)}{\Msrc}.
\end{equation}
\end{corollary}

\begin{proof}
The quality-factor identity is immediate because the common factor \(1/(1+z)\) cancels between real and imaginary parts. For the isolation margin, Proposition~\ref{prop:redshift-covariance} gives
\[
|\omega_j^{\mathrm{det}}(\pdet)-\omega_k^{\mathrm{det}}(\pdet)|
=
\frac{1}{\Mdet}
|\widehat\omega_j(\chif)-\widehat\omega_k(\chif)|.
\]
Taking the minimum over distinct pairs \((j,k)\) proves the first identity in \eqref{eq:appendixB-isolation-scaling}; the second is identical with \(\Mdet\) replaced by \(\Msrc\).
\end{proof}

The preceding proposition also identifies the exact limitation of ringdown-only measurements as far as source-frame masses are concerned.

\begin{proposition}[Redshift degeneracy of ringdown frequencies]\label{prop:redshift-degeneracy}
Fix \(\chif\in[0,1)\) and a mode label \(j\). Let \((M_{\mathrm f,1}^{\mathrm{src}},z_1)\) and \((M_{\mathrm f,2}^{\mathrm{src}},z_2)\) satisfy
\[
(1+z_1)M_{\mathrm f,1}^{\mathrm{src}}=(1+z_2)M_{\mathrm f,2}^{\mathrm{src}}.
\]
Then
\[
\omega_j^{\mathrm{det}}\bigl((1+z_1)M_{\mathrm f,1}^{\mathrm{src}},\chif\bigr)
=
\omega_j^{\mathrm{det}}\bigl((1+z_2)M_{\mathrm f,2}^{\mathrm{src}},\chif\bigr).
\]
Hence, at fixed spin and without external redshift information, ringdown frequencies determine only the redshifted mass \(\Mdet\), not \(\Msrc\) and \(z\) separately.
\end{proposition}

\begin{proof}
The hypothesis says precisely that the two pairs \((M_{\mathrm f,1}^{\mathrm{src}},z_1)\) and \((M_{\mathrm f,2}^{\mathrm{src}},z_2)\) correspond to the same detector-frame mass \(\Mdet\). Applying Proposition~\ref{prop:redshift-covariance} to each pair yields
\[
\omega_j^{\mathrm{det}}
=
\frac{1}{\Mdet}\widehat\omega_j(\chif)
\]
in both cases, so the detector-frame frequencies coincide. Since the detector-frame spectrum depends only on \(\Mdet\) and \(\chif\), the source mass and the redshift are not separately identifiable from ringdown frequencies alone.
\end{proof}

Proposition~\ref{prop:redshift-degeneracy} is the reason all inversion and all trust-region bounds are carried out in the detector frame. Source-frame masses are used only for astrophysical reporting and comparison with public inspiral-merger-ringdown products. They are not the primitive objects of the finite-window inference.

\subsection{Peak-relative times and dimensionless windows}

The detector data come with an absolute GPS time, but absolute source-frame time is not an observable used anywhere in the analysis. Only offsets from a chosen fiducial peak time enter the dynamics.

\begin{definition}[Peak-relative times]\label{def:peak-relative-times}
Let \(\wpeak\) be the detector-frame fiducial peak time fixed in Appendix~\ref{app:assumption-ledger}. For any detector-frame time \(t^{\mathrm{det}}\), define the peak-relative detector-frame offset
\[
\Delta t^{\mathrm{det}}:=t^{\mathrm{det}}-\wpeak.
\]
For a source at redshift \(z\), the corresponding source-frame offset is defined by
\[
\Delta t^{\mathrm{src}}:=\frac{\Delta t^{\mathrm{det}}}{1+z}.
\]
Only these offsets are converted between frames. The absolute GPS timestamp itself is not assigned a source-frame meaning here.
\end{definition}

If \(t_0^{\mathrm{det}}\) is the detector-frame ringdown start time and \(T^{\mathrm{det}}\) is the detector-frame window length, we define
\[
t_0^{\mathrm{src}}:=\frac{t_0^{\mathrm{det}}}{1+z},
\qquad
T^{\mathrm{src}}:=\frac{T^{\mathrm{det}}}{1+z}.
\]
The useful variables are the dimensionless ratios.

\begin{proposition}[Frame invariance of dimensionless time variables]\label{prop:dimensionless-time-invariance}
Let \(\pdet=\Zmap_z(\psrc)\). Then the dimensionless time coordinate
\[
\tau:=\frac{\Delta t^{\mathrm{det}}}{\Mdet}=\frac{\Delta t^{\mathrm{src}}}{\Msrc}
\]
is frame-invariant. In particular,
\begin{equation}\label{eq:appendixB-window-invariants}
\tau_0:=\frac{t_0^{\mathrm{det}}}{\Mdet}=\frac{t_0^{\mathrm{src}}}{\Msrc},
\qquad
\Theta:=\frac{T^{\mathrm{det}}}{\Mdet}=\frac{T^{\mathrm{src}}}{\Msrc}.
\end{equation}
If \(t_0^{\mathrm{det}}\) is expressed in seconds and \(\Mdet\) in solar masses, then
\begin{equation}\label{eq:appendixB-dimensionless-time-seconds}
\tau_0
=
\frac{t_0^{\mathrm{det}}/\mathrm{s}}{(\Mdet/M_\odot)\tauSun}.
\end{equation}
\end{proposition}

\begin{proof}
By Definition~\ref{def:peak-relative-times},
\[
\Delta t^{\mathrm{det}}=(1+z)\Delta t^{\mathrm{src}},
\]
while Definition~\ref{def:frame-parameters} gives \(\Mdet=(1+z)\Msrc\). Dividing the first identity by the second yields
\[
\frac{\Delta t^{\mathrm{det}}}{\Mdet}
=
\frac{\Delta t^{\mathrm{src}}}{\Msrc}.
\]
Applying this to \(\Delta t=t_0\) and \(\Delta t=T\) proves \eqref{eq:appendixB-window-invariants}. Formula \eqref{eq:appendixB-dimensionless-time-seconds} is simply the conversion of \(\Mdet\) into seconds via \eqref{eq:appendixB-solar-time}.
\end{proof}

The preceding proposition is more than notational convenience. It means that the statement ``the analysis starts at \(t_0=5M_{\mathrm f}\) after the peak'' is frame-independent, provided the mass in the denominator is taken in the same frame as the numerator. Accordingly, every later scan over start time may be expressed either in detector-frame seconds or in the invariant dimensionless coordinate \(\tau_0\), but never by mixing the two.

\begin{theorem}[Frame-invariant finite-window template]\label{thm:frame-invariant-template}
Assume the detector-frame decomposition of Appendix~\ref{app:assumption-ledger}, namely \eqref{eq:appendixA-signal-decomposition}, holds on the window \([0,T^{\mathrm{det}}]\) with remnant parameter \(\pdet=(\Mdet,\chif)\). Set \(\Theta=T^{\mathrm{det}}/\Mdet\) and define, for \(0\le \tau\le \Theta\),
\[
\widetilde y_I(\tau):=y_I^{(t_0^{\mathrm{det}},T^{\mathrm{det}})}(\Mdet\tau),
\]
and likewise
\[
\widetilde r_{I,\mathrm{tail}}(\tau):=r_{I,\mathrm{tail}}^{(t_0^{\mathrm{det}},T^{\mathrm{det}})}(\Mdet\tau),
\quad
\widetilde r_{I,\mathrm{mm}}(\tau):=r_{I,\mathrm{mm}}^{(t_0^{\mathrm{det}},T^{\mathrm{det}})}(\Mdet\tau),
\quad
\widetilde n_I(\tau):=n_I^{(t_0^{\mathrm{det}},T^{\mathrm{det}})}(\Mdet\tau).
\]
Then
\begin{equation}\label{eq:appendixB-dimensionless-decomposition}
\widetilde y_I(\tau)
=
\sum_{j\in\mathcal M} A_{I,j}e^{-i\widehat\omega_j(\chif)\tau}
+
\widetilde r_{I,\mathrm{tail}}(\tau)
+
\widetilde r_{I,\mathrm{mm}}(\tau)
+
\widetilde n_I(\tau),
\qquad 0\le \tau\le \Theta.
\end{equation}
Moreover, if
\[
\widetilde{\mathcal H}_\Theta:=L^2([0,\Theta];\mathbb C^2),
\qquad
\|\widetilde y\|_{\widetilde{\mathcal H}_\Theta}^2
:=
\sum_{I\in\Dnet}\int_0^\Theta |\widetilde y_I(\tau)|^2\,d\tau,
\]
then
\begin{equation}\label{eq:appendixB-norm-scaling}
\|y\|_{\mathcal H_{T^{\mathrm{det}}}}^2
=
\Mdet\,\|\widetilde y\|_{\widetilde{\mathcal H}_\Theta}^2.
\end{equation}
The same dimensionless representation \eqref{eq:appendixB-dimensionless-decomposition} is obtained if one starts from the source-frame window \([0,T^{\mathrm{src}}]\) and rescales by \(\Msrc\) instead of \(\Mdet\).
\end{theorem}

\begin{proof}
Starting from \eqref{eq:appendixA-signal-decomposition} and substituting \(u=\Mdet\tau\), we obtain
\[
\widetilde y_I(\tau)
=
\sum_{j\in\mathcal M}
A_{I,j}e^{-i\omega_j^{\mathrm{det}}(\pdet)\Mdet\tau}
+
\widetilde r_{I,\mathrm{tail}}(\tau)
+
\widetilde r_{I,\mathrm{mm}}(\tau)
+
\widetilde n_I(\tau).
\]
By Proposition~\ref{prop:redshift-covariance},
\[
\Mdet\,\omega_j^{\mathrm{det}}(\pdet)=\widehat\omega_j(\chif),
\]
so the oscillatory factor becomes \(e^{-i\widehat\omega_j(\chif)\tau}\). This proves \eqref{eq:appendixB-dimensionless-decomposition}.

For the norm identity, use the change of variables \(u=\Mdet\tau\):
\[
\|y\|_{\mathcal H_{T^{\mathrm{det}}}}^2
=
\sum_{I\in\Dnet}\int_0^{T^{\mathrm{det}}}|y_I(u)|^2\,du
=
\sum_{I\in\Dnet}\int_0^\Theta |\widetilde y_I(\tau)|^2 \Mdet\,d\tau
=
\Mdet\,\|\widetilde y\|_{\widetilde{\mathcal H}_\Theta}^2.
\]
Finally, if the same physical window is described in the source frame, then \(u^{\mathrm{det}}=(1+z)u^{\mathrm{src}}\) and \(\Mdet=(1+z)\Msrc\), so
\[
\tau=\frac{u^{\mathrm{det}}}{\Mdet}=\frac{u^{\mathrm{src}}}{\Msrc}.
\]
Applying Proposition~\ref{prop:redshift-covariance} in the source frame shows that the same dimensionless exponent \(e^{-i\widehat\omega_j(\chif)\tau}\) appears. Hence the dimensionless template is the same in both frames.
\end{proof}

Theorem~\ref{thm:frame-invariant-template} explains why the later theory is naturally formulated in the detector frame but can still be expressed cleanly in terms of frame-free dimensionless variables. The data-analysis pipeline sees frequencies in hertz and times in seconds, yet the underlying mode geometry is entirely encoded by \(\widehat\omega_j(\chif)\), \(\tau_0\), and \(\Theta\).

\subsection{Residuals, uncertainty radii, and auxiliary consistency}

The trust-region inequalities in Appendix~\ref{app:assumption-ledger} are written in detector-frame quantities because the actual fitting is performed on detector-frame data. Their behavior under frame changes is nevertheless explicit.

\begin{proposition}[Scaling of frequency errors and auxiliary residuals]\label{prop:residual-scaling}
Let \(j\) be a mode label and assume that a detector-frame estimate \(\widehat\omega_j^{\mathrm{det}}\) is reported together with the corresponding source-frame rescaling
\[
\widehat\omega_j^{\mathrm{src}}:=(1+z)\widehat\omega_j^{\mathrm{det}}.
\]
Let \(\widehat\pdet=(\widehat M^{\mathrm{det}},\widehat\chi)\) and \(\widehat\psrc=((1+z)^{-1}\widehat M^{\mathrm{det}},\widehat\chi)\) denote the corresponding remnant estimates. Then
\begin{equation}\label{eq:appendixB-residual-scaling}
\left|\widehat\omega_j^{\mathrm{src}}-\omega_j^{\mathrm{src}}(\widehat\psrc)\right|
=
(1+z)\left|\widehat\omega_j^{\mathrm{det}}-\omega_j^{\mathrm{det}}(\widehat\pdet)\right|,
\end{equation}
and therefore
\begin{equation}\label{eq:appendixB-dimensionless-residual-invariant}
\widehat M^{\mathrm{src}}
\left|\widehat\omega_j^{\mathrm{src}}-\omega_j^{\mathrm{src}}(\widehat\psrc)\right|
=
\widehat M^{\mathrm{det}}
\left|\widehat\omega_j^{\mathrm{det}}-\omega_j^{\mathrm{det}}(\widehat\pdet)\right|.
\end{equation}
The same scaling law holds for total frequency uncertainty radii:
\[
\varepsilon_j^{\mathrm{src}}=(1+z)\varepsilon_j^{\mathrm{det}},
\qquad
\widehat M^{\mathrm{src}}\varepsilon_j^{\mathrm{src}}
=
\widehat M^{\mathrm{det}}\varepsilon_j^{\mathrm{det}}.
\]
\end{proposition}

\begin{proof}
By construction,
\[
\widehat\omega_j^{\mathrm{src}}=(1+z)\widehat\omega_j^{\mathrm{det}}.
\]
On the other hand, Proposition~\ref{prop:redshift-covariance} gives
\[
\omega_j^{\mathrm{src}}(\widehat\psrc)
=
(1+z)\,\omega_j^{\mathrm{det}}(\widehat\pdet).
\]
Subtracting the two identities proves \eqref{eq:appendixB-residual-scaling}. Multiplying both sides by \(\widehat M^{\mathrm{src}}=\widehat M^{\mathrm{det}}/(1+z)\) yields \eqref{eq:appendixB-dimensionless-residual-invariant}. The statement for \(\varepsilon_j\) is identical, because each uncertainty radius is a bound on an absolute frequency difference and therefore scales in the same way.
\end{proof}

Proposition~\ref{prop:residual-scaling} shows that the auxiliary consistency tests may be written either in dimensional detector-frame units, which is convenient for direct comparison with the whitened data, or in dimensionless form, which is often conceptually cleaner. The pass/fail content is the same once the tolerance is scaled accordingly.

\begin{remark}\label{rem:detector-frame-trust-tests}
In the detector-frame presentation above, we keep the trust test in detector-frame form, exactly as stated in Definition~\ref{def:tolerance-vector-and-trust-set}. This is the natural choice because the noise power spectral density, whitening filter, and finite-window likelihood all live in detector-frame units. Whenever a dimensionless reformulation is advantageous, the correct invariant combinations are \(\tau_0\), \(\Theta\), \(\widehat\omega_j\), \(\widehat\delta_{\mathrm{iso}}\), and \(\Mf\varepsilon_j\).
\end{remark}

\subsection{Source-frame reporting and error propagation}

Although the inversion itself is carried out in the detector frame, the final astrophysical summary of the remnant mass is often reported in the source frame. Once a redshift posterior is supplied, the exact interval propagation formula used in later tables is elementary.

\begin{proposition}[Propagation of detector-frame mass and redshift intervals]\label{prop:source-mass-interval-propagation}
Let
\[
M^{\mathrm{det}}\in [M_-,M_+],
\qquad
z\in[z_-,z_+],
\qquad
0\le z_- \le z_+,
\]
and define
\[
M^{\mathrm{src}}=\frac{M^{\mathrm{det}}}{1+z}.
\]
Then
\begin{equation}\label{eq:appendixB-source-interval}
M^{\mathrm{src}}
\in
\Bigl[\frac{M_-}{1+z_+},\frac{M_+}{1+z_-}\Bigr].
\end{equation}
If, moreover, \(M^{\mathrm{det}}=\bar M+\delta M\) and \(z=\bar z+\delta z\) with \(|\delta M|\ll \bar M\) and \(|\delta z|\ll 1+\bar z\), then the first-order variation is
\begin{equation}\label{eq:appendixB-linearized-source-mass}
\delta M^{\mathrm{src}}
=
\frac{\delta M}{1+\bar z}
-
\frac{\bar M}{(1+\bar z)^2}\,\delta z
+
O\!\left(|\delta M\,\delta z|+|\delta z|^2\right),
\end{equation}
and therefore
\begin{equation}\label{eq:appendixB-linearized-source-mass-bound}
|\delta M^{\mathrm{src}}|
\le
\frac{|\delta M|}{1+\bar z}
+
\frac{\bar M}{(1+\bar z)^2}|\delta z|
+
O\!\left(|\delta M\,\delta z|+|\delta z|^2\right).
\end{equation}
\end{proposition}

\begin{proof}
The map \((M,z)\mapsto M/(1+z)\) is increasing in \(M\) and decreasing in \(z\) on \((0,\infty)\times[0,\infty)\). Therefore its minimum on the rectangle \([M_-,M_+]\times[z_-,z_+]\) is attained at \((M_-,z_+)\), and its maximum is attained at \((M_+,z_-)\). This proves \eqref{eq:appendixB-source-interval}.

For the linearization, expand
\[
M^{\mathrm{src}}
=
\frac{\bar M+\delta M}{1+\bar z+\delta z}
=
(\bar M+\delta M)\Bigl[(1+\bar z)^{-1}-(1+\bar z)^{-2}\delta z+O(|\delta z|^2)\Bigr].
\]
Multiplying out and subtracting \(\bar M/(1+\bar z)\) gives
\[
\delta M^{\mathrm{src}}
=
\frac{\delta M}{1+\bar z}
-
\frac{\bar M}{(1+\bar z)^2}\delta z
-
\frac{\delta M\,\delta z}{(1+\bar z)^2}
+
O(|\delta z|^2),
\]
which is \eqref{eq:appendixB-linearized-source-mass}. Taking absolute values yields \eqref{eq:appendixB-linearized-source-mass-bound}.
\end{proof}

For GW250114 the public GWOSC products place the redshift near \(z\approx 0.08\text{--}0.09\) \cite{GWOSC_GW250114}. At such redshifts the distinction between detector-frame and source-frame masses is numerically modest but mathematically non-negligible. Since the central theme here is finite-window reliability rather than cosmological inference, we do not re-estimate the source-frame mass from luminosity distance and a background cosmology. Instead, we report source-frame values by propagating our detector-frame remnant estimates through the same redshift convention used in the public GWOSC parameter-estimation products.

\subsection{Frame convention used below}

Every quantity tied directly to detector strain is handled in the detector frame, every astrophysical mass quoted for comparison with public inspiral-merger-ringdown products is reported in the source frame, and every dimensionless quantity of the form \(t/M\), \(T/M\), or \(M\omega\) is frame-invariant provided the same frame is used in numerator and denominator.

The practical consequences are immediate. First, every start-time scan is a scan in \(t_0^{\mathrm{det}}\) or, equivalently, in the invariant variable \(\tau_0=t_0^{\mathrm{det}}/\Mdet\). Second, every fitted QNM frequency appearing in the trust-region theorem is a detector-frame frequency unless the superscript \(\mathrm{src}\) is shown explicitly. Third, every comparison between our remnant estimate and the public GWOSC remnant summaries is made twice when helpful: once in detector-frame units, which is the natural frame of the fit, and once in source-frame units, which is the natural frame of the astrophysical interpretation. Fourth, no absolute event time is ever converted into a source-frame timestamp. Only offsets from the chosen fiducial peak are meaningful for the windowed ringdown problem.

With this convention fixed, later sections may write \(t_0/M_{\mathrm f}\), \(T/M_{\mathrm f}\), \(F_{220}(p)\), and \(R_j(\widehat p,\widehat\omega_j)\) without hidden frame ambiguity. Detector-frame inference comes first, source-frame reporting comes afterwards, and the invariant bridge between the two is the dimensionless Kerr spectrum \(\widehat\omega_j(\chif)\).

\providecommand{\tevt}{t_{\mathrm{evt}}}
\providecommand{\fs}{f_{\mathrm s}}
\providecommand{\dt}{\Delta t}
\providecommand{\fny}{f_{\mathrm{Ny}}}
\providecommand{\Toff}{\mathcal T_{\mathrm{off}}}
\providecommand{\Wop}{\mathcal W}
\providecommand{\Shannon}{\mathcal S_{\Delta t}}
\providecommand{\etaTap}{\eta_{\mathrm{tap}}}
\providecommand{\Lpsd}{L_{\mathrm{PSD}}}
\providecommand{\Hpsd}{H_{\mathrm{PSD}}}
\providecommand{\Mpsd}{M_{\mathrm{PSD}}}
\providecommand{\Uhann}{U_{\mathrm H}}

\section{Power spectral density estimation, whitening, windowing, and tapering}\label{app:preprocessing}

This appendix fixes the preprocessing map that takes the public detector strain to the finite-window network object \(y^{(t_0,T)}\in\mathcal H_T\) used throughout the analysis. Later sections compare deterministic extraction errors, statistical errors, and model-mismatch errors on the same footing. That comparison is meaningful only after the noise weighting, the time localization, and the tapering convention have been fixed. We therefore specify the preprocessing at the same level of precision as the later inverse and extraction maps.

The default input streams are the public H1 and L1 cleaned strains released on the GWOSC detail page for GW250114\_082203. The same page provides both \(16\,\mathrm{kHz}\) and \(4\,\mathrm{kHz}\) products. We adopt the \(4\,\mathrm{kHz}\) strain products as the default input because the corresponding Nyquist frequency \(2048\,\mathrm{Hz}\) remains comfortably above the ringdown band of interest while keeping the preprocessing transparent and exactly reproducible from the public release \cite{GWOSC_GW250114,VallisneriEtAl2015GWOSC}. The general relation between detector-frame times and the dimensionless variables \(t_0/\Mf\) and \(T/\Mf\) was fixed in Appendix~\ref{app:frame-conventions}; only the preprocessing prescription is new here. Standard background on windowed spectral estimation and whitening may be found in \cite{Welch1967FFTPSD,Harris1978Windows,AllenEtAl2012Findchirp,ChatziioannouEtAl2019PSD,AbbottEtAl2020NoiseGuide}.

\subsection{Public strain streams, sampling, and continuous representatives}

Let \(\tevt\) denote the public GPS event time of GW250114\_082203 and let \(d_I[n]\), with \(I\in\Dnet\), denote the corresponding sampled strain on the default \(4\,\mathrm{kHz}\) product. Thus
\[
\fs = 4096\,\mathrm{Hz},
\qquad
\dt = \fs^{-1} = \frac{1}{4096}\,\mathrm s,
\qquad
\fny = \frac{\fs}{2}=2048\,\mathrm{Hz}.
\]
We work on the centered \(1024\,\mathrm s\) interval
\[
[\tevt-512\,\mathrm s,\tevt+512\,\mathrm s],
\]
extracted from the public GWOSC strain release. The off-source region used for power-spectral-density estimation is
\begin{equation}\label{eq:appendixC-offsource-set}
\Toff
=
[\tevt-512\,\mathrm s,\tevt-16\,\mathrm s]
\cup
[\tevt+16\,\mathrm s,\tevt+512\,\mathrm s].
\end{equation}
The central \(32\,\mathrm s\) excision is wide enough to remove the merger signal and the immediate nonstationary neighborhood from the noise estimate without forcing the PSD to be learned from a remote epoch.

The numerical implementation is discrete, but the later deterministic estimates are written in continuous-time notation. To connect the two descriptions we identify each sampled sequence with its Shannon reconstruction on the Nyquist band.

\begin{definition}[Shannon reconstruction]\label{def:appendixC-shannon}
Let \(x=\{x[n]\}_{n\in\mathbb Z}\) be a square-summable sequence sampled at rate \(\fs\). Its Shannon reconstruction is
\begin{equation}\label{eq:appendixC-shannon}
(\Shannon x)(t)
:=
\sum_{n\in\mathbb Z}
x[n]\,
\operatorname{sinc}\!\left(\frac{t-n\dt}{\dt}\right),
\qquad
\operatorname{sinc}(u)=\frac{\sin(\pi u)}{\pi u}.
\end{equation}
Whenever a processed detector stream is written as a continuous function of time, this is the representative we use.
\end{definition}

The next proposition is the precise bridge between the discrete arrays manipulated by FFTs and the \(L^2\) norms that appear in Appendix~\ref{app:assumption-ledger}.

\begin{proposition}[Sample-space and time-space norm identity]\label{prop:appendixC-shannon-parseval}
Let \(x\) be any finitely supported sampled sequence and let \(x^\sharp=\Shannon x\). Then \(x^\sharp\) is band-limited to \([-\fny,\fny]\), satisfies \(x^\sharp(n\dt)=x[n]\), and obeys
\begin{equation}\label{eq:appendixC-shannon-parseval}
\|x^\sharp\|_{L^2(\mathbb R)}^2
=
\dt\sum_{n\in\mathbb Z}|x[n]|^2.
\end{equation}
\end{proposition}

\begin{proof}
The Shannon functions
\[
\phi_n(t):=\operatorname{sinc}\!\left(\frac{t-n\dt}{\dt}\right)
\]
form an orthogonal family in the Paley--Wiener space of functions band-limited to \([-\fny,\fny]\). More precisely,
\[
\int_{\mathbb R}\phi_n(t)\,\overline{\phi_m(t)}\,dt
=
\dt\,\delta_{nm}.
\]
Since \(x^\sharp=\sum_n x[n]\phi_n\), orthogonality gives
\[
\|x^\sharp\|_{L^2(\mathbb R)}^2
=
\sum_{n,m}x[n]\overline{x[m]}
\int_{\mathbb R}\phi_n(t)\overline{\phi_m(t)}\,dt
=
\dt\sum_n |x[n]|^2.
\]
The interpolation identity \(x^\sharp(n\dt)=x[n]\) is the standard cardinal property of the sinc kernel. Band limitation follows from the Fourier transform of the sinc basis.
\end{proof}

In the formulas below we often suppress the sharp and simply write a processed detector stream as a continuous function \(x_I(t)\). This is always understood in the sense of Definition~\ref{def:appendixC-shannon}.

\subsection{Welch power spectral density estimate}

The default PSD estimate is local in time, detector by detector, and is obtained from the off-source data \(\Toff\) by the Welch method \cite{Welch1967FFTPSD}. We use a Hann window and \(50\%\) overlap. The choices are conventional, but we record them explicitly because later robustness studies will vary them and only this appendix fixes the default.

Let
\[
\Lpsd = 8\,\mathrm s,
\qquad
\Hpsd = 4\,\mathrm s.
\]
Each connected component of \(\Toff\) in \eqref{eq:appendixC-offsource-set} has length \(496\,\mathrm s\), so the default prescription yields
\[
123
\]
Hann-windowed segments on the left of the event and
\[
123
\]
on the right, for a total of
\[
\Mpsd=246
\]
off-source segments per detector. The segment window is the Hann function
\begin{equation}\label{eq:appendixC-hann}
h_{\mathrm H}(u)
=
\frac12\Bigl(1-\cos\frac{2\pi u}{\Lpsd}\Bigr),
\qquad
0\le u\le \Lpsd,
\end{equation}
with normalization constant
\begin{equation}\label{eq:appendixC-hann-normalization}
\Uhann
:=
\int_0^{\Lpsd} h_{\mathrm H}(u)^2\,du
=
\frac{3\Lpsd}{8}.
\end{equation}

For the \(m\)-th off-source segment of detector \(I\), written as a continuous function \(x_{I,m}(u)\) on \([0,\Lpsd]\), we define the two-sided modified periodogram
\begin{equation}\label{eq:appendixC-two-sided-periodogram}
\mathcal I^{(2)}_{I,m}(f)
:=
\frac{1}{\Uhann}
\left|
\int_0^{\Lpsd}
h_{\mathrm H}(u)\,x_{I,m}(u)\,e^{-2\pi i f u}\,du
\right|^2,
\qquad
f\in\mathbb R.
\end{equation}
The two-sided Welch estimate is
\begin{equation}\label{eq:appendixC-two-sided-welch}
\widehat S^{(2)}_I(f)
:=
\frac{1}{\Mpsd}\sum_{m=1}^{\Mpsd}\mathcal I^{(2)}_{I,m}(f).
\end{equation}
The corresponding one-sided PSD estimate on \([0,\fny]\) is
\begin{equation}\label{eq:appendixC-one-sided-welch}
\widehat S_I(f)
=
\begin{cases}
\widehat S_I^{(2)}(0), & f=0,\\[0.5ex]
2\widehat S_I^{(2)}(f), & 0<f<\fny,\\[0.5ex]
\widehat S_I^{(2)}(\fny), & f=\fny.
\end{cases}
\end{equation}
In the discrete implementation \(\widehat S_I\) is computed on the FFT frequency grid and then extended to a continuous positive function on \([0,\fny]\) by piecewise-linear interpolation. The same symbol \(\widehat S_I\) is used for the discrete values and for this interpolant.

The next proposition isolates the exact bias mechanism of the windowed periodogram. It is completely standard and explains why segment length and taper choice matter even before statistical variance enters.

\begin{proposition}[Welch expectation as spectral convolution]\label{prop:appendixC-welch-bias}
Let \(n_I(t)\) be a zero-mean second-order stationary detector-noise process with two-sided PSD \(S_I^{(2)}\in L^1_{\mathrm{loc}}(\mathbb R)\). Define
\begin{equation}\label{eq:appendixC-spectral-kernel}
K_{\mathrm H}(\xi)
:=
\frac{|\widehat h_{\mathrm H}(\xi)|^2}{\Uhann},
\qquad
\widehat h_{\mathrm H}(\xi)
=
\int_0^{\Lpsd} h_{\mathrm H}(u)e^{-2\pi i \xi u}\,du .
\end{equation}
Then, for every \(f\in\mathbb R\),
\begin{equation}\label{eq:appendixC-welch-bias}
\mathbb E\,\widehat S_I^{(2)}(f)
=
\int_{\mathbb R}
S_I^{(2)}(\nu)\,K_{\mathrm H}(f-\nu)\,d\nu.
\end{equation}
In particular, the expected Welch spectrum is the true two-sided spectrum convolved with the spectral window \(K_{\mathrm H}\). The overlap pattern does not alter the expectation; it affects only variance and finite-sample dependence.
\end{proposition}

\begin{proof}
Because expectation is linear and every off-source segment is a translate of the same stationary process, it suffices to prove the formula for a single modified periodogram \(\mathcal I_{I,m}^{(2)}(f)\). Using the covariance kernel
\[
R_I(\tau)
:=
\mathbb E[n_I(t+\tau)\overline{n_I(t)}]
=
\int_{\mathbb R} e^{2\pi i \nu \tau}S_I^{(2)}(\nu)\,d\nu,
\]
we obtain
\begin{align*}
\mathbb E\,\mathcal I_{I,m}^{(2)}(f)
&=
\frac{1}{\Uhann}
\mathbb E\left[
\left(\int_0^{\Lpsd} h_{\mathrm H}(u)n_I(u)e^{-2\pi i f u}\,du\right)
\overline{
\left(\int_0^{\Lpsd} h_{\mathrm H}(v)n_I(v)e^{-2\pi i f v}\,dv\right)}
\right]\\
&=
\frac{1}{\Uhann}
\int_0^{\Lpsd}\int_0^{\Lpsd}
h_{\mathrm H}(u)h_{\mathrm H}(v)e^{-2\pi i f(u-v)}
R_I(u-v)\,du\,dv.
\end{align*}
Insert the spectral representation of \(R_I\) and apply Fubini:
\begin{align*}
\mathbb E\,\mathcal I_{I,m}^{(2)}(f)
&=
\frac{1}{\Uhann}
\int_{\mathbb R}
S_I^{(2)}(\nu)
\left(
\int_0^{\Lpsd} h_{\mathrm H}(u)e^{-2\pi i (f-\nu)u}\,du
\right)
\overline{
\left(
\int_0^{\Lpsd} h_{\mathrm H}(v)e^{-2\pi i (f-\nu)v}\,dv
\right)}
\,d\nu\\
&=
\int_{\mathbb R}
S_I^{(2)}(\nu)\,
\frac{|\widehat h_{\mathrm H}(f-\nu)|^2}{\Uhann}
\,d\nu\\
&=
\int_{\mathbb R}
S_I^{(2)}(\nu)\,K_{\mathrm H}(f-\nu)\,d\nu.
\end{align*}
Averaging over \(m\) leaves the right-hand side unchanged, so \eqref{eq:appendixC-welch-bias} follows.
\end{proof}

Equation~\eqref{eq:appendixC-welch-bias} is the reason we later vary \(\Lpsd\) and the taper family only as robustness checks rather than as hidden tuning knobs. A longer segment narrows the spectral window and reduces bias, but at the price of weaker local stationarity; a shorter segment does the opposite. The default \(8\,\mathrm s\) choice is a conservative compromise for a nearby \(1024\,\mathrm s\) off-source neighborhood.

For numerical stability we regularize the PSD estimate by a detector-dependent positive floor. Let
\begin{equation}\label{eq:appendixC-psd-floor}
\lambda_I
:=
10^{-4}\,
\operatorname{med}\bigl\{\widehat S_I(f)\colon 20\,\mathrm{Hz}\le f\le 1024\,\mathrm{Hz}\bigr\}.
\end{equation}
We then define
\begin{equation}\label{eq:appendixC-psd-regularized}
\widehat S_I^\sharp(f)
:=
\max\{\widehat S_I(f),\lambda_I\},
\qquad
0\le f\le \fny.
\end{equation}
The floor is a numerical safeguard and not a physical model. Because \(\lambda_I\) is several orders of magnitude below the relevant off-source median, it never influences the ringdown band in ordinary circumstances; it merely excludes division by an accidentally tiny spectral estimate.

\subsection{Whitening as a noise-weighted isometry}

Applying the PSD estimate yields a whitening map. Since the one-sided PSD \(\widehat S_I^\sharp\) is defined only on \([0,\fny]\), we extend it evenly to the Nyquist band by
\begin{equation}\label{eq:appendixC-even-extension}
\widehat S_{I,\mathrm e}^\sharp(f)
:=
\widehat S_I^\sharp(|f|),
\qquad
|f|\le \fny.
\end{equation}
For any Nyquist-band-limited continuous signal \(x\), with Fourier transform
\[
\widetilde x(f)
=
\int_{\mathbb R} x(t)e^{-2\pi i f t}\,dt,
\]
the detector-wise whitening operator is defined by
\begin{equation}\label{eq:appendixC-whitening}
\widetilde{(\Wop_I x)}(f)
:=
\frac{\sqrt{2}\,\widetilde x(f)}
{\sqrt{\widehat S_{I,\mathrm e}^\sharp(f)}},
\qquad
|f|\le \fny.
\end{equation}
The factor \(\sqrt{2}\) is the standard conversion from a one-sided PSD to an even two-sided weighting. When \(x\) is a sampled detector stream we apply \eqref{eq:appendixC-whitening} on the FFT grid and then return to the time domain by inverse FFT; the resulting sampled sequence is again identified with its Shannon reconstruction.

\begin{proposition}[Exact noise-weighted norm identity]\label{prop:appendixC-whitening-isometry}
Let \(x\) be a real-valued Nyquist-band-limited signal. Then
\begin{equation}\label{eq:appendixC-whitening-isometry}
\|\Wop_I x\|_{L^2(\mathbb R)}^2
=
4\int_0^{\fny}
\frac{|\widetilde x(f)|^2}{\widehat S_I^\sharp(f)}\,df.
\end{equation}
More generally, for any two real Nyquist-band-limited signals \(x\) and \(z\),
\begin{equation}\label{eq:appendixC-whitening-innerproduct}
\langle \Wop_I x,\Wop_I z\rangle_{L^2(\mathbb R)}
=
4\,\Re\!\int_0^{\fny}
\frac{\widetilde x(f)\overline{\widetilde z(f)}}{\widehat S_I^\sharp(f)}\,df.
\end{equation}
\end{proposition}

\begin{proof}
Parseval's identity gives
\[
\|\Wop_I x\|_{L^2(\mathbb R)}^2
=
\int_{-\fny}^{\fny}
\left|
\frac{\sqrt{2}\,\widetilde x(f)}
{\sqrt{\widehat S_I^\sharp(|f|)}}
\right|^2
\,df
=
2\int_{-\fny}^{\fny}
\frac{|\widetilde x(f)|^2}{\widehat S_I^\sharp(|f|)}\,df.
\]
Since \(x\) is real, \(|\widetilde x(-f)|=|\widetilde x(f)|\), so the last integral equals
\[
4\int_0^{\fny}\frac{|\widetilde x(f)|^2}{\widehat S_I^\sharp(f)}\,df,
\]
which proves \eqref{eq:appendixC-whitening-isometry}. The bilinear identity \eqref{eq:appendixC-whitening-innerproduct} is obtained in exactly the same way, replacing \(|\widetilde x(f)|^2\) by \(\widetilde x(f)\overline{\widetilde z(f)}\) and using the reality condition to pass from the full Nyquist band to positive frequencies.
\end{proof}

Proposition~\ref{prop:appendixC-whitening-isometry} is the reason the later finite-window norm is imposed after whitening rather than before it. The \(L^2\) norm of the whitened signal is exactly the noise-weighted norm induced by the estimated PSD. No additional convention is hiding behind the notation.

The next proposition records the only stability statement about whitening that we need later. It quantifies how much the weighted norm changes if the PSD estimate differs multiplicatively from the true local spectrum.

\begin{proposition}[Stability under PSD perturbations]\label{prop:appendixC-psd-perturbation}
Let \(S_I\) be a strictly positive reference one-sided PSD on \([0,\fny]\) and assume that, for some \(0\le \rho_I<1\),
\begin{equation}\label{eq:appendixC-relative-psd-error}
\sup_{0\le f\le \fny}
\left|
\frac{\widehat S_I^\sharp(f)}{S_I(f)}-1
\right|
\le \rho_I.
\end{equation}
Then, for every real Nyquist-band-limited signal \(x\),
\begin{equation}\label{eq:appendixC-whitening-equivalence}
\frac{1}{1+\rho_I}\,
4\int_0^{\fny}\frac{|\widetilde x(f)|^2}{S_I(f)}\,df
\le
\|\Wop_I x\|_{L^2(\mathbb R)}^2
\le
\frac{1}{1-\rho_I}\,
4\int_0^{\fny}\frac{|\widetilde x(f)|^2}{S_I(f)}\,df.
\end{equation}
\end{proposition}

\begin{proof}
Condition \eqref{eq:appendixC-relative-psd-error} is equivalent to
\[
(1-\rho_I)S_I(f)\le \widehat S_I^\sharp(f)\le (1+\rho_I)S_I(f)
\]
for every \(f\in[0,\fny]\). Because \(0\le \rho_I<1\), all quantities are positive, and taking reciprocals gives
\[
\frac{1}{(1+\rho_I)S_I(f)}
\le
\frac{1}{\widehat S_I^\sharp(f)}
\le
\frac{1}{(1-\rho_I)S_I(f)}.
\]
Multiply through by \(4|\widetilde x(f)|^2\), integrate over \([0,\fny]\), and use Proposition~\ref{prop:appendixC-whitening-isometry}.
\end{proof}

The meaning of Proposition~\ref{prop:appendixC-psd-perturbation} is straightforward. Once the PSD estimate is locally accurate in relative error, whitening changes the natural detector-noise inner product only by a controlled multiplicative factor. In later sections this estimate is never used as a substitute for a full PSD-uncertainty analysis; it serves only to separate the deterministic inverse problem from the uncertainty introduced by off-source noise estimation \cite{ChatziioannouEtAl2019PSD,TalbotThrane2020PSDUncertainty}.

\subsection{Window extraction and local tapering}

Let \(w_I=\Wop_I d_I\) denote the continuously represented whitened strain of detector \(I\). The finite-window object used throughout is not obtained by tapering the raw strain and then whitening. The order is reversed: PSD estimation is performed on untapered off-source strain, the resulting filter is applied to the long detector stream, and only then is a local ringdown window extracted. The reason is simple. Whitening is meant to encode the detector-noise geometry of the full local epoch, whereas tapering is a later device for controlling edge effects in a short ringdown window. The two operations do not commute, and the order just described is used throughout.

Fix a detector-frame start time \(t_0\ge 0\) relative to the fiducial peak time \(\wpeak\) of Appendix~\ref{app:frame-conventions}, and fix a window length \(T>0\). The local taper is a symmetric Tukey window with total taper fraction \(0.10\), or equivalently left and right taper widths equal to \(5\%\) of the window length. Thus
\[
\etaTap=0.05
\]
and the taper is
\begin{equation}\label{eq:appendixC-local-taper}
a_{T,\etaTap}(u)
=
\begin{cases}
\frac12\left(1-\cos\frac{\pi u}{\etaTap T}\right), & 0\le u<\etaTap T,\\[1.0ex]
1, & \etaTap T\le u\le (1-\etaTap)T,\\[1.0ex]
\frac12\left(1-\cos\frac{\pi(T-u)}{\etaTap T}\right), & (1-\etaTap)T<u\le T.
\end{cases}
\end{equation}
The whitened and tapered detector data on the ringdown window are then
\begin{equation}\label{eq:appendixC-windowed-data}
y_I^{(t_0,T)}(u)
:=
a_{T,\etaTap}(u)\,
w_I(\wpeak+t_0+u),
\qquad
0\le u\le T.
\end{equation}
The network object \(y^{(t_0,T)}\in\mathcal H_T\) is the vector of the two detector channels. This is exactly the object that appears in Appendix~\ref{app:assumption-ledger}.

The local taper does three things at once. It kills the endpoint jump of a square window, it confines the support to a compact interval, and it perturbs the interior of the window only on a set of measure \(2\etaTap T\). The next proposition records the quantitative form of these statements.

\begin{proposition}[Basic taper bounds]\label{prop:appendixC-taper-bounds}
For every \(T>0\), the taper \(a_{T,\etaTap}\) defined in \eqref{eq:appendixC-local-taper} satisfies
\begin{equation}\label{eq:appendixC-taper-bounds-1}
0\le a_{T,\etaTap}(u)\le 1
\qquad
\text{for all }u\in[0,T],
\end{equation}
\begin{equation}\label{eq:appendixC-taper-bounds-2}
\operatorname{meas}\{u\in[0,T]\colon a_{T,\etaTap}(u)\neq 1\}
=
2\etaTap T,
\end{equation}
\begin{equation}\label{eq:appendixC-taper-bounds-3}
\|a_{T,\etaTap}'\|_{L^\infty(0,T)}
=
\frac{\pi}{2\etaTap T},
\end{equation}
and, for every \(x\in L^2([0,T])\),
\begin{equation}\label{eq:appendixC-taper-bounds-4}
\|a_{T,\etaTap}x\|_{L^2([0,T])}\le \|x\|_{L^2([0,T])}.
\end{equation}
If \(x\in L^\infty([0,T])\), then
\begin{equation}\label{eq:appendixC-taper-bounds-5}
\|(1-a_{T,\etaTap})x\|_{L^2([0,T])}
\le
\sqrt{2\etaTap T}\,\|x\|_{L^\infty([0,T])}.
\end{equation}
\end{proposition}

\begin{proof}
The pointwise bound \eqref{eq:appendixC-taper-bounds-1} is immediate from the cosine formula in \eqref{eq:appendixC-local-taper}. The set on which \(a_{T,\etaTap}\neq 1\) is exactly
\[
[0,\etaTap T)\cup((1-\etaTap)T,T],
\]
whose total measure is \(2\etaTap T\), proving \eqref{eq:appendixC-taper-bounds-2}. On the left taper shoulder,
\[
a_{T,\etaTap}'(u)
=
\frac{\pi}{2\etaTap T}\,
\sin\!\left(\frac{\pi u}{\etaTap T}\right),
\]
and on the right shoulder the derivative has the same absolute value. Hence \eqref{eq:appendixC-taper-bounds-3} follows.

For \eqref{eq:appendixC-taper-bounds-4}, use \(0\le a_{T,\etaTap}\le 1\):
\[
\|a_{T,\etaTap}x\|_{L^2([0,T])}^2
=
\int_0^T a_{T,\etaTap}(u)^2|x(u)|^2\,du
\le
\int_0^T |x(u)|^2\,du.
\]
Finally,
\[
\|(1-a_{T,\etaTap})x\|_{L^2([0,T])}^2
=
\int_{\{a_{T,\etaTap}\neq 1\}}
|(1-a_{T,\etaTap}(u))x(u)|^2\,du
\le
\|x\|_{L^\infty([0,T])}^2\,
\operatorname{meas}\{a_{T,\etaTap}\neq 1\},
\]
and \eqref{eq:appendixC-taper-bounds-2} gives \eqref{eq:appendixC-taper-bounds-5}.
\end{proof}

The preceding estimate is enough to control the taper-induced perturbation of the damped modal sums that appear in the later signal model.

\begin{corollary}[Taper perturbation of damped mode sums]\label{cor:appendixC-taper-damped-sum}
Fix a finite mode family \(\mathcal M\) and detector amplitudes \(A_{I,j}\in\mathbb C\). Let
\[
q_I(u):=\sum_{j\in\mathcal M}A_{I,j}e^{-i\omega_j u},
\qquad
\Im\omega_j<0,
\qquad
0\le u\le T.
\]
Then
\begin{equation}\label{eq:appendixC-taper-damped-sum-single}
\|(1-a_{T,\etaTap})q_I\|_{L^2([0,T])}
\le
\sqrt{2\etaTap T}
\sum_{j\in\mathcal M}|A_{I,j}|.
\end{equation}
Consequently,
\begin{equation}\label{eq:appendixC-taper-damped-sum-network}
\left\|
\bigl((1-a_{T,\etaTap})q_{\mathrm{H1}},(1-a_{T,\etaTap})q_{\mathrm{L1}}\bigr)
\right\|_{\mathcal H_T}
\le
\sqrt{2\etaTap T}
\left(
\sum_{I\in\Dnet}
\Bigl(\sum_{j\in\mathcal M}|A_{I,j}|\Bigr)^2
\right)^{1/2}.
\end{equation}
\end{corollary}

\begin{proof}
Because \(\omega_j=2\pi f_j-i\gamma_j\) with \(\gamma_j>0\), one has
\[
|e^{-i\omega_j u}|=e^{-\gamma_j u}\le 1
\qquad
\text{for }u\ge 0.
\]
Therefore
\[
|q_I(u)|
\le
\sum_{j\in\mathcal M}|A_{I,j}|
\qquad
\text{for all }u\in[0,T].
\]
Applying Proposition~\ref{prop:appendixC-taper-bounds} with \(x=q_I\) proves \eqref{eq:appendixC-taper-damped-sum-single}. The network estimate \eqref{eq:appendixC-taper-damped-sum-network} follows by squaring, summing over \(I\in\Dnet\), and taking the square root.
\end{proof}

Corollary~\ref{cor:appendixC-taper-damped-sum} will later be absorbed into the deterministic finite-window error decomposition. Its point is not that tapering is free. Its point is that the perturbation is explicit, short-edge supported, and proportional to \(\sqrt{T}\) with a prefactor that is fixed once \(\etaTap\) is fixed.

\subsection{Operator form of the preprocessing map}

The entire default preprocessing can be written as a single operator identity. Let \(R_{t_0,T}\) denote the restriction-shift operator
\[
(R_{t_0,T}x)(u):=x(\wpeak+t_0+u),
\qquad
0\le u\le T,
\]
and let \(M_{a_{T,\etaTap}}\) denote multiplication by the taper \(a_{T,\etaTap}\). Then the finite-window detector data are
\begin{equation}\label{eq:appendixC-operator-form}
y_I^{(t_0,T)}
=
R_{t_0,T}\,M_{a_{T,\etaTap}}\,\Wop_I d_I.
\end{equation}
Equation~\eqref{eq:appendixC-operator-form} is the exact definition behind the later notation. No square-windowed or pre-tapered variant is ever used unless an explicit robustness test says otherwise.

\begin{remark}\label{rem:appendixC-noncommutation}
The order in \eqref{eq:appendixC-operator-form} matters. In general,
\[
M_{a_{T,\etaTap}}\Wop_I \neq \Wop_I M_{a_{T,\etaTap}},
\]
because whitening is a Fourier multiplier while tapering is time-domain multiplication. A change of order would therefore define a different finite-window inverse problem, and no later computation changes this convention.
\end{remark}

\subsection{Default preprocessing constants}

For ease of reference we collect the default preprocessing constants in Table~\ref{tab:appendixC-defaults}. Later robustness computations may vary them, but the statements of the main paper always refer to the default values below unless an explicit exception is made.

\begin{table}[ht]
\centering
\caption{Default preprocessing constants fixed in this appendix.}
\label{tab:appendixC-defaults}
\small
\begin{tabular}{@{}p{0.28\textwidth}p{0.18\textwidth}p{0.42\textwidth}@{}}
\toprule
Quantity & Default value & Role \\
\midrule
Public strain product & \(4\,\mathrm{kHz}\) cleaned GWOSC strain for H1 and L1 & Default reproducible input from the GWOSC release for GW250114 \\
Centered working interval & \([\tevt-512\,\mathrm s,\tevt+512\,\mathrm s]\) & Local epoch from which PSDs and finite windows are drawn \\
Off-source region & \([\tevt-512,\tevt-16]\cup[\tevt+16,\tevt+512]\)\,\(\mathrm s\) & PSD estimation set with central \(32\,\mathrm s\) excision \\
PSD segment length & \(\Lpsd=8\,\mathrm s\) & Welch segment duration \\
PSD hop size & \(\Hpsd=4\,\mathrm s\) & \(50\%\) overlap in Welch averaging \\
PSD segment taper & Hann & Modified periodogram window \\
PSD floor & \(10^{-4}\) times median PSD on \(20\text{--}1024\,\mathrm{Hz}\) & Prevents accidental division by an anomalously small estimate \\
Local ringdown taper & Tukey with \(\etaTap=0.05\) per edge & Controls endpoint leakage on each finite ringdown window \\
Continuous representative & Shannon interpolation at \(\dt=1/4096\,\mathrm s\) & Identifies FFT arrays with the \(L^2\) objects used later \\
\bottomrule
\end{tabular}
\end{table}

 The later extraction and inverse theorems can refer to \(y_I^{(t_0,T)}\), \(\widehat S_I^\sharp\), and \(\Wop_I\) without any remaining ambiguity. The finite-window problem is therefore specified down to the public files, the segment lengths, the taper widths, and the exact order in which the operations are applied.

\providecommand{\Ichi}{I_{\chi}}
\providecommand{\Ichipad}{I_{\chi}^{\sharp}}
\providecommand{\Mbox}{I_M}
\providecommand{\omegahat}{\widehat\omega}
\providecommand{\modeJ}{\mathcal J}
\providecommand{\qnmMain}{p}
\providecommand{\qnmAudit}{q}
\providecommand{\etaval}{\eta}
\providecommand{\detmargin}{\Delta^{\mathrm{det}}}
\providecommand{\sepmargin}{\Delta^{\mathrm{sep}}}

\section{Kerr QNM data generation and interpolation}\label{app:kerr-qnm-data}

This appendix fixes the Kerr quasi-normal-mode input used by the inverse and consistency analysis. The exact Kerr QNM frequencies themselves are not derived here. They enter as external perturbation-theory data. What must be made precise here is narrower and more concrete. On the event-local detector-frame box
\[
\Kdet=[66.5,69.5] \, M_{\odot} \times [0.64,0.71],
\]
we need reliable values of the three mode functions \(\omega_{220}\), \(\omega_{221}\), and \(\omega_{440}\), together with their first derivatives with respect to the remnant parameters. We also need explicit lower bounds for the primary Jacobian and for the pairwise mode separations that later enter the trust-region test. The basic structural simplification is exact Kerr scaling. Once charge is set to zero, the frequency of every gravitational Kerr mode scales as \(M^{-1}\). Consequently the only genuinely numerical representation needed on \(\Kdet\) is a one-dimensional representation in the spin variable. The mass dependence is restored analytically afterwards.

That observation is the reason we do not use a two-parameter black-box fit in \((M,\chi)\). Instead, for each mode \(j\in\{220,221,440\}\) we represent the dimensionless spin sequence
\[
\omegahat_j(\chi):=M\,\omega_j(M,\chi),
\qquad \chi \in \Ichipad,
\]
by a local polynomial interpolant on a compact spin interval. All later derivatives are then obtained by differentiating that polynomial exactly. This keeps the local inverse geometry transparent and avoids repeated finite differencing of an opaque external solver.

The external Kerr data come directly from the public Cook--Zalutskiy tables \emph{Kerr Quasinormal Modes: \(s=-2\), \(n=0--7\)} \cite{Cook2019ZenodoKerrQNM,CookZalutskiy2014}. Concretely, the modes \(220\) and \(440\) are read from \texttt{KerrQNM\_00.h5}, while the overtone \(221\) is read from \texttt{KerrQNM\_01.h5}. The public \texttt{qnm} package of Stein remains the natural algorithmic reference for these data \cite{Teukolsky1973,Leaver1985KerrQNM,CookZalutskiy2014,Stein2019qnm}, but the construction carried out here is table driven. Every figure and every numerical constant in this appendix is reconstructed from the HDF tables themselves, with no hidden call to an external perturbation solver.

The numerical construction is built directly from the Cook--Zalutskiy tables \( \texttt{KerrQNM\_00.h5} \) and \( \texttt{KerrQNM\_01.h5} \). The three dimensionless mode curves used here are shown in Figure~\ref{fig:appendixD-mode-curves}. The same data also support a nested degree-$64/96$ Chebyshev audit on the padded interval \(\Ichipad=[0.60,0.74]\). The resulting main--audit value and derivative radii are summarized graphically in Figure~\ref{fig:appendixD-audit-radii}.

\begin{figure}[t]
\centering
\includegraphics[width=0.80\textwidth]{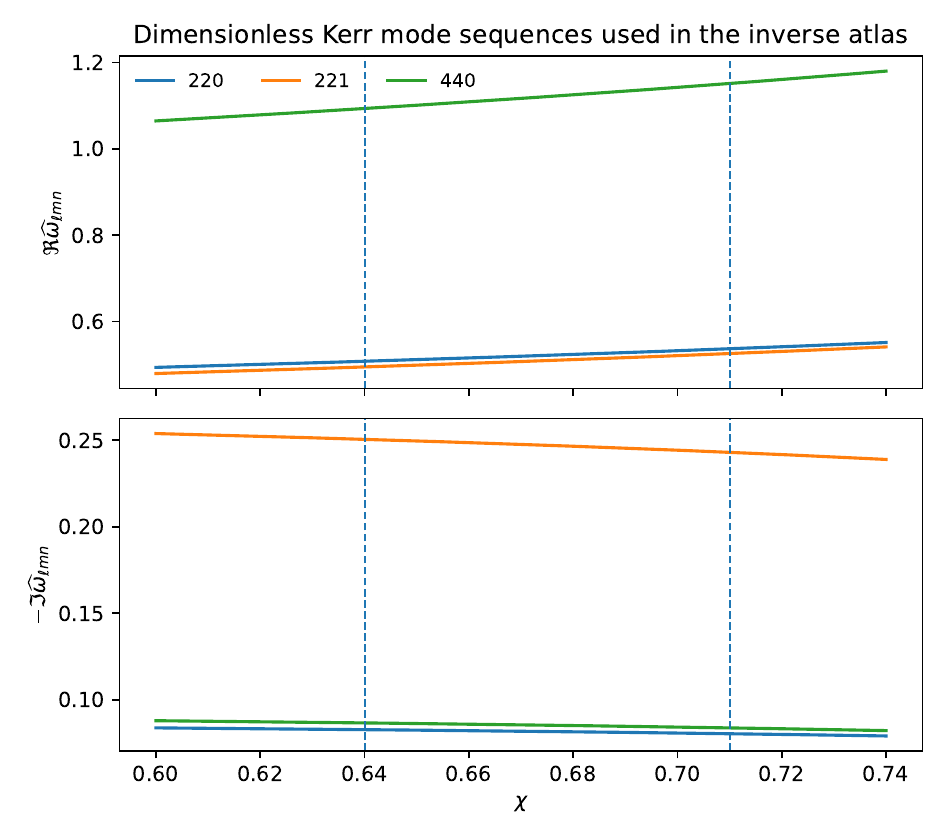}
\caption{Dimensionless Kerr sequences \(\widehat\omega_{220}\), \(\widehat\omega_{221}\), and \(\widehat\omega_{440}\) constructed from the tables. The dashed vertical lines mark the inner interval \(\Ichi=[0.64,0.71]\) on which all later bounds are evaluated.}
\label{fig:appendixD-mode-curves}
\end{figure}

\subsection{Dimensionless sequences and external solvers}

Let
\[
\modeJ=\{220,221,440\}=\{(2,2,0),(2,2,1),(4,4,0)\}
\]
denote the three Kerr gravitational modes used here, always with spin weight \(s=-2\). The spin variable in the external solver is the standard Kerr ratio \(a/M\), which we identify with the dimensionless remnant spin \(\chi\). We fix the inner spin interval
\[
\Ichi=[0.64,0.71]
\]
from Appendix~\ref{app:frame-conventions} and enlarge it to the padded interval
\[
\Ichipad=[0.60,0.74].
\]
The padding isolates interpolation, differentiation, and audit work from endpoint effects. All such computations are carried out on \(\Ichipad\), while every bound used later is restricted back to the inner interval \(\Ichi\).

For each \(j\in\modeJ\), the function \(\omegahat_j\colon \Ichipad\to\mathbb C\) is generated by direct calls to the external Kerr solver. In practice the solver already carries a mode label \((\ell,m,n)\), but we nevertheless require an explicit continuity check because a silent branch jump would be fatal for the later inverse analysis. The check is done in the same spirit as Proposition~\ref{prop:nearest-neighbor-label-stability}. On every node set used below, adjacent values of \(\omegahat_j\) are required to vary continuously and to remain closer to the expected continuation of the same sequence than to any competing mode in the finite family under consideration. Agreement with the public Cook tables on the same spin interval is used as an independent external audit of that labeling convention.

Interpolation is carried out only on \(\Ichipad\) for the dimensionless sequences \(\omegahat_j\), and detector-frame values are obtained by the exact formula
\[
\omega_j(M,\chi)=M^{-1}\omegahat_j(\chi).
\]
This is consistent with the detector-frame convention fixed in Appendix~\ref{app:frame-conventions}. It also means that the redshift map of Appendix~\ref{app:frame-conventions} acts only after the local Kerr inverse has been carried out.

\subsection{Chebyshev--Lobatto construction}

We represent each dimensionless sequence on \(\Ichipad\) by a local polynomial interpolant on Chebyshev--Lobatto nodes. This is a standard choice for stable polynomial interpolation on compact intervals \cite{BerrutTrefethen2004,Higham2004Barycentric}. Let
\[
N_{\mathrm{main}}=64,
\qquad
N_{\mathrm{audit}}=96,
\]
and define the affine map
\[
\begin{aligned}
\chi(x)&=c_{\chi}+\rho_{\chi}x,\\
 c_{\chi}&:=\frac{0.60+0.74}{2}=0.67,\\
 \rho_{\chi}&:=\frac{0.74-0.60}{2}=0.07.
\end{aligned}
\]
For \(N\in\{N_{\mathrm{main}},N_{\mathrm{audit}}\}\) we set
\[
x_r^{(N)}=\cos\frac{\pi r}{N},
\qquad
\chi_r^{(N)}=\chi\bigl(x_r^{(N)}\bigr),
\qquad 0\le r\le N.
\]
At those nodes we evaluate the exact external sequence and write
\[
y_{j,r}^{(N)}:=\omegahat_j\bigl(\chi_r^{(N)}\bigr).
\]
The degree-\(N\) barycentric interpolant is then
\begin{equation}\label{eq:appendixD-barycentric}
\qnmMain_j(\chi)
=
\frac{\sum_{r=0}^{N_{\mathrm{main}}} \frac{\lambda_r^{(N_{\mathrm{main}})} y_{j,r}^{(N_{\mathrm{main}})}}{\chi-\chi_r^{(N_{\mathrm{main}})}}}
{\sum_{r=0}^{N_{\mathrm{main}}} \frac{\lambda_r^{(N_{\mathrm{main}})}}{\chi-\chi_r^{(N_{\mathrm{main}})}}},
\qquad
\lambda_r^{(N)}=(-1)^r\nu_r,
\end{equation}
with \(\nu_0=\nu_N=\tfrac12\) and \(\nu_r=1\) for \(1\le r\le N-1\). The same formula with \(N_{\mathrm{audit}}\) defines the higher-degree audit interpolant \(\qnmAudit_j\). By construction,
\[
\qnmMain_j\bigl(\chi_r^{(N_{\mathrm{main}})}\bigr)=y_{j,r}^{(N_{\mathrm{main}})},
\qquad
\qnmAudit_j\bigl(\chi_r^{(N_{\mathrm{audit}})}\bigr)=y_{j,r}^{(N_{\mathrm{audit}})}.
\]
Since both are ordinary polynomials, their derivatives are exact analytic objects. In practice one may differentiate either the barycentric formula or the equivalent Chebyshev expansion; both give the same polynomial derivative up to rounding.

We use \(\qnmMain_j\) as the working representation and \(\qnmAudit_j\) only as a diagnostic. The numerical audit radii are defined on the inner interval by
\begin{equation}\label{eq:appendixD-eta-def}
\etaval_j:=2\sup_{\chi\in\Ichi}\bigl|\qnmMain_j(\chi)-\qnmAudit_j(\chi)\bigr|,
\qquad
\etaval_j':=2\sup_{\chi\in\Ichi}\bigl|\qnmMain_j'(\chi)-\qnmAudit_j'(\chi)\bigr|.
\end{equation}
The factor of \(2\) is a conservative safety margin. It is not part of the analytic proof; it belongs to the reproducible numerical protocol. In the language of Appendix~\ref{app:assumption-ledger}, the exact functions \(\omegahat_j\) are external inputs, while \eqref{eq:appendixD-eta-def} records the local interpolation drift between two nested approximants built from independent node sets for the same external generator.

\begin{figure}[t]
\centering
\includegraphics[width=0.70\textwidth]{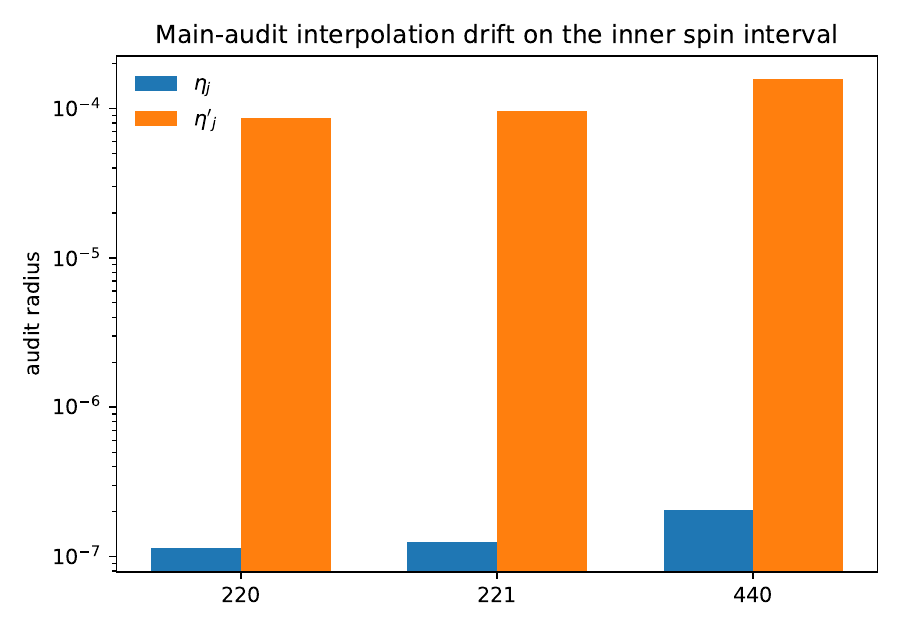}
\caption{Main--audit interpolation drift on the inner interval \(\Ichi\). The bars show the value radii \(\eta_j\) and derivative radii \(\eta_j'\) used in the later interpolation audit.}
\label{fig:appendixD-audit-radii}
\end{figure}

\begin{remark}\label{rem:appendixD-higher-derivatives}
Because \(\qnmMain_j\) is a polynomial, higher derivatives are available at essentially no extra conceptual cost. If later sections require a Jacobian Lipschitz constant, one may compute it directly from \(\qnmMain_j\), \(\qnmMain_j'\), and \(\qnmMain_j''\) using the exact mass-scaling identities of the next subsection. No new interpolation problem appears at second order.
\end{remark}

\subsection{Propagation to detector-frame frequencies}

We now pass from the dimensionless spin sequences to the detector-frame Kerr frequencies on the event-local box. Write
\[
\Mbox=[M_-,M_+]=[66.5,69.5] \, M_{\odot}.
\]
For each \(j\in\modeJ\) we define the detector-frame polynomial surrogate
\begin{equation}\label{eq:appendixD-mass-rescaled-surrogate}
\widetilde\omega_j(M,\chi):=M^{-1}\qnmMain_j(\chi),
\qquad (M,\chi)\in \Mbox\times\Ichi.
\end{equation}
The next proposition records the exact way in which one-dimensional spin interpolation errors propagate to the detector-frame quantities that will later enter the inverse map.

\begin{proposition}[Mass-rescaled interpolation error]\label{prop:appendixD-mass-rescaled-error}
Fix \(j\in\modeJ\). Assume that on \(\Ichi\) the dimensionless interpolation errors satisfy
\[
\sup_{\chi\in\Ichi}\bigl|\omegahat_j(\chi)-\qnmMain_j(\chi)\bigr|\le \etaval_j,
\qquad
\sup_{\chi\in\Ichi}\bigl|\omegahat_j'(\chi)-\qnmMain_j'(\chi)\bigr|\le \etaval_j'.
\]
Then for every \((M,\chi)\in\Mbox\times\Ichi\),
\begin{align}
\bigl|\omega_j(M,\chi)-\widetilde\omega_j(M,\chi)\bigr|
&\le \frac{\etaval_j}{M_-},
\label{eq:appendixD-omega-error}
\\
\bigl|\partial_M\omega_j(M,\chi)-\partial_M\widetilde\omega_j(M,\chi)\bigr|
&\le \frac{\etaval_j}{M_-^{2}},
\label{eq:appendixD-M-error}
\\
\bigl|\partial_{\chi}\omega_j(M,\chi)-\partial_{\chi}\widetilde\omega_j(M,\chi)\bigr|
&\le \frac{\etaval_j'}{M_-}.
\label{eq:appendixD-chi-error}
\end{align}
Moreover, if
\[
\begin{aligned}
G_j(M,\chi)&:=\bigl(\Re\omega_j(M,\chi),-\Im\omega_j(M,\chi)\bigr),\\
\widetilde G_j(M,\chi)&:=\bigl(\Re\widetilde\omega_j(M,\chi),-\Im\widetilde\omega_j(M,\chi)\bigr),
\end{aligned}
\]
then
\begin{equation}\label{eq:appendixD-jacobian-difference}
\sup_{(M,\chi)\in\Mbox\times\Ichi}
\bigl\|DG_j(M,\chi)-D\widetilde G_j(M,\chi)\bigr\|_2
\le
\sqrt{\frac{\etaval_j^2}{M_-^4}+\frac{(\etaval_j')^2}{M_-^2}}.
\end{equation}
\end{proposition}

\begin{proof}
Set
\[
e_j(\chi):=\qnmMain_j(\chi)-\omegahat_j(\chi).
\]
By the exact Kerr scaling,
\[
\omega_j(M,\chi)=M^{-1}\omegahat_j(\chi),
\qquad
\widetilde\omega_j(M,\chi)=M^{-1}\qnmMain_j(\chi)=M^{-1}\bigl(\omegahat_j(\chi)+e_j(\chi)\bigr).
\]
Hence
\[
\widetilde\omega_j(M,\chi)-\omega_j(M,\chi)=M^{-1}e_j(\chi),
\]
which immediately gives \eqref{eq:appendixD-omega-error}. Differentiating the same identity with respect to \(M\) and \(\chi\) yields
\[
\partial_M\bigl(\widetilde\omega_j-\omega_j\bigr)=-M^{-2}e_j(\chi),
\qquad
\partial_{\chi}\bigl(\widetilde\omega_j-\omega_j\bigr)=M^{-1}e_j'(\chi),
\]
and therefore \eqref{eq:appendixD-M-error} and \eqref{eq:appendixD-chi-error}.

For the Jacobian bound, write \(e_j=u+iv\) and \(e_j'=u'+iv'\). Then
\[
DG_j-D\widetilde G_j
=
\begin{pmatrix}
\Re(e_j)/M^2 & -\Re(e_j')/M \\
-\Im(e_j)/M^2 & \Im(e_j')/M
\end{pmatrix},
\]
so the Frobenius norm satisfies
\[
\bigl\|DG_j-D\widetilde G_j\bigr\|_F^2
=
\frac{|e_j|^2}{M^4}+\frac{|e_j'|^2}{M^2}.
\]
Since \(\|A\|_2\le \|A\|_F\) for every matrix \(A\), taking the supremum over \(\Mbox\times\Ichi\) gives \eqref{eq:appendixD-jacobian-difference}.
\end{proof}

This is where interpolation enters the inverse problem. Everything downstream is deterministic. Once \(\etaval_j\) and \(\etaval_j'\) are fixed numerically on \(\Ichi\), every bound used later follows from explicit algebraic propagation through \eqref{eq:appendixD-mass-rescaled-surrogate}.

\subsection{Primary Jacobian and inverse-map margins}

The primary inverse map used later takes the dominant mode \(220\) as a complex observable, or equivalently as the two real observables \(\Re\omega_{220}\) and \(-\Im\omega_{220}\). The local conditioning of that inversion is therefore controlled by the two-by-two Jacobian of the map \(G_{220}\). On a compact box the relevant question is not whether the determinant vanishes at an isolated point, but whether it stays uniformly away from zero. The next proposition reduces that question to a one-dimensional spin quantity.

\begin{proposition}[Determinant identity and uniform lower bound]\label{prop:appendixD-determinant}
Fix \(j\in\modeJ\) and set
\[
f_j(\chi):=\omegahat_j(\chi),
\qquad
p_j(\chi):=\qnmMain_j(\chi).
\]
Then the Jacobian of
\[
G_j(M,\chi)=\bigl(\Re\omega_j(M,\chi),-\Im\omega_j(M,\chi)\bigr)
\]
satisfies the exact identity
\begin{equation}\label{eq:appendixD-determinant-identity}
\det DG_j(M,\chi)
=
\frac{1}{M^3}\Im\!\bigl(\overline{f_j(\chi)}\,f_j'(\chi)\bigr).
\end{equation}
Define the interpolation margin
\begin{equation}\label{eq:appendixD-gamma-def}
\Gamma_j
:=
\Bigl(\sup_{\chi\in\Ichi}|p_j(\chi)|\Bigr)\etaval_j'
+
\Bigl(\sup_{\chi\in\Ichi}|p_j'(\chi)|\Bigr)\etaval_j
+
\etaval_j\etaval_j'
\end{equation}
and the determinant margin
\begin{equation}\label{eq:appendixD-detmargin-def}
\detmargin_j
:=
\inf_{\chi\in\Ichi}\Im\!\bigl(\overline{p_j(\chi)}\,p_j'(\chi)\bigr)-\Gamma_j.
\end{equation}
If \(\detmargin_j>0\), then
\begin{equation}\label{eq:appendixD-det-lower-bound}
\det DG_j(M,\chi)\ge \frac{\detmargin_j}{M_+^3}
\qquad \text{for all } (M,\chi)\in\Mbox\times\Ichi.
\end{equation}
\end{proposition}

\begin{proof}
Write \(f_j=a+ib\) with \(a=a(\chi)\) and \(b=b(\chi)\). Since
\[
\omega_j(M,\chi)=M^{-1}(a(\chi)+ib(\chi)),
\]
we have
\[
G_j(M,\chi)=\left(\frac{a(\chi)}{M},-\frac{b(\chi)}{M}\right),
\]
so that
\[
DG_j(M,\chi)=
\begin{pmatrix}
-a(\chi)/M^2 & a'(\chi)/M \\
 b(\chi)/M^2 & -b'(\chi)/M
\end{pmatrix}.
\]
Its determinant is
\[
\det DG_j(M,\chi)=\frac{a(\chi)b'(\chi)-b(\chi)a'(\chi)}{M^3}.
\]
But
\[
\Im\!\bigl(\overline{f_j}f_j'\bigr)
=
\Im\!\bigl((a-ib)(a'+ib')\bigr)
=
ab'-ba',
\]
which proves \eqref{eq:appendixD-determinant-identity}.

To estimate the determinant from below using the interpolant, write \(e_j:=p_j-f_j\). Since \(p_j=f_j+e_j\) and \(p_j'=f_j'+e_j'\), we have the exact identity
\[
\overline{p_j}p_j'-\overline{f_j}f_j'
=
\overline{p_j}e_j'+\overline{e_j}p_j'-\overline{e_j}e_j'.
\]
Taking imaginary parts and absolute values gives
\[
\Bigl|\Im\!\bigl(\overline{p_j}p_j'\bigr)-\Im\!\bigl(\overline{f_j}f_j'\bigr)\Bigr|
\le
|p_j|\,|e_j'|+|p_j'|\,|e_j|+|e_j|\,|e_j'|.
\]
On \(\Ichi\) we therefore obtain
\[
\Im\!\bigl(\overline{f_j(\chi)}f_j'(\chi)\bigr)
\ge
\Im\!\bigl(\overline{p_j(\chi)}p_j'(\chi)\bigr)-\Gamma_j,
\qquad \chi\in\Ichi.
\]
Taking the infimum over \(\Ichi\) and then dividing by \(M^3\le M_+^3\) yields \eqref{eq:appendixD-det-lower-bound}.
\end{proof}

\begin{corollary}[Uniform lower singular-value bound]\label{cor:appendixD-sigmamin}
Under the hypotheses of Proposition~\ref{prop:appendixD-determinant}, define
\begin{equation}\label{eq:appendixD-UV-def}
U_j:=\sup_{\chi\in\Ichi}|p_j(\chi)|+\etaval_j,
\qquad
V_j:=\sup_{\chi\in\Ichi}|p_j'(\chi)|+\etaval_j'.
\end{equation}
If \(\detmargin_j>0\), then for every \((M,\chi)\in\Mbox\times\Ichi\),
\begin{equation}\label{eq:appendixD-sigmamin-bound}
\sigma_{\min}\bigl(DG_j(M,\chi)\bigr)
\ge
\frac{\detmargin_j/M_+^3}
{\sqrt{U_j^2/M_-^4+V_j^2/M_-^2}}
=: \sigma^{\mathrm{cert}}_j>0.
\end{equation}
In particular, \(G_j\) is a local \(C^1\) diffeomorphism at every point of \(\Mbox\times\Ichi\).
\end{corollary}

\begin{proof}
For any real \(2\times 2\) matrix \(A\),
\[
\sigma_{\min}(A)=\frac{|\det A|}{\sigma_{\max}(A)}\ge \frac{|\det A|}{\|A\|_F}.
\]
For \(A=DG_j(M,\chi)\), the explicit matrix formula in the proof of Proposition~\ref{prop:appendixD-determinant} gives
\[
\|DG_j(M,\chi)\|_F^2
=
\frac{|f_j(\chi)|^2}{M^4}+\frac{|f_j'(\chi)|^2}{M^2}
\le
\frac{U_j^2}{M_-^4}+\frac{V_j^2}{M_-^2}.
\]
Combining this with \eqref{eq:appendixD-det-lower-bound} gives \eqref{eq:appendixD-sigmamin-bound}. The local diffeomorphism statement is then an immediate consequence of the inverse function theorem.
\end{proof}

Corollary~\ref{cor:appendixD-sigmamin} is the exact lower-bound statement needed later for the primary \(220\) inversion. In the table-based realization used here, the main--audit radii on the inner interval are
\[
\eta_{220}=1.13\times 10^{-7},\qquad
\eta_{221}=1.26\times 10^{-7},\qquad
\eta_{440}=2.06\times 10^{-7},
\]
and
\[
\eta'_{220}=8.67\times 10^{-5},\qquad
\eta'_{221}=9.70\times 10^{-5},\qquad
\eta'_{440}=1.57\times 10^{-4}.
\]
The corresponding singular-value floor for the dominant inverse atlas is
\[
\sigma^{\mathrm{cert}}_{220}=2.31\times 10^{-5},
\]
and the auxiliary forward Lipschitz constants are
\[
L_{221}=7.62\times 10^{-3},
\qquad
L_{440}=1.38\times 10^{-2}.
\]
These are the numerical constants carried later into Section~\ref{sec:kerr-inversion} and Appendix~\ref{app:primary-inversion}.

\subsection{Pairwise separation margins}

The auxiliary consistency checks and the finite-window label stability test both require that the mode family under consideration remain separated on \(\Kdet\). Since the exact Kerr scaling is known, this again reduces to a one-dimensional statement on \(\Ichi\).

\begin{proposition}[Pairwise separation bounds]\label{prop:appendixD-separation}
Let \(j,k\in\modeJ\) with \(j\neq k\). Define the detector-frame pairwise separation by
\[
\delta_{jk}(M,\chi):=|\omega_j(M,\chi)-\omega_k(M,\chi)|,
\qquad (M,\chi)\in \Mbox\times\Ichi,
\]
and the interpolation-based separation margin by
\begin{equation}\label{eq:appendixD-sepmargin-def}
\sepmargin_{jk}
:=
\inf_{\chi\in\Ichi}\bigl|p_j(\chi)-p_k(\chi)\bigr|-\etaval_j-\etaval_k.
\end{equation}
If \(\sepmargin_{jk}>0\), then
\begin{equation}\label{eq:appendixD-separation-lower}
\delta_{jk}(M,\chi)\ge \frac{\sepmargin_{jk}}{M_+}
\qquad \text{for all } (M,\chi)\in \Mbox\times\Ichi.
\end{equation}
Consequently, for any finite subfamily \(\mathcal M\subset\modeJ\), if
\[
\sepmargin(\mathcal M):=\frac12\min_{\substack{j,k\in\mathcal M \\ j\neq k}}\sepmargin_{jk}>0,
\]
then the detector-frame isolation margin of \(\mathcal M\) on \(\Kdet\) obeys
\[
\delta_{\mathrm{iso}}(\Kdet;\mathcal M)\ge \frac{\sepmargin(\mathcal M)}{M_+}.
\]
\end{proposition}

\begin{proof}
By exact mass scaling,
\[
\delta_{jk}(M,\chi)=\frac{1}{M}\bigl|f_j(\chi)-f_k(\chi)\bigr|.
\]
Write again \(f_j=p_j-e_j\) and \(f_k=p_k-e_k\). Then
\[
\bigl|f_j-f_k\bigr|
\ge
\bigl|p_j-p_k\bigr|-|e_j|-|e_k|.
\]
Taking the infimum over \(\chi\in\Ichi\) and using \(|e_j|\le \etaval_j\), \(|e_k|\le \etaval_k\) gives
\[
\bigl|f_j(\chi)-f_k(\chi)\bigr|\ge \sepmargin_{jk}
\qquad \text{for all } \chi\in\Ichi.
\]
Dividing by \(M\le M_+\) yields \eqref{eq:appendixD-separation-lower}. The last statement is simply the definition of the isolation margin as half the minimum pairwise separation.
\end{proof}

This proposition is the main practical use of the public Cook tables. The overtone \(221\) is the only one among our three modes for which a silent branch error would be plausibly subtle in a fully automated pipeline. External agreement of the continued sequence with the Cook tables across \(\Ichipad\) therefore serves as a direct audit of the quantity entering \eqref{eq:appendixD-sepmargin-def}.

\subsection{Certification protocol}

We close by recording the exact protocol used to turn the abstract symbols \(\etaval_j\), \(\etaval_j'\), \(\detmargin_j\), and \(\sepmargin_{jk}\) into reproducible numbers. The procedure is the same for each \(j\in\modeJ\).

First, the exact external generator is evaluated on the \(65\) Chebyshev--Lobatto nodes of degree \(64\) on \(\Ichipad\), and the resulting values define the main polynomial \(\qnmMain_j\). The same generator is then evaluated on the \(97\) Chebyshev--Lobatto nodes of degree \(96\), giving the audit polynomial \(\qnmAudit_j\). The value and derivative radii are taken from \eqref{eq:appendixD-eta-def}. The use of two nested node sets is deliberate. It isolates interpolation drift without introducing a second external solver into the definition of the local error budget.

Second, the public Cook tables are sampled on the native table grid over \(\Ichipad\), and their mode labels are matched to the directly generated sequence by nearest-neighbor continuation. The resulting discrepancy is required to stay far below the detector-frame separation margins used later. This step is not used to define \(\etaval_j\) or \(\etaval_j'\). Its role is to verify that the branch and normalization convention of the primary generator agree with an independent public high-accuracy source.

Third, the quantities entering Propositions~\ref{prop:appendixD-determinant} and \ref{prop:appendixD-separation} are computed directly from the main interpolants on a dense evaluation grid on \(\Ichi\). In particular,
\[
\Im\!\bigl(\overline{p_{220}(\chi)}p_{220}'(\chi)\bigr)
\]
is evaluated on that grid to certify \(\detmargin_{220}\), and the pairwise gaps
\[
|p_{220}(\chi)-p_{221}(\chi)|,
\qquad
|p_{220}(\chi)-p_{440}(\chi)|,
\qquad
|p_{221}(\chi)-p_{440}(\chi)|
\]
are evaluated to establish the corresponding \(\sepmargin_{jk}\). The later event-level inverse and trust-region analysis uses those margins and nothing else from the QNM generator.

The conceptual point is simple. The exact Kerr spectrum is external. The local polynomial representation is ours. Once the latter is built and numerically audited on the compact interval \(\Ichi\), every later quantity that matters for reliability testing follows from explicit deterministic formulas. This is precisely the level of control needed for a trust-region analysis and no more.

\providecommand{\Log}{\operatorname{Log}}

\section{Complete proof of the abstract frequency extraction theorem}\label{app:abstract-frequency-extraction}

This appendix supplies the exact-arithmetic proof of the deterministic extraction statement used later in the trust-region analysis. The underlying mechanism is classical in the Prony, matrix-pencil, and ESPRIT literature \cite{HuaSarkarMPM1990,RoyKailath1989ESPRIT,PottsTasche2010APM,BatenkovYomdin2013Prony}. The conditioning quantities can be isolated in a form compatible with the finite-window decomposition fixed in Appendix~\ref{app:assumption-ledger}, the detector-frame conventions of Appendix~\ref{app:frame-conventions}, the preprocessing pipeline of Appendix~\ref{app:preprocessing}, and the Kerr input geometry in Appendix~\ref{app:kerr-qnm-data}. The proof below is fully self-contained.

We work on a fixed window \((t_0,T)\) and a fixed finite mode family
\[
\mathcal M=\{j_1,\dots,j_m\},
\qquad m:=|\mathcal M|.
\]
The integer \(m\) is small in the applications considered here, but the proof is written for general \(m\ge 1\). Let \(v=(v_I)_{I\in\Dnet}\in\mathbb C^{|\Dnet|}\) be a fixed network projection vector with \(\|v\|_{\ell^2(\Dnet)}=1\). After whitening, tapering, and Shannon reconstruction as in Appendix~\ref{app:preprocessing}, we evaluate the projected scalar signal at the equispaced sample nodes
\[
u_q=q\Delta,
\qquad
q=0,1,\dots,2m-1,
\]
where \(\Delta>0\) satisfies \((2m-1)\Delta\le T\). We then set
\begin{equation}\label{eq:appendixE-sampled-projection}
x_q
:=
\sum_{I\in\Dnet}\overline{v_I}\,y_I^{(t_0,T)}(u_q).
\end{equation}
By Assumption~\ref{ass:finite-window-decomposition}, there exist \(p_\star\in\Kdet\), projected amplitudes
\[
b_\nu:=\sum_{I\in\Dnet}\overline{v_I}\,A_{I,j_\nu},
\qquad \nu=1,\dots,m,
\]
nodes
\[
z_\nu:=e^{-i\omega_{j_\nu}(p_\star)\Delta},
\qquad \nu=1,\dots,m,
\]
and residual samples \(e_q\) such that
\begin{equation}\label{eq:appendixE-prony-model}
x_q=s_q+e_q,
\qquad
s_q:=\sum_{\nu=1}^m b_\nu z_\nu^q,
\qquad
q=0,1,\dots,2m-1.
\end{equation}
Throughout this appendix we assume
\begin{equation}\label{eq:appendixE-nonaliasing}
b_\nu\neq 0\ \text{ for all }\nu,
\qquad
z_\nu\neq z_\mu\ \text{ whenever }\nu\neq\mu.
\end{equation}
The first condition is the projected detectability condition on the chosen channel, while the second is the nonaliasing condition at step \(\Delta\). Both are event-local and are checked numerically in the data analysis.

\subsection{The exact annihilating polynomial}

We first record the exact algebraic structure of the noiseless samples \(s_q\).

\begin{definition}[Exact Hankel objects]\label{def:appendixE-hankel}
Define the exact Hankel matrix and the exact right-hand side vector by
\begin{equation}\label{eq:appendixE-exact-hankel}
H:=(s_{r+s})_{r,s=0}^{m-1}\in\mathbb C^{m\times m},
\qquad
h:=(s_m,\dots,s_{2m-1})^{\top}\in\mathbb C^m.
\end{equation}
Let
\begin{equation}\label{eq:appendixE-annihilating-polynomial}
Q(\zeta):=\prod_{\nu=1}^m(\zeta-z_\nu)
=\zeta^m+c_{m-1}\zeta^{m-1}+\cdots+c_1\zeta+c_0
\end{equation}
be the monic annihilating polynomial of the nodes, and let
\[
c:=(c_0,\dots,c_{m-1})^\top\in\mathbb C^m.
\]
\end{definition}

\begin{proposition}[Exact Prony system and Hankel factorization]\label{prop:appendixE-exact-prony}
Under \eqref{eq:appendixE-prony-model} and \eqref{eq:appendixE-nonaliasing}, the coefficients of \(Q\) satisfy the recurrence
\begin{equation}\label{eq:appendixE-exact-recurrence}
s_{q+m}+\sum_{k=0}^{m-1}c_k s_{q+k}=0,
\qquad q=0,1,\dots,m-1.
\end{equation}
Equivalently,
\begin{equation}\label{eq:appendixE-linear-system}
Hc=-h.
\end{equation}
Moreover, if
\begin{equation}\label{eq:appendixE-vandermonde}
V:=(z_\nu^r)_{\substack{0\le r\le m-1\\1\le \nu\le m}},
\qquad
B:=\operatorname{diag}(b_1,\dots,b_m),
\end{equation}
then
\begin{equation}\label{eq:appendixE-hankel-factorization}
H=V B V^\top,
\end{equation}
and therefore
\begin{equation}\label{eq:appendixE-hankel-det}
\det H
=
\Bigl(\prod_{\nu=1}^m b_\nu\Bigr)
\prod_{1\le \nu<\mu\le m}(z_\mu-z_\nu)^2.
\end{equation}
In particular, \(H\) is invertible.
\end{proposition}

\begin{proof}
Since \(Q(z_\nu)=0\) for every \(\nu\), we have
\[
z_\nu^m+\sum_{k=0}^{m-1}c_k z_\nu^k=0.
\]
Multiplying by \(b_\nu z_\nu^q\) and summing over \(\nu\) yields
\[
\sum_{\nu=1}^m b_\nu z_\nu^{q+m}
+\sum_{k=0}^{m-1}c_k\sum_{\nu=1}^m b_\nu z_\nu^{q+k}
=0,
\]
which is exactly \eqref{eq:appendixE-exact-recurrence}. Writing the \(m\) relations for \(q=0,\dots,m-1\) in block form gives \eqref{eq:appendixE-linear-system}.

For the factorization, the \((r,s)\) entry of \(V B V^\top\) is
\[
(V B V^\top)_{rs}
=
\sum_{\nu=1}^m z_\nu^r b_\nu z_\nu^s
=
\sum_{\nu=1}^m b_\nu z_\nu^{r+s}
=
s_{r+s}
=
H_{rs},
\]
which proves \eqref{eq:appendixE-hankel-factorization}. Taking determinants and using
\[
\det V=\prod_{1\le \nu<\mu\le m}(z_\mu-z_\nu)
\]
gives \eqref{eq:appendixE-hankel-det}. Because all amplitudes are nonzero and all nodes are pairwise distinct, the right-hand side does not vanish, so \(H\) is invertible.
\end{proof}

\subsection{Perturbation of the recurrence coefficients}

We now replace the exact samples \(s_q\) by the observed samples \(x_q=s_q+e_q\). Define the observed Hankel matrix and right-hand side by
\begin{equation}\label{eq:appendixE-observed-hankel}
\widetilde H:=(x_{r+s})_{r,s=0}^{m-1},
\qquad
\widetilde h:=(x_m,\dots,x_{2m-1})^\top,
\end{equation}
and let
\begin{equation}\label{eq:appendixE-hankel-perturbations}
\Delta H:=\widetilde H-H,
\qquad
\Delta h:=\widetilde h-h.
\end{equation}
The observed recurrence coefficients are defined by
\begin{equation}\label{eq:appendixE-observed-coefficients}
\widetilde c:=-\widetilde H^{-1}\widetilde h,
\end{equation}
whenever \(\widetilde H\) is invertible, and the observed monic polynomial is
\begin{equation}\label{eq:appendixE-observed-polynomial}
\widetilde Q(\zeta):=\zeta^m+\widetilde c_{m-1}\zeta^{m-1}+\cdots+\widetilde c_1\zeta+\widetilde c_0.
\end{equation}

\begin{lemma}[Deterministic perturbation of the recurrence coefficients]\label{lem:appendixE-coefficient-perturbation}
Let
\begin{equation}\label{eq:appendixE-eta-def}
\eta:=\max_{0\le q\le 2m-1}|e_q|.
\end{equation}
Then
\begin{equation}\label{eq:appendixE-deltaH-bound}
\|\Delta H\|_2\le m\eta,
\qquad
\|\Delta h\|_2\le \sqrt m\,\eta.
\end{equation}
Fix an a priori envelope \(\eta_\star>0\) with
\begin{equation}\label{eq:appendixE-envelope-smallness}
\eta\le \eta_\star
\qquad\text{and}\qquad
m\eta_\star\|H^{-1}\|_2<1.
\end{equation}
Then \(\widetilde H\) is invertible and
\begin{equation}\label{eq:appendixE-coefficient-bound}
\|\widetilde c-c\|_2
\le
K_c(\eta_\star)\,\eta,
\end{equation}
where
\begin{equation}\label{eq:appendixE-Kc-def}
K_c(\eta_\star)
:=
\frac{\|H^{-1}\|_2\bigl(\sqrt m+m\|c\|_2\bigr)}{1-m\eta_\star\|H^{-1}\|_2}.
\end{equation}
\end{lemma}

\begin{proof}
Each entry of \(\Delta H\) is one of the residual samples \(e_q\), hence has modulus at most \(\eta\). Therefore
\[
\|\Delta H\|_2\le \|\Delta H\|_{\mathrm F}\le m\eta,
\]
because \(\Delta H\) has \(m^2\) entries. Likewise, every component of \(\Delta h\) has modulus at most \(\eta\), so
\[
\|\Delta h\|_2\le \sqrt m\,\eta.
\]
This proves \eqref{eq:appendixE-deltaH-bound}.

By \eqref{eq:appendixE-envelope-smallness} and \eqref{eq:appendixE-deltaH-bound},
\[
\|H^{-1}\Delta H\|_2
\le
\|H^{-1}\|_2\|\Delta H\|_2
\le
m\eta_\star\|H^{-1}\|_2
<1.
\]
Hence
\[
\widetilde H
=
H(I+H^{-1}\Delta H),
\]
and the Neumann-series criterion shows that \(I+H^{-1}\Delta H\) and therefore \(\widetilde H\) are invertible. Moreover,
\begin{equation}\label{eq:appendixE-Htildeinv-bound}
\|\widetilde H^{-1}\|_2
\le
\frac{\|H^{-1}\|_2}{1-\|H^{-1}\Delta H\|_2}
\le
\frac{\|H^{-1}\|_2}{1-m\eta_\star\|H^{-1}\|_2}.
\end{equation}

Subtracting \(Hc=-h\) from \(\widetilde H\widetilde c=-\widetilde h\) gives
\[
(H+\Delta H)(\widetilde c-c)
=
-\Delta h-\Delta H\,c.
\]
Multiplying by \(\widetilde H^{-1}\) and taking norms yields
\[
\|\widetilde c-c\|_2
\le
\|\widetilde H^{-1}\|_2
\bigl(\|\Delta h\|_2+\|\Delta H\|_2\|c\|_2\bigr).
\]
Using \eqref{eq:appendixE-deltaH-bound} and \eqref{eq:appendixE-Htildeinv-bound}, we obtain
\[
\|\widetilde c-c\|_2
\le
\frac{\|H^{-1}\|_2}{1-m\eta_\star\|H^{-1}\|_2}
\bigl(\sqrt m\,\eta + m\eta\|c\|_2\bigr),
\]
which is exactly \eqref{eq:appendixE-coefficient-bound}.
\end{proof}

\subsection{Localization of the perturbed roots}

The coefficient bound is not yet the desired frequency bound. We still have to translate polynomial-coefficient perturbations into labeled root perturbations. This is where node separation enters.

\begin{definition}[Root-separation constants]\label{def:appendixE-separation-constants}
For each \(\nu=1,\dots,m\), define
\begin{equation}\label{eq:appendixE-d-gamma}
d_\nu:=\min_{\mu\neq \nu}|z_\nu-z_\mu|,
\qquad
\Gamma_\nu:=|Q'(z_\nu)|=\prod_{\mu\neq \nu}|z_\nu-z_\mu|.
\end{equation}
Set
\begin{equation}\label{eq:appendixE-R-Lambda}
R_\nu:=|z_\nu|+\frac{d_\nu}{2},
\qquad
\Lambda_\nu:=\left(\sum_{k=0}^{m-1}R_\nu^{2k}\right)^{1/2}.
\end{equation}
For \(\eta_\star\) as in Lemma~\ref{lem:appendixE-coefficient-perturbation}, define the admissible root radius
\begin{equation}\label{eq:appendixE-root-radius}
r_\nu(\eta_\star)
:=
\frac{2^m\Lambda_\nu K_c(\eta_\star)}{\Gamma_\nu}\,\eta_\star.
\end{equation}
\end{definition}

\begin{lemma}[Rouch\'e localization of the perturbed roots]\label{lem:appendixE-root-localization}
Assume \eqref{eq:appendixE-envelope-smallness} and suppose that, for every \(\nu=1,\dots,m\),
\begin{equation}\label{eq:appendixE-root-smallness}
r_\nu(\eta_\star)<\min\Bigl\{\frac{d_\nu}{2},\frac{|z_\nu|}{2}\Bigr\}.
\end{equation}
Then, for every realized \(\eta\le \eta_\star\), the polynomial \(\widetilde Q\) has exactly one root \(\widehat z_\nu\) in the disk
\[
D\Bigl(z_\nu,\frac{2^m\Lambda_\nu K_c(\eta_\star)}{\Gamma_\nu}\eta\Bigr).
\]
In particular,
\begin{equation}\label{eq:appendixE-node-error}
|\widehat z_\nu-z_\nu|
\le
\frac{2^m\Lambda_\nu K_c(\eta_\star)}{\Gamma_\nu}\,\eta.
\end{equation}
\end{lemma}

\begin{proof}
Fix \(\nu\). Write
\[
\delta c:=\widetilde c-c,
\qquad
\delta Q(\zeta):=\widetilde Q(\zeta)-Q(\zeta)=\sum_{k=0}^{m-1}\delta c_k\zeta^k.
\]
Set
\begin{equation}\label{eq:appendixE-realized-radius}
r_\nu
:=
\frac{2^m\Lambda_\nu K_c(\eta_\star)}{\Gamma_\nu}\,\eta.
\end{equation}
Because \(\eta\le \eta_\star\), we have \(r_\nu\le r_\nu(\eta_\star)\), and therefore \(r_\nu<d_\nu/2\) and \(r_\nu<|z_\nu|/2\).

Consider the circle
\[
\Gamma_\nu^{\circ}:=\{\zeta\in\mathbb C:|\zeta-z_\nu|=r_\nu\}.
\]
If \(\zeta\in\Gamma_\nu^{\circ}\), then for every \(\mu\neq \nu\),
\[
|\zeta-z_\mu|
\ge
|z_\nu-z_\mu|-|\zeta-z_\nu|
\ge
d_\nu-r_\nu
>
\frac{d_\nu}{2}.
\]
Hence
\begin{equation}\label{eq:appendixE-Q-lower}
|Q(\zeta)|
=
|\zeta-z_\nu|\prod_{\mu\neq \nu}|\zeta-z_\mu|
\ge
r_\nu \Bigl(\frac12\Bigr)^{m-1}\Gamma_\nu.
\end{equation}
On the same circle we also have \(|\zeta|\le |z_\nu|+r_\nu\le |z_\nu|+d_\nu/2=R_\nu\), so by Cauchy--Schwarz and Lemma~\ref{lem:appendixE-coefficient-perturbation},
\[
|\delta Q(\zeta)|
\le
\|\delta c\|_2\left(\sum_{k=0}^{m-1}|\zeta|^{2k}\right)^{1/2}
\le
\|\delta c\|_2\Lambda_\nu
\le
K_c(\eta_\star)\eta\,\Lambda_\nu.
\]
Using the definition of \(r_\nu\), we may rewrite the last quantity as
\[
K_c(\eta_\star)\eta\,\Lambda_\nu
=
\frac{r_\nu\Gamma_\nu}{2^m}
<
\frac{r_\nu\Gamma_\nu}{2^{m-1}}
\le
|Q(\zeta)|
\]
by \eqref{eq:appendixE-Q-lower}. Thus \(|\delta Q(\zeta)|<|Q(\zeta)|\) on \(\Gamma_\nu^{\circ}\). By Rouch\'e's theorem, \(Q\) and \(\widetilde Q=Q+\delta Q\) have the same number of zeros in the disk bounded by \(\Gamma_\nu^{\circ}\). Since \(Q\) has exactly one zero there, namely \(z_\nu\), the polynomial \(\widetilde Q\) has exactly one zero \(\widehat z_\nu\) there as well. This proves \eqref{eq:appendixE-node-error}.
\end{proof}

\subsection{Local logarithm and frequency recovery}

The node \(z_\nu=e^{-i\omega_{j_\nu}(p_\star)\Delta}\) determines the frequency only up to a logarithm branch. For the abstract theorem this is a local issue: once the perturbed root remains in a disk that avoids the origin, there is a unique holomorphic branch through the true value.

\begin{lemma}[Local logarithm estimate]\label{lem:appendixE-logarithm}
Let \(z\in\mathbb C\setminus\{0\}\), let \(r<|z|/2\), and let \(D(z,r)\) be the open disk centered at \(z\) of radius \(r\). Then \(D(z,r)\cap\{0\}=\varnothing\). Consequently there exists a unique holomorphic branch \(\Log_z\) of the logarithm on \(D(z,r)\) satisfying \(\Log_z(z)=\log|z|+i\arg_z\), where \(\arg_z\) is any prescribed argument of \(z\). Moreover, for every \(\widehat z\in D(z,r)\),
\begin{equation}\label{eq:appendixE-log-bound}
|\Log_z(\widehat z)-\Log_z(z)|
\le
\frac{2}{|z|}\,|\widehat z-z|.
\end{equation}
\end{lemma}

\begin{proof}
The inequality \(r<|z|/2\) implies that every \(\zeta\in D(z,r)\) satisfies
\[
|\zeta|\ge |z|-|\zeta-z|>|z|-r>\frac{|z|}{2}>0,
\]
so the disk avoids the origin and therefore admits a holomorphic logarithm. The normalization at \(z\) fixes the branch uniquely.

For \(\widehat z\in D(z,r)\), consider the segment \(\gamma(s)=z+s(\widehat z-z)\), \(0\le s\le 1\). The whole segment lies in \(D(z,r)\), so
\[
\Log_z(\widehat z)-\Log_z(z)
=
\int_0^1 \frac{\widehat z-z}{\gamma(s)}\,ds.
\]
Taking absolute values and using \(|\gamma(s)|\ge |z|/2\) gives
\[
|\Log_z(\widehat z)-\Log_z(z)|
\le
\int_0^1 \frac{|\widehat z-z|}{|\gamma(s)|}\,ds
\le
\frac{2}{|z|}\,|\widehat z-z|,
\]
which is \eqref{eq:appendixE-log-bound}.
\end{proof}

We now state and prove the theorem used later.

\begin{theorem}[Abstract sampled frequency extraction theorem]\label{thm:appendixE-abstract-frequency}
Assume the sampled model \eqref{eq:appendixE-prony-model}, the nonaliasing and detectability condition \eqref{eq:appendixE-nonaliasing}, and the envelope condition \eqref{eq:appendixE-envelope-smallness}. Suppose in addition that the root-smallness condition \eqref{eq:appendixE-root-smallness} holds for every \(\nu=1,\dots,m\). Then, for every realized residual level \(\eta\le \eta_\star\), the perturbed polynomial \(\widetilde Q\) has uniquely labeled roots \(\widehat z_\nu\), one in each disk of Lemma~\ref{lem:appendixE-root-localization}, and these roots determine frequencies
\begin{equation}\label{eq:appendixE-frequency-estimator}
\widehat\omega_{j_\nu}
:=
\frac{i}{\Delta}\Log_\nu(\widehat z_\nu),
\qquad
\nu=1,\dots,m,
\end{equation}
where \(\Log_\nu\) is the unique holomorphic logarithm branch on \(D(z_\nu,r_\nu(\eta_\star))\) satisfying
\[
\Log_\nu(z_\nu)=-i\omega_{j_\nu}(p_\star)\Delta.
\]
These frequency estimates obey
\begin{equation}\label{eq:appendixE-frequency-bound}
|\widehat\omega_{j_\nu}-\omega_{j_\nu}(p_\star)|
\le
K_{\omega,\nu}(\eta_\star)\,\eta,
\end{equation}
with
\begin{equation}\label{eq:appendixE-Komega-def}
K_{\omega,\nu}(\eta_\star)
:=
\frac{2^{m+1}\Lambda_\nu K_c(\eta_\star)}{\Delta\,|z_\nu|\,\Gamma_\nu}.
\end{equation}
\end{theorem}

\begin{proof}
Lemma~\ref{lem:appendixE-root-localization} gives the unique labeled root \(\widehat z_\nu\) and the node error estimate
\[
|\widehat z_\nu-z_\nu|
\le
\frac{2^m\Lambda_\nu K_c(\eta_\star)}{\Gamma_\nu}\,\eta.
\]
Because \(r_\nu(\eta_\star)<|z_\nu|/2\), the disk \(D(z_\nu,r_\nu(\eta_\star))\) avoids the origin, so Lemma~\ref{lem:appendixE-logarithm} gives the required branch \(\Log_\nu\). Applying \eqref{eq:appendixE-log-bound} with \(z=z_\nu\) and \(\widehat z=\widehat z_\nu\), we obtain
\[
|\Log_\nu(\widehat z_\nu)-\Log_\nu(z_\nu)|
\le
\frac{2}{|z_\nu|}\,|\widehat z_\nu-z_\nu|.
\]
Since \(\Log_\nu(z_\nu)=-i\omega_{j_\nu}(p_\star)\Delta\), multiplying by \(1/\Delta\) yields
\begin{align*}
|\widehat\omega_{j_\nu}-\omega_{j_\nu}(p_\star)|
&=
\frac{1}{\Delta}
|\Log_\nu(\widehat z_\nu)-\Log_\nu(z_\nu)|\\
&\le
\frac{2}{\Delta|z_\nu|}
|\widehat z_\nu-z_\nu|\\
&\le
\frac{2}{\Delta|z_\nu|}
\cdot
\frac{2^m\Lambda_\nu K_c(\eta_\star)}{\Gamma_\nu}\,\eta\\
&=
K_{\omega,\nu}(\eta_\star)\,\eta,
\end{align*}
which is \eqref{eq:appendixE-frequency-bound}.
\end{proof}

\begin{remark}[Where the conditioning enters]\label{rem:appendixE-conditioning}
The theorem separates four distinct conditioning mechanisms. The factor \(\|H^{-1}\|_2\) controls the stability of the recurrence coefficients. The factor \(\Gamma_\nu^{-1}=|Q'(z_\nu)|^{-1}\) controls the sensitivity of the labeled root \(z_\nu\). The factor \(|z_\nu|^{-1}\) enters when the logarithm is applied, and therefore becomes worse for very rapidly damped modes or for very large shift steps. Finally, the requirement \(r_\nu(\eta_\star)<d_\nu/2\) shows explicitly that near-collisions of nodes are the first obstruction to reliable mode labeling. This is precisely the mechanism that will later be turned into a trust-region exclusion criterion.
\end{remark}

\subsection{Splitting the sample error into physical components}

To connect Theorem~\ref{thm:appendixE-abstract-frequency} with the organization of Appendix~\ref{app:assumption-ledger}, we now separate the residual samples into the pieces corresponding to detector noise, omitted linear content, model mismatch, and implementation error.

\begin{corollary}[Additive frequency-error decomposition at the sample level]\label{cor:appendixE-error-ledger}
Assume the hypotheses of Theorem~\ref{thm:appendixE-abstract-frequency}. Suppose that the sample residuals admit the decomposition
\begin{equation}\label{eq:appendixE-residual-split}
e_q
=
e_q^{\mathrm{stat}}
+
e_q^{\mathrm{tail}}
+
e_q^{\mathrm{mm}}
+
e_q^{\mathrm{alg}},
\qquad q=0,1,\dots,2m-1,
\end{equation}
and define the corresponding envelopes by
\begin{equation}\label{eq:appendixE-component-envelopes}
\eta^{\mathrm{stat}}:=\max_q |e_q^{\mathrm{stat}}|,
\quad
\eta^{\mathrm{tail}}:=\max_q |e_q^{\mathrm{tail}}|,
\quad
\eta^{\mathrm{mm}}:=\max_q |e_q^{\mathrm{mm}}|,
\quad
\eta^{\mathrm{alg}}:=\max_q |e_q^{\mathrm{alg}}|.
\end{equation}
Then
\begin{equation}\label{eq:appendixE-additive-frequency-bound}
|\widehat\omega_{j_\nu}-\omega_{j_\nu}(p_\star)|
\le
K_{\omega,\nu}(\eta_\star)
\bigl(
\eta^{\mathrm{stat}}
+
\eta^{\mathrm{tail}}
+
\eta^{\mathrm{mm}}
+
\eta^{\mathrm{alg}}
\bigr).
\end{equation}
Consequently one may set
\begin{align}
\varepsilon_{j_\nu}^{\mathrm{stat}}&:=K_{\omega,\nu}(\eta_\star)\eta^{\mathrm{stat}},\label{eq:appendixE-eps-stat}\\
\varepsilon_{j_\nu}^{\mathrm{tail}}&:=K_{\omega,\nu}(\eta_\star)\eta^{\mathrm{tail}},\label{eq:appendixE-eps-tail}\\
\varepsilon_{j_\nu}^{\mathrm{mm}}&:=K_{\omega,\nu}(\eta_\star)\eta^{\mathrm{mm}},\label{eq:appendixE-eps-mm}\\
\varepsilon_{j_\nu}^{\mathrm{alg}}&:=K_{\omega,\nu}(\eta_\star)\eta^{\mathrm{alg}},\label{eq:appendixE-eps-alg}
\end{align}
so that
\begin{equation}\label{eq:appendixE-total-ledger}
|\widehat\omega_{j_\nu}-\omega_{j_\nu}(p_\star)|
\le
\varepsilon_{j_\nu}^{\mathrm{stat}}
+
\varepsilon_{j_\nu}^{\mathrm{tail}}
+
\varepsilon_{j_\nu}^{\mathrm{mm}}
+
\varepsilon_{j_\nu}^{\mathrm{alg}}.
\end{equation}
\end{corollary}

\begin{proof}
From \eqref{eq:appendixE-residual-split} and the triangle inequality,
\[
|e_q|
\le
|e_q^{\mathrm{stat}}|
+
|e_q^{\mathrm{tail}}|
+
|e_q^{\mathrm{mm}}|
+
|e_q^{\mathrm{alg}}|
\]
for each \(q\). Taking the maximum over \(q\) gives
\[
\eta
\le
\eta^{\mathrm{stat}}
+
\eta^{\mathrm{tail}}
+
\eta^{\mathrm{mm}}
+
\eta^{\mathrm{alg}}.
\]
Substituting this into \eqref{eq:appendixE-frequency-bound} proves \eqref{eq:appendixE-additive-frequency-bound}. The definitions \eqref{eq:appendixE-eps-stat}--\eqref{eq:appendixE-eps-alg} then yield \eqref{eq:appendixE-total-ledger}.
\end{proof}

\begin{remark}[Exact arithmetic versus implementation]\label{rem:appendixE-implementation}
The theorem itself is an exact-arithmetic statement. If the linear system for \(\widetilde c\) is solved by backward-stable linear algebra and the roots of \(\widetilde Q\) are computed with an explicit enclosure, the corresponding implementation error is naturally absorbed into \(e_q^{\mathrm{alg}}\) or, equivalently, into a backward perturbation of \(\widetilde H\) and \(\widetilde h\). This appendix isolates the structural dependence on node separation and on the Prony map. Numerical choices such as overdetermined least-squares fitting, regularization, or posterior sampling enter only through the constants that replace \(K_c(\eta_\star)\) and through the definition of the algorithmic ledger term.
\end{remark}

\providecommand{\AliasLat}{\Lambda_\Delta}
\providecommand{\unwrapop}{\operatorname{unwrap}}
\providecommand{\Arg}{\operatorname{Arg}}
\providecommand{\wpr}{\omega^{\mathrm{pr}}_\Delta}

\section{Branch selection and frequency unwrapping}\label{app:branch-selection}

Appendix~\ref{app:abstract-frequency-extraction} produces labeled roots \(\widehat z_\nu\) near the true nodes
\[
z_\nu=e^{-i\omega_{j_\nu}(p_\star)\Delta},
\]
and shows that each labeled root determines a frequency once one chooses a holomorphic logarithm branch through the true node. That statement is sufficient for the abstract perturbation theorem, but it is not yet the form needed in an actual pipeline. A numerical analysis must decide, for each recovered root, which representative of the alias class of the complex frequency is being reported. It must also keep that choice consistent when the ringdown start time changes, when the projection channel changes, or when the model family is enlarged. If this step is left implicit, a harmless addition of \(2\pi/\Delta\) to the real part can masquerade as a physical drift and contaminate the later trust-region diagnostics. This appendix removes that ambiguity completely.

The problem is classical in Prony-type methods, matrix-pencil methods, and related exponential-fitting schemes \cite{HuaSarkarMPM1990,RoyKailath1989ESPRIT,PottsTasche2010APM,BatenkovYomdin2013Prony}. The contribution here is a deterministic organization rule adapted to the event-local Kerr setting of Appendices~\ref{app:assumption-ledger} and \ref{app:kerr-qnm-data}. Mode labels are fixed first by the separation arguments of Appendices~\ref{app:assumption-ledger} and \ref{app:abstract-frequency-extraction}. Once the label is fixed, the remaining ambiguity is an alias ambiguity in the real part of the frequency, and that ambiguity is resolved by a prior. At the initial window the prior is supplied by a Kerr prediction \(p^\sharp\mapsto \omega_j(p^\sharp)\). At later windows it may be supplied recursively by the recovered representative from the previous window. Both cases are treated by the same formulas.

\subsection{Alias classes of the shift map}

Fix a sampling step \(\Delta>0\). The map
\[
\omega\longmapsto e^{-i\omega\Delta}
\]
is injective in imaginary part and periodic in real part. To make this exact, define the real alias lattice
\begin{equation}\label{eq:appendixF-alias-lattice}
\AliasLat:=\frac{2\pi}{\Delta}\mathbb Z\subset\mathbb R
\end{equation}
and the principal strip
\begin{equation}\label{eq:appendixF-principal-strip}
\mathfrak S_\Delta:=
\left\{
\omega\in\mathbb C:\;
-\frac{\pi}{\Delta}\le \Re\omega<\frac{\pi}{\Delta}
\right\}.
\end{equation}

\begin{proposition}[Alias classes of the exponential shift map]\label{prop:appendixF-alias}
Let \(\Delta>0\) and let \(\omega,\omega'\in\mathbb C\). Then
\begin{equation}\label{eq:appendixF-alias-equivalence}
e^{-i\omega\Delta}=e^{-i\omega'\Delta}
\quad\Longleftrightarrow\quad
\omega'-\omega\in\AliasLat.
\end{equation}
Equivalently, if \(z\in\mathbb C\setminus\{0\}\) and
\begin{equation}\label{eq:appendixF-principal-representative}
\wpr(z):=\frac{i}{\Delta}\Log z,
\end{equation}
where \(\Log\) denotes the principal logarithm on \(\mathbb C\setminus(-\infty,0]\), then
\begin{equation}\label{eq:appendixF-all-representatives}
\{\omega\in\mathbb C:\; e^{-i\omega\Delta}=z\}
=
\wpr(z)+\AliasLat.
\end{equation}
In particular, every representative of \(z\) has the same imaginary part,
\begin{equation}\label{eq:appendixF-imaginary-part-fixed}
\Im\omega=\frac{\log|z|}{\Delta},
\end{equation}
and only the real part is aliased.
\end{proposition}

\begin{proof}
If \(\omega'-\omega=2\pi k/\Delta\) for some \(k\in\mathbb Z\), then
\[
e^{-i\omega'\Delta}
=
e^{-i\omega\Delta}e^{-i2\pi k}
=
e^{-i\omega\Delta},
\]
so the implication from right to left is immediate.

Conversely, assume \(e^{-i\omega\Delta}=e^{-i\omega'\Delta}\). Then
\[
e^{-i(\omega'-\omega)\Delta}=1.
\]
Write \(\omega'-\omega=a+ib\) with \(a,b\in\mathbb R\). Taking absolute values gives
\[
1=\left|e^{-i(a+ib)\Delta}\right|=e^{b\Delta},
\]
hence \(b=0\). Therefore \(\omega'-\omega\) is real. Since
\[
e^{-ia\Delta}=1,
\]
we must have \(a\Delta\in2\pi\mathbb Z\), that is, \(a\in\AliasLat\). This proves \eqref{eq:appendixF-alias-equivalence}.

Now fix \(z\neq0\). By definition of the principal logarithm,
\[
e^{-i\wpr(z)\Delta}
=
e^{-i(i/\Delta)\Log z\,\Delta}
=
e^{\Log z}
=
z.
\]
Hence \(\wpr(z)\) is one representative of \(z\). Formula \eqref{eq:appendixF-all-representatives} then follows from \eqref{eq:appendixF-alias-equivalence}. Finally, every representative has the form \(\wpr(z)+2\pi k/\Delta\), and the lattice term is real, so the imaginary part is independent of \(k\). Since
\[
\Im \wpr(z)
=
\Im\!\left(\frac{i}{\Delta}\bigl(\log|z|+i\Arg z\bigr)\right)
=
\frac{\log|z|}{\Delta},
\]
we obtain \eqref{eq:appendixF-imaginary-part-fixed}.
\end{proof}

The preceding proposition isolates the exact source of the ambiguity. Once a labeled root \(\widehat z_\nu\) is known, its damping rate is fixed. The only remaining freedom is the addition of a real lattice element \(2\pi k/\Delta\) to the oscillation frequency. The rest of the appendix explains how that lattice index is chosen and how the choice is propagated from one window to the next.

\subsection{Prior-centered logarithm branches}

Fix a prior frequency \(\omega^\sharp\in\mathbb C\) and let
\begin{equation}\label{eq:appendixF-prior-node}
z^\sharp:=e^{-i\omega^\sharp\Delta}.
\end{equation}
The prior may be a Kerr prediction \(\omega_j(p^\sharp)\) at a reference point \(p^\sharp\in\Kdet\), or it may be a previously recovered representative from an adjacent ringdown window. Whenever the ratio \(z/z^\sharp\) avoids the branch cut of the principal logarithm, the prior selects a unique representative of the alias class of \(z\).

\begin{definition}[Prior-selected representative]\label{def:appendixF-prior-selected}
Let \(z\in\mathbb C\setminus\{0\}\). If
\begin{equation}\label{eq:appendixF-cut-avoidance}
\frac{z}{z^\sharp}\notin(-\infty,0],
\end{equation}
we define the prior-selected representative of \(z\) relative to \(\omega^\sharp\) by
\begin{equation}\label{eq:appendixF-prior-selected-representative}
\omega_{\Delta}^{(\omega^\sharp)}(z)
:=
\omega^\sharp+\frac{i}{\Delta}\Log\!\left(\frac{z}{z^\sharp}\right).
\end{equation}
\end{definition}

\begin{proposition}[Characterization of the prior-selected representative]\label{prop:appendixF-prior-characterization}
Assume \eqref{eq:appendixF-cut-avoidance}. Then \(\omega_{\Delta}^{(\omega^\sharp)}(z)\) is well defined, satisfies
\begin{equation}\label{eq:appendixF-prior-characterization-1}
e^{-i\omega_{\Delta}^{(\omega^\sharp)}(z)\Delta}=z,
\end{equation}
and is the unique representative of \(z\) lying in the shifted principal strip
\begin{equation}\label{eq:appendixF-shifted-strip}
\omega^\sharp+\mathfrak S_\Delta
=
\left\{
\omega\in\mathbb C:\;
-\frac{\pi}{\Delta}\le \Re(\omega-\omega^\sharp)<\frac{\pi}{\Delta}
\right\}.
\end{equation}
Consequently, if \(\omega\) is any representative of \(z\) such that
\begin{equation}\label{eq:appendixF-strip-compatibility}
-\frac{\pi}{\Delta}\le \Re(\omega-\omega^\sharp)<\frac{\pi}{\Delta},
\end{equation}
then
\begin{equation}\label{eq:appendixF-prior-equals-physical}
\omega_{\Delta}^{(\omega^\sharp)}(z)=\omega.
\end{equation}
\end{proposition}

\begin{proof}
Because \(z/z^\sharp\) lies in the domain of the principal logarithm, \(\omega_{\Delta}^{(\omega^\sharp)}(z)\) is well defined. Using \eqref{eq:appendixF-prior-node},
\[
e^{-i\omega_{\Delta}^{(\omega^\sharp)}(z)\Delta}
=
e^{-i\omega^\sharp\Delta}
e^{-i(i/\Delta)\Log(z/z^\sharp)\Delta}
=
z^\sharp e^{\Log(z/z^\sharp)}
=
z,
\]
which proves \eqref{eq:appendixF-prior-characterization-1}.

Next,
\[
\omega_{\Delta}^{(\omega^\sharp)}(z)-\omega^\sharp
=
\frac{i}{\Delta}\Log\!\left(\frac{z}{z^\sharp}\right).
\]
By definition of the principal logarithm, the real part of the right-hand side belongs to \([-\pi/\Delta,\pi/\Delta)\). Hence \(\omega_{\Delta}^{(\omega^\sharp)}(z)\in\omega^\sharp+\mathfrak S_\Delta\).

To prove uniqueness, let \(\omega'\) be another representative of \(z\) in the same shifted strip. By Proposition~\ref{prop:appendixF-alias},
\[
\omega'-\omega_{\Delta}^{(\omega^\sharp)}(z)\in\AliasLat.
\]
On the other hand, both \(\omega'-\omega^\sharp\) and \(\omega_{\Delta}^{(\omega^\sharp)}(z)-\omega^\sharp\) belong to \(\mathfrak S_\Delta\), so their real parts differ by less than \(2\pi/\Delta\). The only lattice element with this property is \(0\). Therefore \(\omega'=\omega_{\Delta}^{(\omega^\sharp)}(z)\), proving uniqueness.

If \(\omega\) satisfies \eqref{eq:appendixF-strip-compatibility}, then \(\omega\) is a representative of \(z\) in the shifted principal strip, so uniqueness gives \eqref{eq:appendixF-prior-equals-physical}.
\end{proof}

Proposition~\ref{prop:appendixF-prior-characterization} is the exact bridge between the abstract logarithm of Appendix~\ref{app:abstract-frequency-extraction} and the representative actually reported by the pipeline. One never needs to know the true branch in advance. The prior determines a strip, and the reported representative is simply the unique representative in that strip.

The next result gives a fully explicit sufficient condition under which the prior-selected representative varies Lipschitz-continuously with the node. This is the form used later with the root disks coming from Appendix~\ref{app:abstract-frequency-extraction}.

\begin{proposition}[Local Lipschitz control in a prior disk]\label{prop:appendixF-prior-Lipschitz}
Let \(z,\widehat z\in\mathbb C\setminus\{0\}\). Assume that
\begin{equation}\label{eq:appendixF-prior-disk-condition}
|z-z^\sharp|+|\widehat z-z|<|z^\sharp|.
\end{equation}
Then both prior-selected representatives \(\omega_{\Delta}^{(\omega^\sharp)}(z)\) and \(\omega_{\Delta}^{(\omega^\sharp)}(\widehat z)\) are well defined, and
\begin{equation}\label{eq:appendixF-prior-Lipschitz-bound}
\left|
\omega_{\Delta}^{(\omega^\sharp)}(\widehat z)
-
\omega_{\Delta}^{(\omega^\sharp)}(z)
\right|
\le
\frac{|\widehat z-z|}
{\Delta\bigl(|z^\sharp|-|z-z^\sharp|-|\widehat z-z|\bigr)}.
\end{equation}
\end{proposition}

\begin{proof}
Set
\[
\rho:=\frac{|z-z^\sharp|+|\widehat z-z|}{|z^\sharp|}.
\]
By \eqref{eq:appendixF-prior-disk-condition}, \(0\le \rho<1\). The triangle inequality gives
\[
|z-z^\sharp|\le \rho|z^\sharp|,
\qquad
|\widehat z-z^\sharp|\le |z-z^\sharp|+|\widehat z-z|\le \rho|z^\sharp|.
\]
Hence both \(z\) and \(\widehat z\) lie in the closed disk
\[
\overline{D}\bigl(z^\sharp,\rho|z^\sharp|\bigr).
\]
For any point \(\zeta\) on the straight-line segment between \(z\) and \(\widehat z\),
\[
|\zeta-z^\sharp|
\le
|z-z^\sharp|+|\zeta-z|
\le
|z-z^\sharp|+|\widehat z-z|
\le
\rho|z^\sharp|.
\]
Therefore
\[
|\zeta|
\ge
|z^\sharp|-|\zeta-z^\sharp|
\ge
(1-\rho)|z^\sharp|
=
|z^\sharp|-|z-z^\sharp|-|\widehat z-z|
>0.
\]
In particular, the entire segment lies in \(\mathbb C\setminus\{0\}\) and the ratios \(\zeta/z^\sharp\) lie in the disk \(D(1,\rho)\), which is contained in the right half-plane because \(\rho<1\). Hence the principal logarithm is holomorphic on a neighborhood of the segment connecting \(z/z^\sharp\) and \(\widehat z/z^\sharp\).

Now apply the fundamental theorem of calculus along the straight-line segment:
\[
\Log\!\left(\frac{\widehat z}{z^\sharp}\right)
-
\Log\!\left(\frac{z}{z^\sharp}\right)
=
\int_0^1
\frac{d}{ds}
\Log\!\left(\frac{z+s(\widehat z-z)}{z^\sharp}\right)\,ds
=
\int_0^1
\frac{\widehat z-z}{z+s(\widehat z-z)}\,ds.
\]
Taking absolute values and using the lower bound above yields
\[
\left|
\Log\!\left(\frac{\widehat z}{z^\sharp}\right)
-
\Log\!\left(\frac{z}{z^\sharp}\right)
\right|
\le
\frac{|\widehat z-z|}
{|z^\sharp|-|z-z^\sharp|-|\widehat z-z|}.
\]
Multiplying by \(1/\Delta\) proves \eqref{eq:appendixF-prior-Lipschitz-bound}.
\end{proof}

The preceding proposition is the practical branch-selection estimate. It does not require access to the true frequency. It requires only a prior node \(z^\sharp\), a explicit node-error radius, and the elementary geometric condition \eqref{eq:appendixF-prior-disk-condition}. In applications \(z^\sharp\) is generated from the Kerr surrogate of Appendix~\ref{app:kerr-qnm-data}.

\begin{corollary}[Direct use with the Appendix E root disks]\label{cor:appendixF-from-appendixE}
Fix a labeled mode \(\nu\) from Theorem~\ref{thm:appendixE-abstract-frequency}. Let
\[
z_\nu=e^{-i\omega_{j_\nu}(p_\star)\Delta},
\qquad
\widehat z_\nu
\]
be the true and recovered nodes, and let
\[
z_\nu^\sharp:=e^{-i\omega_{j_\nu}(p^\sharp)\Delta}
\]
be a prior node generated from some guide parameter \(p^\sharp\in\Kdet\). If
\begin{equation}\label{eq:appendixF-E-compatibility}
|z_\nu-z_\nu^\sharp|+r_\nu(\eta_\star)<|z_\nu^\sharp|,
\end{equation}
then the prior-selected representatives of \(z_\nu\) and \(\widehat z_\nu\) are well defined, and
\begin{equation}\label{eq:appendixF-E-compatibility-bound}
\left|
\omega_{\Delta}^{(\omega_{j_\nu}(p^\sharp))}(\widehat z_\nu)
-
\omega_{\Delta}^{(\omega_{j_\nu}(p^\sharp))}(z_\nu)
\right|
\le
\frac{r_\nu(\eta_\star)}
{\Delta\bigl(|z_\nu^\sharp|-|z_\nu-z_\nu^\sharp|-r_\nu(\eta_\star)\bigr)}.
\end{equation}
\end{corollary}

\begin{proof}
Theorem~\ref{thm:appendixE-abstract-frequency} gives
\[
|\widehat z_\nu-z_\nu|\le r_\nu(\eta_\star).
\]
Substituting \(\widehat z=\widehat z_\nu\), \(z=z_\nu\), and \(z^\sharp=z_\nu^\sharp\) into Proposition~\ref{prop:appendixF-prior-Lipschitz}, and using \eqref{eq:appendixF-E-compatibility}, yields \eqref{eq:appendixF-E-compatibility-bound}.
\end{proof}

Corollary~\ref{cor:appendixF-from-appendixE} is the point at which the abstract node localization of Appendix~\ref{app:abstract-frequency-extraction} becomes an implementable frequency estimate. Once a guide parameter \(p^\sharp\) is chosen, the prior-selected representative is a concrete formula and the branch is no longer existential.

\subsection{Explicit unwrapping of the principal representative}

One may implement the same choice rule without forming the ratio \(z/z^\sharp\) explicitly. Starting from any representative of the recovered node, one simply adds the unique alias correction that places the result inside the strip centered at the prior.

\begin{definition}[Unwrapping around a prior]\label{def:appendixF-unwrap}
Let \(\widetilde\omega\in\mathbb C\). We define
\begin{equation}\label{eq:appendixF-unwrap-definition}
\unwrapop_{\Delta,\omega^\sharp}(\widetilde\omega)
:=
\widetilde\omega+\frac{2\pi}{\Delta}k_\star,
\end{equation}
where \(k_\star\in\mathbb Z\) is the unique integer such that
\begin{equation}\label{eq:appendixF-unwrap-strip}
-\frac{\pi}{\Delta}
\le
\Re\!\left(
\widetilde\omega+\frac{2\pi}{\Delta}k_\star-\omega^\sharp
\right)
<
\frac{\pi}{\Delta}.
\end{equation}
\end{definition}

The existence and uniqueness of \(k_\star\) are immediate because the half-open intervals
\[
\left[-\frac{\pi}{\Delta}+\frac{2\pi k}{\Delta},\frac{\pi}{\Delta}+\frac{2\pi k}{\Delta}\right),
\qquad k\in\mathbb Z,
\]
form a partition of \(\mathbb R\).

\begin{proposition}[Equivalence of ratio branching and strip unwrapping]\label{prop:appendixF-unwrap-equivalence}
Let \(\widehat z\in\mathbb C\setminus\{0\}\) and assume that \(\widehat z/z^\sharp\notin(-\infty,0]\). Then
\begin{equation}\label{eq:appendixF-unwrap-equivalence}
\unwrapop_{\Delta,\omega^\sharp}\bigl(\wpr(\widehat z)\bigr)
=
\omega_{\Delta}^{(\omega^\sharp)}(\widehat z).
\end{equation}
\end{proposition}

\begin{proof}
Set
\[
\widetilde\omega:=\wpr(\widehat z)=\frac{i}{\Delta}\Log \widehat z.
\]
By Proposition~\ref{prop:appendixF-alias},
\[
e^{-i(\widetilde\omega-\omega^\sharp)\Delta}
=
\frac{e^{-i\widetilde\omega\Delta}}{e^{-i\omega^\sharp\Delta}}
=
\frac{\widehat z}{z^\sharp}.
\]
Thus \(\widetilde\omega-\omega^\sharp\) is one representative of the ratio \(\widehat z/z^\sharp\). Again by Proposition~\ref{prop:appendixF-alias}, all representatives of that ratio are obtained by adding elements of \(\AliasLat\). The principal representative of the ratio is
\[
\frac{i}{\Delta}\Log\!\left(\frac{\widehat z}{z^\sharp}\right),
\]
and by definition it is the unique representative lying in \(\mathfrak S_\Delta\). Therefore there exists a unique integer \(k_\star\) such that
\[
\frac{i}{\Delta}\Log\!\left(\frac{\widehat z}{z^\sharp}\right)
=
\widetilde\omega-\omega^\sharp+\frac{2\pi}{\Delta}k_\star,
\]
where \(k_\star\) is precisely the integer selected by \eqref{eq:appendixF-unwrap-strip}. Rearranging gives
\[
\omega^\sharp+\frac{i}{\Delta}\Log\!\left(\frac{\widehat z}{z^\sharp}\right)
=
\widetilde\omega+\frac{2\pi}{\Delta}k_\star
=
\unwrapop_{\Delta,\omega^\sharp}(\widetilde\omega),
\]
which is \eqref{eq:appendixF-unwrap-equivalence}.
\end{proof}

Proposition~\ref{prop:appendixF-unwrap-equivalence} shows that there are two exactly equivalent ways to implement branch selection. One may either form the prior-centered ratio and apply the principal logarithm once, or one may compute the principal representative \(\wpr(\widehat z)\) and then add the unique lattice correction that places it in the strip centered at the prior. The second form is often more convenient in code. The first form is often more convenient in proofs.

\subsection{Recursive continuation across ringdown windows}

The start-time scans used later require continuation from one window to the next. Once the representative at one window has been fixed, the representative at the next window can be selected by using the previous estimate as the new prior. The following theorem shows that, under a sharp and easily checked condition, this continuation cannot suffer a branch slip.

\begin{theorem}[Recursive continuation without branch slips]\label{thm:appendixF-recursive}
Fix \(\Delta>0\) and let \(n=0,1,\dots,N\) index a sequence of ringdown windows. For each \(n\), let \(z_n\in\mathbb C\setminus\{0\}\) be the true node of a fixed labeled mode and let \(\widehat z_n\in\mathbb C\setminus\{0\}\) be its recovered node. Let \(\omega_n\) be a chosen physical representative of \(z_n\), and let \(\widetilde\omega_n\) be any representative of \(\widehat z_n\); in practice one takes \(\widetilde\omega_n=\wpr(\widehat z_n)\).

Assume that for each \(n\) there exists a representative \(\omega_n^{\mathrm{loc}}\) of \(\widehat z_n\) such that
\begin{equation}\label{eq:appendixF-recursive-local-error}
|\omega_n^{\mathrm{loc}}-\omega_n|\le \epsilon_n
\end{equation}
for some \(\epsilon_n\ge0\). Suppose in addition that
\begin{equation}\label{eq:appendixF-recursive-safe-step}
|\Re(\omega_n-\omega_{n-1})|+\epsilon_{n-1}+\epsilon_n<\frac{\pi}{\Delta}
\qquad \text{for } n=1,\dots,N.
\end{equation}
Define recursively
\begin{equation}\label{eq:appendixF-recursive-estimator}
\widehat\omega_0:=\omega_0^{\mathrm{loc}},
\qquad
\widehat\omega_n:=\unwrapop_{\Delta,\widehat\omega_{n-1}}(\widetilde\omega_n),
\quad n\ge1.
\end{equation}
Then
\begin{equation}\label{eq:appendixF-recursive-conclusion}
\widehat\omega_n=\omega_n^{\mathrm{loc}}
\qquad \text{for all } n=0,1,\dots,N.
\end{equation}
In particular,
\begin{equation}\label{eq:appendixF-recursive-error}
|\widehat\omega_n-\omega_n|\le \epsilon_n
\qquad \text{for all } n=0,1,\dots,N.
\end{equation}
\end{theorem}

\begin{proof}
We argue by induction on \(n\). The claim is true for \(n=0\) by definition.

Assume now that \(\widehat\omega_{n-1}=\omega_{n-1}^{\mathrm{loc}}\) has been proved for some \(n\ge1\). Then
\[
\bigl|\Re(\omega_n^{\mathrm{loc}}-\widehat\omega_{n-1})\bigr|
\le
|\Re(\omega_n-\omega_{n-1})|
+
|\Re(\omega_n^{\mathrm{loc}}-\omega_n)|
+
|\Re(\widehat\omega_{n-1}-\omega_{n-1})|.
\]
Using the induction hypothesis and \eqref{eq:appendixF-recursive-local-error},
\[
\bigl|\Re(\omega_n^{\mathrm{loc}}-\widehat\omega_{n-1})\bigr|
\le
|\Re(\omega_n-\omega_{n-1})|+\epsilon_n+\epsilon_{n-1}
<
\frac{\pi}{\Delta}
\]
by \eqref{eq:appendixF-recursive-safe-step}. Hence \(\omega_n^{\mathrm{loc}}\) lies in the shifted principal strip
\[
\widehat\omega_{n-1}+\mathfrak S_\Delta.
\]

Now \(\omega_n^{\mathrm{loc}}\) and \(\widetilde\omega_n\) are both representatives of \(\widehat z_n\). By Proposition~\ref{prop:appendixF-alias},
\[
\omega_n^{\mathrm{loc}}-\widetilde\omega_n\in\AliasLat.
\]
Therefore \(\omega_n^{\mathrm{loc}}\) is one of the candidates produced by adding lattice corrections to \(\widetilde\omega_n\). Because it lies in \(\widehat\omega_{n-1}+\mathfrak S_\Delta\), Definition~\ref{def:appendixF-unwrap} shows that
\[
\unwrapop_{\Delta,\widehat\omega_{n-1}}(\widetilde\omega_n)=\omega_n^{\mathrm{loc}}.
\]
Thus \(\widehat\omega_n=\omega_n^{\mathrm{loc}}\), completing the induction and proving \eqref{eq:appendixF-recursive-conclusion}. Estimate \eqref{eq:appendixF-recursive-error} is then immediate from \eqref{eq:appendixF-recursive-local-error}.
\end{proof}

Theorem~\ref{thm:appendixF-recursive} is the statement that makes the later start-time drift plots meaningful. If the safe-step condition \eqref{eq:appendixF-recursive-safe-step} holds, then a window-to-window scan cannot jump by a lattice element without violating the theorem's hypotheses. Any observed mid-window drift is therefore attributable to actual estimator behavior or model mismatch, not to a silent branch slip.

\begin{corollary}[Event-local Kerr prior criterion]\label{cor:appendixF-kerr-prior}
Fix a mode index \(j\in\modeJ\), a compact guide set \(\mathcal U\subset\Kdet\), and a reference point \(p^\sharp\in\mathcal U\). Define
\begin{equation}\label{eq:appendixF-kerr-range}
R_j(p^\sharp;\mathcal U)
:=
\sup_{p\in\mathcal U}
\left|
\Re\omega_j(p)-\Re\omega_j(p^\sharp)
\right|.
\end{equation}
Assume that every recovered representative of mode \(j\) on \(\mathcal U\) obeys the deterministic or empirical error bound
\begin{equation}\label{eq:appendixF-kerr-local-error}
|\omega_j^{\mathrm{loc}}(p)-\omega_j(p)|\le \epsilon_j
\qquad \text{for all } p\in\mathcal U.
\end{equation}
If
\begin{equation}\label{eq:appendixF-kerr-safe}
R_j(p^\sharp;\mathcal U)+\epsilon_j<\frac{\pi}{\Delta},
\end{equation}
then, for every \(p\in\mathcal U\), unwrapping the recovered representative around the prior \(\omega_j(p^\sharp)\) returns the physical representative \(\omega_j^{\mathrm{loc}}(p)\) and no alias ambiguity remains.
\end{corollary}

\begin{proof}
For \(p\in\mathcal U\),
\[
\bigl|\Re(\omega_j^{\mathrm{loc}}(p)-\omega_j(p^\sharp))\bigr|
\le
\bigl|\Re(\omega_j(p)-\omega_j(p^\sharp))\bigr|
+
\bigl|\Re(\omega_j^{\mathrm{loc}}(p)-\omega_j(p))\bigr|
\le
R_j(p^\sharp;\mathcal U)+\epsilon_j
<
\frac{\pi}{\Delta}.
\]
Thus \(\omega_j^{\mathrm{loc}}(p)\) lies in the shifted principal strip centered at \(\omega_j(p^\sharp)\). Since it is a representative of the recovered node, Proposition~\ref{prop:appendixF-prior-characterization} gives the claim.
\end{proof}

Corollary~\ref{cor:appendixF-kerr-prior} is the initialization rule used later. Appendix~\ref{app:kerr-qnm-data} provides numerically interpolants with error control for the maps \(p\mapsto \omega_j(p)\). Hence the range \(R_j(p^\sharp;\mathcal U)\) can itself be on the chosen event-local box. Once the first window is initialized in this way, Theorem~\ref{thm:appendixF-recursive} takes over and propagates the same branch along the window sequence.

\begin{remark}[When a window is branch-ambiguous]\label{rem:appendixF-ambiguous}
The hypotheses of Proposition~\ref{prop:appendixF-prior-Lipschitz} and Theorem~\ref{thm:appendixF-recursive} are meant to be used as acceptance tests, not as hidden assumptions. If a recovered node or representative fails the relevant prior-disk condition or the safe-step inequality, the window is marked as branch ambiguous and excluded from the trust region. This is the failure mode that later enters the trust-region definition.
\end{remark}

Once a mode has been labeled, branch selection reduces to the unique alias representative inside a prior-centered strip, provided the explicit sufficient conditions derived above hold. Under those conditions the chosen representative is stable and propagates consistently across the ringdown-window scan.

\providecommand{\sepz}{\delta_z}
\providecommand{\amin}{a_{\min}}
\providecommand{\amax}{a_{\max}}
\providecommand{\Rnode}{R_z}
\providecommand{\Yzero}{Y_0}
\providecommand{\Yone}{Y_1}
\providecommand{\Zmat}{Z}

\section{Prony and matrix-pencil conditioning near coalescence}\label{app:prony-conditioning}

Appendix~\ref{app:abstract-frequency-extraction} already identifies the two quantities that control deterministic extraction: the inverse Hankel norm and the root factor \(\Gamma_\nu^{-1}=|Q'(z_\nu)|^{-1}\). What that appendix does not yet say explicitly is how fast those quantities deteriorate when two nodes approach each other. This is the instability behind high-overtone correlations, near-degenerate fits, and start-time windows in which a nominally two-mode model becomes numerically fragile. This appendix makes that mechanism completely explicit in the simplest nontrivial case.

There are two distinct messages. The first is nonasymptotic and rigorous. If one reconstructs the nodes through the monic annihilating polynomial, then the deterministic radius obtained from the coefficient route grows like \(\sepz^{-3}\), where
\[
\sepz:=|z_1-z_2|
\]
is the node separation. The second is local and differential. If one studies the same two-node problem through the structured matrix pencil itself, then the simple-eigenvalue condition number grows like \(\sepz^{-2}\). The difference is real rather than cosmetic. The coefficient route is the robust bound needed for trust-region exclusion rules. The matrix-pencil derivative describes the leading-order behavior of a structured implementation once a branch has already been fixed. Both viewpoints are classical in Prony-type analysis, Vandermonde inversion, and matrix-pencil perturbation theory \cite{Gautschi1962Vandermonde,HuaSarkarMPM1990,RoyKailath1989ESPRIT,PottsTasche2010APM,BatenkovYomdin2013Prony}. This appendix does not introduce a new algorithm. It records the exact separation exponents that later guide the interpretation of spectroscopy windows.

\subsection{Two-node Hankel and pencil factorizations}

Consider the noiseless two-node model
\begin{equation}\label{eq:appendixG-two-node-model}
 y_q=a_1 z_1^q+a_2 z_2^q,
 \qquad q=0,1,2,3,
\end{equation}
with
\begin{equation}\label{eq:appendixG-nondegenerate-data}
 a_1a_2\neq 0,
 \qquad z_1\neq z_2.
\end{equation}
In the ringdown application one has \(z_j=e^{-i\omega_j\Delta}\) with \(\Im\omega_j<0\), hence \(|z_j|<1\). For the algebra below we keep a general radius parameter
\begin{equation}\label{eq:appendixG-radius-def}
 \Rnode:=\max\{|z_1|,|z_2|\}
\end{equation}
and the amplitude scale
\begin{equation}\label{eq:appendixG-amin-def}
 \amin:=\min\{|a_1|,|a_2|\},
 \qquad
 \amax:=\max\{|a_1|,|a_2|\}.
\end{equation}
The two basic matrices are the order-two Hankel block and its one-step shift,
\begin{equation}\label{eq:appendixG-pencils}
 \Yzero:=\begin{pmatrix} y_0 & y_1\\ y_1 & y_2\end{pmatrix},
 \qquad
 \Yone:=\begin{pmatrix} y_1 & y_2\\ y_2 & y_3\end{pmatrix}.
\end{equation}
We also write
\begin{equation}\label{eq:appendixG-VAZ}
 V:=\begin{pmatrix} 1 & 1\\ z_1 & z_2\end{pmatrix},
 \qquad
 A:=\operatorname{diag}(a_1,a_2),
 \qquad
 \Zmat:=\operatorname{diag}(z_1,z_2).
\end{equation}

\begin{proposition}[Exact factorizations for the two-node model]\label{prop:appendixG-exact-factorizations}
Under \eqref{eq:appendixG-two-node-model} and \eqref{eq:appendixG-nondegenerate-data},
\begin{equation}\label{eq:appendixG-Y0-factorization}
 \Yzero=VAV^\top,
 \qquad
 \Yone=VA\Zmat V^\top.
\end{equation}
Consequently,
\begin{equation}\label{eq:appendixG-detY0}
 \det \Yzero=a_1a_2(z_1-z_2)^2,
\end{equation}
and the matrix pencil is diagonalizable in the exact form
\begin{equation}\label{eq:appendixG-pencil-diagonalization}
 M:=\Yzero^{-1}\Yone=V^{-\top}\Zmat V^\top.
\end{equation}
If
\begin{equation}\label{eq:appendixG-annihilating-poly}
 Q(\zeta)=\zeta^2+c_1\zeta+c_0=(\zeta-z_1)(\zeta-z_2),
\end{equation}
then
\begin{equation}\label{eq:appendixG-c0c1}
 c_0=z_1z_2,
 \qquad
 c_1=-(z_1+z_2),
\end{equation}
and the four samples satisfy the recurrence
\begin{equation}\label{eq:appendixG-recurrence}
 y_{q+2}+c_1 y_{q+1}+c_0 y_q=0,
 \qquad q=0,1.
\end{equation}
Equivalently, with
\begin{equation}\label{eq:appendixG-h-def}
 h:=\begin{pmatrix} y_2\\ y_3\end{pmatrix},
 \qquad
 c:=\begin{pmatrix} c_0\\ c_1\end{pmatrix},
\end{equation}
we have
\begin{equation}\label{eq:appendixG-hankel-system}
 \Yzero c=-h.
\end{equation}
\end{proposition}

\begin{proof}
The \((r,s)\) entry of \(VAV^\top\) equals
\[
\sum_{\nu=1}^2 z_\nu^r a_\nu z_\nu^s
=
\sum_{\nu=1}^2 a_\nu z_\nu^{r+s}
=
 y_{r+s},
\]
for \(r,s\in\{0,1\}\), which gives \eqref{eq:appendixG-Y0-factorization} for \(\Yzero\). The same computation with one extra factor of \(\Zmat\) gives \(\Yone=VA\Zmat V^\top\).

Since \(\det V=z_2-z_1\), formula \eqref{eq:appendixG-detY0} follows at once:
\[
\det \Yzero=(\det V)^2\det A=a_1a_2(z_2-z_1)^2.
\]
Because \eqref{eq:appendixG-nondegenerate-data} implies \(\det \Yzero\neq 0\), the inverse exists, and
\[
M=\Yzero^{-1}\Yone=(VAV^\top)^{-1}(VA\Zmat V^\top)=V^{-\top}\Zmat V^\top,
\]
which is \eqref{eq:appendixG-pencil-diagonalization}.

The factorization
\[
Q(\zeta)=(\zeta-z_1)(\zeta-z_2)=\zeta^2-(z_1+z_2)\zeta+z_1z_2
\]
gives \eqref{eq:appendixG-c0c1}. Multiplying \(Q(z_\nu)=0\) by \(a_\nu z_\nu^q\) and summing over \(\nu=1,2\) yields
\[
\sum_{\nu=1}^2 a_\nu z_\nu^{q+2}
+c_1\sum_{\nu=1}^2 a_\nu z_\nu^{q+1}
+c_0\sum_{\nu=1}^2 a_\nu z_\nu^q=0,
\]
which is exactly \eqref{eq:appendixG-recurrence}. Writing the two relations for \(q=0\) and \(q=1\) in vector form gives \eqref{eq:appendixG-hankel-system}.
\end{proof}

The exact factorization already exhibits the two dangerous denominators. Solving the coefficient system involves \(\Yzero^{-1}\), and mode labeling through the polynomial involves
\[
Q'(z_1)=z_1-z_2,
\qquad
Q'(z_2)=z_2-z_1.
\]
Thus coalescence appears once at the Hankel level and once again at the root-selection level.

\subsection{Explicit separation bounds}

The next lemma makes the Hankel loss completely explicit. We keep the notation
\begin{equation}\label{eq:appendixG-sepz-def}
\sepz:=|z_1-z_2|.
\end{equation}

\begin{lemma}[Vandermonde and Hankel bounds]\label{lem:appendixG-vandermonde-bounds}
Under \eqref{eq:appendixG-two-node-model} and \eqref{eq:appendixG-nondegenerate-data},
\begin{equation}\label{eq:appendixG-V-bound}
 \|V\|_2\le \sqrt{2(1+\Rnode^2)},
\end{equation}

a closed formula for the inverse is
\begin{equation}\label{eq:appendixG-Vinv-formula}
 V^{-1}=\frac{1}{z_2-z_1}
 \begin{pmatrix}
 z_2 & -1\\
 -z_1 & 1
 \end{pmatrix},
\end{equation}
and therefore
\begin{equation}\label{eq:appendixG-Vinv-bound}
 \|V^{-1}\|_2\le \frac{\sqrt{2(1+\Rnode^2)}}{\sepz},
 \qquad
 \kappa_2(V)\le \frac{2(1+\Rnode^2)}{\sepz}.
\end{equation}
Finally,
\begin{equation}\label{eq:appendixG-Y0inv-bound}
 \|\Yzero^{-1}\|_2
 \le
 \frac{2(1+\Rnode^2)}{\amin\sepz^2}.
\end{equation}
\end{lemma}

\begin{proof}
The Frobenius norm bound gives
\[
\|V\|_2\le \|V\|_{\mathrm F}
=\bigl(2+|z_1|^2+|z_2|^2\bigr)^{1/2}
\le \sqrt{2(1+\Rnode^2)},
\]
which is \eqref{eq:appendixG-V-bound}.

The explicit inverse \eqref{eq:appendixG-Vinv-formula} is the usual \(2\times2\) formula. Its Frobenius norm satisfies
\[
\left\|\begin{pmatrix}
 z_2 & -1\\
 -z_1 & 1
 \end{pmatrix}\right\|_{\mathrm F}^2
=|z_2|^2+1+|z_1|^2+1
\le 2(1+\Rnode^2),
\]
so
\[
\|V^{-1}\|_2\le \|V^{-1}\|_{\mathrm F}
\le \frac{\sqrt{2(1+\Rnode^2)}}{\sepz}.
\]
Combining this with \eqref{eq:appendixG-V-bound} yields the condition-number estimate in \eqref{eq:appendixG-Vinv-bound}.

Now use Proposition~\ref{prop:appendixG-exact-factorizations}. Since
\[
\Yzero^{-1}=V^{-\top}A^{-1}V^{-1},
\]
we obtain
\[
\|\Yzero^{-1}\|_2
\le \|V^{-1}\|_2^2\|A^{-1}\|_2
\le \frac{2(1+\Rnode^2)}{\sepz^2}\cdot \frac{1}{\amin},
\]
which is \eqref{eq:appendixG-Y0inv-bound}.
\end{proof}

Lemma~\ref{lem:appendixG-vandermonde-bounds} is the first half of the near-coalescence story. It says that the coefficient system becomes quadratically ill-conditioned as the two nodes collide. The second half is the loss incurred when one turns coefficients into labeled roots.

\subsection{A quantitative Prony bound}

We now perturb the four samples. Let
\begin{equation}\label{eq:appendixG-perturbed-samples}
 \widetilde y_q=y_q+e_q,
 \qquad q=0,1,2,3,
\end{equation}
with
\begin{equation}\label{eq:appendixG-eta-def}
 |e_q|\le \eta.
\end{equation}
Define
\begin{equation}\label{eq:appendixG-tilde-objects}
 \widetilde \Yzero:=\begin{pmatrix} \widetilde y_0 & \widetilde y_1\\ \widetilde y_1 & \widetilde y_2\end{pmatrix},
 \qquad
 \widetilde h:=\begin{pmatrix} \widetilde y_2\\ \widetilde y_3\end{pmatrix},
\end{equation}
and let
\begin{equation}\label{eq:appendixG-tilde-c}
 \widetilde c:=-\widetilde \Yzero^{-1}\widetilde h
\end{equation}
whenever \(\widetilde\Yzero\) is invertible. The observed Prony polynomial is then
\begin{equation}\label{eq:appendixG-tildeQ}
 \widetilde Q(\zeta)=\zeta^2+\widetilde c_1\zeta+\widetilde c_0.
\end{equation}

\begin{lemma}[Coefficient perturbation in the two-node case]\label{lem:appendixG-coefficient-perturbation}
Assume
\begin{equation}\label{eq:appendixG-smallness1}
 \eta\le \frac{\amin\sepz^2}{8(1+\Rnode^2)}.
\end{equation}
Then \(\widetilde\Yzero\) is invertible and
\begin{equation}\label{eq:appendixG-deltac-bound}
 \|\widetilde c-c\|_2
 \le
 \frac{C_{\mathrm{coef}}(\Rnode)}{\amin\sepz^2}\,\eta,
\end{equation}
where
\begin{equation}\label{eq:appendixG-Ccoef}
 C_{\mathrm{coef}}(\Rnode):=4(1+\Rnode^2)\bigl(\sqrt2+2\Rnode(\Rnode+2)\bigr).
\end{equation}
\end{lemma}

\begin{proof}
Set
\[
\Delta\Yzero:=\widetilde\Yzero-\Yzero,
\qquad
\Delta h:=\widetilde h-h.
\]
Because each entry of \(\Delta\Yzero\) has modulus at most \(\eta\),
\[
\|\Delta\Yzero\|_2\le \|\Delta\Yzero\|_{\mathrm F}\le 2\eta.
\]
Likewise,
\[
\|\Delta h\|_2\le \sqrt2\,\eta.
\]
By Lemma~\ref{lem:appendixG-vandermonde-bounds} and \eqref{eq:appendixG-smallness1},
\[
\|\Yzero^{-1}\Delta\Yzero\|_2
\le \|\Yzero^{-1}\|_2\|\Delta\Yzero\|_2
\le \frac{2(1+\Rnode^2)}{\amin\sepz^2}\cdot 2\eta
\le \frac12.
\]
Hence \(I+\Yzero^{-1}\Delta\Yzero\) is invertible, so \(\widetilde\Yzero=\Yzero(I+\Yzero^{-1}\Delta\Yzero)\) is invertible and
\begin{equation}\label{eq:appendixG-Y0tildeinv-bound}
 \|\widetilde\Yzero^{-1}\|_2
 \le \frac{\|\Yzero^{-1}\|_2}{1-\|\Yzero^{-1}\Delta\Yzero\|_2}
 \le 2\|\Yzero^{-1}\|_2.
\end{equation}

Since \(\Yzero c=-h\),
\[
\widetilde c-c
=-\widetilde\Yzero^{-1}\widetilde h+\Yzero^{-1}h
=-\widetilde\Yzero^{-1}(\Delta h+\Delta\Yzero\,c).
\]
Taking norms and using \eqref{eq:appendixG-Y0tildeinv-bound},
\begin{equation}\label{eq:appendixG-deltac-prebound}
 \|\widetilde c-c\|_2
 \le 2\|\Yzero^{-1}\|_2\bigl(\sqrt2\,\eta+2\eta\,\|c\|_2\bigr).
\end{equation}
It remains to bound \(\|c\|_2\). From \eqref{eq:appendixG-c0c1},
\[
\|c\|_2^2=|z_1z_2|^2+|z_1+z_2|^2
\le \Rnode^4+4\Rnode^2
=\Rnode^2(\Rnode^2+4),
\]
so
\begin{equation}\label{eq:appendixG-c-bound}
 \|c\|_2\le \Rnode\sqrt{\Rnode^2+4}\le \Rnode(\Rnode+2).
\end{equation}
Substituting \eqref{eq:appendixG-c-bound} and Lemma~\ref{lem:appendixG-vandermonde-bounds} into \eqref{eq:appendixG-deltac-prebound} gives
\[
\|\widetilde c-c\|_2
\le 2\cdot \frac{2(1+\Rnode^2)}{\amin\sepz^2}
\bigl(\sqrt2+2\Rnode(\Rnode+2)\bigr)\eta,
\]
which is exactly \eqref{eq:appendixG-deltac-bound}.
\end{proof}

The next theorem turns the coefficient bound into a labeled root bound. This is the step where the extra factor \(\sepz^{-1}\) appears.

\begin{theorem}[Certified two-node Prony bound]\label{thm:appendixG-prony-conditioning}
Assume \eqref{eq:appendixG-two-node-model}, \eqref{eq:appendixG-nondegenerate-data}, and \eqref{eq:appendixG-eta-def}. Let
\begin{equation}\label{eq:appendixG-Lambda-star}
 \Lambda_\star:=\sqrt{1+\bigl(\Rnode+\sepz/2\bigr)^2}
\end{equation}
and
\begin{equation}\label{eq:appendixG-Cprony}
 C_{\mathrm{Pr}}(\Rnode,\sepz):=4\Lambda_\star\,C_{\mathrm{coef}}(\Rnode).
\end{equation}
If
\begin{equation}\label{eq:appendixG-smallness2}
 \eta<\min\left\{
 \frac{\amin\sepz^2}{8(1+\Rnode^2)},
 \frac{\amin\sepz^4}{8C_{\mathrm{Pr}}(\Rnode,\sepz)}
 \right\},
\end{equation}
then the polynomial \(\widetilde Q\) has exactly one root in each disk
\begin{equation}\label{eq:appendixG-root-disks}
 D\!\left(z_j,\frac{C_{\mathrm{Pr}}(\Rnode,\sepz)}{\amin\sepz^3}\,\eta\right),
 \qquad j=1,2,
\end{equation}
and its roots can therefore be labeled as \(\widetilde z_1,\widetilde z_2\) so that
\begin{equation}\label{eq:appendixG-root-error}
 \max_{j=1,2}|\widetilde z_j-z_j|
 \le
 \frac{C_{\mathrm{Pr}}(\Rnode,\sepz)}{\amin\sepz^3}\,\eta.
\end{equation}
\end{theorem}

\begin{proof}
Let
\[
\delta c:=\widetilde c-c,
\qquad
\delta Q(\zeta):=\widetilde Q(\zeta)-Q(\zeta)=\delta c_1\zeta+\delta c_0.
\]
Define
\begin{equation}\label{eq:appendixG-rstar-def}
 r_\star:=\frac{4\Lambda_\star}{\sepz}\,\|\delta c\|_2.
\end{equation}
By Lemma~\ref{lem:appendixG-coefficient-perturbation},
\[
 r_\star
 \le \frac{4\Lambda_\star C_{\mathrm{coef}}(\Rnode)}{\amin\sepz^3}\,\eta
 =\frac{C_{\mathrm{Pr}}(\Rnode,\sepz)}{\amin\sepz^3}\,\eta.
\]
The second inequality in \eqref{eq:appendixG-smallness2} implies \(r_\star<\sepz/2\).

Fix \(j\in\{1,2\}\) and consider the circle \(|\zeta-z_j|=r_\star\). Since \(r_\star<\sepz/2\),
\[
|\zeta-z_{3-j}|
\ge |z_j-z_{3-j}|-|\zeta-z_j|
=\sepz-r_\star
>\frac{\sepz}{2}.
\]
Therefore
\begin{equation}\label{eq:appendixG-Q-lower}
|Q(\zeta)|
=|\zeta-z_j|\,|\zeta-z_{3-j}|
\ge r_\star\frac{\sepz}{2}.
\end{equation}
On the same circle,
\[
|\zeta|\le |z_j|+r_\star\le \Rnode+\frac{\sepz}{2},
\]
so by Cauchy--Schwarz,
\begin{equation}\label{eq:appendixG-deltaQ-upper}
|\delta Q(\zeta)|
\le \|\delta c\|_2\sqrt{1+|\zeta|^2}
\le \Lambda_\star\|\delta c\|_2
=\frac{r_\star\sepz}{4}.
\end{equation}
Comparing \eqref{eq:appendixG-deltaQ-upper} with \eqref{eq:appendixG-Q-lower}, we obtain
\[
|\delta Q(\zeta)|<|Q(\zeta)|
\qquad \text{on } |\zeta-z_j|=r_\star.
\]
Rouch\'e's theorem therefore implies that \(Q\) and \(\widetilde Q\) have the same number of zeros in the disk \(D(z_j,r_\star)\). Since \(Q\) has exactly one zero there, namely \(z_j\), the perturbed polynomial \(\widetilde Q\) also has exactly one zero there. The two disks are disjoint because \(r_\star<\sepz/2\), so the roots are uniquely labeled. The bound \eqref{eq:appendixG-root-error} follows immediately from the definition of the disks.
\end{proof}

\begin{remark}[Why the exponent is three]\label{rem:appendixG-three}
The proof makes the exponent completely transparent. Solving the Hankel system costs \(\sepz^{-2}\) through \(\|\Yzero^{-1}\|_2\). Passing from coefficients to labeled roots costs one more power through \(|Q'(z_j)|^{-1}=\sepz^{-1}\). The resulting Prony radius therefore scales like \(\sepz^{-3}\). This is the exponent relevant for deterministic exclusion tests because it is obtained without any hidden asymptotics and without assuming that the perturbation points in a particularly favorable direction.
\end{remark}

\subsection{The local matrix-pencil condition number}

The same two-node model can be analyzed directly at the level of the structured pencil. The resulting local sensitivity is one power of \(\sepz\) better than the Prony radius above. That improvement does not contradict Theorem~\ref{thm:appendixG-prony-conditioning}; it reflects the fact that the matrix-pencil derivative keeps the nonlinear structure of the generalized eigenvalue problem intact instead of passing through the monic coefficient map.

\begin{proposition}[Left and right eigenvectors of the exact pencil]\label{prop:appendixG-eigenvectors}
For \(j=1,2\), let \(e_j\) denote the \(j\)-th Euclidean basis vector of \(\mathbb C^2\), and define
\begin{equation}\label{eq:appendixG-xjyj}
 x_j:=V^{-\top}e_j,
 \qquad
 y_j^\top:=e_j^\top A^{-1}V^{-1}.
\end{equation}
Then
\begin{equation}\label{eq:appendixG-right-left}
 \Yone x_j=z_j\Yzero x_j,
 \qquad
 y_j^\top \Yone=z_j y_j^\top \Yzero,
\end{equation}
and the normalization
\begin{equation}\label{eq:appendixG-normalization}
 y_j^\top \Yzero x_j=1
\end{equation}
holds. Moreover,
\begin{equation}\label{eq:appendixG-xy-bound}
 \|x_j\|_2\,\|y_j\|_2
 \le
 \frac{2(1+\Rnode^2)}{\amin\sepz^2}.
\end{equation}
\end{proposition}

\begin{proof}
Using Proposition~\ref{prop:appendixG-exact-factorizations},
\[
\Yone x_j=VA\Zmat V^\top V^{-\top}e_j=VA\Zmat e_j=z_j VAe_j,
\]
while
\[
\Yzero x_j=VAV^\top V^{-\top}e_j=VAe_j.
\]
This proves the first identity in \eqref{eq:appendixG-right-left}. The second is similar:
\[
 y_j^\top \Yone=e_j^\top A^{-1}V^{-1}VA\Zmat V^\top=e_j^\top \Zmat V^\top=z_j e_j^\top V^\top,
\]
whereas
\[
 y_j^\top \Yzero=e_j^\top A^{-1}V^{-1}VAV^\top=e_j^\top V^\top.
\]
The normalization is immediate:
\[
 y_j^\top \Yzero x_j=e_j^\top A^{-1}V^{-1}VAV^\top V^{-\top}e_j=e_j^\top e_j=1.
\]
Finally,
\[
\|x_j\|_2\le \|V^{-\top}\|_2=\|V^{-1}\|_2,
\qquad
\|y_j\|_2\le \|A^{-1}\|_2\|V^{-1}\|_2=\frac{\|V^{-1}\|_2}{\amin}.
\]
Applying Lemma~\ref{lem:appendixG-vandermonde-bounds} yields \eqref{eq:appendixG-xy-bound}.
\end{proof}

\begin{theorem}[Local matrix-pencil condition number]\label{thm:appendixG-pencil-local}
Let \(\Delta_0,\Delta_1\in\mathbb C^{2\times2}\) be fixed perturbation directions and consider the analytic pencil
\begin{equation}\label{eq:appendixG-analytic-pencil}
 P_\varepsilon(\lambda):=(\Yone+\varepsilon\Delta_1)-\lambda(\Yzero+\varepsilon\Delta_0).
\end{equation}
For each \(j=1,2\), there is an analytic eigenvalue branch \(\lambda_j(\varepsilon)\) defined for \(|\varepsilon|\) small, with \(\lambda_j(0)=z_j\), and
\begin{equation}\label{eq:appendixG-derivative-formula}
 \lambda_j'(0)=y_j^\top(\Delta_1-z_j\Delta_0)x_j,
\end{equation}
where \(x_j\) and \(y_j\) are given by Proposition~\ref{prop:appendixG-eigenvectors}. Consequently,
\begin{equation}\label{eq:appendixG-pencil-condition-bound}
 |\lambda_j'(0)|
 \le
 \frac{2(1+\Rnode^2)}{\amin\sepz^2}
 \bigl(\|\Delta_1\|_2+|z_j|\,\|\Delta_0\|_2\bigr).
\end{equation}
If the perturbation is induced by sample directions \(d_q\in\mathbb C\) through
\begin{equation}\label{eq:appendixG-sample-directions}
 \Delta_0=\begin{pmatrix} d_0 & d_1\\ d_1 & d_2\end{pmatrix},
 \qquad
 \Delta_1=\begin{pmatrix} d_1 & d_2\\ d_2 & d_3\end{pmatrix},
\end{equation}
then
\begin{equation}\label{eq:appendixG-sample-direction-bound}
 |\lambda_j'(0)|
 \le
 \frac{4(1+\Rnode)(1+\Rnode^2)}{\amin\sepz^2}
 \max_{0\le q\le 3}|d_q|.
\end{equation}
\end{theorem}

\begin{proof}
The exact characteristic determinant is
\[
\det(\Yone-\lambda\Yzero)=\det(\Yzero)(z_1-\lambda)(z_2-\lambda)
\]
by Proposition~\ref{prop:appendixG-exact-factorizations}. Since \(z_1\neq z_2\), both eigenvalues are simple. The analytic implicit function theorem applied to \(\det P_\varepsilon(\lambda)=0\) therefore gives analytic branches \(\lambda_j(\varepsilon)\) with \(\lambda_j(0)=z_j\).

Because the generalized eigenvalue is simple, one may choose a differentiable right-eigenvector branch \(x_j(\varepsilon)\) satisfying
\begin{equation}\label{eq:appendixG-eig-branch}
 P_\varepsilon(\lambda_j(\varepsilon))x_j(\varepsilon)=0,
\end{equation}
with \(x_j(0)=x_j\). Differentiate \eqref{eq:appendixG-eig-branch} at \(\varepsilon=0\):
\[
(\Delta_1-z_j\Delta_0)x_j-\lambda_j'(0)\Yzero x_j+(\Yone-z_j\Yzero)x_j'(0)=0.
\]
Left-multiplying by \(y_j^\top\) and using Proposition~\ref{prop:appendixG-eigenvectors},
\[
 y_j^\top(\Yone-z_j\Yzero)=0,
 \qquad
 y_j^\top\Yzero x_j=1,
\]
we obtain \eqref{eq:appendixG-derivative-formula}.

Now apply the Cauchy--Schwarz inequality and Proposition~\ref{prop:appendixG-eigenvectors}:
\[
|\lambda_j'(0)|
\le \|y_j\|_2\|x_j\|_2\bigl(\|\Delta_1\|_2+|z_j|\,\|\Delta_0\|_2\bigr)
\le \frac{2(1+\Rnode^2)}{\amin\sepz^2}\bigl(\|\Delta_1\|_2+|z_j|\,\|\Delta_0\|_2\bigr),
\]
which is \eqref{eq:appendixG-pencil-condition-bound}.

If \eqref{eq:appendixG-sample-directions} holds, then each of \(\Delta_0\) and \(\Delta_1\) has Frobenius norm at most \(2\max_q |d_q|\), hence operator norm at most the same quantity. Since \(|z_j|\le \Rnode\), substituting into \eqref{eq:appendixG-pencil-condition-bound} gives
\[
|\lambda_j'(0)|
\le \frac{2(1+\Rnode^2)}{\amin\sepz^2}\cdot 2(1+\Rnode)\max_q |d_q|,
\]
which is exactly \eqref{eq:appendixG-sample-direction-bound}.
\end{proof}

Theorem~\ref{thm:appendixG-pencil-local} is the local counterpart of Theorem~\ref{thm:appendixG-prony-conditioning}. Its exponent is \(\sepz^{-2}\) rather than \(\sepz^{-3}\). The reason is that the structured pencil derivative does not first pass through the monic coefficient map. It perturbs the generalized eigenvalue problem directly.

\begin{proposition}[A sharp structured direction]\label{prop:appendixG-sharp-direction}
Consider the perturbation that changes only the third sample, namely
\begin{equation}\label{eq:appendixG-y2-direction}
 d_0=d_1=d_3=0,
 \qquad d_2=1.
\end{equation}
Equivalently,
\begin{equation}\label{eq:appendixG-y2-matrices}
 \Delta_0=\begin{pmatrix} 0 & 0\\ 0 & 1\end{pmatrix},
 \qquad
 \Delta_1=\begin{pmatrix} 0 & 1\\ 1 & 0\end{pmatrix}.
\end{equation}
Then the local eigenvalue derivatives are
\begin{equation}\label{eq:appendixG-sharp-derivatives}
 \lambda_1'(0)=-\frac{z_1+2z_2}{a_1(z_1-z_2)^2},
 \qquad
 \lambda_2'(0)=-\frac{2z_1+z_2}{a_2(z_1-z_2)^2}.
\end{equation}
In particular, unless one of the numerators vanishes accidentally, the \(\sepz^{-2}\) exponent in Theorem~\ref{thm:appendixG-pencil-local} is sharp.
\end{proposition}

\begin{proof}
Apply \eqref{eq:appendixG-derivative-formula}. For \(j=1\), the explicit formulas
\[
 x_1=\frac{1}{z_2-z_1}\begin{pmatrix} z_2\\ -1\end{pmatrix},
 \qquad
 y_1^\top=\frac{1}{a_1(z_2-z_1)}\begin{pmatrix} z_2 & -1\end{pmatrix}
\]
follow from \eqref{eq:appendixG-xjyj}. Hence
\[
(\Delta_1-z_1\Delta_0)x_1
=
\frac{1}{z_2-z_1}
\begin{pmatrix}
-1\\ z_2+z_1
\end{pmatrix}.
\]
Taking the dot product with \(y_1\) gives
\[
\lambda_1'(0)
=
\frac{1}{a_1(z_2-z_1)^2}\bigl(-z_2-(z_1+z_2)\bigr)
=-\frac{z_1+2z_2}{a_1(z_1-z_2)^2}.
\]
The computation for \(j=2\) is identical:
\[
 x_2=\frac{1}{z_2-z_1}\begin{pmatrix} -z_1\\ 1\end{pmatrix},
 \qquad
 y_2^\top=\frac{1}{a_2(z_2-z_1)}\begin{pmatrix} -z_1 & 1\end{pmatrix},
\]
and therefore
\[
\lambda_2'(0)
=-\frac{2z_1+z_2}{a_2(z_1-z_2)^2}.
\]
This proves \eqref{eq:appendixG-sharp-derivatives}.
\end{proof}

\begin{remark}[Certified and local exponents serve different roles]\label{rem:appendixG-different-roles}
The Prony exponent \(\sepz^{-3}\) and the structured pencil exponent \(\sepz^{-2}\) are not competing claims. They answer different questions. Theorem~\ref{thm:appendixG-prony-conditioning} is a finite, nonasymptotic guarantee that survives arbitrary perturbation directions compatible with the sample bound \(|e_q|\le \eta\). Theorem~\ref{thm:appendixG-pencil-local} describes the first-order condition number of a simple generalized eigenvalue inside a fixed smooth perturbation family. For the trust-region bounds the former is the safe object. For interpreting actual estimator behavior in windows that are already well inside the trusted regime, the latter often captures the leading sensitivity more faithfully.
\end{remark}

\subsection{More than two nodes}

Although the sharpest formulas are easiest to read in the two-node case, the same mechanism is already encoded in the root factor of Appendix~\ref{app:abstract-frequency-extraction}. The next observation records the basic cluster scaling.

\begin{proposition}[Cluster scaling of the root factor]\label{prop:appendixG-cluster}
Let \(z_1,\dots,z_m\in\mathbb C\) be distinct, let
\[
Q(\zeta)=\prod_{\mu=1}^m(\zeta-z_\mu),
\qquad
\Gamma_\nu=|Q'(z_\nu)|=\prod_{\mu\ne \nu}|z_\nu-z_\mu|,
\]
and fix an index \(\nu\). Suppose there is a subset \(\mathcal C\subset\{1,\dots,m\}\) of cardinality \(q\ge 2\) such that \(\nu\in\mathcal C\) and
\begin{equation}\label{eq:appendixG-cluster-hyp}
 |z_\alpha-z_\beta|\le \delta
 \qquad \text{for all } \alpha,\beta\in\mathcal C.
\end{equation}
Then
\begin{equation}\label{eq:appendixG-cluster-upper}
 \Gamma_\nu
 \le
 \delta^{q-1}\prod_{\mu\notin\mathcal C}|z_\nu-z_\mu|.
\end{equation}
In particular, if the factors with \(\mu\notin\mathcal C\) remain uniformly bounded above and below, then
\begin{equation}\label{eq:appendixG-cluster-lower}
 \Gamma_\nu^{-1}=\Omega\!\bigl(\delta^{-(q-1)}\bigr)
 \qquad \text{as } \delta\to 0.
\end{equation}
\end{proposition}

\begin{proof}
For every \(\mu\in\mathcal C\setminus\{\nu\}\), assumption \eqref{eq:appendixG-cluster-hyp} gives \(|z_\nu-z_\mu|\le \delta\). Multiplying those \(q-1\) factors and leaving the remaining factors unchanged yields \eqref{eq:appendixG-cluster-upper}. Statement \eqref{eq:appendixG-cluster-lower} is immediate when the nonclustered factors stay within fixed positive bounds.
\end{proof}

Proposition~\ref{prop:appendixG-cluster} is the direct higher-mode analogue of the two-node discussion. A cluster of three nearby nodes costs at least two powers of the cluster scale in \(\Gamma_\nu^{-1}\), a cluster of four nearby nodes costs at least three, and so on. Accordingly, overtone-rich windows can become numerically unstable even before any individual mode has obviously vanished in amplitude.

Near coalescence is a first-order obstruction to trustworthy spectroscopy. Appendix~\ref{app:abstract-frequency-extraction} translates it into explicit extraction radii, Appendix~\ref{app:branch-selection} prevents the resulting estimates from being polluted by alias slips, and this appendix explains why the resulting trust-region test must reject windows in which the effective node separation becomes too small.

\section{Primary inversion and auxiliary consistency}\label{app:primary-inversion}

This appendix proves the deterministic inverse statements behind the $220$-based remnant recovery and the auxiliary $221$ and $440$ consistency tests. The basic geometric input has already been isolated in Appendix~\ref{app:kerr-qnm-data}. There we established a strictly positive lower singular-value margin for the real primary map associated with the dominant Kerr mode. The remaining work is to turn that local nondegeneracy into a quantitative inverse theorem on the event-local parameter box and then to propagate the resulting parameter error into the auxiliary mode predictions.

Throughout this appendix we write
\[
\Omega_{\mathrm{phys}}:=(0,\infty)\times(-1,1)
\]
for the open physical parameter domain and keep the event-local compact box
\[
\Kdet=[66.5,69.5]\,M_\odot\times[0.64,0.71]\subset \Omega_{\mathrm{phys}}
\]
fixed. The primary mode is $220$, and the auxiliary modes are $221$ and $440$.

\subsection{The primary Kerr map}

It is convenient to pass from the complex dominant mode frequency to its equivalent real two-vector.

\begin{definition}\label{def:appendixH-primary-map}
Define the real-linear isometry
\[
\Xi:\mathbb C\to\mathbb R^2,
\qquad
\Xi(z):=(\Re z,-\Im z),
\]
and the primary Kerr maps
\begin{equation}\label{eq:appendixH-primary-maps}
F_{220}(p):=\omega_{220}(p)\in\mathbb C,
\qquad
G_{220}(p):=\Xi\bigl(F_{220}(p)\bigr)
=\bigl(\Re\omega_{220}(p),-\Im\omega_{220}(p)\bigr)
\in\mathbb R^2,
\qquad p\in\Omega_{\mathrm{phys}}.
\end{equation}
The Euclidean norm on $\mathbb R^2$ and the modulus on $\mathbb C$ are compatible through $\Xi$:
\begin{equation}\label{eq:appendixH-Xi-isometry}
|\Xi(z_1)-\Xi(z_2)|=|z_1-z_2|,
\qquad z_1,z_2\in\mathbb C.
\end{equation}
\end{definition}

The following constant is the primary nondegeneracy margin furnished by Corollary~\ref{cor:appendixD-sigmamin}.
\begin{equation}\label{eq:appendixH-sigma0-def}
\sigma_0:=\sigma^{\mathrm{cert}}_{220}>0.
\end{equation}
By construction,
\begin{equation}\label{eq:appendixH-primary-sigmamin-lower}
\sigma_{\min}\bigl(DG_{220}(p)\bigr)\ge \sigma_0
\qquad\text{for all }p\in\Kdet.
\end{equation}

\begin{corollary}[Explicit primary chart constants for the Kerr tables]\label{cor:appendixH-explicit-chart}
For the Kerr tables used in Appendix~\ref{app:kerr-qnm-data}, one may take
\[
\sigma_0=2.31\times 10^{-5},
\qquad
L_{220}=\frac{2}{\sigma_0}=8.66\times 10^{4},
\qquad
\rho_{220}=2.47\times 10^{-4}.
\]
The last value is obtained from an explicit global Jacobian-Lipschitz bound for \(DG_{220}\) on \(\Kdet\) rather than from compactness alone. In particular, the four public detector-frame medians all lie more than \(10^2\rho_{220}\) away from the boundary of \(\Kdet\).
\end{corollary}

The first task is to make the continuity of the primary Jacobian uniform on the whole event-local box at the scale set by $\sigma_0$.

\begin{lemma}[Uniform local Jacobian control]\label{lem:appendixH-jacobian-control}
There exists a radius $\rho_{220}>0$ such that for every $p\in\Kdet$ and every pair of points
\[
q,r\in B_{\mathbb R^2}(p,\rho_{220})\cap\Kdet,
\]
one has
\begin{equation}\label{eq:appendixH-jacobian-control}
\bigl\|DG_{220}(q)-DG_{220}(r)\bigr\|_2\le \frac{\sigma_0}{2}.
\end{equation}
\end{lemma}

\begin{proof}
Fix $p\in\Kdet$. Since $G_{220}$ is $C^1$ on the open set $\Omega_{\mathrm{phys}}$, its Jacobian is continuous at $p$. Hence there exists $r_p>0$ such that
\[
\overline B_{\mathbb R^2}(p,2r_p)\subset\Omega_{\mathrm{phys}}
\qquad\text{and}\qquad
\bigl\|DG_{220}(x)-DG_{220}(p)\bigr\|_2\le \frac{\sigma_0}{4}
\quad\text{for all }x\in \overline B_{\mathbb R^2}(p,2r_p).
\]
Therefore, whenever $x,y\in \overline B_{\mathbb R^2}(p,2r_p)$,
\[
\bigl\|DG_{220}(x)-DG_{220}(y)\bigr\|_2
\le
\bigl\|DG_{220}(x)-DG_{220}(p)\bigr\|_2
+
\bigl\|DG_{220}(y)-DG_{220}(p)\bigr\|_2
\le \frac{\sigma_0}{2}.
\]
The family of open balls $\{B_{\mathbb R^2}(p,r_p):p\in\Kdet\}$ covers the compact set $\Kdet$, so there is a finite subcover
\[
\Kdet\subset \bigcup_{\alpha=1}^N B_{\mathbb R^2}(p_\alpha,r_{p_\alpha}).
\]
Set
\[
\rho_{220}:=\min_{1\le \alpha\le N} r_{p_\alpha}>0.
\]
Now fix $p\in\Kdet$ and choose $\alpha$ such that $p\in B_{\mathbb R^2}(p_\alpha,r_{p_\alpha})$. If $q,r\in B_{\mathbb R^2}(p,\rho_{220})\cap\Kdet$, then
\[
|q-p_\alpha|\le |q-p|+|p-p_\alpha|<\rho_{220}+r_{p_\alpha}\le 2r_{p_\alpha},
\]
and the same estimate holds for $r$. Thus both $q$ and $r$ lie in $B_{\mathbb R^2}(p_\alpha,2r_{p_\alpha})$, where the pairwise Jacobian bound proved above applies. This yields \eqref{eq:appendixH-jacobian-control}.
\end{proof}

The next theorem is the quantitative local inverse statement used throughout the analysis.

\begin{theorem}[Quantitative local inversion for the dominant mode]\label{thm:appendixH-primary-inverse}
Let $\rho_{220}$ be as in Lemma~\ref{lem:appendixH-jacobian-control}. Then for every $p\in\Kdet$ and every pair of points
\[
q,r\in B_{\mathbb R^2}(p,\rho_{220})\cap\Kdet,
\]
one has
\begin{equation}\label{eq:appendixH-primary-bilipschitz}
\bigl|G_{220}(q)-G_{220}(r)\bigr|
\ge
\frac{\sigma_0}{2}\,\|q-r\|_{\mathbb R^2}.
\end{equation}
Consequently, the restriction of $G_{220}$ to $B_{\mathbb R^2}(p,\rho_{220})\cap\Kdet$ is injective, and its inverse on the image satisfies
\begin{equation}\label{eq:appendixH-primary-inverse-Lipschitz}
\bigl\|G_{220}^{-1}(y_1)-G_{220}^{-1}(y_2)\bigr\|_{\mathbb R^2}
\le
\frac{2}{\sigma_0}\,|y_1-y_2|
\qquad
\text{for all }y_1,y_2\in G_{220}\bigl(B_{\mathbb R^2}(p,\rho_{220})\cap\Kdet\bigr).
\end{equation}
If we set
\begin{equation}\label{eq:appendixH-L220-def}
L_{220}:=\frac{2}{\sigma_0},
\end{equation}
then $L_{220}$ is a valid local inverse Lipschitz constant for the primary map on every event-local chart of radius $\rho_{220}$.
\end{theorem}

\begin{proof}
Fix $p\in\Kdet$ and $q,r\in B_{\mathbb R^2}(p,\rho_{220})\cap\Kdet$. Since both the ball and the rectangle $\Kdet$ are convex, the segment
\[
\gamma(t):=r+t(q-r),\qquad 0\le t\le 1,
\]
lies in $B_{\mathbb R^2}(p,\rho_{220})\cap\Kdet$. By the fundamental theorem of calculus,
\[
G_{220}(q)-G_{220}(r)
=
\int_0^1 DG_{220}(\gamma(t))(q-r)\,dt.
\]
Subtract and add the constant matrix $DG_{220}(p)$:
\[
G_{220}(q)-G_{220}(r)
=
DG_{220}(p)(q-r)
+
\int_0^1\bigl(DG_{220}(\gamma(t))-DG_{220}(p)\bigr)(q-r)\,dt.
\]
Taking norms, using the reverse triangle inequality, and then applying Lemma~\ref{lem:appendixH-jacobian-control}, we obtain
\begin{align*}
\bigl|G_{220}(q)-G_{220}(r)\bigr|
&\ge
\bigl|DG_{220}(p)(q-r)\bigr|
-
\int_0^1 \bigl\|DG_{220}(\gamma(t))-DG_{220}(p)\bigr\|_2\,dt\,\|q-r\|_{\mathbb R^2} \\
&\ge
\sigma_{\min}\bigl(DG_{220}(p)\bigr)\,\|q-r\|_{\mathbb R^2}
-
\frac{\sigma_0}{2}\,\|q-r\|_{\mathbb R^2}.
\end{align*}
Since $p\in\Kdet$, the uniform lower bound \eqref{eq:appendixH-primary-sigmamin-lower} gives
\[
\sigma_{\min}\bigl(DG_{220}(p)\bigr)\ge \sigma_0,
\]
and therefore
\[
\bigl|G_{220}(q)-G_{220}(r)\bigr|
\ge
\left(\sigma_0-\frac{\sigma_0}{2}\right)\|q-r\|_{\mathbb R^2}
=
\frac{\sigma_0}{2}\,\|q-r\|_{\mathbb R^2}.
\]
This proves \eqref{eq:appendixH-primary-bilipschitz}. If $G_{220}(q)=G_{220}(r)$, the lower bound forces $q=r$, so the restriction of $G_{220}$ to the chart is injective. Let $y_i=G_{220}(q_i)$ with $q_i\in B_{\mathbb R^2}(p,\rho_{220})\cap\Kdet$. Applying \eqref{eq:appendixH-primary-bilipschitz} to $q_1$ and $q_2$ yields
\[
\|q_1-q_2\|_{\mathbb R^2}
\le
\frac{2}{\sigma_0}\,|G_{220}(q_1)-G_{220}(q_2)|
=
\frac{2}{\sigma_0}\,|y_1-y_2|,
\]
which is exactly \eqref{eq:appendixH-primary-inverse-Lipschitz}.
\end{proof}

The complex form of the same statement is immediate because $\Xi$ is an isometry.

\begin{corollary}[Complex-frequency form of the primary inverse]\label{cor:appendixH-primary-complex}
For $p\in\Kdet$, define the local primary image
\begin{equation}\label{eq:appendixH-primary-image}
V_{220}(p):=F_{220}\bigl(B_{\mathbb R^2}(p,\rho_{220})\cap\Kdet\bigr)\subset\mathbb C.
\end{equation}
Then $F_{220}$ is injective on $B_{\mathbb R^2}(p,\rho_{220})\cap\Kdet$, and its inverse on $V_{220}(p)$ satisfies
\begin{equation}\label{eq:appendixH-primary-complex-Lipschitz}
\bigl\|F_{220}^{-1}(z_1)-F_{220}^{-1}(z_2)\bigr\|_{\mathbb R^2}
\le
L_{220}\,|z_1-z_2|,
\qquad z_1,z_2\in V_{220}(p).
\end{equation}
In particular, if $\widehat\omega_{220}\in V_{220}(p)$ and
\begin{equation}\label{eq:appendixH-primary-error-hyp}
\bigl|\widehat\omega_{220}-\omega_{220}(p)\bigr|\le \varepsilon_{220},
\end{equation}
then the primary estimate
\begin{equation}\label{eq:appendixH-primary-estimate}
\widehat p:=F_{220}^{-1}(\widehat\omega_{220})
\end{equation}
is well-defined in $B_{\mathbb R^2}(p,\rho_{220})\cap\Kdet$ and obeys
\begin{equation}\label{eq:appendixH-primary-parameter-error}
\|\widehat p-p\|_{\mathbb R^2}
\le
L_{220}\,\varepsilon_{220}.
\end{equation}
\end{corollary}

\begin{proof}
Because $\Xi$ is bijective and satisfies \eqref{eq:appendixH-Xi-isometry}, the injectivity of $F_{220}$ on the chart is equivalent to the injectivity of $G_{220}$ there. For $z_i\in V_{220}(p)$ write $y_i=\Xi(z_i)$. Then
\[
F_{220}^{-1}(z_i)=G_{220}^{-1}(y_i),
\qquad
|y_1-y_2|=|z_1-z_2|.
\]
Applying \eqref{eq:appendixH-primary-inverse-Lipschitz} gives \eqref{eq:appendixH-primary-complex-Lipschitz}. The parameter error bound \eqref{eq:appendixH-primary-parameter-error} follows by taking $z_1=\widehat\omega_{220}$ and $z_2=\omega_{220}(p)$.
\end{proof}

Corollary~\ref{cor:appendixH-primary-complex} is the missing existence statement behind the local inverse map used abstractly in Appendix~\ref{app:assumption-ledger}. The constant $L_{220}=2/\sigma_0$ is not an ad hoc condition number. It is the deterministic inverse Lipschitz constant produced by the Jacobian margin of Appendix~\ref{app:kerr-qnm-data}.

\subsection{Auxiliary mode Lipschitz control}

To turn the primary parameter error into a tolerance for the auxiliary checks, one needs a forward Lipschitz bound for the $221$ and $440$ Kerr maps on $\Kdet$.

\begin{proposition}[Certified forward Lipschitz constants for the auxiliary modes]\label{prop:appendixH-aux-lipschitz}
For each auxiliary mode $j\in\{221,440\}$ define
\begin{equation}\label{eq:appendixH-Lj-def}
L_j:=\sqrt{\frac{U_j^2}{M_-^4}+\frac{V_j^2}{M_-^2}},
\end{equation}
where $U_j$ and $V_j$ are the quantities introduced in \eqref{eq:appendixD-UV-def}. Then for all $p_1,p_2\in\Kdet$,
\begin{equation}\label{eq:appendixH-aux-lipschitz}
\bigl|\omega_j(p_1)-\omega_j(p_2)\bigr|
\le
L_j\,\|p_1-p_2\|_{\mathbb R^2}.
\end{equation}
Thus $L_{221}$ and $L_{440}$ are valid global Lipschitz constants for the auxiliary Kerr maps on the event-local box.
\end{proposition}

\begin{proof}
Fix $j\in\{221,440\}$ and let
\[
G_j(M,\chi):=\bigl(\Re\omega_j(M,\chi),-\Im\omega_j(M,\chi)\bigr).
\]
By the explicit Jacobian formula used in the proof of Corollary~\ref{cor:appendixD-sigmamin}, one has
\[
\|DG_j(M,\chi)\|_F^2
=
\frac{|\widehat\omega_j(\chi)|^2}{M^4}
+
\frac{|\widehat\omega_j'(\chi)|^2}{M^2}.
\]
The definitions of $U_j$ and $V_j$ imply
\[
|\widehat\omega_j(\chi)|\le U_j,
\qquad
|\widehat\omega_j'(\chi)|\le V_j
\qquad\text{for all }\chi\in\Ichi,
\]
so on $\Kdet=\Mbox\times\Ichi$ we obtain
\[
\|DG_j(M,\chi)\|_2
\le
\|DG_j(M,\chi)\|_F
\le
\sqrt{\frac{U_j^2}{M_-^4}+\frac{V_j^2}{M_-^2}}
=L_j.
\]
Now let $p_1,p_2\in\Kdet$. Since $\Kdet$ is a rectangle, the segment
\[
\gamma(t):=p_2+t(p_1-p_2),\qquad 0\le t\le 1,
\]
remains in $\Kdet$. Using the real map $G_j=\Xi\circ\omega_j$, the fundamental theorem of calculus, and the fact that $\Xi$ is an isometry, we find
\begin{align*}
|\omega_j(p_1)-\omega_j(p_2)|
&=|G_j(p_1)-G_j(p_2)| \\
&=\left|\int_0^1 DG_j\bigl(\gamma(t)\bigr)(p_1-p_2)\,dt\right| \\
&\le \int_0^1 \bigl\|DG_j\bigl(\gamma(t)\bigr)\bigr\|_2\,dt\,\|p_1-p_2\|_{\mathbb R^2} \\
&\le L_j\,\|p_1-p_2\|_{\mathbb R^2},
\end{align*}
which is \eqref{eq:appendixH-aux-lipschitz}.
\end{proof}

\subsection{Propagation from primary inversion to auxiliary consistency}

We now combine the primary inverse theorem with the forward Lipschitz control of the auxiliary modes.

\begin{theorem}[Primary error propagation into auxiliary residuals]\label{thm:appendixH-aux-propagation}
Fix $p\in\Kdet$ and let $\widehat\omega_{220}\in V_{220}(p)$ with primary estimate $\widehat p=F_{220}^{-1}(\widehat\omega_{220})\in B_{\mathbb R^2}(p,\rho_{220})\cap\Kdet$. Let $j\in\{221,440\}$ and let $\widehat\omega_j\in\mathbb C$ be an extracted auxiliary frequency. Assume that
\begin{equation}\label{eq:appendixH-aux-hypotheses}
\bigl|\widehat\omega_{220}-\omega_{220}(p)\bigr|\le \varepsilon_{220},
\qquad
\bigl|\widehat\omega_j-\omega_j(p)\bigr|\le \varepsilon_j.
\end{equation}
Then
\begin{equation}\label{eq:appendixH-aux-primary-error}
\|\widehat p-p\|_{\mathbb R^2}\le L_{220}\,\varepsilon_{220},
\end{equation}
and the auxiliary residual
\begin{equation}\label{eq:appendixH-aux-residual-def}
R_j(\widehat p,\widehat\omega_j):=\bigl|\widehat\omega_j-\omega_j(\widehat p)\bigr|
\end{equation}
satisfies
\begin{equation}\label{eq:appendixH-aux-residual-bound}
R_j(\widehat p,\widehat\omega_j)
\le
\varepsilon_j+L_jL_{220}\,\varepsilon_{220}.
\end{equation}
Equivalently, if one defines the deterministic auxiliary tolerance
\begin{equation}\label{eq:appendixH-aux-threshold}
\tau_j(\varepsilon_{220},\varepsilon_j):=\varepsilon_j+L_jL_{220}\,\varepsilon_{220},
\end{equation}
then every auxiliary mode compatible with the same underlying remnant must lie in the disk
\begin{equation}\label{eq:appendixH-aux-disk}
\widehat\omega_j\in \overline D\bigl(\omega_j(\widehat p),\tau_j(\varepsilon_{220},\varepsilon_j)\bigr).
\end{equation}
\end{theorem}

\begin{proof}
The primary parameter bound \eqref{eq:appendixH-aux-primary-error} is precisely Corollary~\ref{cor:appendixH-primary-complex}. For the auxiliary residual, insert and subtract $\omega_j(p)$:
\[
R_j(\widehat p,\widehat\omega_j)
\le
\bigl|\widehat\omega_j-\omega_j(p)\bigr|
+
\bigl|\omega_j(p)-\omega_j(\widehat p)\bigr|.
\]
The first term is bounded by $\varepsilon_j$ from \eqref{eq:appendixH-aux-hypotheses}. The second term is controlled by Proposition~\ref{prop:appendixH-aux-lipschitz}:
\[
\bigl|\omega_j(p)-\omega_j(\widehat p)\bigr|
\le
L_j\,\|\widehat p-p\|_{\mathbb R^2}
\le
L_jL_{220}\,\varepsilon_{220}.
\]
Combining the two inequalities yields \eqref{eq:appendixH-aux-residual-bound}. Statement \eqref{eq:appendixH-aux-disk} is just a rewriting of the same bound.
\end{proof}

Theorem~\ref{thm:appendixH-aux-propagation} shows that even if the auxiliary extraction itself were exact, the comparison radius could not be smaller than $L_jL_{220}\,\varepsilon_{220}$, because the primary inversion uncertainty already moves the Kerr prediction for the auxiliary mode by that amount. The auxiliary threshold is therefore forced by deterministic geometry.

\subsection{Consistency tubes and failure certificates}

The previous theorem is still phrased relative to an underlying parameter point $p$. For the trust-region test one also needs the practical contrapositive: if the observed auxiliary residual is too large, then no single remnant in the same local primary chart can account for both the primary and the auxiliary observables within the stated budgets.

\begin{proposition}[Local common-remnant failure certificate]\label{prop:appendixH-failure-certificate}
Fix $p\in\Kdet$ and let $\widehat p\in B_{\mathbb R^2}(p,\rho_{220})\cap\Kdet$. Set
\[
\widehat\omega_{220}:=F_{220}(\widehat p).
\]
Let $j\in\{221,440\}$ and let $\widehat\omega_j\in\mathbb C$. Assume there exists a point
\[
q\in B_{\mathbb R^2}(p,\rho_{220})\cap\Kdet
\]
such that
\begin{equation}\label{eq:appendixH-failure-hypotheses}
\bigl|\omega_{220}(q)-\widehat\omega_{220}\bigr|\le \varepsilon_{220},
\qquad
\bigl|\omega_j(q)-\widehat\omega_j\bigr|\le \varepsilon_j.
\end{equation}
Then
\begin{equation}\label{eq:appendixH-failure-bound}
R_j(\widehat p,\widehat\omega_j)
\le
\varepsilon_j+L_jL_{220}\,\varepsilon_{220}.
\end{equation}
Consequently, if the computed residual violates
\begin{equation}\label{eq:appendixH-failure-violation}
R_j(\widehat p,\widehat\omega_j)
>
\varepsilon_j+L_jL_{220}\,\varepsilon_{220},
\end{equation}
then there is no point $q$ in the same local primary chart for which both inequalities in \eqref{eq:appendixH-failure-hypotheses} hold simultaneously.
\end{proposition}

\begin{proof}
Because $q$ and $\widehat p$ both belong to $B_{\mathbb R^2}(p,\rho_{220})\cap\Kdet$, Theorem~\ref{thm:appendixH-primary-inverse} applies on that chart. Since $F_{220}(\widehat p)=\widehat\omega_{220}$,
\[
\|q-\widehat p\|_{\mathbb R^2}
\le
L_{220}\,\bigl|F_{220}(q)-F_{220}(\widehat p)\bigr|
=
L_{220}\,\bigl|\omega_{220}(q)-\widehat\omega_{220}\bigr|
\le
L_{220}\,\varepsilon_{220}.
\]
Now insert and subtract $\omega_j(q)$:
\begin{align*}
R_j(\widehat p,\widehat\omega_j)
&=
\bigl|\widehat\omega_j-\omega_j(\widehat p)\bigr| \\
&\le
\bigl|\widehat\omega_j-\omega_j(q)\bigr|
+
\bigl|\omega_j(q)-\omega_j(\widehat p)\bigr| \\
&\le
\varepsilon_j+L_j\,\|q-\widehat p\|_{\mathbb R^2} \\
&\le
\varepsilon_j+L_jL_{220}\,\varepsilon_{220},
\end{align*}
which is \eqref{eq:appendixH-failure-bound}. The final statement is the contrapositive.
\end{proof}

Proposition~\ref{prop:appendixH-failure-certificate} is the precise deterministic reason that a large auxiliary residual is informative. It does not say merely that the fit is suboptimal. It says that, within the local primary chart, no single Kerr remnant can simultaneously explain the dominant mode and the tested auxiliary mode at the declared error levels.

When both auxiliary modes are available, the consistency region is the product of the two deterministic disks.

\begin{corollary}[Joint auxiliary consistency tube]\label{cor:appendixH-joint-tube}
Under the hypotheses of Theorem~\ref{thm:appendixH-aux-propagation} for both auxiliary modes $221$ and $440$, define
\begin{equation}\label{eq:appendixH-joint-tube-def}
\mathcal C_{\mathrm{aux}}(\widehat p)
:=
\overline D\bigl(\omega_{221}(\widehat p),\tau_{221}(\varepsilon_{220},\varepsilon_{221})\bigr)
\times
\overline D\bigl(\omega_{440}(\widehat p),\tau_{440}(\varepsilon_{220},\varepsilon_{440})\bigr)
\subset \mathbb C^2.
\end{equation}
Then
\begin{equation}\label{eq:appendixH-joint-tube-membership}
(\widehat\omega_{221},\widehat\omega_{440})\in \mathcal C_{\mathrm{aux}}(\widehat p).
\end{equation}
Conversely, if either coordinate residual exceeds its deterministic radius, then no single remnant in the same local primary chart can satisfy the corresponding primary and auxiliary error budgets simultaneously.
\end{corollary}

\begin{proof}
Apply Theorem~\ref{thm:appendixH-aux-propagation} separately to $j=221$ and $j=440$. The converse is the corresponding application of Proposition~\ref{prop:appendixH-failure-certificate}.
\end{proof}

Corollary~\ref{cor:appendixH-joint-tube} gives the mathematical statement behind the auxiliary-consistency panels used later. Those panels record the deterministic requirement that the $221$ and $440$ observables lie in a primary-error-enlarged Kerr prediction tube if they are to support the same remnant inferred from $220$.

The proof of this appendix is now complete. Appendix~\ref{app:kerr-qnm-data} supplied the Jacobian and derivative margins; the current appendix turns them into a concrete inverse theorem and a concrete consistency test. Those are the only ingredients needed later to convert frequency error bars into parameter error bars and then into trust-region pass or fail statements.

\section{Secondary results for joint two-mode inversion}\label{app:joint-two-mode-inversion}

The asymmetric strategy used above is simple. The dominant mode $220$ is used for primary inversion, and the subdominant modes $221$ and $440$ are used only as auxiliary consistency checks. On the event-local chart relevant for GW250114, the dominant map is already injective with a explicit inverse Lipschitz constant, so an exact two-mode inverse carries no additional local parameter information. The corresponding overdetermined geometry is nevertheless recorded because it shows why an exact pair inverse is locally redundant and because it yields a clean projected least-squares estimator when the observed pair does not lie exactly on the Kerr image.

Throughout this appendix we fix an auxiliary index
\[
j\in\{221,440\}
\]
and write
\[
\Omega_{\mathrm{phys}}:=(0,\infty)\times(-1,1),
\qquad
\Kdet=[66.5,69.5] \, M_{\odot}\times[0.64,0.71]\subset\Omega_{\mathrm{phys}}.
\]
We retain the notation of Appendix~\ref{app:primary-inversion}. In particular, $F_{220}$ and $G_{220}$ denote the complex and real dominant-mode maps, $\rho_{220}$ is the dominant-mode local chart radius from Lemma~\ref{lem:appendixH-jacobian-control}, $\sigma_0=\sigma^{\mathrm{cert}}_{220}>0$ is the lower singular-value margin from \eqref{eq:appendixH-sigma0-def}, and
\[
L_{220}=\frac{2}{\sigma_0}
\]
is the dominant inverse Lipschitz constant from \eqref{eq:appendixH-L220-def}.

\subsection{Pair maps and local geometry}

We begin by packaging the dominant mode and one auxiliary mode into a single overdetermined forward map.

\begin{definition}\label{def:appendixI-pair-map}
Let
\[
\Xi_2:\mathbb C^2\to\mathbb R^4,
\qquad
\Xi_2(z_1,z_2):=\bigl(\Xi(z_1),\Xi(z_2)\bigr),
\]
where $\Xi$ is the isometry from Definition~\ref{def:appendixH-primary-map}. Equip $\mathbb C^2$ with the Euclidean norm
\[
\|(z_1,z_2)\|_{\mathbb C^2}:=\bigl(|z_1|^2+|z_2|^2\bigr)^{1/2}
\]
and $\mathbb R^4$ with its standard Euclidean norm. For each $j\in\{221,440\}$ define the complex pair map
\begin{equation}\label{eq:appendixI-pair-map-complex}
W_j(p):=\bigl(\omega_{220}(p),\omega_j(p)\bigr)\in\mathbb C^2,
\qquad p\in\Omega_{\mathrm{phys}},
\end{equation}
and its real representative
\begin{equation}\label{eq:appendixI-pair-map-real}
H_j(p):=\Xi_2\bigl(W_j(p)\bigr)
=\bigl(G_{220}(p),G_j(p)\bigr)
\in\mathbb R^4,
\qquad
G_j(p):=\Xi\bigl(\omega_j(p)\bigr).
\end{equation}
The map $\Xi_2$ is a real-linear isometry:
\begin{equation}\label{eq:appendixI-Xi2-isometry}
\bigl|\Xi_2(z)-\Xi_2(w)\bigr|=\|z-w\|_{\mathbb C^2},
\qquad z,w\in\mathbb C^2.
\end{equation}
\end{definition}

The overdetermined pair map inherits a uniform lower singular-value margin from the dominant component.

\begin{proposition}[Pair singular values dominate the primary singular value]\label{prop:appendixI-pair-singular}
For each $j\in\{221,440\}$ and each $p\in\Kdet$, define
\begin{equation}\label{eq:appendixI-Sigma-pointwise}
\Sigma_j(p):=\sigma_{\min}\bigl(DH_j(p)\bigr).
\end{equation}
Then
\begin{equation}\label{eq:appendixI-pair-singular-lower-pointwise}
\Sigma_j(p)\ge \sigma_{\min}\bigl(DG_{220}(p)\bigr)\ge \sigma_0
\qquad\text{for all }p\in\Kdet.
\end{equation}
Consequently the compact-box infimum
\begin{equation}\label{eq:appendixI-Sigma-global}
\Sigma_j:=\inf_{p\in\Kdet}\Sigma_j(p)
\end{equation}
is strictly positive and satisfies
\begin{equation}\label{eq:appendixI-pair-singular-lower-global}
\Sigma_j\ge \sigma_0>0.
\end{equation}
\end{proposition}

\begin{proof}
Fix $j\in\{221,440\}$ and $p\in\Kdet$. By construction,
\[
DH_j(p)=
\begin{pmatrix}
DG_{220}(p)\\
DG_j(p)
\end{pmatrix},
\]
so for every $v\in\mathbb R^2$ one has
\begin{equation}\label{eq:appendixI-block-norm-identity}
\bigl|DH_j(p)v\bigr|^2
=
\bigl|DG_{220}(p)v\bigr|^2+
\bigl|DG_j(p)v\bigr|^2.
\end{equation}
If $|v|=1$, then \eqref{eq:appendixI-block-norm-identity} implies
\[
\bigl|DH_j(p)v\bigr|^2\ge \bigl|DG_{220}(p)v\bigr|^2.
\]
Taking the infimum over all unit vectors gives
\[
\Sigma_j(p)^2
=\inf_{|v|=1}\bigl|DH_j(p)v\bigr|^2
\ge
\inf_{|v|=1}\bigl|DG_{220}(p)v\bigr|^2
=\sigma_{\min}\bigl(DG_{220}(p)\bigr)^2.
\]
The uniform lower bound \eqref{eq:appendixH-primary-sigmamin-lower} from Appendix~\ref{app:primary-inversion} therefore yields
\[
\Sigma_j(p)\ge \sigma_0,
\qquad p\in\Kdet.
\]
This proves \eqref{eq:appendixI-pair-singular-lower-pointwise}. Since $H_j$ is $C^1$ on the open set $\Omega_{\mathrm{phys}}$, the Jacobian $DH_j$ is continuous, and therefore the singular-value map $p\mapsto\Sigma_j(p)$ is continuous on the compact set $\Kdet$. Its infimum is therefore attained and is strictly positive by the pointwise lower bound. This proves \eqref{eq:appendixI-pair-singular-lower-global}.
\end{proof}

\begin{corollary}[Explicit pair singular margins]\label{cor:appendixI-explicit-pair}
For the Kerr tables, the pair singular-value floors are
\[
\Sigma_{221}^{\mathrm{cert}}=7.65\times 10^{-5},
\qquad
\Sigma_{440}^{\mathrm{cert}}=3.46\times 10^{-5}.
\]
The corresponding explicit pair-chart radii obtained from global Jacobian-Lipschitz bounds are
\[
\rho_{(220,221)}=5.46\times 10^{-4},
\qquad
\rho_{(220,440)}=1.78\times 10^{-4}.
\]
Since both exceed the primary radius \(\rho_{220}\), the common local chart used later remains limited by the dominant-mode inverse rather than by the pair geometry.
\end{corollary}

\begin{figure}[t]
\centering
\includegraphics[width=0.94\textwidth]{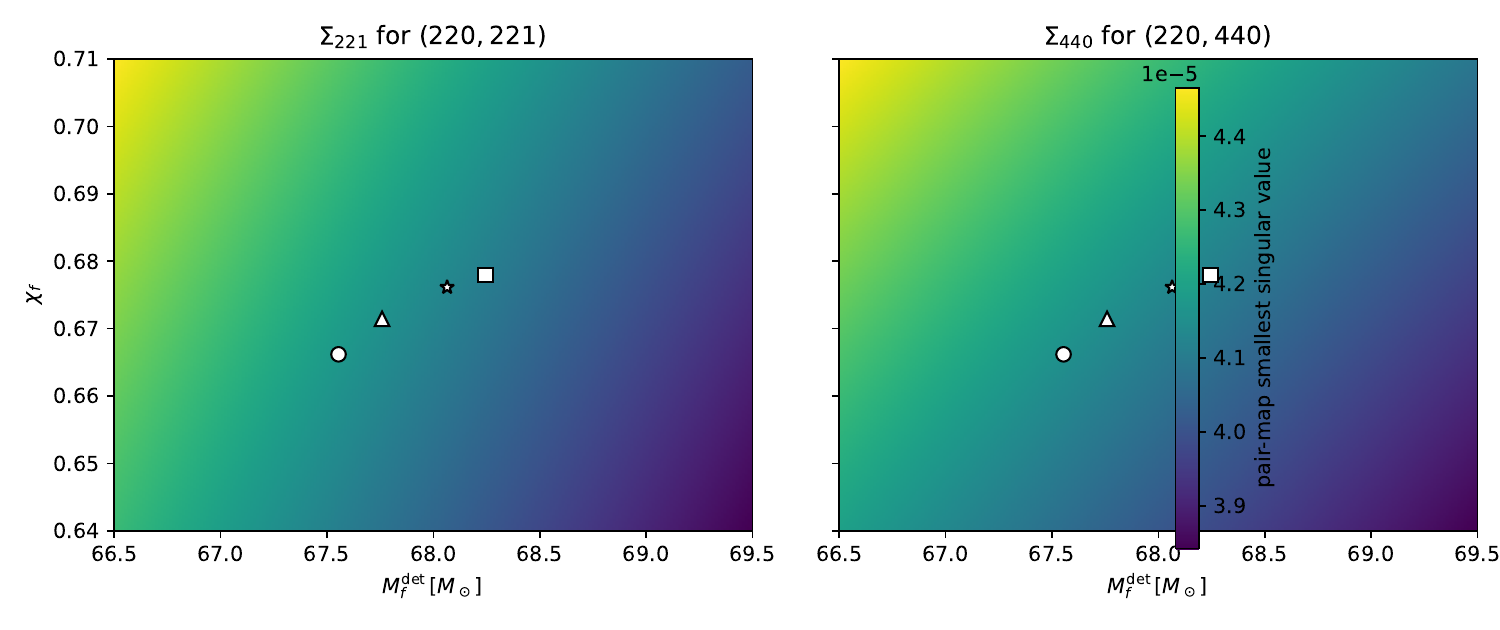}
\caption{Smallest singular-value fields of the realified pair maps \(H_{221}\) and \(H_{440}\) on the detector-frame box \(\Kdet\). The public detector-frame medians stay well inside the region where the overdetermined pair geometry is uniformly nondegenerate.}
\label{fig:appendixI-pair-sigma-atlas}
\end{figure}

The next lemma gives a uniform chart radius on which the pair Jacobian varies by at most half of its global lower margin.

\begin{lemma}[Uniform local Jacobian control for the pair map]\label{lem:appendixI-jacobian-control}
For each $j\in\{221,440\}$ there exists a radius
\begin{equation}\label{eq:appendixI-rho-pair-def}
\rho_{(220,j)}>0
\end{equation}
such that for every $p\in\Kdet$ and every pair of points
\[
q,r\in B_{\mathbb R^2}\bigl(p,\rho_{(220,j)}\bigr)\cap\Kdet,
\]
one has
\begin{equation}\label{eq:appendixI-jacobian-control}
\bigl\|DH_j(q)-DH_j(r)\bigr\|_2\le \frac{\Sigma_j}{2}.
\end{equation}
\end{lemma}

\begin{proof}
Fix $p\in\Kdet$. Since $H_j$ is $C^1$ on $\Omega_{\mathrm{phys}}$, its Jacobian is continuous at $p$. Hence there exists $r_p>0$ such that
\[
\overline B_{\mathbb R^2}(p,2r_p)\subset \Omega_{\mathrm{phys}}
\qquad\text{and}\qquad
\bigl\|DH_j(x)-DH_j(p)\bigr\|_2\le \frac{\Sigma_j}{4}
\]
for all $x\in \overline B_{\mathbb R^2}(p,2r_p)$. Therefore, whenever $x,y\in \overline B_{\mathbb R^2}(p,2r_p)$,
\[
\bigl\|DH_j(x)-DH_j(y)\bigr\|_2
\le
\bigl\|DH_j(x)-DH_j(p)\bigr\|_2+
\bigl\|DH_j(y)-DH_j(p)\bigr\|_2
\le \frac{\Sigma_j}{2}.
\]
The family of open balls $\{B_{\mathbb R^2}(p,r_p):p\in\Kdet\}$ covers the compact set $\Kdet$, so there is a finite subcover
\[
\Kdet\subset \bigcup_{\alpha=1}^{N} B_{\mathbb R^2}(p_\alpha,r_{p_\alpha}).
\]
Set
\[
\rho_{(220,j)}:=\min_{1\le \alpha\le N} r_{p_\alpha}>0.
\]
Now fix $p\in\Kdet$ and choose $\alpha$ such that $p\in B_{\mathbb R^2}(p_\alpha,r_{p_\alpha})$. If
\[
q,r\in B_{\mathbb R^2}\bigl(p,\rho_{(220,j)}\bigr)\cap\Kdet,
\]
then
\[
|q-p_\alpha|\le |q-p|+|p-p_\alpha|<\rho_{(220,j)}+r_{p_\alpha}\le 2r_{p_\alpha},
\]
and the same estimate holds for $r$. Thus $q,r\in B_{\mathbb R^2}(p_\alpha,2r_{p_\alpha})$, where the pairwise Jacobian bound proved above applies. This yields \eqref{eq:appendixI-jacobian-control}.
\end{proof}

For later comparison with the primary inverse, it is convenient to work on a common local chart on which both the dominant map and the pair map are simultaneously controlled. We therefore set
\begin{equation}\label{eq:appendixI-common-radius}
\bar\rho_j:=\min\{\rho_{220},\rho_{(220,j)}\}>0,
\end{equation}
and for each base point $p\in\Kdet$ define the common chart
\begin{equation}\label{eq:appendixI-common-chart}
\mathcal C_j(p):=B_{\mathbb R^2}(p,\bar\rho_j)\cap\Kdet.
\end{equation}

\begin{theorem}[Local bilipschitz theorem for the overdetermined pair map]\label{thm:appendixI-pair-inverse}
Fix $j\in\{221,440\}$. For every $p\in\Kdet$ and every pair of points
\[
q,r\in \mathcal C_j(p),
\]
one has
\begin{equation}\label{eq:appendixI-pair-bilipschitz}
\bigl\|W_j(q)-W_j(r)\bigr\|_{\mathbb C^2}
=
\bigl|H_j(q)-H_j(r)\bigr|
\ge
\frac{\Sigma_j}{2}\,\|q-r\|_{\mathbb R^2}.
\end{equation}
Consequently the restriction of $W_j$ to $\mathcal C_j(p)$ is injective. If we define
\begin{equation}\label{eq:appendixI-pair-Lipschitz-constant}
L_{(220,j)}:=\frac{2}{\Sigma_j},
\end{equation}
then its inverse on the local image
\begin{equation}\label{eq:appendixI-local-image}
\mathcal S_j(p):=W_j\bigl(\mathcal C_j(p)\bigr)\subset\mathbb C^2
\end{equation}
satisfies
\begin{equation}\label{eq:appendixI-pair-inverse-Lipschitz}
\bigl\|W_j^{-1}(z)-W_j^{-1}(w)\bigr\|_{\mathbb R^2}
\le
L_{(220,j)}\,\|z-w\|_{\mathbb C^2},
\qquad z,w\in\mathcal S_j(p).
\end{equation}
Moreover,
\begin{equation}\label{eq:appendixI-pair-constant-dominates}
L_{(220,j)}\le L_{220}.
\end{equation}
\end{theorem}

\begin{proof}
Fix $p\in\Kdet$ and $q,r\in\mathcal C_j(p)$. Since $\mathcal C_j(p)$ is convex, the segment
\[
\gamma(t):=r+t(q-r),\qquad 0\le t\le 1,
\]
remains in $\mathcal C_j(p)$. Because $\bar\rho_j\le \rho_{(220,j)}$, Lemma~\ref{lem:appendixI-jacobian-control} applies throughout the chart. The fundamental theorem of calculus gives
\[
H_j(q)-H_j(r)=\int_0^1 DH_j\bigl(\gamma(t)\bigr)(q-r)\,dt.
\]
Subtract and add the constant matrix $DH_j(p)$:
\[
H_j(q)-H_j(r)
=
DH_j(p)(q-r)
+
\int_0^1\bigl(DH_j(\gamma(t))-DH_j(p)\bigr)(q-r)\,dt.
\]
Taking norms, using the reverse triangle inequality, and then using Lemma~\ref{lem:appendixI-jacobian-control}, we obtain
\begin{align*}
\bigl|H_j(q)-H_j(r)\bigr|
&\ge
\bigl|DH_j(p)(q-r)\bigr|
-
\int_0^1\bigl\|DH_j(\gamma(t))-DH_j(p)\bigr\|_2\,dt\,\|q-r\|_{\mathbb R^2} \\
&\ge
\Sigma_j(p)\,\|q-r\|_{\mathbb R^2}
-
\frac{\Sigma_j}{2}\,\|q-r\|_{\mathbb R^2}.
\end{align*}
By Proposition~\ref{prop:appendixI-pair-singular}, one has $\Sigma_j(p)\ge \Sigma_j$, and therefore
\[
\bigl|H_j(q)-H_j(r)\bigr|
\ge
\left(\Sigma_j-\frac{\Sigma_j}{2}\right)\|q-r\|_{\mathbb R^2}
=
\frac{\Sigma_j}{2}\,\|q-r\|_{\mathbb R^2}.
\]
Since $\Xi_2$ is an isometry, this is exactly \eqref{eq:appendixI-pair-bilipschitz}. If $W_j(q)=W_j(r)$, then the lower bound forces $q=r$, so $W_j$ is injective on $\mathcal C_j(p)$. If $z=W_j(q)$ and $w=W_j(r)$, the same estimate implies
\[
\|q-r\|_{\mathbb R^2}
\le
\frac{2}{\Sigma_j}\,\|W_j(q)-W_j(r)\|_{\mathbb C^2}
=
L_{(220,j)}\,\|z-w\|_{\mathbb C^2},
\]
which is \eqref{eq:appendixI-pair-inverse-Lipschitz}. Finally, Proposition~\ref{prop:appendixI-pair-singular} gives $\Sigma_j\ge \sigma_0$, hence
\[
L_{(220,j)}=\frac{2}{\Sigma_j}\le \frac{2}{\sigma_0}=L_{220}.
\]
This proves \eqref{eq:appendixI-pair-constant-dominates}.
\end{proof}

The last inequality is conceptually important. The pair inverse is never worse than the primary inverse on the common chart, but on that same chart it is also estimatorially redundant in the exact-image setting.

\begin{corollary}[Exact pair inversion agrees with the primary inverse]\label{cor:appendixI-exact-pair-agrees-primary}
Fix $j\in\{221,440\}$ and $p\in\Kdet$. Let
\[
\widehat z=(\widehat\omega_{220},\widehat\omega_j)\in\mathcal S_j(p)
\]
and define the exact pair inverse
\begin{equation}\label{eq:appendixI-exact-pair-estimator}
\widetilde p_j:=W_j^{-1}(\widehat z)\in\mathcal C_j(p).
\end{equation}
Then $\widehat\omega_{220}\in V_{220}(p)$, the primary inverse
\begin{equation}\label{eq:appendixI-primary-estimator}
\widehat p:=F_{220}^{-1}(\widehat\omega_{220})
\end{equation}
is well-defined on the dominant-mode chart of Appendix~\ref{app:primary-inversion}, and one has the exact identity
\begin{equation}\label{eq:appendixI-exact-agreement}
\widetilde p_j=\widehat p.
\end{equation}
In particular, if
\begin{equation}\label{eq:appendixI-exact-primary-error}
\bigl|\widehat\omega_{220}-\omega_{220}(p)\bigr|\le \varepsilon_{220},
\end{equation}
then
\begin{equation}\label{eq:appendixI-exact-pair-parameter-bound}
\|\widetilde p_j-p\|_{\mathbb R^2}
=
\|\widehat p-p\|_{\mathbb R^2}
\le
L_{220}\,\varepsilon_{220}.
\end{equation}
\end{corollary}

\begin{proof}
Since $\widetilde p_j\in\mathcal C_j(p)$ and $\mathcal C_j(p)\subset B_{\mathbb R^2}(p,\rho_{220})\cap\Kdet$, we have
\[
\widehat\omega_{220}=\omega_{220}(\widetilde p_j)\in V_{220}(p),
\]
where $V_{220}(p)$ is the dominant local image defined in \eqref{eq:appendixH-primary-image}. Thus the primary inverse in \eqref{eq:appendixI-primary-estimator} is well-defined. Moreover,
\[
F_{220}(\widetilde p_j)=\omega_{220}(\widetilde p_j)=\widehat\omega_{220}=F_{220}(\widehat p).
\]
Both $\widetilde p_j$ and $\widehat p$ lie in the same dominant-mode chart, and Corollary~\ref{cor:appendixH-primary-complex} says that $F_{220}$ is injective there. Hence $\widetilde p_j=\widehat p$, proving \eqref{eq:appendixI-exact-agreement}. The parameter bound \eqref{eq:appendixI-exact-pair-parameter-bound} is then just \eqref{eq:appendixH-primary-parameter-error} from Corollary~\ref{cor:appendixH-primary-complex}.
\end{proof}

Corollary~\ref{cor:appendixI-exact-pair-agrees-primary} explains why exact joint inversion does not appear in the main theorem surface. On the common event-local chart, exact two-mode inversion does not improve or alter the deterministic remnant estimate; it collapses to the dominant-mode inverse. In practice the relevant case is the off-image situation, namely the size of the pairwise residual when the data do not lie exactly on the Kerr image.

\subsection{Distance to the local pair image}

Once the observed pair is allowed to lie off the exact Kerr image, the relevant object is its distance to the local image surface $\mathcal S_j(p)$.

\begin{definition}\label{def:appendixI-distance-to-image}
Fix $j\in\{221,440\}$, a base point $p\in\Kdet$, and observed data
\[
\widehat z=(\widehat\omega_{220},\widehat\omega_j)\in\mathbb C^2.
\]
For each $q\in\mathcal C_j(p)$ define the pair misfit
\begin{equation}\label{eq:appendixI-pair-misfit}
\mathfrak e_j(\widehat z;q):=\bigl\|\widehat z-W_j(q)\bigr\|_{\mathbb C^2}
=
\Bigl(|\widehat\omega_{220}-\omega_{220}(q)|^2+|\widehat\omega_j-\omega_j(q)|^2\Bigr)^{1/2},
\end{equation}
and the distance from $\widehat z$ to the local pair image
\begin{equation}\label{eq:appendixI-distance-to-image}
\mathrm{dist}_j(\widehat z;p):=\inf_{q\in\mathcal C_j(p)} \mathfrak e_j(\widehat z;q).
\end{equation}
\end{definition}

\begin{proposition}[Local pair-image failure certificate]\label{prop:appendixI-distance-certificate}
Fix $j\in\{221,440\}$, a base point $p\in\Kdet$, and observed data $\widehat z=(\widehat\omega_{220},\widehat\omega_j)\in\mathbb C^2$. If there exists a point $q\in\mathcal C_j(p)$ such that
\begin{equation}\label{eq:appendixI-distance-hypothesis}
\mathfrak e_j(\widehat z;q)\le \varepsilon_{\mathrm{pair}},
\end{equation}
then
\begin{equation}\label{eq:appendixI-distance-bound}
\mathrm{dist}_j(\widehat z;p)\le \varepsilon_{\mathrm{pair}}.
\end{equation}
Conversely, if
\begin{equation}\label{eq:appendixI-distance-violation}
\mathrm{dist}_j(\widehat z;p)>\varepsilon_{\mathrm{pair}},
\end{equation}
then there is no $q\in\mathcal C_j(p)$ for which \eqref{eq:appendixI-distance-hypothesis} holds.
In particular, if there exists $q\in\mathcal C_j(p)$ such that
\begin{equation}\label{eq:appendixI-coordinate-budgets}
\bigl|\widehat\omega_{220}-\omega_{220}(q)\bigr|\le \varepsilon_{220},
\qquad
\bigl|\widehat\omega_j-\omega_j(q)\bigr|\le \varepsilon_j,
\end{equation}
then
\begin{equation}\label{eq:appendixI-coordinate-distance-bound}
\mathrm{dist}_j(\widehat z;p)
\le
\bigl(\varepsilon_{220}^2+\varepsilon_j^2\bigr)^{1/2}.
\end{equation}
\end{proposition}

\begin{proof}
If there exists $q\in\mathcal C_j(p)$ satisfying \eqref{eq:appendixI-distance-hypothesis}, then by definition of the infimum,
\[
\mathrm{dist}_j(\widehat z;p)
=\inf_{r\in\mathcal C_j(p)}\mathfrak e_j(\widehat z;r)
\le
\mathfrak e_j(\widehat z;q)
\le \varepsilon_{\mathrm{pair}},
\]
which is \eqref{eq:appendixI-distance-bound}. The converse is the contrapositive of the same implication. Finally, if \eqref{eq:appendixI-coordinate-budgets} holds, then by \eqref{eq:appendixI-pair-misfit}
\[
\mathfrak e_j(\widehat z;q)
=
\Bigl(|\widehat\omega_{220}-\omega_{220}(q)|^2+|\widehat\omega_j-\omega_j(q)|^2\Bigr)^{1/2}
\le
\bigl(\varepsilon_{220}^2+\varepsilon_j^2\bigr)^{1/2},
\]
so \eqref{eq:appendixI-coordinate-distance-bound} follows from the first part.
\end{proof}

Proposition~\ref{prop:appendixI-distance-certificate} is the pair-space analogue of the auxiliary failure certificate in Appendix~\ref{app:primary-inversion}. It reformulates common-remnant compatibility as a geometric statement about the distance of the observed pair to the local Kerr surface in $\mathbb C^2$.

\begin{theorem}[Projected pair estimator]\label{thm:appendixI-projected-estimator}
Fix $j\in\{221,440\}$, a base point $p\in\Kdet$, and observed data $\widehat z\in\mathbb C^2$. There exists at least one minimizer
\begin{equation}\label{eq:appendixI-projected-estimator-def}
\widehat p_j^{\mathrm{proj}}
\in
\operatorname*{argmin}_{q\in\mathcal C_j(p)} \mathfrak e_j(\widehat z;q).
\end{equation}
If, for the same base point $p$, one has
\begin{equation}\label{eq:appendixI-projected-hypothesis}
\mathfrak e_j(\widehat z;p)\le \varepsilon_{\mathrm{pair}},
\end{equation}
then every such minimizer obeys the deterministic parameter bound
\begin{equation}\label{eq:appendixI-projected-parameter-bound}
\bigl\|\widehat p_j^{\mathrm{proj}}-p\bigr\|_{\mathbb R^2}
\le
2L_{(220,j)}\,\varepsilon_{\mathrm{pair}}.
\end{equation}
In particular, if the coordinate bounds
\begin{equation}\label{eq:appendixI-projected-coordinate-hypothesis}
\bigl|\widehat\omega_{220}-\omega_{220}(p)\bigr|\le \varepsilon_{220},
\qquad
\bigl|\widehat\omega_j-\omega_j(p)\bigr|\le \varepsilon_j
\end{equation}
hold, then
\begin{equation}\label{eq:appendixI-projected-coordinate-bound}
\bigl\|\widehat p_j^{\mathrm{proj}}-p\bigr\|_{\mathbb R^2}
\le
2L_{(220,j)}\,\bigl(\varepsilon_{220}^2+\varepsilon_j^2\bigr)^{1/2}.
\end{equation}
\end{theorem}

\begin{proof}
The chart $\mathcal C_j(p)$ is compact and the map $q\mapsto\mathfrak e_j(\widehat z;q)$ is continuous, so the minimum is attained, proving the existence of \eqref{eq:appendixI-projected-estimator-def}. Let $\widehat p_j^{\mathrm{proj}}$ be any minimizer. By minimality and \eqref{eq:appendixI-projected-hypothesis},
\begin{equation}\label{eq:appendixI-projected-minimality}
\mathfrak e_j\bigl(\widehat z;\widehat p_j^{\mathrm{proj}}\bigr)
\le
\mathfrak e_j(\widehat z;p)
\le \varepsilon_{\mathrm{pair}}.
\end{equation}
Using the triangle inequality in $\mathbb C^2$, we obtain
\begin{align*}
\bigl\|W_j(\widehat p_j^{\mathrm{proj}})-W_j(p)\bigr\|_{\mathbb C^2}
&\le
\bigl\|\widehat z-W_j(\widehat p_j^{\mathrm{proj}})\bigr\|_{\mathbb C^2}
+
\bigl\|\widehat z-W_j(p)\bigr\|_{\mathbb C^2} \\
&=
\mathfrak e_j\bigl(\widehat z;\widehat p_j^{\mathrm{proj}}\bigr)
+
\mathfrak e_j(\widehat z;p)
\le 2\varepsilon_{\mathrm{pair}}.
\end{align*}
Since both $p$ and $\widehat p_j^{\mathrm{proj}}$ lie in the same chart $\mathcal C_j(p)$, the inverse Lipschitz estimate \eqref{eq:appendixI-pair-inverse-Lipschitz} from Theorem~\ref{thm:appendixI-pair-inverse} applies and yields
\[
\bigl\|\widehat p_j^{\mathrm{proj}}-p\bigr\|_{\mathbb R^2}
\le
L_{(220,j)}\,\bigl\|W_j(\widehat p_j^{\mathrm{proj}})-W_j(p)\bigr\|_{\mathbb C^2}
\le 2L_{(220,j)}\,\varepsilon_{\mathrm{pair}},
\]
which proves \eqref{eq:appendixI-projected-parameter-bound}. If \eqref{eq:appendixI-projected-coordinate-hypothesis} holds, then
\[
\mathfrak e_j(\widehat z;p)
\le
\bigl(\varepsilon_{220}^2+\varepsilon_j^2\bigr)^{1/2},
\]
so \eqref{eq:appendixI-projected-coordinate-bound} follows from \eqref{eq:appendixI-projected-parameter-bound}.
\end{proof}

The factor of $2$ in Theorem~\ref{thm:appendixI-projected-estimator} is the price of projecting off-image data back to the local Kerr surface. It disappears only in the exact-image setting, where the observed pair already lies on $\mathcal S_j(p)$ and Corollary~\ref{cor:appendixI-exact-pair-agrees-primary} applies.

\subsection{Comparison with the primary estimate}

The main theorem surface is expressed in terms of the primary estimate from $220$ and the auxiliary residual from $221$ or $440$. The projected pair estimator can be compared directly to that primary estimate.

\begin{proposition}[Projected pair inversion is controlled by the auxiliary residual]\label{prop:appendixI-projection-vs-primary}
Fix $j\in\{221,440\}$ and $p\in\Kdet$. Let observed data
\[
\widehat z=(\widehat\omega_{220},\widehat\omega_j)\in\mathbb C^2
\]
be such that $\widehat\omega_{220}\in F_{220}\bigl(\mathcal C_j(p)\bigr)$, and define the primary estimate
\begin{equation}\label{eq:appendixI-primary-estimate-on-common-chart}
\widehat p:=F_{220}^{-1}(\widehat\omega_{220})\in\mathcal C_j(p).
\end{equation}
Let $\widehat p_j^{\mathrm{proj}}$ be any projected pair estimator from \eqref{eq:appendixI-projected-estimator-def}. Then
\begin{equation}\label{eq:appendixI-projection-vs-primary-bound}
\bigl\|\widehat p_j^{\mathrm{proj}}-\widehat p\bigr\|_{\mathbb R^2}
\le
2L_{(220,j)}\,R_j\bigl(\widehat p,\widehat\omega_j\bigr),
\end{equation}
where the residual $R_j$ is defined in \eqref{eq:appendixH-aux-residual-def}. In particular, if
\begin{equation}\label{eq:appendixI-primary-exact-residual-zero}
R_j\bigl(\widehat p,\widehat\omega_j\bigr)=0,
\end{equation}
then
\begin{equation}\label{eq:appendixI-primary-projection-equality}
\widehat p_j^{\mathrm{proj}}=\widehat p.
\end{equation}
\end{proposition}

\begin{proof}
Since $\widehat p\in\mathcal C_j(p)$, it is an admissible competitor in the minimization problem defining $\widehat p_j^{\mathrm{proj}}$. Therefore
\begin{align*}
\bigl\|\widehat z-W_j(\widehat p_j^{\mathrm{proj}})\bigr\|_{\mathbb C^2}
&\le
\bigl\|\widehat z-W_j(\widehat p)\bigr\|_{\mathbb C^2} \\
&=
\Bigl(|\widehat\omega_{220}-\omega_{220}(\widehat p)|^2+|\widehat\omega_j-\omega_j(\widehat p)|^2\Bigr)^{1/2} \\
&=
\Bigl(0^2+R_j\bigl(\widehat p,\widehat\omega_j\bigr)^2\Bigr)^{1/2}
=
R_j\bigl(\widehat p,\widehat\omega_j\bigr),
\end{align*}
because $F_{220}(\widehat p)=\widehat\omega_{220}$. Applying the triangle inequality in $\mathbb C^2$ then gives
\begin{align*}
\bigl\|W_j(\widehat p_j^{\mathrm{proj}})-W_j(\widehat p)\bigr\|_{\mathbb C^2}
&\le
\bigl\|\widehat z-W_j(\widehat p_j^{\mathrm{proj}})\bigr\|_{\mathbb C^2}
+
\bigl\|\widehat z-W_j(\widehat p)\bigr\|_{\mathbb C^2} \\
&\le
2R_j\bigl(\widehat p,\widehat\omega_j\bigr).
\end{align*}
Both points belong to the same chart $\mathcal C_j(p)$, so Theorem~\ref{thm:appendixI-pair-inverse} yields
\[
\bigl\|\widehat p_j^{\mathrm{proj}}-\widehat p\bigr\|_{\mathbb R^2}
\le
L_{(220,j)}\,\bigl\|W_j(\widehat p_j^{\mathrm{proj}})-W_j(\widehat p)\bigr\|_{\mathbb C^2}
\le
2L_{(220,j)}\,R_j\bigl(\widehat p,\widehat\omega_j\bigr),
\]
which is \eqref{eq:appendixI-projection-vs-primary-bound}. If the residual vanishes, the same argument gives
\[
\bigl\|\widehat p_j^{\mathrm{proj}}-\widehat p\bigr\|_{\mathbb R^2}=0,
\]
proving \eqref{eq:appendixI-primary-projection-equality}.
\end{proof}

The next corollary turns the abstract comparison bound into a deterministic error statement using Appendix~\ref{app:primary-inversion}.

\begin{corollary}[Projected pair inversion under common-remnant error budgets]\label{cor:appendixI-projection-budget}
Fix $j\in\{221,440\}$ and $p\in\Kdet$. Assume the hypotheses of Theorem~\ref{thm:appendixH-aux-propagation} hold, and assume moreover that the primary estimate
\[
\widehat p=F_{220}^{-1}(\widehat\omega_{220})
\]
lies in the common chart $\mathcal C_j(p)$. Let $\widehat p_j^{\mathrm{proj}}$ be any projected pair estimator from \eqref{eq:appendixI-projected-estimator-def}. Then
\begin{equation}\label{eq:appendixI-projection-budget-bound}
\bigl\|\widehat p_j^{\mathrm{proj}}-\widehat p\bigr\|_{\mathbb R^2}
\le
2L_{(220,j)}\,\tau_j(\varepsilon_{220},\varepsilon_j)
=
2L_{(220,j)}\,\bigl(\varepsilon_j+L_jL_{220}\,\varepsilon_{220}\bigr),
\end{equation}
where $\tau_j$ is the deterministic auxiliary tolerance defined in \eqref{eq:appendixH-aux-threshold}. Since $L_{(220,j)}\le L_{220}$, one also has the weaker but simpler estimate
\begin{equation}\label{eq:appendixI-projection-budget-bound-simpler}
\bigl\|\widehat p_j^{\mathrm{proj}}-\widehat p\bigr\|_{\mathbb R^2}
\le
2L_{220}\,\bigl(\varepsilon_j+L_jL_{220}\,\varepsilon_{220}\bigr).
\end{equation}
\end{corollary}

\begin{proof}
By Theorem~\ref{thm:appendixH-aux-propagation}, the residual of the primary estimate satisfies
\[
R_j\bigl(\widehat p,\widehat\omega_j\bigr)
\le
\tau_j(\varepsilon_{220},\varepsilon_j)
=
\varepsilon_j+L_jL_{220}\,\varepsilon_{220}.
\]
Substituting this bound into Proposition~\ref{prop:appendixI-projection-vs-primary} proves \eqref{eq:appendixI-projection-budget-bound}. The simpler estimate \eqref{eq:appendixI-projection-budget-bound-simpler} follows from \eqref{eq:appendixI-pair-constant-dominates}.
\end{proof}

Exact two-mode inversion is locally redundant because it agrees with the dominant-mode inverse on the common chart. Once the data move off the exact Kerr image, however, the projected pair estimator becomes meaningful, and its deviation from the primary estimate is controlled precisely by the same auxiliary residual that already appears in the trust-region test. The analysis therefore keeps the $220$ inversion and the auxiliary consistency checks separate, so that the place where off-image or model-mismatch effects enter remains explicit.

\section{Synthetic bank construction and mismatch calibration}\label{app:synthetic-bank}

This appendix fixes the event-local synthetic bank used in Section~\ref{sec:numerical-calibration}. The bank is event local. It converts the GW250114 public products into explicit window-dependent calibration radii and uses only the public posterior samples, the H1/L1 strain, and the Kerr QNM tables for \(220\), \(221\), and \(440\).

\begin{definition}\label{def:appendixJ-source-families}
The public source-family set is
\[
\mathfrak F_{\mathrm{bank}}
=
\{\mathrm{NRSur7dq4},\mathrm{PhenomXO4a},\mathrm{PhenomXPHM},\mathrm{SEOBNRv5PHM}\}.
\]
For each \(\mathfrak f\in\mathfrak F_{\mathrm{bank}}\), let \(\mathcal P_{\mathfrak f}\) denote the detector-frame posterior sample set stored in the corresponding HDF5 file. Let \(\mathcal P_{\mathfrak f}^{\mathrm{loc}}\subset \mathcal P_{\mathfrak f}\) be the restriction to the compact event-local detector-frame box \(\Kdet\).
\end{definition}

The posterior files carry no direct overtone amplitudes, so the synthetic bank supplements the remnant samples with an explicit finite subdominant-mode audit box. These amplitudes are treated only as audit variables, not as inferred quantities. We therefore fix three event-local configurations,
\[
\mathcal C_{\mathrm{amp}}
=
\{\mathrm{mild},\mathrm{medium},\mathrm{strong}\},
\]
with amplitude ratios
\[
(a_{221},a_{440})
=
(0.25,0.05),\ (0.50,0.10),\ (0.75,0.20),
\]
and fixed phases
\[
(\phi_{221},\phi_{440})
=
(0,0),\ (\pi/4,0),\ (\pi/2,\pi/4).
\]
These values are not used as astrophysical claims. They define a compact event-local audit box that spans weak, moderate, and strong subdominant content.

\begin{definition}\label{def:appendixJ-bank-design}
For each public source family \(\mathfrak f\), choose the three representative detector-frame posterior samples at the 20th, 50th, and 80th order-statistic locations in \(M_f^{\mathrm{det}}\) inside \(\mathcal P_{\mathfrak f}^{\mathrm{loc}}\). The resulting set of remnant anchors is denoted by \(\mathfrak B_{\mathrm{rem}}\). The zero-noise bank is
\begin{equation}\label{eq:appendixJ-zero-noise-bank}
\mathfrak B_0
=
\mathfrak B_{\mathrm{rem}}\times \mathcal C_{\mathrm{amp}},
\end{equation}
so \(|\mathfrak B_0|=36\).
\end{definition}

For \(b\in\mathfrak B_0\), write
\[
p_b=(M_b,\chi_b)\in\Kdet
\]
for its detector-frame remnant parameter. The reference synthetic signal is generated in the Kerr family
\begin{equation}\label{eq:appendixJ-reference-family}
\mathcal M^{\sharp}:=\Mtwo=\{220,221,440\}.
\end{equation}
The fitted families remain \(\Mzero\), \(\Mone\), and \(\Mtwo\). For a fixed window \(W=(t_0,T)\), each bank element \(b\) determines a network signal
\begin{equation}\label{eq:appendixJ-bank-signal}
y_b^{0,W}
=
\sum_{j\in\mathcal M^{\sharp}}
A_{b,j}^{H1} e^{-i\omega_j(p_b)u}
\oplus
\sum_{j\in\mathcal M^{\sharp}}
A_{b,j}^{L1} e^{-i\omega_j(p_b)u},
\qquad
u\in[0,T],
\end{equation}
where the detector weights are fixed by the H1/L1 optimal-SNR split carried by the corresponding posterior sample, and the overall normalization is chosen so that the discrete network norm over the reference \(24M_\star\) window matches the network optimal SNR of that sample.

\begin{definition}\label{def:appendixJ-model-subspace}
For \(\mathcal A\subset \mathcal M^{\sharp}\), \(p\in\Kdet\), and \(W=(t_0,T)\), let
\begin{equation}\label{eq:appendixJ-model-subspace}
\mathcal S_{\mathcal A}(p;W)
=
\left\{
\sum_{j\in\mathcal A}
B_j^{H1}e^{-i\omega_j(p)u}
\oplus
\sum_{j\in\mathcal A}
B_j^{L1}e^{-i\omega_j(p)u}
:
B_j^{H1},B_j^{L1}\in\mathbb C
\right\}
\subset \mathcal H_W
\end{equation}
be the network model space over the window \(W\).
\end{definition}

Because \(\Mzero\subset\Mone\subset\Mtwo=\mathcal M^{\sharp}\), every fitted family is a literal subspace truncation of the reference family. The omitted-linear content of \(y_b^{0,W}\) relative to \(\mathcal A\) is therefore
\begin{equation}\label{eq:appendixJ-tail-def}
r_{b,\mathcal A}^{\mathrm{tail},W}
=
\bigl(\Pi_{\mathcal M^{\sharp},p_b}^{W}-\Pi_{\mathcal A,p_b}^{W}\bigr)y_b^{0,W},
\end{equation}
and the implementation sets the non-Kerr mismatch term to zero,
\begin{equation}\label{eq:appendixJ-mm-def}
r_b^{\mathrm{mm},W}=0,
\end{equation}
because the calibration bank is generated inside the same Kerr family \(\mathcal M^{\sharp}\) that is used as the synthetic reference. This choice isolates finite-window truncation and colored-noise effects before any wider NR/CCE mismatch budget is introduced.

\begin{proposition}\label{prop:appendixJ-orthogonal-split}
For every \(b\in\mathfrak B_0\), every fixed window \(W\), and every fitted family \(\mathcal A\subset\mathcal M^{\sharp}\),
\[
y_b^{0,W}
=
\Pi_{\mathcal A,p_b}^{W}y_b^{0,W}
+
r_{b,\mathcal A}^{\mathrm{tail},W},
\]
with
\[
\left\langle
\Pi_{\mathcal A,p_b}^{W}y_b^{0,W},
r_{b,\mathcal A}^{\mathrm{tail},W}
\right\rangle_{\mathcal H_W}=0.
\]
Consequently,
\begin{equation}\label{eq:appendixJ-pythagoras}
\|y_b^{0,W}\|_{\mathcal H_W}^2
=
\|\Pi_{\mathcal A,p_b}^{W}y_b^{0,W}\|_{\mathcal H_W}^2
+
\|r_{b,\mathcal A}^{\mathrm{tail},W}\|_{\mathcal H_W}^2.
\end{equation}
\end{proposition}

\begin{proof}
This is immediate from the orthogonal projector identity and the fact that \(\mathcal S_{\mathcal A}(p_b;W)\subset \mathcal S_{\mathcal M^{\sharp}}(p_b;W)\).
\end{proof}

The bank does not calibrate every possible window. It calibrates a finite grid. Here,
\[
\mathfrak W_{\mathrm{cal}}
=
\{(t_0,24M_\star): t_0/M_\star\in\{3.0,3.5,\dots,11.0\}\}.
\]
This is the exact grid used to generate the new figures in Section~\ref{sec:numerical-calibration}. The synthetic map is assembled pointwise from this finite collection of fixed windows.

\subsection{Zero-noise radii}

For \(b\in\mathfrak B_0\), \(W\in\mathfrak W_{\mathrm{cal}}\), and \(\mathcal A\in\{\Mzero,\Mone,\Mtwo\}\), the default extraction pipeline returns a fitted detector-frame remnant
\[
\widehat p_{b,\mathcal A}^{\,0}(W)
=
\bigl(\widehat M_{b,\mathcal A}^{\,0}(W),\widehat\chi_{b,\mathcal A}^{\,0}(W)\bigr).
\]
The corresponding zero-noise remnant errors are
\[
R_{M,b}^{\,0,\mathcal A,W}
=
\bigl|\widehat M_{b,\mathcal A}^{\,0}(W)-M_b\bigr|,
\qquad
R_{\chi,b}^{\,0,\mathcal A,W}
=
\bigl|\widehat\chi_{b,\mathcal A}^{\,0}(W)-\chi_b\bigr|.
\]
The bankwise deterministic radii are the maxima
\begin{equation}\label{eq:appendixJ-zero-noise-bias-bound-pointwise}
\overline R_{M}^{\,0,\mathcal A,W}
=
\max_{b\in\mathfrak B_0}R_{M,b}^{\,0,\mathcal A,W},
\qquad
\overline R_{\chi}^{\,0,\mathcal A,W}
=
\max_{b\in\mathfrak B_0}R_{\chi,b}^{\,0,\mathcal A,W}.
\end{equation}

In this bank these maxima are explicitly computable. The strongest deterministic hierarchy is already visible at this stage. For the full synthetic family \(\Mtwo\), the zero-noise mass radius stays at the level of the grid and local-chart discretization, while for \(\Mzero\) and \(\Mone\) it is dominated by omitted \(221\) and \(440\) content. Figure~\ref{fig:sec7-zero-noise-score} plots the public-envelope normalization of \eqref{eq:appendixJ-zero-noise-bias-bound-pointwise}.

\subsection{Colored-noise quantiles}

The H1/L1 strain provides the colored-noise component. We extract eight off-source detector chunks, apply the same \(50\)–\(500\) Hz stabilization used in the real-event analysis, and normalize each chunk to unit sample standard deviation on its own support. For a fixed \(b\), \(W\), and noise draw index \(r\), the noisy bank element is
\begin{equation}\label{eq:appendixJ-noisy-bank-element}
y_{b,r}^{W}
=
y_b^{0,W}
+
n_r^{H1,W}\oplus n_r^{L1,W}.
\end{equation}
Applying the same extraction pipeline yields
\[
\widehat p_{b,\mathcal A}^{\,r}(W)
=
\bigl(\widehat M_{b,\mathcal A}^{\,r}(W),\widehat\chi_{b,\mathcal A}^{\,r}(W)\bigr)
\]
and the colored-noise increments
\[
\Delta_{M,b,r}^{\mathcal A,W}
=
\bigl|\widehat M_{b,\mathcal A}^{\,r}(W)-\widehat M_{b,\mathcal A}^{\,0}(W)\bigr|,
\qquad
\Delta_{\chi,b,r}^{\mathcal A,W}
=
\bigl|\widehat \chi_{b,\mathcal A}^{\,r}(W)-\widehat \chi_{b,\mathcal A}^{\,0}(W)\bigr|.
\]

For each fixed \(W\) and \(\mathcal A\), the calibration therefore produces \(N=288\) noisy realizations. Let \(Q_M^{\mathcal A,W}(q)\) and \(Q_\chi^{\mathcal A,W}(q)\) denote the left \(q\)-quantiles of the pooled increments \(\Delta_{M,b,r}^{\mathcal A,W}\) and \(\Delta_{\chi,b,r}^{\mathcal A,W}\). Their empirical versions are built in the obvious way from the \(288\) samples.

\begin{theorem}[Per-window upper-confidence quantiles]\label{thm:appendixJ-dkw}
Fix a fitted family \(\mathcal A\in\{\Mzero,\Mone,\Mtwo\}\), a window \(W\in\mathfrak W_{\mathrm{cal}}\), a nominal quantile level \(0<q<1\), and a confidence level \(1-\alpha\in(0,1)\). Let
\begin{equation}\label{eq:appendixJ-dkw-delta}
\delta_N(\alpha)
=
\sqrt{\frac{\log(2/\alpha)}{2N}},
\qquad N=288.
\end{equation}
Then, simultaneously for the mass and spin increment distributions,
\begin{equation}\label{eq:appendixJ-dkw-quantile-bound}
Q_M^{\mathcal A,W}(q)
\le
\widehat Q_M^{\mathcal A,W}\!\bigl(q+\delta_N(\alpha)\bigr),
\qquad
Q_\chi^{\mathcal A,W}(q)
\le
\widehat Q_\chi^{\mathcal A,W}\!\bigl(q+\delta_N(\alpha)\bigr)
\end{equation}
with probability at least \(1-\alpha\).
\end{theorem}

\begin{proof}
Apply the Dvoretzky--Kiefer--Wolfowitz inequality to the empirical distribution functions of the pooled mass and spin increment samples separately, and then invert the empirical distribution functions at the shifted level \(q+\delta_N(\alpha)\) \cite{DvoretzkyKieferWolfowitz1956,Massart1990DKW}. The statement is per fixed \(W\) and \(\mathcal A\); no union over \(|\mathfrak W_{\mathrm{cal}}|\) is claimed here.
\end{proof}

This theorem fixes the probability semantics. The synthetic map above is assembled from these fixed-window objects. It is not itself a single simultaneous probability statement over the entire start-time grid.

Define the high-confidence statistical radii
\begin{equation}\label{eq:appendixJ-bankwise-stat-envelope}
\overline R_{M}^{\,\mathrm{stat},\mathcal A,W}(q;\alpha)
=
\widehat Q_M^{\mathcal A,W}\!\bigl(q+\delta_N(\alpha)\bigr),
\qquad
\overline R_{\chi}^{\,\mathrm{stat},\mathcal A,W}(q;\alpha)
=
\widehat Q_\chi^{\mathcal A,W}\!\bigl(q+\delta_N(\alpha)\bigr).
\end{equation}

\subsection{Total calibrated radii}

The synthetic calibration enters the trust inequalities through remnant-space radii and their induced modewise envelopes. For each fixed \(W\) and \(\mathcal A\), define the total detector-frame remnant radii
\begin{equation}\label{eq:appendixJ-total-radius-def}
\overline R_{M}^{\,\mathrm{tot},\mathcal A,W}(q;\alpha)
=
\overline R_{M}^{\,0,\mathcal A,W}
+
\overline R_{M}^{\,\mathrm{stat},\mathcal A,W}(q;\alpha),
\end{equation}
\[
\overline R_{\chi}^{\,\mathrm{tot},\mathcal A,W}(q;\alpha)
=
\overline R_{\chi}^{\,0,\mathcal A,W}
+
\overline R_{\chi}^{\,\mathrm{stat},\mathcal A,W}(q;\alpha).
\]

The main-text figures report these radii through the normalized synthetic ratio
\[
\mathfrak S_{\mathcal A}(W)
=
\max\!\left\{
\frac{\overline R_{M}^{\,\mathrm{tot},\mathcal A,W}(q;\alpha)}{\Delta M_{\mathrm{pub}}/2},
\frac{\overline R_{\chi}^{\,\mathrm{tot},\mathcal A,W}(q;\alpha)}{\Delta \chi_{\mathrm{pub}}/2}
\right\},
\]
where \(\Delta M_{\mathrm{pub}}/2=1.315\,M_\odot\) and \(\Delta\chi_{\mathrm{pub}}/2=0.01877\) are the half-widths of the public detector-frame comparison envelope.

To reconnect this remnant-space calibration to the symbolic frequency decomposition used in Section~\ref{sec:trust-region}, let \(L_{220}^{W}\), \(L_{221}^{W}\), and \(L_{440}^{W}\) denote any detector-frame local forward Lipschitz constants for the Kerr maps \(G_{220}\), \(G_{221}\), and \(G_{440}\) on the event-local chart used at \(W\). Then the induced modewise calibrated radii are
\begin{align}
\label{eq:appendixJ-eps-tail-cal}
\varepsilon_{220}^{\mathrm{cal},\mathcal A,W}(q;\alpha)
&=
L_{220}^{W}\,
\bigl\|
(\overline R_{M}^{\,\mathrm{tot},\mathcal A,W}(q;\alpha),
 \overline R_{\chi}^{\,\mathrm{tot},\mathcal A,W}(q;\alpha))
\bigr\|_{\mathbb R^2},
\\
\label{eq:appendixJ-eps-mm-cal}
\varepsilon_{221}^{\mathrm{cal},\mathcal A,W}(q;\alpha)
&=
L_{221}^{W}\,
\bigl\|
(\overline R_{M}^{\,\mathrm{tot},\mathcal A,W}(q;\alpha),
 \overline R_{\chi}^{\,\mathrm{tot},\mathcal A,W}(q;\alpha))
\bigr\|_{\mathbb R^2},
\\
\label{eq:appendixJ-eps-stat-cal}
\varepsilon_{440}^{\mathrm{cal},\mathcal A,W}(q;\alpha)
&=
L_{440}^{W}\,
\bigl\|
(\overline R_{M}^{\,\mathrm{tot},\mathcal A,W}(q;\alpha),
 \overline R_{\chi}^{\,\mathrm{tot},\mathcal A,W}(q;\alpha))
\bigr\|_{\mathbb R^2}.
\end{align}
Collecting them yields
\begin{equation}\label{eq:appendixJ-total-calibrated-radius}
\varepsilon_{j}^{\mathrm{cal},\mathcal A,W}(q;\alpha),
\qquad
j\in\{220,221,440\}.
\end{equation}
This is the object referenced symbolically in Section~\ref{sec:trust-region}. The numerically instantiated calibration is therefore remnant local first and modewise second.

Figure~\ref{fig:appJ-anchor-decomposition} separates the zero-noise and colored-noise parts at the anchor windows \(t_0/M_\star\in\{3,6,9,11\}\). The plot confirms the qualitative picture already visible above: the earliest window is expensive in both pieces, while the later windows are controlled mainly by the noise floor, and the full family \(\Mtwo\) remains strictly below \(\Mzero\) and \(\Mone\) at all four anchor windows.

\begin{figure}[t]
\centering
\includegraphics[width=0.74\textwidth]{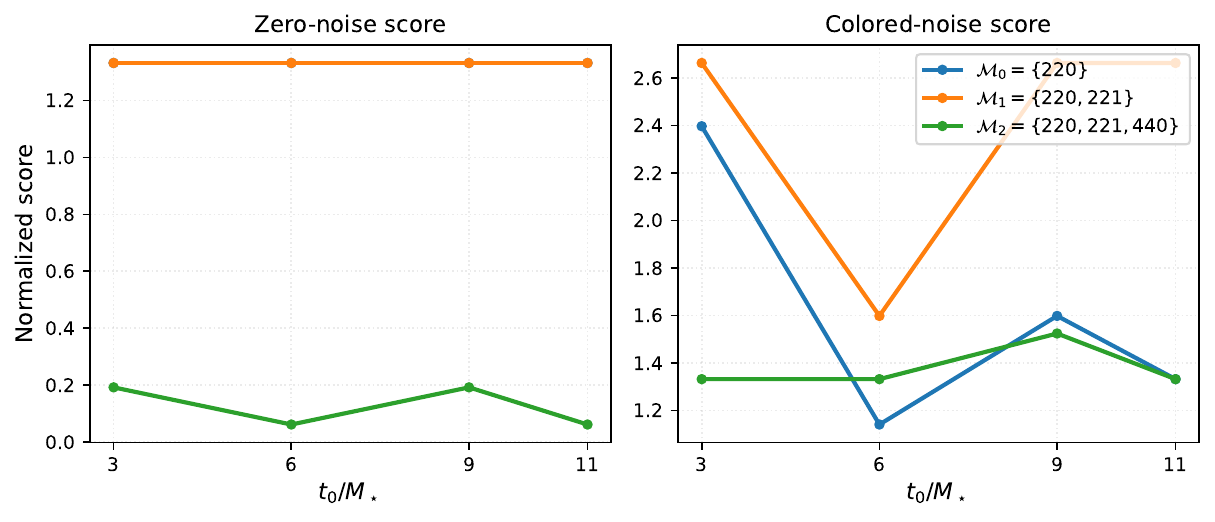}
\caption{Zero-noise and colored-noise contributions at the four anchor windows \(t_0/M_\star\in\{3,6,9,11\}\).}
\label{fig:appJ-anchor-decomposition}
\end{figure}

\providecommand{\ScanTau}{\mathfrak W_{\mathrm{scan}}^{\tau}}
\providecommand{\LenSet}{\mathfrak L}
\providecommand{\PsdSet}{\mathfrak P}
\providecommand{\TapSet}{\mathfrak T}
\providecommand{\GridSet}{\mathfrak G}
\providecommand{\PEset}{\mathfrak R_{\mathrm{PE}}}
\providecommand{\NuSet}{\mathfrak N}
\providecommand{\HullPE}{\mathcal H_{\mathrm{PE}}}
\providecommand{\Mstar}{M_{\bullet}}
\providecommand{\BaseB}{\mathfrak B^{\mathrm{base}}}

\section{GW250114 robustness checks}\label{app:gw250114-robustness}

This appendix fixes the finite real-data rerun suite used to test how the GW250114 event-level diagnostics respond to small changes in preprocessing and to the fixed nuisance dictionary of Section~\ref{sec:model-hierarchy}. The reruns do not introduce a new discovery model. It reruns the same event under an explicitly listed family of PSD prescriptions, taper fractions, and scan grids, while keeping the direct-wave order \(\Ndir\), the left-edge envelope \(\Enu_W\), and the quadratic pair set \(\Pquad\) fixed.

\subsection{Fixed diagnostic classes and finite specification sets}

The nuisance dictionary is exactly the one fixed in Definition~\ref{def:sec3-nuisance-families}. In particular,
\[
\Ndir=4,
\qquad
\Enu_W(u)=\exp\!\Bigl[-16\Bigl(\frac{u}{T}\Bigr)^2\Bigr],
\qquad
\Pquad=\{(220,220)\}.
\]
No additional direct-wave atom and no additional quadratic pair is introduced anywhere in the robustness checks.

For the real-data reruns we fix the finite preprocessing sets
\begin{equation}\label{eq:appK-psd-set}
\PsdSet=
\{4\mathrm{sym},8\mathrm{sym},16\mathrm{sym},8\mathrm{left},8\mathrm{right}\},
\end{equation}
\begin{equation}\label{eq:appK-taper-set}
\TapSet=
\{0.05,0.10,0.20\},
\end{equation}
and the finite scan-grid set
\begin{equation}\label{eq:appK-grid-set}
\GridSet=
\{g_{\mathrm{anc}},g_{\mathrm{dense}}\},
\end{equation}
where
\begin{equation}\label{eq:appK-anchor-grid}
g_{\mathrm{anc}}=\{3,6,9,11\},
\end{equation}
and
\begin{equation}\label{eq:appK-dense-grid}
g_{\mathrm{dense}}=\{0,0.25,0.50,\dots,15\}.
\end{equation}
Both grids use the fixed detector-frame window length \(\Theta=24\). The nuisance labels are
\begin{equation}\label{eq:appK-nuisance-set}
\NuSet=\{\varnothing,\mathrm{dir},\mathrm{quad},\mathrm{dir+quad}\}.
\end{equation}

For every baseline family \(\mathcal A\subset\Mref\) with \(220\in\mathcal A\), every nuisance label \(\nu\in\NuSet\), and every grid label \(g\in\GridSet\), the event-level rerun family is indexed by
\begin{equation}\label{eq:appK-rerun-set}
\mathfrak B_{\mathcal A}^{\nu}(g)
:=
\PsdSet\times\TapSet\times\{g\}\times\{\nu\}.
\end{equation}
The corresponding baseline specification set is
\begin{equation}\label{eq:appK-base-set}
\mathfrak B_{\mathcal A}^{\mathrm{base}}(g)
:=
\PsdSet\times\TapSet\times\{g\}.
\end{equation}
The reference specification used for the baseline curves in Section~\ref{sec:gw250114-empirical} is
\begin{equation}\label{eq:appK-ref-spec}
\sigma_{\mathrm{ref}}:=(8\mathrm{sym},0.10,g_{\mathrm{dense}}).
\end{equation}
The anchor reruns displayed below restricts the grid label to \(g_{\mathrm{anc}}\), so each public anchor window is rerun \(5\times 3=15\) times.

\begin{definition}[Allowed nuisance map]\label{def:appendixK-allowed-map}
For every baseline family \(\mathcal A\in\{\Mzero,\Mone,\Mtwo\}\), the only admissible nuisance extensions in the rerun suite are the four spaces \(\mathcal S_{\mathcal A}^{\nu}(p;W)\) with \(\nu\in\NuSet\) from Definition~\ref{def:sec3-nuisance-families}. The values of \(\Ndir\), \(\Enu_W\), and \(\Pquad\) are fixed globally and are never retuned across scan points, reruns, or baseline families.
\end{definition}

The combined class \(\nu=\mathrm{dir+quad}\) is computed throughout the present GW250114 release. The later normalized ledgers sometimes keep its appearance conditional only to allow comparison with reduced releases. In the actual rerun suite behind this manuscript, the combined nuisance diagnostic is always available.

\subsection{Specwise residuals and nuisance gains}

Let \(\Xi=(\tau_0,\Theta)\in\ScanTau\) be a master scan point, let \(g\in\GridSet\), and let
\[
\sigma=(P,\vartheta,g)\in\mathfrak B_{\mathcal A}^{\mathrm{base}}(g)
\]
be a baseline specification. Denote by \(W_g(\Xi)\) the detector-frame window induced by \(\Xi\) on grid \(g\), and by \(y^{\sigma,\Xi}\) the corresponding preprocessed H1/L1 data vector.

For \(\nu\in\NuSet\), define the specwise normalized residual
\begin{equation}\label{eq:appK-rho-nu}
\rho_\nu^{\sigma,\mathcal A}(\Xi)
:=
\inf_{p\in\Kdet}
\frac{
\dist_{\mathcal H_{W_g(\Xi)}}\!\bigl(
y^{\sigma,\Xi},
\mathcal S_{\mathcal A}^{\nu}(p;W_g(\Xi))
\bigr)
}{
\|y^{\sigma,\Xi}\|_{\mathcal H_{W_g(\Xi)}} }.
\end{equation}
For \(\nu\in\NuSet\setminus\{\varnothing\}\), the corresponding nuisance gain is
\begin{equation}\label{eq:appK-gamma-spec}
\Gamma_\nu^{\sigma,\mathcal A}(\Xi)
:=
\rho_{\varnothing}^{\sigma,\mathcal A}(\Xi)
-
\rho_\nu^{\sigma,\mathcal A}(\Xi).
\end{equation}
The robust gain entering the event-level inequalities is the worst-case gain across the finite rerun family,
\begin{equation}\label{eq:appK-gamma-agg}
\Gamma_\nu^{\mathcal A}(\Xi)
:=
\sup_{\sigma\in\mathfrak B_{\mathcal A}^{\mathrm{base}}(g_{\mathrm{dense}})}
\Gamma_\nu^{\sigma,\mathcal A}(\Xi),
\qquad
\nu\in\NuSet\setminus\{\varnothing\},
\end{equation}
while the anchor-display median is
\begin{equation}\label{eq:appK-gamma-med}
\operatorname{med}\Gamma_\nu^{\mathcal A}(\tau_0)
:=
\operatorname{median}
\Bigl\{
\Gamma_\nu^{\sigma,\mathcal A}\bigl((\tau_0,24)\bigr)
:
\sigma\in\mathfrak B_{\mathcal A}^{\mathrm{base}}(g_{\mathrm{anc}})
\Bigr\}.
\end{equation}

\begin{proposition}[Nuisance gains are monotone residual improvements]\label{prop:appendixK-gain-monotonicity}
For every admissible \(\Xi\), every baseline family \(\mathcal A\), every baseline specification \(\sigma\), and every \(\nu\in\NuSet\setminus\{\varnothing\}\),
\[
\Gamma_\nu^{\sigma,\mathcal A}(\Xi)\ge 0.
\]
Consequently \(\Gamma_\nu^{\mathcal A}(\Xi)\ge 0\).
\end{proposition}

\begin{proof}
For fixed \(p\) and \(W_g(\Xi)\), one has
\[
\mathcal S_{\mathcal A}^{\varnothing}(p;W_g(\Xi))
\subset
\mathcal S_{\mathcal A}^{\nu}(p;W_g(\Xi)).
\]
By Proposition~\ref{prop:sec3-residual-monotonicity}, enlarging the fitted space can only decrease the distance of the data vector to that space. Hence
\[
\rho_\nu^{\sigma,\mathcal A}(\Xi)\le \rho_{\varnothing}^{\sigma,\mathcal A}(\Xi),
\]
which is exactly \(\Gamma_\nu^{\sigma,\mathcal A}(\Xi)\ge 0\). Taking the supremum over the finite set \(\mathfrak B_{\mathcal A}^{\mathrm{base}}(g_{\mathrm{dense}})\) proves the second claim.
\end{proof}

The supremum in \eqref{eq:appK-gamma-agg} is the quantity that enters the event-level nuisance inequalities. The medians in \eqref{eq:appK-gamma-med} are purely descriptive summaries of the anchor reruns and should not be confused with acceptance thresholds.

\subsection{Compatibility with the numerical acceptance notation}

Later appendices summarize the event-level decision rule through normalized acceptance diagnostics. To keep the notation consistent, we record the corresponding baseline and event-level conditions here. Fix a baseline family \(\mathcal A\subset\mathcal M^{\sharp}\) with \(220\in\mathcal A\), a master scan point \(\Xi\in\ScanTau\), and a finite baseline specification set \(\mathfrak B_{\mathcal A}^{\mathrm{base}}(g_{\mathrm{dense}})\).

\begin{definition}[Robust baseline trust]\label{def:appendixK-robust-base-trust}
A master scan point \(\Xi\) belongs to \(\mathcal T_{\eta}^{\mathrm{rob,base}}(\mathcal A)\) if there exists \(\sigma\in\mathfrak B_{\mathcal A}^{\mathrm{base}}(g_{\mathrm{dense}})\) such that
\begin{equation}\label{eq:appendixK-rob-sep}
\overline\varepsilon_{220}^{\mathcal A}(\Xi)+\Delta_{220}^{\mathcal A}(\Xi)\le \eta_{\mathrm{sep}},
\qquad
\delta_{\mathrm{iso}}\bigl(\widehat p^{\sigma,\mathcal A}(\Xi);\mathcal A\bigr)-L_{\mathrm{iso}}(\mathcal A)L_{220}\Delta_{220}^{\mathcal A}(\Xi)>\eta_{\mathrm{sep}},
\end{equation}
\begin{equation}\label{eq:appendixK-rob-param}
L_{220}\bigl(\overline\varepsilon_{220}^{\mathcal A}(\Xi)+\Delta_{220}^{\mathcal A}(\Xi)\bigr)\le \eta_{\mathrm p},
\end{equation}
\begin{equation}\label{eq:appendixK-rob-drift}
\overline D_{\Delta}^{\mathcal A}(\Xi)\le \eta_{\mathrm d},
\end{equation}
and, whenever the relevant auxiliary mode is included,
\begin{equation}\label{eq:appendixK-rob-221}
R_{221}^{\sigma,\mathcal A}(\Xi)+\Delta_{221}^{\mathcal A}(\Xi)+L_{221}L_{220}\Delta_{220}^{\mathcal A}(\Xi)\le \eta_{221},
\end{equation}
\begin{equation}\label{eq:appendixK-rob-440}
R_{440}^{\sigma,\mathcal A}(\Xi)+\Delta_{440}^{\mathcal A}(\Xi)+L_{440}L_{220}\Delta_{220}^{\mathcal A}(\Xi)\le \eta_{440}.
\end{equation}
\end{definition}

\begin{definition}[Event-level acceptance]\label{def:appendixK-event-acceptance}
A master scan point \(\Xi\) belongs to the numerical GW250114 acceptance set \(\mathcal T_{\eta,\eta_{\mathrm{PE}}}^{\mathrm{GW250114}}(\mathcal A)\) if \(\Xi\in\mathcal T_{\eta}^{\mathrm{rob,base}}(\mathcal A)\) and, in addition,
\begin{equation}\label{eq:appendixK-dpe}
d_{\mathrm{PE}}^{\mathcal A}(\Xi)\le \eta_{\mathrm{PE}},
\end{equation}
\begin{equation}\label{eq:appendixK-gamma-nu}
\Gamma_{\mathrm{dir}}^{\mathcal A}(\Xi)\le \etammstar(\Xi),
\qquad
\Gamma_{\mathrm{quad}}^{\mathcal A}(\Xi)\le \etammstar(\Xi),
\end{equation}
and, when the combined nuisance certificate is invoked,
\begin{equation}\label{eq:appendixK-gamma-fail}
\Gamma_{\mathrm{dir+quad}}^{\mathcal A}(\Xi)\le \etammstar(\Xi).
\end{equation}
The present release computes \(\Gamma_{\mathrm{dir+quad}}^{\mathcal A}\) at every scan point. The formal distinction is kept only because the minimal acceptance set already requires the two single-class inequalities, while the combined inequality is used as an additional failure certificate and as a stronger release option.
\end{definition}

\subsection{Observed anchor-window spreads}

The real-data anchor reruns is centered on the same public H1/L1 strain used in Section~\ref{sec:gw250114-empirical}. We keep the detector-frame window length fixed at \(\Theta=24\) and inspect the public anchor windows
\begin{equation}\label{eq:appK-anchor-set}
\tau_0\in\{3,6,9,11\}.
\end{equation}
Because the displayed reruns fix the grid label to \(g_{\mathrm{anc}}\), each anchor window is rerun fifteen times, once for every pair in \(\PsdSet\times\TapSet\). The labels have literal meaning. The symmetric prescriptions use both off-source wings in the \(1024\,\mathrm{s}\) working epoch, while the left and right prescriptions use only one wing with an eight-second Welch segment. The taper parameter is the Tukey fraction applied after window extraction. No other aspect of the problem changes in the anchor reruns. In particular, the detector-local peak proxies, the detector-frame box \(\Kdet\), and the fitted family remain fixed.

The first displayed quantity is the anchor-window spread of the primary detector-frame remnant estimate. The second is the median direct-wave gain \(\operatorname{med}\Gamma_{\mathrm{dir}}(\tau_0)\). For compactness the figure below plots only the direct-wave medians; the quadratic and combined medians are defined by the same rule \eqref{eq:appK-gamma-med} and enter the event-level scan through the robust gains \(\Gamma_\nu^{\mathcal A}\).

The resulting summary is encoded in Figure~\ref{fig:appK-anchor-rerun} and in the numerical envelopes
\begin{equation}\label{eq:appK-anchor-spreads}
\begin{aligned}
 t_0=3M_\star &: \, M_f^{\mathrm{det}}\in[66.8,69.5],\quad \chi_f\in[0.6400,0.7100],\quad \operatorname{med}\Gamma_{\mathrm{dir}}=0.316,\\
 t_0=6M_\star &: \, M_f^{\mathrm{det}}\in[66.5,69.5],\quad \chi_f\in[0.6400,0.6975],\quad \operatorname{med}\Gamma_{\mathrm{dir}}=0.129,\\
 t_0=9M_\star &: \, M_f^{\mathrm{det}}\in[68.1,69.5],\quad \chi_f\in[0.6400,0.6675],\quad \operatorname{med}\Gamma_{\mathrm{dir}}=0.110,\\
 t_0=11M_\star &: \, M_f^{\mathrm{det}}\in[68.6,69.5],\quad \chi_f\in[0.6675,0.7100],\quad \operatorname{med}\Gamma_{\mathrm{dir}}=0.069.
\end{aligned}
\end{equation}
These are finite observed ranges, not asymptotic uncertainty bars.

\begin{figure}[ht]
\centering
\includegraphics[width=0.92\textwidth]{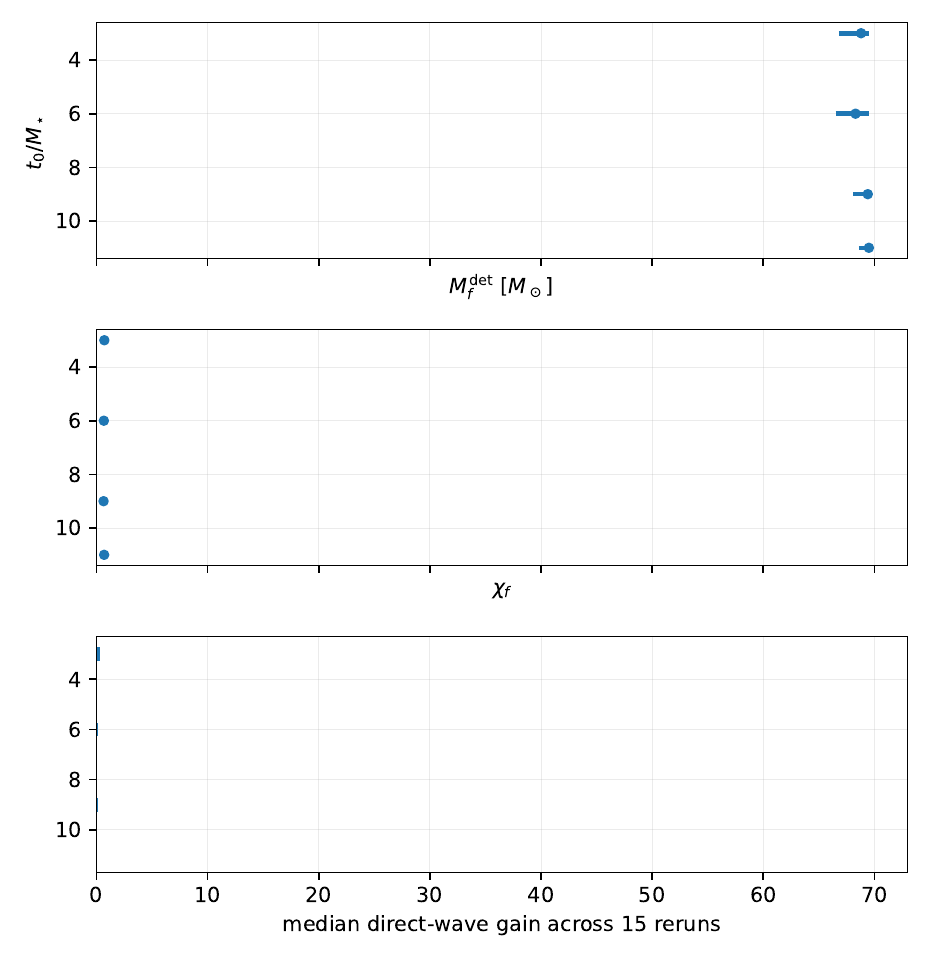}
\caption{Finite anchor-window robustness checks from the real public H1/L1 data. The upper two panels show the observed ranges of the primary detector-frame remnant estimate across the fifteen reruns at each public anchor window. The lower panel shows the corresponding median direct-wave gain. The dots mark the baseline specification at each anchor.}
\label{fig:appK-anchor-rerun}
\end{figure}

Three empirical facts are immediate. First, the earliest public anchor at \(3M_\star\) exhibits the widest direct-wave sensitivity of the four anchors. Second, the intermediate anchors at \(6M_\star\) and \(9M_\star\) have materially smaller direct-wave medians while still keeping the remnant estimate inside a compact subset of the event-local box. Third, the later \(11M_\star\) anchor has the smallest direct-wave median, but its remnant estimate remains concentrated against the upper edge of the detector-frame box and therefore should not be read as a cleaner version of the \(6M_\star\) window.

\subsection{Use of the reruns in the event map}

The event map of Section~\ref{sec:gw250114-empirical} does not turn Figure~\ref{fig:appK-anchor-rerun} into theorem-level constants. Its role is narrower. The finite reruns determine the scale on which preprocessing changes are treated as routine, while the fixed nuisance dictionary prevents those reruns from absorbing arbitrarily flexible early-time structure.

The dense-scan cutoffs of \eqref{eq:sec8-trusted-cutoffs}--\eqref{eq:sec8-transitional-cutoffs} are interpreted against two empirical constraints. A trusted window should look at least as stable as the \(6M_\star\) and \(9M_\star\) anchors, which lie near the center of the finite rerun suite. A rejected window should fail clearly enough that it resembles the early \(3M_\star\) anchor rather than a marginal perturbation of the accepted band. The real-data dense scan follows exactly that pattern. The \(3M_\star\) anchor has both large direct-wave gain and large quadratic gain. The \(6M_\star\) anchor minimizes the baseline residual and lies inside the compact part of the public detector-frame envelope. The \(9M_\star\) anchor remains in the same intermediate band, although it already lies closer to the edge of the public comparison. The \(11M_\star\) anchor is conservative but mixed: direct-wave sensitivity is small, yet the baseline residual and the quadratic gain are no longer minimal.

In this sense the rerun suite supports the numerical trust map. It fixes the finite real-data scale on which the accepted band is drawn.

\end{document}